\documentclass{article}

\usepackage{arxiv}

\usepackage{lineno,hyperref}
\modulolinenumbers[5]

\usepackage{makeidx}  
\usepackage[english]{babel}
\usepackage{latexsym}
\usepackage{amssymb}
\usepackage{amsfonts}
\usepackage{amsmath}
\usepackage{graphicx}
\usepackage{xcolor}
\usepackage{multicol}
\usepackage[latin1]{inputenc}
\usepackage{enumitem}
\usepackage{amsthm}
\usepackage{amsmath}
\usepackage{tikz}
\usepackage{mathdots}
\usepackage{diagbox}
\usepackage{adjustbox}
\usepackage{dashbox}

\usepackage{cancel}
\usepackage{color}

\usepackage{array}
\usepackage{multirow}
\usepackage{amssymb}
\usepackage{gensymb}
\usepackage{tabularx}
\usepackage{booktabs}
\usetikzlibrary{fadings}
\usetikzlibrary{patterns}
\usetikzlibrary{shadows.blur}
\usetikzlibrary{shapes}
\usetikzlibrary{arrows,automata,positioning}
\tikzset{>=latex}

\newcommand{\N}{\mathbb N}

\newcommand\G{\mathcal{G}}
\newcommand{\DLOG}{\textbf{DLOGSPACE}}
\newcommand{\PREDU}{\textsf{U-PRED}}
\newcommand{\PREDB}{\textsf{B-PRED}}
\newcommand{\PREDC}{\textsf{PRED-CHG}}
\newcommand{\PRED}{\textsf{PRED}}

\newcommand{\NC}{\textbf{NC}}
\DeclareMathOperator{\lcm}{lcm}

\newtheorem{theorem}{Theorem}
\newtheorem{corollary}{Corollary}

\newtheorem{lemma}[theorem]{Lemma}
\newtheorem{proposition}{Proposition}

\newtheorem*{problem}{Problem}
\newtheorem{definition}[theorem]{Definition}
\newtheorem{remark}{Remark}

\usepackage{tabularx,lipsum,environ,amsmath,amssymb}

\usepackage{tabularx,lipsum,environ,amsmath,amssymb}

\makeatletter

\newcommand{\problemtitle}[1]{\gdef\@problemtitle{#1}}
\newcommand{\probleminput}[1]{\gdef\@probleminput{#1}}
\newcommand{\problemquestion}[1]{\gdef\@problemquestion{#1}}
\newcounter{ProblemCounter}
\NewEnviron{problem1}{
	\problemtitle{}\probleminput{}\problemquestion{} 
	\BODY
	\par\addvspace{.5\baselineskip}
	\stepcounter{ProblemCounter}
	\centering
	\noindent
	\begin{tabularx}{0.8\textwidth}{@{\hspace{\parindent}} l X c}
		\hline
		\multicolumn{2}{@{\hspace{\parindent}}l}{{\bf Problem \arabic{ProblemCounter}: }\@problemtitle} \\
		\hline	
		\textbf{Input:} & \@probleminput \\
		\textbf{Question:} & \@problemquestion \\
		\hline	
	\end{tabularx}
	\par\addvspace{.5\baselineskip}
}

\title{On Symmetry versus Asynchronism: at the Edge of Universality in Automata Networks.  \thanks{This research was partially supported by French ANR project FANs ANR-18-CE40-0002 (G.T., M.R.W.) and ECOS project C19E02 (G.T., M.R.W.), ANID via PFCHA/DOCTORADO NACIONAL/2018 - 21180910 + PIA AFB 170001 (M.R.W).}}

\author{ Mart\'in R\'ios Wilson \\
Departamento de Ingenier\'ia Matem\'atica, Universidad de Chile, Santiago, Chile\\ 
Aix Marseille Univ, Universit\'e de Toulon, CNRS, LIS, Marseille, France.\\
	\texttt{mrios@dim.uchile.cl} \\
	\And
	Guillaume Theyssier \\
Aix-Marseille Universit\'e, CNRS, I2M (UMR 7373), Marseille, France.\\
	\texttt{guillaume.theyssier@cnrs.fr} \\
}

\date{}

\renewcommand{\shorttitle}{\textit{arXiv} Template}

\begin{document}

\maketitle
\begin{abstract}
  An automata network (AN) is a finite graph where each node holds a state from a finite alphabet and is equipped with a local map defining the evolution of the state of the node depending on its neighbors. The  global dynamics of the network is then induced by an {update scheme} describing which nodes are updated at each time step. We study how update schemes can compensate the limitations coming from symmetric local interactions. Our approach is based on intrinsic simulations and universality and we study both dynamical and computational complexity. By considering several families of {concrete symmetric AN} under several different update schemes, we explore the edge of universality in this two-dimensional landscape. On the way, we develop a proof technique based on an operation of glueing of networks, which allows to produce complex orbits in large networks from compatible pseudo-orbits in small networks.
\end{abstract} 
\newpage
\section{Introduction}

Automata networks introduced in the 40s \cite{McCulloch_1943} are both a family of dynamical systems frequently used in the modeling of biological networks \cite{Thomas_1973,KAUFFMAN_1969} and a computational model \cite{GR15,ChM14,goles1997reaction,WR79ii,WR79i}.
An automata network is a (finite) graph where each node holds a state from a finite set $Q$ and is equipped with a local transition map that determines how the state of the node evolves depending on the states of neighboring nodes.
Alternatively, they can be described (in the deterministic case) through a global map ${F:Q^V\rightarrow Q^V}$ that defines the collective evolution of all nodes of the network.
However, this global map hides two fundamental aspects at the heart of automata network literature \cite{DemongeotS20,Gadouleau_2019,chatain2018most,ARS17b,GN12,Aracena_2009,Robert_1969}: the interaction graph (knowing on which nodes effectively depends the behavior of a given node) and the update schedule (knowing in which order and with which degree of synchrony are local transition maps of each node applied).
The major influence of these two aspects on the dynamics of automata networks is clear, but the detailed understanding remains largely open.

The initial motivation of this paper lies in the striking interplay established in some cases between the symmetry of the local interactions and the synchrony of the update schedule: on one hand, a seminal result \cite{GolesO80,PaperGoles} shows that symmetric threshold networks under fully synchronous updates cannot have periodic orbits of period more than 2, and that they have polynomially bounded transients; this implies the existence of a polynomial time algorithm to predict the future of a node from any given initial configuration.
On the other hand, majority networks under partially asynchronous updates (precisely block-sequential update modes) were shown to have super-polynomial periodic orbits and a PSPACE-complete prediction problem \cite{GolesMST16,ChM14}. A similar result was obtained recently for conjunctive networks under a more general update mode called 'firing memory' \cite{goles2020firing}.
Is is well-known that threshold networks with synchronous updates but without the symmetry constraint are as capable as automata networks in general (essentially because they can embed any monotone Boolean circuit).
Thus, the results above can be interpreted as the power of some asynchronous updates to break symmetries in the local rules in such a way that dynamical and computational complexity is recovered. 

The purpose of this paper is to further study the capabilities of asynchronous update schemes to compensate for the limitation coming from constrained symmetric local interactions.
Our approach is based in considering a hierarchy of symmetric local interactions and a hierarchy of update mode which allow us to precisely analyze how the dynamical and computational complexity varies depending on both dimensions.
We are particularly interested in cases where the full richness of automata networks is achieved both dynamically and computationally, or more precisely, the cases where any automata network can be closely simulated, a property that we call universality and that we formalize below.
The main question we ask (and partially answer) is where lies the edge of universality in this bi-dimensional landscape (local symmetry vs. asynchronism).

\paragraph{Detailed framework} We aim at making our approach general and modular, while giving concrete and relevant examples that where considered in literature. Our framework can be decomposed as follows:
\begin{itemize}
\item the family of automata network we consider, called \emph{concrete symmetric automata network} (CSAN), are described by labeled undirected graphs whose labels (both on vertices and edges) describe the local interaction maps; most of the results obtained as an application of our framework are on signed symmetric conjunctive Boolean networks and three subfamilies defined by local sign constraints: all positive (conjunctive networks), locally positive (one neighbor at least has a positive interaction), and no constraint on signs. This latter family has been much studied already (see \textit{e.g.} \cite{ARS17b,De_Schutter_1999}).
\item we consider four types of update modes: parallel, block-sequential and general periodic ones (well-known in the literature) and a much less studied one that we call \emph{local clocks}, which was already considered in the setting of asynchronous cellular automata \cite{Cornforth_2005} and which is close to the recently introduced block-parallel mode \cite{DemongeotS20}; we view all these update modes as particular local mechanisms of activation of nodes based on hidden local finite memory acting like clocks; more precisely, we formalize everything inside deterministic automata network under parallel update mode whose projection on a sub-component of states gives exactly the desired asynchronous behavior; 
\item the key aspect in the above choices is that considering a CSAN family X under an update mode Y is actually formalized as a new CSAN family Z, because the way update modes are coded is directly translated into constraints in the local interaction maps; our study therefore amounts to analyze the complexity of particular CSAN families;
\item our complexity analysis focus on two aspects: dynamical complexity (transient, cycles, etc) and computational complexity (witnessed by various long term or short term prediction problems); we use the notion of intrinsic simulation in order to hit two targets with one bullet: with a proof that A simulates a previously analyzed B we show complexity lower bounds on A both in dynamical and in the computational sense; in particular we have a clear formal notion of universality for a family of automata networks that implies maximal complexity in both aspects; the interest of our approach is its modularity, and the fact that  it is not a priori limited to a small set of benchmark problems or properties: universality results can be used as a black box to then prove new corollaries on the complexity of other decision problems or other dynamical aspects; we stress that knowing that property/aspect X is complex for some family F of automata networks do not generally imply that property/aspect Y is also complex for family F, even in the case where Y is hard for automata networks in general (we actually give concrete examples of this below).
\end{itemize}

\paragraph{Our contributions}  This paper makes two kinds of contributions: it partially answers the main question addressed above, but it also establishes a new formalism and a general proof technique to obtain simulation and universality results suitable for automata network families with symmetric interactions. 
  \begin{itemize}
  \item our main contribution is a detailed study of the trade-off between all local constraints and update modes of our framework described above (Section~\ref{sec:concretesimulationresults}); what we obtain is a series of separation results and markers at the ``edge of universality'' showing that the interplay between local interactions and update modes is rich; the following table give a synthetic view of some of our results:
    \begin{center}
      \begin{tabular}{c|c|c|c|c}
        \diagbox{sign constraint}{update mode}&parallel&block sequential&local clocks &periodic\\
        \hline
        all positive&BPA & BPA & BPA & SPA\\
        locally positive& BPA & BPA & SU & SU\\
        free& BPA & SU & SU & SU \\
      \end{tabular}
    \end{center}
    where BPA means \emph{bounded period attractors} (there is a bound on the period of all attractors), SPA means \emph{super-polynomial attractors} (attractors can be constructed whose period is super-polynomial is the number of nodes of the network) and SU means \emph{strong universality}; strong universality (Definition~\ref{def:universal}) implies the existence of exponential attractors (Theorem~\ref{them:univ-rich-dynamics}) as well as maximal complexity for short term and long term prediction problems (Corollary~\ref{cor:universality}).
  \item our second contribution is a method and proof technique for building complex networks, which is key in obtaining the above results. Indeed, the classical approach to obtain complex networks (both computationally and dynamically) is to design individual building blocs or gadgets that have a specific input/output behavior (e.g. some Boolean operators) and connect them, from outputs to inputs, in such a way that the desired global behavior is achieved by composition (e.g. a Boolean circuit); when the local interactions are symmetric or constrained, it is generally impossible to proceed like this because connecting the output of some gadget to the input of another one is making a two-way link that can induce feedback compromising the behavior of the first gadget. Of course this problem can be dealt with in particular cases (it was done in \cite{goles2020firing,GolesMST16}), but our approach is generic: we replace the connection between gadgets by a glueing operation of two networks on a common part identified in both (Definition~\ref{def:glueing}), and the fundamental objects that compose through this glueing operation are pseudo-orbits (Lemma~\ref{lem:pseudo-orbit-glueing}) and not directly input/output relations. At the end, we obtain a proof method that allows to show strong universality (therefore both dynamical and computational complexity) by just exhibiting a finite set of gadgets and pseudo-orbits verifying suitable compatibility conditions (Definition~\ref{def:coherent-gadgets}).
  \item besides the two main contributions above, we also develop a complete formalization of (intrinsic) simulation between families of automata networks and (intrinsic) universality. This kind of approach is well-known for other models (see \textit{e.g.} \cite{bulk2}), and informally or indirectly present in various contributions on automata networks (see \textup{e.g.} \cite{goles1997reaction}). We however believe that our efforts of formalization clarify important aspects: the choice of a concrete representation when considering decision problems with automata networks as input (see Section~\ref{sec:representation}), the existence of several natural definitions of universality (Definition~\ref{def:universal}) with different implication on dynamical complexity (Theorem~\ref{them:univ-rich-dynamics}), the fact that two widely used prediction problems are actually incomparable (one can be hard while the other is easy in some family of automata networks, and reciprocally, see Theorem~\ref{theo:orthogonalpredictions}), and the fact that some concrete families can exhibit dynamical complexity while failing to be universal (theorems~\ref{theo:non-polyn-cycl} and \ref{theo:non-poly-transient} and \ref{theo:transient-nonuniversal}); to put it shortly, we show that a simulation/universality approach is better than taking individual decision problems or dynamical features as benchmarks. 
  \item finally, the local clocks update mode we consider was not studied theoretically before in automata networks (as far as we know), and  proves to be useful as an intermediate one lying between block sequential and general periodic modes. More generally, we believe that our unified approach which consists in viewing various update modes as local mechanisms of activation using some finite local information is natural and deserves further developments. For instance, it captures the 'firing memory' mode of \cite{goles2020firing} and more general modes can be proposed on this principle. 
  \end{itemize}

\paragraph{Organization of the paper} We start in Section~\ref{sec:automataandfamilies} by giving basic definitions about automata networks and families, including our hierarchy of concrete symmetric automata networks. We then introduce our hierarchy of update schemes in Section~\ref{sec:updateschemes}, including its formalization as asynchronous extensions. In Section~\ref{sec:simuniv}, we detail our notions of instrinsic simulations between individual automata networks, then between families. From there, we introduce the notion of intrinsic universality and we study their consequences both in terms of dynamics and computational complexity. In Section~\ref{sec:gadgetsandglueing}, we formalize our toolbox based on glueing, $\G$-networks and gadgets. We also study various families of $\G$-networks, including the canonical universal ones that will serve as a base to establish universality results later on. In Section~\ref{sec:concretesimulationresults}, we show our main results on local interaction rules versus update modes classification that leads to the table presented above.  Finally, we conclude by discussing some research perspectives in Section~\ref{sec:perspectives}.
  
\section{Automata networks and families}
\label{sec:automataandfamilies}

A \emph{graph} is a pair $G = (V,E)$  where $V$ and $E$ are finite sets satisfying $E \subseteq V \times V.$ We will call $V$ the set of \emph{nodes} and the set $E$ of \emph{edges}. We call $|V|$ the \emph{order} of $G$ and we usually identify this quantity by the letter $n$. Usually, as $E$ and $V$ are finite sets we will implicitly assume that there exists an ordering of the vertices in $V$ from $1$ to $n$ (or from $0$ to $n-1$). Sometimes we will denote the latter set as $[n].$ If $G = (V,E)$ and $V' \subseteq V, E' \subseteq E$ we say that $G'$ is a \emph{subgraph} of $G.$ We call a graph $P=(V,E)$ of the form $V= \{v_{1},\hdots, v_{n}\}$ $E = \{(v_{1}v_{2}),\hdots,(v_{n-1},v_{n})\}$  a \emph{path graph}, or simply a \emph{path}. We often refer to a path by simply denoting its sequence of vertices $\{v_{1},\hdots,v_{n}\}$. We denote the \emph{length} of a path by its number of edges. Whenever $P = (V= \{v_{1},\hdots, v_{n}\},E = \{(v_{1}v_{2}),\hdots,(v_{n-1},v_{n})\}$ is a path we call the graph in which we add the edge $\{v_{n},v_{1}\}$
a \emph{cycle graph} or simply a \emph{cycle} and we call it $C$ where $C = P + \{v_{n},v_{1}\}.$ Analogously, a cycle  is denoted usually by a sequence of nodes and its length is also given by the amount of edges (or vertices) in the cycle. Depending of the length of $C$ we call it a $k$-cycle when $k$ is its length. A non-empty graph is called \emph{connected} if any pair of two vertices $u,v$ are linked by some path. Given any non-empty graph, a maximal connected subgraph is called a connected component.

We call \emph{directed graph} a pair $G = (V,E)$ together with two functions $\text{init}:E \to V$ and $\text{ter}:E \to V$ where each edge $e \in E$ is said to be directed from $\text{init}(e)$ to $\text{ter}(e)$ and we write $e=(u,v)$ whenever $\text{init}(e) = u$ and $\text{ter}(e) = v.$ There is also a natural extension of the definition of paths, cycles and connectivity for directed graphs in the obvious way. We say a directed graph is strongly connected if there is a directed path between any two nodes. A strongly connected component of a directed graph $G = (V,E)$  is a maximal strongly connected subgraph.

Given a (non-directed) graph $G=(V,E)$ and two vertices $u,v$ we say that $u$ and $v$ are neighbors if $(u,v) \in E$.  Remark that abusing notations, an edge $(u,v)$ is also denoted by $uv$. Let $v \in V,$ we call $N_{G}(v) = \{u \in V: uv \in E\}$ (or simply $N(v)$ when the context is clear)  the set of neighbors (or \emph{neighborhood}) of $v$ and $\delta(G)_v = |N_{G}(v)|$ to the \emph{degree} of $v$. Observe that if $G'=(V',E')$ is a subgraph of $G$ and $v \in V'$, we can also denote by $N_{G'}(v)$  the set of its neighbors in $G'$ and the degree of $v$ in $G'$ as $\delta(G')_{v} =|N_{G'}(v)|.$  In addition, we define the \emph{closed neighborhood} of $v$ as the set $N[v] = N(v) \cup \{v\}$ and we use the following notation $\Delta(G) = \max \limits_{v \in V} \delta_v$  for the \emph{maximum degree} of $G$. Additionally, given $v \in V$, we will denote by $E_{v}$ to its set of \emph{incident edges}, i.e., $E_{v} = \{e \in E: e=uv\}.$ We will use the letter $n$ to denote the order of $G$, i.e. $n = |V|$.  Also, if $G$ is a graph whose sets of nodes and edges are not specified, we use the notation $V(G)$ and $E(G)$ for the set of vertices and the set of edges of $G$ respectively. In the case of a directed graph $G = (V,E)$ we define for a node $v \in V$ the set of its \emph{in-neighbors} by $N^{-}(v) = \{ u \in V: (u,v) \in E\}$ and its \emph{out-neighbors} as $N^{+}(v) = \{ u \in V: (v,u) \in E\}.$ We have also in this context the indegree of $v$ given by $ \delta^{-} = |N^{-}(v)|$ and its \emph{outdegree} given by $ \delta^{+} = |N^{+}(v)|$

During the most part of of the text, and unless explicitly stated otherwise,  every graph $G$ will be assumed to be connected and undirected.
We start by stating the following basic definitions, notations and properties that we will be using in the next sections. In general, $Q$ and $V$ will denote finite sets representing the alphabet and the set of  nodes respectively. We define $\Sigma(Q)$ as the set of all possible permutations over alphabet $Q$. We call an \textit{abstract automata network} any function $F:Q^V \to Q^V$. Note that $F$ induces a dynamics in $Q^V$ and thus we can see $(Q^V,F)$ as dynamical system.
In this regard, we recall some classical definitions.  We call a \emph{configuration} to any element $x \in Q^{V}.$ If $S \subseteq V$ we define the restriction of a configuration $x$ to $V$ as the function $x|_{S} \in Q^{S}$ such that $(x|_{S})_{v} = x_{v}$ for all $v \in S$. In particular, if $S = \{v\},$ we write $x_{v}.$
 
Given an initial configuration $x \in Q^V$,  we define the \textit{orbit} of $x$ as  the sequence $\mathcal{O}(x) = (F^t(x))_{t\geq 0}$. 
We define the set of \emph{limit configurations} or \emph{recurrent configurations} of $F$ as $L(F) = \bigcap_{t\geq 0}F^t(Q^V)$.  Observe that since $Q$ is finite and $F$ is deterministic, each orbit is eventually periodic, i.e. for each $x \in Q^{V}$ there exist some $\tau, p \in \N$ such that $F^{\tau+p}(x) = F^{\tau}(x)$ for all $x \in Q^{V}$. Note that if $x$ is a limit configuration then, its orbit is periodic. In addition, any configuration $x \in Q^{V}$ eventually reaches a limit configuration in finite time. We denote the set of orbits corresponding to periodic configurations as $\text{Att}(F)=\{\mathcal{O}(x): x \in L(F)\}$ and we call it the set of $\emph{attractors}$ of $F.$
We define the \emph{global period} or simply the \emph{period} of $\overline{x} \in \text{Att}(F)$  by $p(\overline{x}) = \min \{p \in \N : \overline{x}(p) = \overline{x}(0)\}$. If $p(\overline{x}) = 1$ we say that $\overline{x}$ is a \emph{fixed point} and otherwise, we say that $\overline{x}$ is a \emph{limit cycle}.

Given a node $v$, its behavior ${x\mapsto F(x)_v}$ might depend or not on another node $u$.
This dependencies can be captured by a graph structure which plays an important role in the theory of automata networks (see \cite{Gadouleau_2019} for a review of known results on this aspect).
This motivates the following definitions.
\begin{definition}
 Let $F: Q^{V} \to Q^{V}$ be an abstract automata network and $G = (V,E)$ a directed graph. We say $G$ is a \emph{communication graph} of $F$ if for all $v \in V$ there exist $D \subseteq N^{-}_v$ and some function $f_v:  Q^{D} \to Q$ such that $F(x)_v = f_v(x|_{D}).$ The \emph{interaction graph} of $F$ is its minimal communication graph.
\end{definition}
Note that by minimality, for any node $v$ and any in-neighbor $u$ of $v$ in the interaction graph of some $F$, then the next state at node $v$ effectively depends on the actual state at node $u$.
More precisely, there is some configuration $c\in Q^V$ and some $q\in Q$ with ${q\neq c_u}$ such that ${F(c)_v\neq F(c')_v}$ where $c'$ is the configuration $c$ where the state of node $u$ is changed to $q$.
This notion of effective dependency is sometimes taken as a definition of edges of the interaction graph.

From now on, for an abstract automata network $F$ and some communication graph $G$ of $F$ we use the notation $\mathcal{A} = (G,F)$. In addition, by abuse of notation. we also call $\mathcal{A}$ an  abstract automata network. We define a set of automata networks or a \emph{abstract family} of automata networks on some alphabet $Q$ as a set $\mathcal{F} \subseteq \bigcup \limits_{n \in \N} \{F:Q^{V} \to Q^{V}:  V\subseteq [n]\}.$
Note that the latter definition provides a general framework of study as it allows us to analyze an automata network as an abstract dynamical system.
However, as we are going to be working also with a computational complexity framework, it is necessary to be more precise in how we represent them.
In this regard, one possible slant is to start defining an automata network from a communication graph.
One of the main definition used all along this paper is that of \emph{concrete symmetric automata network}.
Roughly, they are non-directed labeled graph $G$ (both on nodes and edges) that represent an automata network.
They are \emph{concrete} because the labeled graph is a natural concrete representation upon which we can formalize decision problems and develop a computational complexity analysis.
They are \emph{symmetric} in two ways: first their communication graph is non-directed, meaning that an influence of node $u$ on node $v$ implies an influence of node $v$ on node $u$; second, the behavior of a given node is blind to the ordering of its neighbors in the communication graph, and it can only differentiate its dependence on neighbors when the labels of corresponding edges differ. 

\begin{definition}\label{def:csan}
  Given a non-directed graph $G = (V,E)$, a vertex label map $ \lambda: V \to ( Q\times 2^Q\to Q)$ and an edge label map $\rho: E \to \Sigma(Q)$, we define the tuple $\mathcal{A} = (G,\lambda,\rho)$ and we call it a  \textit{concrete symmetric automata network} (CSAN) associated to the graph $G$.
  A \emph{family of concrete symmetric automata networks} (CSAN family) $\mathcal{F}$ is given by an alphabet $Q$, a set of local labeling constraints ${\mathcal C\subseteq \Lambda\times R}$ where ${\Lambda = \{\lambda : Q\times 2^Q\to Q\}}$ is the set of possible vertex labels and ${R = 2^{\Sigma(Q)}}$ is the set of possible neighboring edge labels.
  We say a CSAN ${(G,\lambda,\rho)}$ belongs to ${\mathcal{F}}$ if for any vertex $v$ of $G$ with incident edges $E_v$ it holds ${(\lambda(v),\rho(E_v))\in \mathcal C}$.
\end{definition}

Note that the labeling constraints defining a CSAN family are local.
In particular, the communication graph structure is a priori free.
This aspect will play an important role later when building arbitrarily complex objects by composition of simple building blocks inside a CSAN family.

Let us now define the abstract automata network associated to a CSAN, by describing the semantics of labels defined above.
Intuitively, labels on edges are state modifiers, and labels on nodes give a map that describes how the node changes depending on the set of sates appearing in the neighborhood, after application of state modifiers.
We use the following notation: given $k \geq 1$,  $\sigma = (\sigma_1, \hdots, \sigma_k) \in \Sigma(Q)^k$ and $x \in Q^k$ we note
$x_{\sigma} = \{\sigma_1(x_1), \hdots, \sigma_k(x_k)\}$.
In addition, given $x \in  Q^n$ we define the restriction of $x$  to some subset $U \subseteq V$ as the partial configuration $x|_U \in Q^{|U|}$ such that $(x|_U)_u = x_u$ for all $u \in U.$ 
\begin{definition}\label{def:globalcsan}
  Given a CSAN $(G,\lambda,\rho)$, its associated global map $F: Q^V \to Q^V$ is defined as follows.
  For all node $v \in V$ and  for all $x\in Q^n$:
  \begin{equation*}
    F(x)_v = \lambda_v (x_v,(x|_{N(v)})_{\rho^v}),
  \end{equation*}
  where $N(i) = \{u_1,\hdots, u_{\delta_v}\}$ is the neighborhood of $v$ and $\rho^v = (\rho(v,u_1), \hdots, \rho(v,u_{\delta_u})).$
\end{definition}

Note that if $(G,\lambda,\rho)$ is a concrete automata network and $F$ its global rule then, $F$ is an abstract automata network with interaction graph included in $G$.

\paragraph{Our core set of CSAN families.}
Now we present some examples of families of automata networks that we study in this paper.
They differ in the set of allowed labels and their degree of local symmetry.

\begin{definition}\label{def:conjfamily}
  Let $Q =\{0,1\}$.
  The family of \emph{signed conjunctive automata networks (SCN)} is the set of CSAN ${(G,\lambda,\rho)}$ where for each node $v$ we have $\lambda_v(q,X) = \min X = \bigwedge X$ and, for each edge $e$, $\rho_e$ is either the identity map or the map  ${x\mapsto 1-x}$.
  
  The family of \emph{locally positive conjunctive networks (LPCN)} is the set of signed conjunctive automata networks ${(G,\lambda,\rho)}$ where we require, in addition, that for each node $v$ there is at least on edge $e$ incident to $v$ such that $\rho_e$ is the identity.
  
  Finally, the family of \emph{globally positive conjunctive networks (GPCN)} is the set of signed conjunctive automata networks where, for all edge $e$, $\rho_e$ is the identity.
\end{definition}

\begin{definition}\label{def:minmaxfamily}
	Let $Q$ be a totally ordered set. The family of \emph{min-max automata networks} over $Q$ is the set of CSAN ${(G,\lambda,\rho)}$ such that for each edge $e$, $\rho_e$ is the identity map and, for each node $v$, $\lambda_v(q,X)= \max X$ or $\lambda_v(q,X) = \min X$.
\end{definition}

\begin{remark}
  If  $Q = \{0,1\}$) then, $\mathcal{F}$ is the class of AND-OR networks, i.e., $F(x)_i = \bigwedge \limits_{j \in N(i)} x_j $ or $F(x)_i = \bigvee \limits_{j \in N(i)} x_j $ 
\end{remark}

      

\section{Update schemes}
\label{sec:updateschemes}

Through its global rule $F$, an automata network  $\mathcal{A}$ defines a dynamics over $Q^V$ by the subsequent iterations of $F$ .  This is the most natural way to define a dynamics from an automata network and its usually said that in this case the dynamics follows a $\textit{parallel}$ update scheme. The name comes from the fact that, at each time step, each node in the network updates its state according to its local function \emph{at the same time}. Nevertheless, starting from the network structure of $\mathcal{A}$,  one can also induce a dynamical system over $Q^V$ by considering other ways of updating in which not all the nodes are updated at the same time. Observe that, generally speaking, this latter notion demands some sort of temporal information in the nodes that determines which nodes have to be updated at a particular time step. We use the following general definition of \emph{update scheme}.

\begin{definition}
	Consider an abstract automata network  $F : Q^V \to Q^V$.  An update scheme is a sequence $\mu: \N \to 2^V$.   Given an update scheme $\mu$ and a natural number $k \in \N$ we call an intermediate step of the dynamics given by $\mu$ to the function $F^{\mu(k)}$ defined by $F^{\mu(k)}(x)_i  = \begin{cases}
	F(x)_i & \text{If } i \in \mu(k), \\
	x_i & \text{ otherwise.}
	\end{cases}$
	\sloppy We define an orbit given by $\mu$ starting from some $x \in Q^{V}$ as the sequence $\mathcal{O}_{\mu,F}(x) = (x, F^{\mu(0)}(x), F^{\mu(1)}(F^{\mu(0)}(x)),  F^{\mu(2)}(F^{\mu(1)}(F^{\mu(0)}(x))), \hdots)$
\end{definition}

\subsection{Periodic update schemes}
One of the most studied types of update schemes are the \textit{periodic update schemes}, \textit{i.e.} modes where map $\mu$ is periodic.
This class contains well-studied particular cases, for instance: parallel update scheme, in which all the nodes of the networks are updated at the same time (see Figure \ref{fig:syncsequpdate})
and also the block sequential update schemes in which each node is updated once every $p$ steps (but not necessarily all at the same time).
In addition, we explore a new class of update schemes which contains all the rest that it is called \textit{local clocks}. In this class, each node $v$ is updated once every $p_v$ steps but the frequency of update $p_v$ might depend on the node. This scheme can be seen intuitively as follows which justify the name: each node possesses an internal clock that ticks periodically and triggers an update of the node.
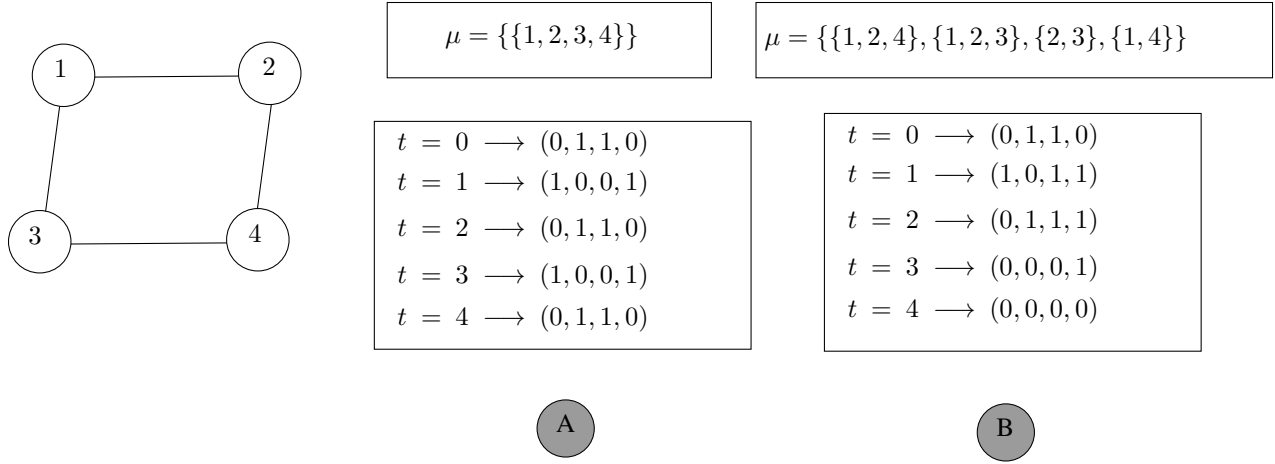
\begin{figure}
\centering
\begin{tikzpicture}[x=0.75pt,y=0.75pt,yscale=-1,xscale=1]

\draw    (42.75,85.75) -- (149.5,85) ;
\draw    (42.75,85.75) -- (31.5,170) ;
\draw    (31.5,170) -- (138.25,169.25) ;
\draw    (149.5,85) -- (138.25,169.25) ;
\draw  [fill={rgb, 255:red, 255; green, 255; blue, 255 }  ,fill opacity=1 ] (27,84.75) .. controls (27,76.05) and (34.05,69) .. (42.75,69) .. controls (51.45,69) and (58.5,76.05) .. (58.5,84.75) .. controls (58.5,93.45) and (51.45,100.5) .. (42.75,100.5) .. controls (34.05,100.5) and (27,93.45) .. (27,84.75) -- cycle ;
\draw  [fill={rgb, 255:red, 255; green, 255; blue, 255 }  ,fill opacity=1 ] (131,83.75) .. controls (131,75.05) and (138.05,68) .. (146.75,68) .. controls (155.45,68) and (162.5,75.05) .. (162.5,83.75) .. controls (162.5,92.45) and (155.45,99.5) .. (146.75,99.5) .. controls (138.05,99.5) and (131,92.45) .. (131,83.75) -- cycle ;
\draw  [fill={rgb, 255:red, 255; green, 255; blue, 255 }  ,fill opacity=1 ] (125,167.75) .. controls (125,159.05) and (132.05,152) .. (140.75,152) .. controls (149.45,152) and (156.5,159.05) .. (156.5,167.75) .. controls (156.5,176.45) and (149.45,183.5) .. (140.75,183.5) .. controls (132.05,183.5) and (125,176.45) .. (125,167.75) -- cycle ;
\draw  [fill={rgb, 255:red, 255; green, 255; blue, 255 }  ,fill opacity=1 ] (15,168.75) .. controls (15,160.05) and (22.05,153) .. (30.75,153) .. controls (39.45,153) and (46.5,160.05) .. (46.5,168.75) .. controls (46.5,177.45) and (39.45,184.5) .. (30.75,184.5) .. controls (22.05,184.5) and (15,177.45) .. (15,168.75) -- cycle ;

\draw   (206,48) -- (369.5,48) -- (369.5,86) -- (206,86) -- cycle ;
\draw   (199.5,108) -- (389.5,108) -- (389.5,223) -- (199.5,223) -- cycle ;

\draw   (392,48) -- (652.5,48) -- (652.5,86) -- (392,86) -- cycle ;
\draw   (426.5,104) -- (616.5,104) -- (616.5,224) -- (426.5,224) -- cycle ;

\draw (210,111.4) node [anchor=north west][inner sep=0.75pt]    {$t\ =\ 0\ \longrightarrow \ ( 0,1,1,0) \ $};
\draw (210,131.4) node [anchor=north west][inner sep=0.75pt]    {$t\ =\ 1\ \longrightarrow \ ( 1,0,0,1) \ $};
\draw (210,154.4) node [anchor=north west][inner sep=0.75pt]    {$t\ =\ 2\ \longrightarrow \ ( 0,1,1,0) \ $};
\draw (210,178.4) node [anchor=north west][inner sep=0.75pt]    {$t\ =\ 3\ \longrightarrow \ ( 1,0,0,1)$};
\draw (210,199.4) node [anchor=north west][inner sep=0.75pt]    {$t\ =\ 4\ \longrightarrow \ ( 0,1,1,0) \ $};
\draw (234,57.4) node [anchor=north west][inner sep=0.75pt]    {$\mu =\{\{1,2,3,4\}\}$};
\draw (36.75,75.9) node [anchor=north west][inner sep=0.75pt]    {$1$};
\draw (141.75,74.9) node [anchor=north west][inner sep=0.75pt]    {$2$};
\draw (134.75,158.9) node [anchor=north west][inner sep=0.75pt]    {$4$};
\draw (23.75,159.9) node [anchor=north west][inner sep=0.75pt]    {$3$};
\draw  [fill={rgb, 255:red, 155; green, 155; blue, 155 }  ,fill opacity=1 ]  (296, 263) circle [x radius= 14.42, y radius= 14.42]   ;
\draw (290,254.4) node [anchor=north west][inner sep=0.75pt]    {$\text{A}$};
\draw  [fill={rgb, 255:red, 155; green, 155; blue, 155 }  ,fill opacity=1 ]  (518, 265) circle [x radius= 14.42, y radius= 14.42]   ;
\draw (512,256.4) node [anchor=north west][inner sep=0.75pt]    {$\text{B}$};
\draw (395,58.4) node [anchor=north west][inner sep=0.75pt]    {$\mu =\{\{1,2,4\} ,\{1,2,3\} ,\{2,3\} ,\{1,4\}\}$};
\draw (437,195.4) node [anchor=north west][inner sep=0.75pt]    {$t\ =\ 4\ \longrightarrow \ ( 0,0,0,0) \ $};
\draw (437,174.4) node [anchor=north west][inner sep=0.75pt]    {$t\ =\ 3\ \longrightarrow \ ( 0,0,0,1)$};
\draw (437,150.4) node [anchor=north west][inner sep=0.75pt]    {$t\ =\ 2\ \longrightarrow \ ( 0,1,1,1) \ $};
\draw (437,127.4) node [anchor=north west][inner sep=0.75pt]    {$t\ =\ 1\ \longrightarrow \ ( 1,0,1,1) \ $};
\draw (437,107.4) node [anchor=north west][inner sep=0.75pt]    {$t\ =\ 0\ \longrightarrow \ ( 0,1,1,0) \ $};

\end{tikzpicture}

\caption{Synchronous update scheme and general periodic update scheme for the same conjunctive automata network. Local function is given by the minimum (AND function) over the set of states of neighbors for each node. A) Synchronous or parallel update scheme. In this case $\mu$ has period $1$ and all nodes are updates simultaneously. Observe that dynamics exhibits an attractor of period $2$ B) General periodic update scheme over a conjunctive network. In this case $\mu$ has period $4$ and dynamics reach a fixed point after $4$ time steps. Observe that there is no restriction on how many times a node is updated. For example, $1$ is updated $3$ times every $4$ time steps but $4$ is updated only $2$ times every $4$ time steps.}

\label{fig:syncsequpdate}

\end{figure}
\begin{definition}
  We say that an update scheme $\mu$ is a \emph{periodic update scheme} if there exists $p \in \N$ such that $\mu(n+p) =\mu(n)$ for all $n \in \N$. Moreover, we say that $\mu$ is
  \begin{itemize}
  \item a \emph{block sequential scheme} if there are subsets (called blocks) ${B_0,\ldots,B_{p-1}\subseteq V}$ forming a partition of $V$ such that ${\mu(n) = B_{n\bmod p}}$,
  \item a \emph{local clocks scheme} if for each ${v\in V}$ there is a local period $\tau_v\in\N$ and a shift ${0\leq \delta_v < \tau_v}$ such that ${v\in\mu(n)\iff \delta_v = n\bmod \tau_v}$.
  \end{itemize}
\end{definition}
For a concrete example on how these update schemes work, see Figure \ref{fig:blocklocalupdate} in which different dynamics for a simple conjunctive network under block sequential and local clocks update schemes are shown.
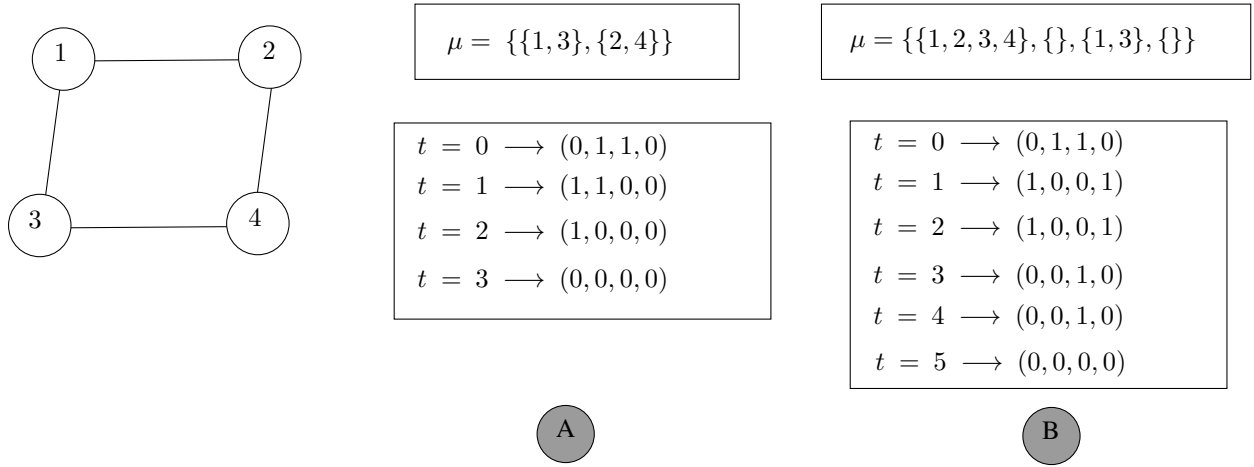
\begin{figure}[h]
\centering
\begin{tikzpicture}[x=0.75pt,y=0.75pt,yscale=-1,xscale=1]

\draw    (38.75,66.75) -- (145.5,66) ;
\draw    (38.75,66.75) -- (27.5,151) ;
\draw    (27.5,151) -- (134.25,150.25) ;
\draw    (145.5,66) -- (134.25,150.25) ;
\draw  [fill={rgb, 255:red, 255; green, 255; blue, 255 }  ,fill opacity=1 ] (23,65.75) .. controls (23,57.05) and (30.05,50) .. (38.75,50) .. controls (47.45,50) and (54.5,57.05) .. (54.5,65.75) .. controls (54.5,74.45) and (47.45,81.5) .. (38.75,81.5) .. controls (30.05,81.5) and (23,74.45) .. (23,65.75) -- cycle ;
\draw  [fill={rgb, 255:red, 255; green, 255; blue, 255 }  ,fill opacity=1 ] (127,64.75) .. controls (127,56.05) and (134.05,49) .. (142.75,49) .. controls (151.45,49) and (158.5,56.05) .. (158.5,64.75) .. controls (158.5,73.45) and (151.45,80.5) .. (142.75,80.5) .. controls (134.05,80.5) and (127,73.45) .. (127,64.75) -- cycle ;
\draw  [fill={rgb, 255:red, 255; green, 255; blue, 255 }  ,fill opacity=1 ] (121,148.75) .. controls (121,140.05) and (128.05,133) .. (136.75,133) .. controls (145.45,133) and (152.5,140.05) .. (152.5,148.75) .. controls (152.5,157.45) and (145.45,164.5) .. (136.75,164.5) .. controls (128.05,164.5) and (121,157.45) .. (121,148.75) -- cycle ;
\draw  [fill={rgb, 255:red, 255; green, 255; blue, 255 }  ,fill opacity=1 ] (11,149.75) .. controls (11,141.05) and (18.05,134) .. (26.75,134) .. controls (35.45,134) and (42.5,141.05) .. (42.5,149.75) .. controls (42.5,158.45) and (35.45,165.5) .. (26.75,165.5) .. controls (18.05,165.5) and (11,158.45) .. (11,149.75) -- cycle ;
\draw   (216,38) -- (379.5,38) -- (379.5,76) -- (216,76) -- cycle ;
\draw   (205.5,98) -- (395.5,98) -- (395.5,197) -- (205.5,197) -- cycle ;
\draw   (421,38) -- (637.5,38) -- (637.5,76) -- (421,76) -- cycle ;
\draw   (435.5,97) -- (625.5,97) -- (625.5,232) -- (435.5,232) -- cycle ;

\draw (32.75,56.9) node [anchor=north west][inner sep=0.75pt]    {$1$};
\draw (137.75,55.9) node [anchor=north west][inner sep=0.75pt]    {$2$};
\draw (130.75,139.9) node [anchor=north west][inner sep=0.75pt]    {$4$};
\draw (19.75,140.9) node [anchor=north west][inner sep=0.75pt]    {$3$};
\draw (216,102.4) node [anchor=north west][inner sep=0.75pt]    {$t\ =\ 0\ \longrightarrow \ ( 0,1,1,0) \ $};
\draw (216,122.4) node [anchor=north west][inner sep=0.75pt]    {$t\ =\ 1\ \longrightarrow \ ( 1,1,0,0) \ $};
\draw (216,145.4) node [anchor=north west][inner sep=0.75pt]    {$t\ =\ 2\ \longrightarrow \ ( 1,0,0,0) \ $};
\draw (216,169.4) node [anchor=north west][inner sep=0.75pt]    {$t\ =\ 3\ \longrightarrow \ ( 0,0,0,0)$};
\draw (231,49.4) node [anchor=north west][inner sep=0.75pt]    {$\mu =\ \{\{1,3\} ,\{2,4\}\}$};
\draw (446,100.4) node [anchor=north west][inner sep=0.75pt]    {$t\ =\ 0\ \longrightarrow \ ( 0,1,1,0) \ $};
\draw (446,120.4) node [anchor=north west][inner sep=0.75pt]    {$t\ =\ 1\ \longrightarrow \ ( 1,0,0,1) \ $};
\draw (446,143.4) node [anchor=north west][inner sep=0.75pt]    {$t\ =\ 2\ \longrightarrow \ ( 1,0,0,1) \ $};
\draw (446,167.4) node [anchor=north west][inner sep=0.75pt]    {$t\ =\ 3\ \longrightarrow \ ( 0,0,1,0)$};
\draw (446,188.4) node [anchor=north west][inner sep=0.75pt]    {$t\ =\ 4\ \longrightarrow \ ( 0,0,1,0) \ $};
\draw (424+10,48.4) node [anchor=north west][inner sep=0.75pt]    {$\mu =\{\{1,2,3,4\} ,\{\} ,\{1,3\} ,\{\}\}$};
\draw (447,211.4) node [anchor=north west][inner sep=0.75pt]    {$t\ =\ 5\ \longrightarrow \ ( 0,0,0,0) \ $};
\draw  [fill={rgb, 255:red, 155; green, 155; blue, 155 }  ,fill opacity=1 ]  (292, 255) circle [x radius= 14.42, y radius= 14.42]   ;
\draw (286,246.4) node [anchor=north west][inner sep=0.75pt]    {$\text{A}$};
\draw  [fill={rgb, 255:red, 155; green, 155; blue, 155 }  ,fill opacity=1 ]  (537, 256) circle [x radius= 14.42, y radius= 14.42]   ;
\draw (531,247.4) node [anchor=north west][inner sep=0.75pt]    {$\text{B}$};

\end{tikzpicture}

\caption{Block sequential and local clocks update schemes over a simple conjunctive network. Local functions are given by the minimum (AND) over the states of the neighbors of each node. A) Block sequential update scheme. In this case function $\mu$ is defined by two blocks: $\{1,3\}$ and $\{2,4\}$. Dynamics reach a fixed point after $3$ time steps. B) Local clocks update scheme. In this case each node has an internal clock with different period. Nodes $1$ and $3$ are updated every two steps ($\tau_{1} = \tau_{3} = 2$) and nodes $2$ and $4$ are updated every $4$ time steps (i.e. $\tau_{2}=\tau_{4}=4$). Shift parameter is $0$ for all nodes $\delta_{1}=\delta_{2}=\delta_{3}=\delta_{4}=0$. Dynamics reach a fixed point after $3$ time steps.}

\label{fig:blocklocalupdate}
\end{figure}

Block sequential and local clocks schemes are clearly periodic schemes. Moreover, any block sequential scheme given by ${B_0,\ldots,B_{p-1}\subseteq V}$  is a local clocks scheme given by $\tau_v = p$ and $\delta_v=i\iff v\in B_i$ for all $v\in V$. As already said, block sequential schemes can thus be seen as local clocks schemes where all nodes share the same update frequency. General periodic update schemes allows different time intervals between two consecutive updates of a node, which local clocks schemes obviously can't do (see Figure \ref{fig:syncsequpdate}, B).  We will see later the tremendous consequences that such subtle differences in time intervals between updates at each node can have. For now let us just make the formal observation that the inclusions between these families of update schedules are strict when focusing on the sets of maps $\mu$.

\begin{remark} \label{rem:blockpar}
  A so-called \emph{block-parallel} scheme has also been considered more recently \cite{DemongeotS20} which is defined by a set of list of nodes ${L_i = (v_{i,j})_{0\leq j<p_i}}$ (for ${1\leq i\leq k}$) forming a partition (\textit{i.e.} such that $v_{i,j}$ are all distinct and ${\cup_{i,j}v_{i,j}=V}$) to which is associated the map ${\mu}$ such that ${v_{i,j}\in\mu(n)\iff j=n\bmod p_i}$. It is a particular case of our definition of local clocks scheme above with the additional constraints that the size of the set ${\mu(n)}$ of updated nodes is constant with $n$ (see Figure \ref{fig:blockparupdate} for an example). We note that, conversely, any local clocks scheme on a given networked can be simulated by a block-parallel scheme by artificially adding disconnected nodes that do nothing but satisfy the constraint of ${\mu(n)}$ being of constant size. 
\end{remark}

\begin{figure}[h]
\centering	

\begin{tikzpicture}[x=0.75pt,y=0.75pt,yscale=-1,xscale=1]

\draw    (88.75,74.75) -- (195.5,74) ;
\draw    (88.75,74.75) -- (77.5,159) ;
\draw    (77.5,159) -- (184.25,158.25) ;
\draw    (195.5,74) -- (184.25,158.25) ;
\draw  [fill={rgb, 255:red, 255; green, 255; blue, 255 }  ,fill opacity=1 ] (73,73.75) .. controls (73,65.05) and (80.05,58) .. (88.75,58) .. controls (97.45,58) and (104.5,65.05) .. (104.5,73.75) .. controls (104.5,82.45) and (97.45,89.5) .. (88.75,89.5) .. controls (80.05,89.5) and (73,82.45) .. (73,73.75) -- cycle ;
\draw  [fill={rgb, 255:red, 255; green, 255; blue, 255 }  ,fill opacity=1 ] (177,72.75) .. controls (177,64.05) and (184.05,57) .. (192.75,57) .. controls (201.45,57) and (208.5,64.05) .. (208.5,72.75) .. controls (208.5,81.45) and (201.45,88.5) .. (192.75,88.5) .. controls (184.05,88.5) and (177,81.45) .. (177,72.75) -- cycle ;
\draw  [fill={rgb, 255:red, 255; green, 255; blue, 255 }  ,fill opacity=1 ] (171,156.75) .. controls (171,148.05) and (178.05,141) .. (186.75,141) .. controls (195.45,141) and (202.5,148.05) .. (202.5,156.75) .. controls (202.5,165.45) and (195.45,172.5) .. (186.75,172.5) .. controls (178.05,172.5) and (171,165.45) .. (171,156.75) -- cycle ;
\draw  [fill={rgb, 255:red, 255; green, 255; blue, 255 }  ,fill opacity=1 ] (61,157.75) .. controls (61,149.05) and (68.05,142) .. (76.75,142) .. controls (85.45,142) and (92.5,149.05) .. (92.5,157.75) .. controls (92.5,166.45) and (85.45,173.5) .. (76.75,173.5) .. controls (68.05,173.5) and (61,166.45) .. (61,157.75) -- cycle ;
\draw   (273,43) -- (533.5,43) -- (533.5,81) -- (273,81) -- cycle ;
\draw   (307.5,99) -- (497.5,99) -- (497.5,198) -- (307.5,198) -- cycle ;

\draw (82.75,64.9) node [anchor=north west][inner sep=0.75pt]    {$1$};
\draw (187.75,63.9) node [anchor=north west][inner sep=0.75pt]    {$2$};
\draw (180.75,147.9) node [anchor=north west][inner sep=0.75pt]    {$4$};
\draw (69.75,148.9) node [anchor=north west][inner sep=0.75pt]    {$3$};
\draw (318+10,102.4+3) node [anchor=north west][inner sep=0.75pt]    {$t\ =\ 0\ \longrightarrow \ ( 0,1,1,0) \ $};
\draw (318+10,122.4+3) node [anchor=north west][inner sep=0.75pt]    {$t\ =\ 1\ \longrightarrow \ ( 1,0,1,1) \ $};
\draw (318+10,145.4+3) node [anchor=north west][inner sep=0.75pt]    {$t\ =\ 2\ \longrightarrow \ ( 0,0,1,0) \ $};
\draw (318+10,169.4+3) node [anchor=north west][inner sep=0.75pt]    {$t\ =\ 3\ \longrightarrow \ ( 0,0,0,0)$};
\draw (319+10,53.4) node [anchor=north west][inner sep=0.75pt]    {$\mu =\{\{1,2,4\} ,\{1,3,4\}\}$};

\end{tikzpicture}
\caption{Block parallel update scheme defined over a conjunctive network. Updating list is given by $L = \{(1),(2,3),(4)\}.$ Observe that a constant amount of nodes (equal to the length of $L$, i.e., $3$) is updated at each time step. Dynamics reach a fixed point after $3$ time-steps.}
\label{fig:blockparupdate}
\end{figure}
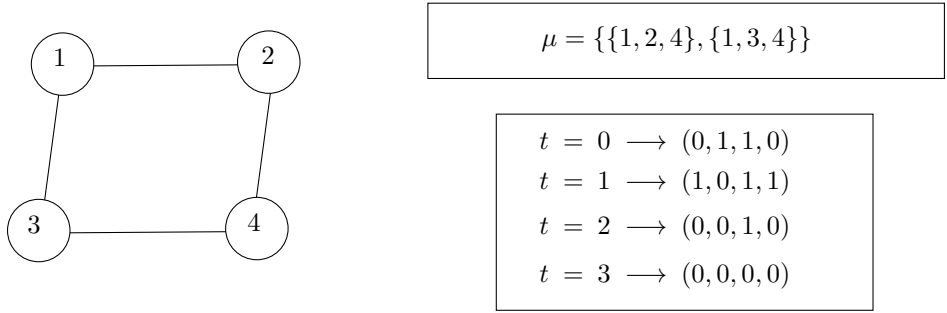

\subsection{Projections and asynchronous extensions}

Now we present a dynamical formalism that allow us to include all the update schemes presented above and possibly other ones into one formalism.
We remark that in all periodic schemes, a given node can take the decision to update or not by simply keeping track of the current value of time modulo the period.
The key observation is that, when we add the knowledge of time modulo the period at each node as a new component of states, the whole system becomes deterministic.
In fact, we are recovering the original dynamics of some automata network with alphabet $Q$ under a periodic update scheme by projecting a specific deterministic automata network with alphabet ${Q'\times Q}$ onto $Q$.  

Let $F:Q^V \to Q^V$ be an abstract automata network.
We define its \textit{asynchronous} version as an automata network that in every time-step (non-deterministically) choose if a node $i$ should be updated or if it will be stay in the same state.
More precisely a \textit{asynchronous} version of $F$ is a non-deterministic function $F^*: Q^V \to (\mathcal{P}(Q))^V$ such that $F^*(x)_i = \{x_i,F(x)_i\}$. Note that, analogously to the deterministic case one can define an orbit starting from $x$ of $F^*$ as a sequence of states $\mathcal{O}_{F^*}(x)= x^0=x,x^1,x^2,\hdots,x^t,\hdots, \in Q^n$ such that $x^s_{i}\in F(x^{s-1})_i$ for $i\in V$ and $s\geq1$. Note also that, given $x \in Q^V$ and an orbit  $\mathcal{O}_{F^*}(x)$ we can see $\mathcal{O}_{F^*}(x)$ as a particular realization of certain update scheme $\mu$. More precisely, there exist an update scheme $\mu$ (which is defined in the obvious way i.e. by updating the corresponding nodes in every time step according to points in $\mathcal{O}_{F^*}(x)$)  such that for every $x^s \in \mathcal{O}_{F^*}(x)$ we have $x^s = (\mathcal{O}_{F,\mu}(x))^{s}.$ In addition, we have that for each update scheme $\mu$ there exist an orbit of $F^*$ which coincides with its dynamics in every time step.
Thus, we could work with $F^*$ in order to globally study all possible update schemes.
However, we are interested in specific update schemes and we would like to continue working in a deterministic framework in order to keep things simple (in particular the notion of simulation that we define later).

In order to achieve this task, we introduce the following notion of asynchronous extension which is a way to produce the dynamics of different update schemes through projection.

\begin{definition}
  \label{def:asyncextension}
  Let $Q$ be a finite alphabet and $Q' = Q \times R$ where $R$ is finite. Let $F:Q^V \to Q^V$ and $F':Q'^V \to Q'^V$ two abstract automata networks. We say that $F$ is a projection system of $F'$ if for $x \in Q'^V$ and $F'$ is an asynchronous extension of $F$
\begin{equation*}
\overline\pi(F'(x)) \in F^{*}(\overline\pi(x))
\end{equation*}
where $\overline\pi$ is the node-wise extension of the projection $\pi: Q' \to Q$ such that $\pi(q,r) = q$ for all $q' = (q,r) \in Q'.$
\end{definition}


We show hereunder that the dynamics associated to any of the previously presented periodic update schemes can be described as an asynchronous extension over a product alphabet. To simplify notations we sometimes identify ${(A\times B\times\cdots)^V}$ with ${A^V\times B^V\times\cdots}$.

\begin{definition}[block sequential extension]
	Let $F:Q^V\to Q^V$ be an abstract automata network. Let $b \leq n$ and let $Q' = Q \times \{0,\hdots,b-1\}$.  We define the \textit{block sequential extension} of $F$ with $b$ blocs as the automata network ${F':(Q')^V\to (Q')^V}$ such that  for all $x = (x_Q,x_b) \in Q'^V$ and all ${v\in V}$:
	\begin{equation*}
	F'(x)_v = \begin{cases}
(F(x_Q)_v, (x_b)_v - 1 \mod b) &\text{ if } (x_b)_v = 0, \\
((x_Q)_v, (x_b)_v - 1 \mod b) & \text{ else.}
\end{cases}
	\end{equation*}
\end{definition}

\begin{definition}[local clocks extension]
		Let $F:Q^V\to Q^V$ be an abstract automata network. Let $c \in \mathbb{N}$ and let $Q' = Q \times \{0,\hdots, c-1\}\times \{1,\hdots, c\}$.  We define the \textit{local clocks} extension of $F$ with clock length $c$ as the automata network ${F':(Q')^V\to (Q')^V}$ such that for all $x = (x_Q,x_c,x_m) \in Q'^V$ and all ${v\in V}$:
	\begin{equation*}
	F'(x)_v = \begin{cases}
	(F(x_Q)_v, (\psi_{(x_m)_v}[(x_c)_v]) + 1 \mod (x_m)_v,(x_m)_v) & \text{ if }(x_c)_v = 0, \\
	((x_Q)_v, (\psi_{(x_m)_v}[(x_c)_v])  + 1 \mod (x_m)_v,(x_m)_v) & \text{ else.} 
	\end{cases}
	\end{equation*}
	 $\psi_{m}(r): \{0,\hdots, c-1\} \to  \{0,\hdots, c-1\}$ is such that $\psi_m(r) = \begin{cases}
	r & \text{ if } r\leq m-1, \\
	m -1 & \text{ else.}
	\end{cases}$
\end{definition}

\begin{definition}[periodic extension]
  Let $F:Q^V\to Q^V$ be an abstract automata network. Let $p \in \mathbb{N}$ and let $Q' = Q \times \{0,\hdots, p-1\}\times 2^{\{0,\hdots,p-1\}}$.  We define the \textit{periodic extension} of $F$ with period  length $p$ as the automata network ${F':(Q')^V\to (Q')^V}$ such that for all $x = (x_Q,x_p,x_s) \in Q'^V$ and all ${v\in V}$:
	\begin{equation*}
	F'(x)_v = \begin{cases}
	(F(x_Q)_v,(x_p)_v + 1 \mod p,(x_s)_v) &\text{ if } (x_p)_v \in (x_s)_v, \\
	((x_Q)_v, (x_p)_v + 1 \mod p,(x_s)_v) & \text{ else.}
	\end{cases}
	\end{equation*}
\end{definition}


\begin{remark}
Observe that, given an abstract automata network $F:Q^{V} \to Q^{V}$ and an asynchronous extension $F': (Q \times R )^{V} \mapsto (Q \times R)^{V}$ of the type previously defined i.e. a block sequential extension, a local clocks extension or a periodic extension, both $F'$ and $F^{\mu}$ (where $\mu$ is some of the latter update schemes) describe the same dynamics. In fact, let us illustrate this fact by analyzing the case of the block sequential extension (the other cases are analogous). Let $b \geq 1$ and $\mu_{b} = (I_{0},\hdots, I_{b-1}).$ Note that $(I_{0},\hdots, I_{b-1})$ is an ordered partition of $V$. On one hand, we consider an arbitrary initial condition $x \in Q^{V}$ and the correspondent block sequential orbit $\mathcal{O}_{\mu,F}(x).$ On the other hand, we consider a block extenstion $F^{'}: (Q \times \{0,\hdots,b-1\})^{V} \mapsto (Q \times \{0,\hdots,b-1\})^{V}$ and an initial condition $z \in (Q \times \{0,\hdots,b-1\})^{V}$ given by $z_{v} = (x_{v},y_{v})$ where $y_{v} = k$ if and only if $v \in I_{k}$ for $1\leq k \leq b.$ By the definition of $F^{'}$, we have that $\overline{\pi}(F'(z)) = (\mathcal{O}_{\mu,F}(x))^{1}$ since the only nodes $v$ in which $F$ is applied (in the first coordinate) are the ones such that $v \in I_{0}$.  In addition, we have that for each $v$ in $V$, $\overline{\pi_{2}}(F^{'}(x))_{v} = y_{v} - 1 \mod b,$ where $\pi_{2}$ is the node-wise extension of the proyection $\pi_{2}: Q \times \{0,\hdots,b-1\} \mapsto \{0,\hdots,b-1\}.$  Thus, we have $\overline{\pi}(F'(F'(x))) =  (\mathcal{O}_{\mu,F}(x))^{2}.$ Iteratively, we deduce $\overline{\pi}(F'^{t}(x)) = (\mathcal{O}_{\mu,F}(x))^{t}$ for each $t \geq 1.$

Conversely, let us choose an arbitrary initial condition $x =(x_{Q},x_{b}) \in (Q \times \{0,\hdots,b-1\})^{V}.$ We define the ordered partition $I_{0} \hdots, I_{b-1}$ given by $v \in I_{k}$ if and only if $(x_{b})_{v} = k$ for $0 \leq k \leq b-1.$ Then, the second coordinate of the initial condition $x_{b}$ induces a block sequential update scheme $\mu_{x_{b}}$ which is defined by the latter ordered partition.

\end{remark}
The previous definitions are formalized for every abstract automata network. We now focus on CSAN families where the extensions are also CSAN as show in the following lemma.

\begin{lemma}
  Let $F$ be a CSAN, then any block sequential extension (resp. local clocks extension, resp. periodic extension) of $F$ is a CSAN. Moreover, for any CSAN family ${\mathcal{F}}$ and any fixed $b$, the set of block sequential extensions with $b$ blocs of networks of $\mathcal{F}$ is again a CSAN family. The same holds for local clocks and periodic extensions.
\end{lemma}
\begin{proof}
  In each case, the definition of the extension $F'$ with alphabet $Q'=Q\times R$ is such that the action of $F'$ on the $R$ component is purely local (the new value of the $R$ component of a node evolves as a function of the old value of this $R$ component) and the value of the $R$ component at a node determines alone if the $Q$ component should be updated according to $F$ or left unchanged. Therefore clearly $F'$ is a CSAN if $F$ is.

  In the context of a CSAN family $\mathcal{F}$, the CSAN definition of $F'$ involves only local constraints coming from $F\in\mathcal{F}$ and the action on the $R$-component is the same at each node. So the second part of the lemma is clear. 
\end{proof}

\begin{remark}
  This approach by asynchronous extensions can also capture non-periodic update schemes. For instance \cite{glrs20} studies an update scheme for Boolean networks called firing memory which uses local delays at each node and, in addition, makes the delay mechanism depend on the state of the current configuration at the node. Firing memory schemes can be captured as an asynchronous extension in such a way that the above lemma for the CSAN case still works. 
\end{remark}

\newcommand{\updateschemefamily}[3]{{#1}^{\textsc{#2},#3}}
\newcommand{\blocfamily}[2]{\updateschemefamily{#1}{block}{#2}}
\newcommand{\clockfamily}[2]{\updateschemefamily{#1}{clock}{#2}}
\newcommand{\periodfamily}[2]{\updateschemefamily{#1}{per}{#2}}

To sum up, our formalism allows to treat variations in the update scheme as a change in the CSAN family considered. Given a CSAN family $\mathcal{F}$ and integers $b,c,p$, we introduce the following notations:
\begin{itemize}
\item ${\blocfamily{\mathcal{F}}{b}}$ is the CSAN family of all block sequential extensions of networks from $\mathcal{F}$ with $b$ blocks,
\item ${\clockfamily{\mathcal{F}}{c}}$ is the CSAN family of all local clocks sequential extensions of networks from $\mathcal{F}$ with clock length $c$,
\item ${\periodfamily{\mathcal{F}}{p}}$ is the CSAN family of all periodic extensions of networks from $\mathcal{F}$ with period $p$.
\end{itemize}

\section{Simulation and universality}
\label{sec:simuniv}

In this section we introduce a key tool used in this paper: simulations.
The goal is to easily prove computational or dynamical complexity of some family of automata networks by showing it can simulate some well-known reference family where the complexity analysis is already established.
It can be thought as a complexity or dynamical reduction.
Simulations of various kinds are often implicitly used in proofs of dynamical or computational hardness.
We are going instead to explicitly define a notion of simulation and establish hardness results as corollaries of simulation results later in the paper.
To be more precise, we will first define a notion of simulation between individual automata networks, and then extend it to a notion of simulation between families.
This latter notion, which is the one we are really interested in requires more care if we want to use it as a notion of reduction for computational complexity.
We introduce all the useful concepts progressively in the next subsections.

\subsection{Simulation between individual automata networks}

At the core of our formalism is the following definition of simulation where an automata network $F$ is simulated by an automata network $G$ with a constant time slowdown and using blocs of nodes in $G$ to represent nodes in $F$.
Our definition is rather strict and requires in particular an injective encoding of configurations of $F$ into configurations of $G$.
We are not aware of a published work with this exact same formal definition, but close variants certainly exist and it is a direct adaptation to finite automata networks of a classical definition of simulation for cellular automata \cite{bulk2}.

\begin{definition}\label{def:bloc-simu}
  Let ${F:Q_F^{V_F}\rightarrow Q_F^{V_F}}$ and ${G:Q_G^{V_G}\rightarrow Q_G^{V_G}}$ be abstract automata networks.
  A \emph{block embedding} of $Q_F^{V_F}$ into $Q_G^{V_G}$ is a collection of blocs ${D_i\subseteq V_G}$ for each ${i\in V_F}$ which forms a partition of ${V_G}$ together with a collection of patterns ${p_{i,q}\in Q_G^{D_i}}$ for each ${i\in V_F}$ and each ${q\in Q_F}$ such that ${p_{i,q}=p_{i,q'}}$ implies ${q=q'}$.
  This defines an injective map ${\phi:Q_F^{V_F}\rightarrow Q_G^{V_G}}$ by ${\phi(x)_{D_i} = p_{i,x_i}}$ for each $i\in V_F$.
  We say that $G$ simulates $F$ via block embedding $\phi$ if there is a time constant $T$ such that the following holds on ${Q_F^{V_F}}$: 
  \[\phi\circ F = G^T\circ\phi.\]
\end{definition}

See Figure  \ref{fig:bsimulation} for a scheme of block simulation. In the following, when useful we represent a block embedding as the list of blocs together with the list of patterns. The size of this representation is linear in the number of nodes (for fixed  alphabet). 


\begin{figure}

\centering
\begin{tikzpicture}[x=0.55pt,y=0.55pt,yscale=-1,xscale=1]

\draw  [fill={rgb, 255:red, 255; green, 255; blue, 255 }  ,fill opacity=1 ] (552,188) -- (637.5,188) -- (637.5,242) -- (552,242) -- cycle ;
\draw  [fill={rgb, 255:red, 255; green, 255; blue, 255 }  ,fill opacity=1 ] (524.25,25) -- (617.5,25) -- (617.5,80.75) -- (524.25,80.75) -- cycle ;
\draw  [fill={rgb, 255:red, 255; green, 255; blue, 255 }  ,fill opacity=1 ] (288,73) -- (373.5,73) -- (373.5,122) -- (288,122) -- cycle ;
\draw  [fill={rgb, 255:red, 255; green, 255; blue, 255 }  ,fill opacity=1 ] (362.5,163) -- (454.5,163) -- (454.5,212) -- (362.5,212) -- cycle ;
\draw  [dash pattern={on 4.5pt off 4.5pt}] (43.5,40) .. controls (63.5,30) and (167,56) .. (147,76) .. controls (127,96) and (198.5,164) .. (218.5,194) .. controls (238.5,224) and (35.5,199) .. (15.5,169) .. controls (-4.5,139) and (23.5,50) .. (43.5,40) -- cycle ;
\draw  [fill={rgb, 255:red, 0; green, 0; blue, 0 }  ,fill opacity=1 ] (18,102.75) .. controls (18,98.47) and (21.47,95) .. (25.75,95) .. controls (30.03,95) and (33.5,98.47) .. (33.5,102.75) .. controls (33.5,107.03) and (30.03,110.5) .. (25.75,110.5) .. controls (21.47,110.5) and (18,107.03) .. (18,102.75) -- cycle ;
\draw  [fill={rgb, 255:red, 0; green, 0; blue, 0 }  ,fill opacity=1 ] (64,146.75) .. controls (64,142.47) and (67.47,139) .. (71.75,139) .. controls (76.03,139) and (79.5,142.47) .. (79.5,146.75) .. controls (79.5,151.03) and (76.03,154.5) .. (71.75,154.5) .. controls (67.47,154.5) and (64,151.03) .. (64,146.75) -- cycle ;
\draw  [fill={rgb, 255:red, 0; green, 0; blue, 0 }  ,fill opacity=1 ] (119,97.75) .. controls (119,93.47) and (122.47,90) .. (126.75,90) .. controls (131.03,90) and (134.5,93.47) .. (134.5,97.75) .. controls (134.5,102.03) and (131.03,105.5) .. (126.75,105.5) .. controls (122.47,105.5) and (119,102.03) .. (119,97.75) -- cycle ;
\draw  [fill={rgb, 255:red, 0; green, 0; blue, 0 }  ,fill opacity=1 ] (162,170.75) .. controls (162,166.47) and (165.47,163) .. (169.75,163) .. controls (174.03,163) and (177.5,166.47) .. (177.5,170.75) .. controls (177.5,175.03) and (174.03,178.5) .. (169.75,178.5) .. controls (165.47,178.5) and (162,175.03) .. (162,170.75) -- cycle ;
\draw    (25.75,102.75) -- (71.75,146.75) ;
\draw    (71.75,146.75) -- (126.75,97.75) ;
\draw    (126.75,97.75) -- (169.75,170.75) ;
\draw    (71.75,146.75) -- (169.75,170.75) ;
\draw    (25.75,102.75) -- (126.75,97.75) ;

\draw [color={rgb, 255:red, 74; green, 144; blue, 226 }  ,draw opacity=1 ][line width=2.25]    (541.75,71.75) -- (571.5,200) ;
\draw [color={rgb, 255:red, 74; green, 144; blue, 226 }  ,draw opacity=1 ][line width=2.25]    (441.75,184.75) -- (571.5,200) ;
\draw [color={rgb, 255:red, 74; green, 144; blue, 226 }  ,draw opacity=1 ][line width=2.25]    (441.75,184.75) -- (541.75,71.75) ;
\draw [color={rgb, 255:red, 74; green, 144; blue, 226 }  ,draw opacity=1 ][line width=2.25]    (354.5,95) -- (541.75,71.75) ;
\draw [color={rgb, 255:red, 74; green, 144; blue, 226 }  ,draw opacity=1 ][line width=2.25]    (354.5,95) -- (441.75,184.75) ;
\draw    (324.25,113.25) -- (352.25,94.25) ;
\draw    (303.25,92.25) -- (352.25,94.25) ;
\draw    (303.25,92.25) -- (324.25,113.25) ;
\draw  [fill={rgb, 255:red, 155; green, 155; blue, 155 }  ,fill opacity=1 ] (297,92.25) .. controls (297,88.8) and (299.8,86) .. (303.25,86) .. controls (306.7,86) and (309.5,88.8) .. (309.5,92.25) .. controls (309.5,95.7) and (306.7,98.5) .. (303.25,98.5) .. controls (299.8,98.5) and (297,95.7) .. (297,92.25) -- cycle ;
\draw    (411.25,201.25) -- (441.75,184.75) ;
\draw    (381.25,173.25) -- (439.25,182.25) ;
\draw  [fill={rgb, 255:red, 155; green, 155; blue, 155 }  ,fill opacity=1 ] (435.5,184.75) .. controls (435.5,181.3) and (438.3,178.5) .. (441.75,178.5) .. controls (445.2,178.5) and (448,181.3) .. (448,184.75) .. controls (448,188.2) and (445.2,191) .. (441.75,191) .. controls (438.3,191) and (435.5,188.2) .. (435.5,184.75) -- cycle ;
\draw    (381.25,173.25) -- (376.5,190) ;
\draw    (568.25,199.25) -- (612.82,201.25) ;
\draw    (571.5,200) -- (585,225.75) ;
\draw  [fill={rgb, 255:red, 155; green, 155; blue, 155 }  ,fill opacity=1 ] (565.25,200) .. controls (565.25,196.55) and (568.05,193.75) .. (571.5,193.75) .. controls (574.95,193.75) and (577.75,196.55) .. (577.75,200) .. controls (577.75,203.45) and (574.95,206.25) .. (571.5,206.25) .. controls (568.05,206.25) and (565.25,203.45) .. (565.25,200) -- cycle ;
\draw    (573.25,56.25) -- (601.25,37.25) ;
\draw    (541.75,71.75) -- (571.5,57) ;
\draw  [fill={rgb, 255:red, 155; green, 155; blue, 155 }  ,fill opacity=1 ] (535.5,71.75) .. controls (535.5,68.3) and (538.3,65.5) .. (541.75,65.5) .. controls (545.2,65.5) and (548,68.3) .. (548,71.75) .. controls (548,75.2) and (545.2,78) .. (541.75,78) .. controls (538.3,78) and (535.5,75.2) .. (535.5,71.75) -- cycle ;
\draw    (376.5,196.25) -- (411.25,201.25) ;
\draw  [fill={rgb, 255:red, 155; green, 155; blue, 155 }  ,fill opacity=1 ] (370.25,196.25) .. controls (370.25,192.8) and (373.05,190) .. (376.5,190) .. controls (379.95,190) and (382.75,192.8) .. (382.75,196.25) .. controls (382.75,199.7) and (379.95,202.5) .. (376.5,202.5) .. controls (373.05,202.5) and (370.25,199.7) .. (370.25,196.25) -- cycle ;
\draw  [fill={rgb, 255:red, 155; green, 155; blue, 155 }  ,fill opacity=1 ] (405,201.25) .. controls (405,197.8) and (407.8,195) .. (411.25,195) .. controls (414.7,195) and (417.5,197.8) .. (417.5,201.25) .. controls (417.5,204.7) and (414.7,207.5) .. (411.25,207.5) .. controls (407.8,207.5) and (405,204.7) .. (405,201.25) -- cycle ;
\draw    (585.25,230.25) -- (613.25,230.25) ;
\draw  [fill={rgb, 255:red, 155; green, 155; blue, 155 }  ,fill opacity=1 ] (579,230.25) .. controls (579,226.8) and (581.8,224) .. (585.25,224) .. controls (588.7,224) and (591.5,226.8) .. (591.5,230.25) .. controls (591.5,233.7) and (588.7,236.5) .. (585.25,236.5) .. controls (581.8,236.5) and (579,233.7) .. (579,230.25) -- cycle ;
\draw  [fill={rgb, 255:red, 155; green, 155; blue, 155 }  ,fill opacity=1 ] (607,230.25) .. controls (607,226.8) and (609.8,224) .. (613.25,224) .. controls (616.7,224) and (619.5,226.8) .. (619.5,230.25) .. controls (619.5,233.7) and (616.7,236.5) .. (613.25,236.5) .. controls (609.8,236.5) and (607,233.7) .. (607,230.25) -- cycle ;
\draw    (573.25,56.25) -- (612.82,201.25) ;
\draw  [fill={rgb, 255:red, 155; green, 155; blue, 155 }  ,fill opacity=1 ] (567,56.25) .. controls (567,52.8) and (569.8,50) .. (573.25,50) .. controls (576.7,50) and (579.5,52.8) .. (579.5,56.25) .. controls (579.5,59.7) and (576.7,62.5) .. (573.25,62.5) .. controls (569.8,62.5) and (567,59.7) .. (567,56.25) -- cycle ;
\draw  [fill={rgb, 255:red, 155; green, 155; blue, 155 }  ,fill opacity=1 ] (607.13,201.25) .. controls (607.13,197.8) and (609.68,195) .. (612.82,195) .. controls (615.96,195) and (618.5,197.8) .. (618.5,201.25) .. controls (618.5,204.7) and (615.96,207.5) .. (612.82,207.5) .. controls (609.68,207.5) and (607.13,204.7) .. (607.13,201.25) -- cycle ;
\draw  [dash pattern={on 4.5pt off 4.5pt}] (217,32) .. controls (237,22) and (676.5,-5) .. (656.5,15) .. controls (636.5,35) and (639.5,207) .. (659.5,237) .. controls (679.5,267) and (419.5,291) .. (399.5,261) .. controls (379.5,231) and (197,42) .. (217,32) -- cycle ;
\draw  [fill={rgb, 255:red, 255; green, 255; blue, 255 }  ,fill opacity=1 ] (223,149.5) -- (256.9,149.5) -- (256.9,143) -- (279.5,156) -- (256.9,169) -- (256.9,162.5) -- (223,162.5) -- cycle ;
\draw  [fill={rgb, 255:red, 155; green, 155; blue, 155 }  ,fill opacity=1 ] (348.25,95) .. controls (348.25,91.55) and (351.05,88.75) .. (354.5,88.75) .. controls (357.95,88.75) and (360.75,91.55) .. (360.75,95) .. controls (360.75,98.45) and (357.95,101.25) .. (354.5,101.25) .. controls (351.05,101.25) and (348.25,98.45) .. (348.25,95) -- cycle ;
\draw  [fill={rgb, 255:red, 155; green, 155; blue, 155 }  ,fill opacity=1 ] (595,37.25) .. controls (595,33.8) and (597.8,31) .. (601.25,31) .. controls (604.7,31) and (607.5,33.8) .. (607.5,37.25) .. controls (607.5,40.7) and (604.7,43.5) .. (601.25,43.5) .. controls (597.8,43.5) and (595,40.7) .. (595,37.25) -- cycle ;
\draw    (324.25,113.25) -- (381.25,173.25) ;
\draw  [fill={rgb, 255:red, 155; green, 155; blue, 155 }  ,fill opacity=1 ] (318,113.25) .. controls (318,109.8) and (320.8,107) .. (324.25,107) .. controls (327.7,107) and (330.5,109.8) .. (330.5,113.25) .. controls (330.5,116.7) and (327.7,119.5) .. (324.25,119.5) .. controls (320.8,119.5) and (318,116.7) .. (318,113.25) -- cycle ;
\draw  [fill={rgb, 255:red, 155; green, 155; blue, 155 }  ,fill opacity=1 ] (375,173.25) .. controls (375,169.8) and (377.8,167) .. (381.25,167) .. controls (384.7,167) and (387.5,169.8) .. (387.5,173.25) .. controls (387.5,176.7) and (384.7,179.5) .. (381.25,179.5) .. controls (377.8,179.5) and (375,176.7) .. (375,173.25) -- cycle ;

\draw (539,29.4) node [anchor=north west][inner sep=0.75pt]    {$B_{i}$};
\draw (113,61.4+7) node [anchor=north west][inner sep=0.75pt]    {$v_{i}$};
\draw (94,171.4) node [anchor=north west][inner sep=0.75pt]    {$F$};
\draw (475,243.4) node [anchor=north west][inner sep=0.75pt]    {$G$};

\end{tikzpicture}

\caption{Scheme of one-to-one block simulation. In this case, network $F$ is simulated by $G$. Each node in $F$ is assigned to a block in $G$ and state coding is injective. Observe that blocks are connected (one edge in the original graph may be represented by a path in the communication graph of $G$) according to connections between nodes in the original network $F$. This connections are represented by blue lines}
\label{fig:bsimulation}
\end{figure}
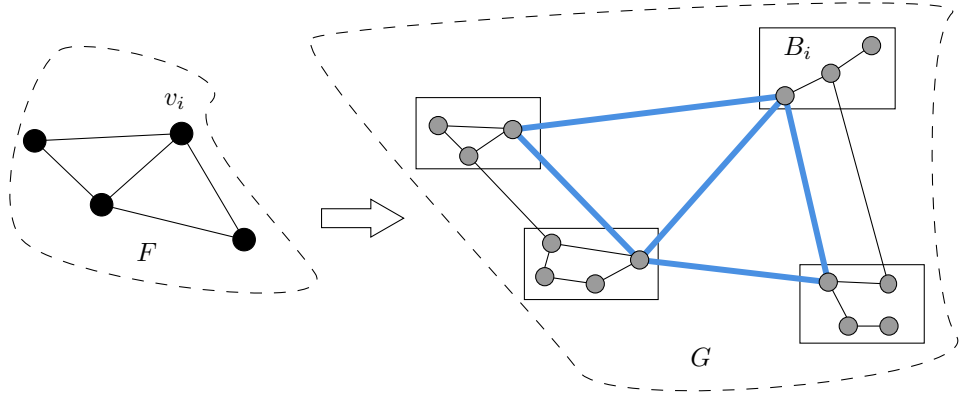

\begin{remark}\label{rem:bloc-simul}
  It is convenient in many concrete cases to define a block embedding through blocs $D_i$ that are disjoint but do not cover ${V_G}$ and add a context block $C$ disjoint from the $D_i$ that completes the covering of ${V_G}$. In this variant a block embedding of $Q_F^{V_F}$ into ${Q_G^{V_G}}$ is given by patterns $p_{i,q}$ and a constant context pattern $p_C\in Q_ G^C$ which define an injective map ${\phi:Q_F^{V_F}\rightarrow Q_G^{V_G}}$ by ${\phi(x)_{D_i} = p_{i,x_i}}$ for each $i\in V_F$ and ${\phi(x)_C = p_C}$. This variant is actually just a particular instance of Definition~\ref{def:bloc-simu} because we can include $C$ in an arbitrary block (${D_i\leftarrow D_i\cup C}$) and define the block embedding as in Definition~\ref{def:bloc-simu}.

  In our proofs of simulations in section~\ref{sec:concretesimulationresults}, the variant blocs/context will be particularly relevant because the size of blocs will be bounded while the context will grow with the size of the sonsidered automata. Said differently, the information about an encoded state will be very localized.

  Another natural particular case of Definition~\ref{def:bloc-simu} corresponding to localized information is when in each block $D_i$, there is a special node $v_i\in D_i$ such that the map ${q\mapsto p_{i,q}(v_i)}$ is injective. It is only possible when $Q_G$ is larger than $Q_F$, but it will be the case in several examples of Boolean automata networks below. Interestingly, this local coding phenomena is forced when some automate network $G$ simulates some Boolean automata network $G$: indeed, in any block $D_i$ of $G$ at least one node $v_i$ must change between patterns $p_{i,0}$ and $p_{i,1}$, but the map ${x\in\{0,1\}\mapsto p_{i,x}}$ being injective, it means that ${x\mapsto p_{i,x}(v_i)}$ is injective too.
\end{remark}

\begin{remark}
  The simulation relation of Definition~\ref{def:bloc-simu} is a pre-order on automata networks.
\end{remark}

The orbit graph $G_F$ associated to network $F$ with nodes $V$ and alphabet $Q$ is the digraph with vertices ${Q^V}$ and an edge from $x$ to $F(x)$ for each $x\in Q^V$.
We also denote $G_F^t = G_{F^t}$.

\begin{lemma}\label{lem:simu-dynamic}
  If $G$ simulates $F$ via block embedding with time constant $T$ then the orbit graph $G_F$ of $F$ is a subgraph of $G_G^T$. In particular if $F$ has an orbit with transient of length $t$ and period of length $p$, then $G$ has an orbit with transient of length $Tt$ and period $Tp$.
\end{lemma}
\begin{proof}
  The embedding of $G_F$ inside $G_G^T$ is realized by definition by the block embedding of the simulation. The consequence on the length of periods and transients comes from the fact that the embedding $\phi$ verifies: $x$ is in a periodic orbit if and only if $\phi(x)$ is in a periodic orbit.   
\end{proof}

\subsection{Representing Automata Networks}
\label{sec:representation}
As we are interested in measuring computational complexity of decision problems related to the dynamics of automata networks belonging to a particular family, we introduce hereunder a general notion for the representation of a family of automata networks.
We can always fix a canonical representation of automata networks as Boolean circuits.
However, as we show in the first part of this section, families can have different natural representation which are closely related to their particular properties.
Considering this fact, we introduce the notion of \textit{standard representation} in order to denote some representation from which we can efficiently obtain a circuit family computing original automata network family.
Finally, we resume previous discussion on different representations for some particular families, showing how difficult it is to transform one particular representation into another one.


\subsubsection{Standard representations}

\sloppy We fix for any alphabet $Q$ an injective map ${m_Q : Q\to \{0,1\}^{k_Q}}$ which we extend cell-wise for each $n$ to ${m_Q:Q^n\to\{0,1\}^{k_Qn}}$. Given an abstract automata network ${F:Q^n\to Q^n}$, a circuit encoding of $F$ is a Boolean circuit ${C:\{0,1\}^{k_Qn}\to \{0,1\}^{k_Qn}}$ such that ${m_Q\circ F = C\circ m_Q}$ on ${Q^n}$. We also fix a canonical way to represent circuits as words of ${\{0,1\}^*}$ (for instance given by a number of vertices, the list of gate type positioned at each vertex and the adjacency matrix of the graph of the circuit).



Let $\mathcal{F}$ be a set of abstract automata network over alphabet $Q$. A standard representation $\mathcal{F}^*$ for $\mathcal{F}$ is a language $L_{\mathcal{F}} \subseteq \{0,1\}^*$ together with a \DLOG{} algorithm such that:
\begin{itemize}
\item the algorithm transforms any $w \in L_{\mathcal{F}}$ into the canonical
  representation of a circuit encoding ${C(w)}$ that code an abstract
  automata network ${F_w\in\mathcal{F}}$;
\item for any ${F\in\mathcal{F}}$ there is ${w\in L_{\mathcal{F}}}$ with ${F=F_w}$.
\end{itemize}

The default general representations we will use are circuit representations, \textit{i.e.} representations where $w\in L_{\mathcal{F}}$ is just a canonical representation of a circuit. In this case the \DLOG{} algorithm is trivial (the identity map). However, we sometimes want to work with more concrete and natural representations for some families of networks: in such a case, the above definition allows any kind of coding as soon as it is easy to deduce the canonical circuit representation from it.

\subsubsection{Example of standard representations of some particular families}

Observe that communication graph is often an essential piece of information for describing an automata network, however, this information is usually not enough.
Thus, in this section three examples of canonical types of families are discussed: bounded degree networks, CSAN and algebraic families. The first one is characterized by a communication graph which has bounded degree, i.e. there exist a constant $\Delta \in \N $ (not depending in the size of the network) such that maximum degree is at most $\Delta$. The second one was  widely presented in previous sections and its roughly the family of all set-valued functions described by labelled communication graphs. Finally, an algebraic family is essentially a family of linear maps in the case in which $Q$ is seen as a finite field and $Q^{n}$ as a linear space. We show in latter cases that there exists a natural representation that is based on its properties. Roughly, it is observed that bounded degree condition provides a way to store information in an efficient way,  set-valued functions in the case of CSAN does not depend on graph structure and linear maps can be represented as matrices

\paragraph{The CSAN case}
A CSAN family is a collection of labeled graphs and thus is naturally represented as a graph $G$ together with some circuit family which represents local functions (i.e. $\lambda$ and $\rho$).
Each local function will depend only in the local configuration composed by a node and its neighbors.
In addition, for a fixed CSAN family, a collection of set-valued functions and edge-labels is finite and provided independently of the structure of each graph as they depend on the possible sets of elements in the alphabet $Q$ (note that what makes different two networks in one family is the position of labels but local functions are chosen from the a same fixed set).
Thus, we can represent a CSAN family $\mathcal{F}$ by a collection of labeled graphs $\mathcal{G}$ and a circuit family $\Lambda$.
We call this representation a succinct representation for CSAN $\mathcal{F}.$
As circuit family $\Lambda$ depends only in the size of the alphabet and we are considering that alphabet is fixed, we will usually omit $\Lambda$ and just write $\mathcal{G}$ as a succinct representation for $\mathcal{F}$ and we will denote a CSAN $(G,\lambda,\rho) \in \mathcal{F}$ by simply $G$ (where of course $G$ is a labeled graph in $\mathcal{G}$).
Note that as alphabet size is constant, a succinct representation is, in particular, a standard representation for CSAN family as an abstract automata network family.
Formally, the language $L_{\mathcal{G}}$ containing the encoding of each labeled graph in $\mathcal{G}$ together with the standard encoding of constant circuit family $\Lambda$ is a standard representation for CSAN family $\mathcal{F}.$

\paragraph{The case of bounded degree communication graphs}
Let us fix some positive constant $\Delta$. It is natural to consider the family of automata networks whose interaction graph has a maximum degree bounded by $\Delta$ (see Remark~\ref{rem:represent-g-networks} below). We associate to this family the following representation: an automata network $F$ is given as pair ${(G,(\tau_v)_{v\in V(g)})}$ where $G$ is a communication graph of $F$ of maximum degree at most $\Delta$ and ${(\tau_v)_{v\in V(G)}}$ is the list for all nodes of $G$ of its local transition map $F_v$ of the form ${Q^d\to Q}$ for ${d\leq\Delta}$ and represented as a plain transition table of size ${|Q|^d\log(|Q|)}$.

\begin{remark}\label{rem:csan-bounded-degree}
  Given any CSAN family, there is a \DLOG{} algorithm that transforms a bounded degree representation of an automata network of the family into a CSAN representation: since all local maps are bounded objects, it is just a matter of making a bounded computation for each node.
\end{remark}

\paragraph{The algebraic case}
When endowing the alphabet $Q$ with a finite field structure, the set of configurations $Q^n$ is a vector space and one can consider automata networks that are actually linear maps. In this case the natural representation is a ${n\times n}$ matrix. It is clearly a standard representation in the above sense since circuit encodings can be easily computed from the matrix. Moreover, as in the CSAN case, when a linear automata network is given as a bounded degree representation, it is easy to recover a matrix in \DLOG{}.

More generally, we can consider matrix representations without field structure on the alphabet. An interesting case is that of Boolean matrices: ${Q=\{0,1\}}$ is endowed with the standard Boolean algebra structure with operations ${\vee,\wedge}$ and matrix multiplication is defined by: 
\[\bigl(AB)_{i,j}=\bigvee_kA_{i,k}\wedge B_{k,j}.\]
They are a standard representation of disjunctive networks (and by switching the role of $0$ and $1$ conjunctive networks), \textit{i.e.} networks $F$ over alphabet ${\{0,1\}}$ whose local maps are of the form ${F_i(x) = \vee_{k\in N(i)}x_k}$ (respectively ${F_i(x) = \wedge_{k\in N(i)}x_k}$) . When their dependency graph is symmetric, disjunctive networks (resp. conjunctive networks) are a particular case of CSAN networks for which $\rho_v$ maps are the identity and $\lambda_v$ are just max (resp. min) maps. For disjunctive networks (resp. conjunctive networks) the CSAN representation and the matrix representation are \DLOG{} equivalent.

\subsubsection{Computing interaction graphs from representations}

One of the key differences between all the representations presented so far is in the information they give about the interaction graph of an automata network. For instance, it is straightforward to deduce the interaction graph of a linear network from its matrix representation in \DLOG{}. At the other extreme, one can see that it is NP-hard to decide whether a given edge belongs to the interaction graph of an automata network given by a circuit representation: indeed, one can build in \DLOG{} from any SAT formula $\phi$ with $n$ variables a circuit representation of an automata network ${F:\{0,1\}^{n+1}\to\{0,1\}^{n+1}}$ with 
\[F(x)_1 =
  \begin{cases}
    x_{n+1} &\text{ if }\phi(x_1,\ldots,x_n)\text{ is true,}\\
    0 &\text{ else.}
  \end{cases}
\]
This $F$ is such that node $1$ depends on node ${n+1}$ if and only if $\phi$ is satisfiable.

For automata networks with communication graphs of degree at most $\Delta$, there is a polynomial time algorithm to compute the interaction graph from a circuit representation: for each node $v$, try all the possible subsets $S$ of nodes of size at most $\Delta$ and find the largest one such that the following map 
\[x\in Q^S\mapsto F_v(\phi(x))\] effectively depends on each node of $S$, where ${\phi(x)_w}$ is $x_{w}$ is ${w\in S}$ and some arbitrary fixed state ${q\in Q}$ else. Note also, that we can compute a bounded degree representation in polynomial time with the same idea.

In the CSAN case, the situation is ambivalent. On one hand, the interaction graph can be computed in polynomial time from a CSAN representation because for any given node $v$ the following holds: either $\lambda_v$ is a constant map and then $v$ has no dependence, or ${\lambda_v(q,X)}$ depends only on $q$ and in this case node $v$ depends on itself but no other node, or $\lambda_v(q,X)$ depends only on $X$ and in this case node $v$ depends exactly on its neighbors in the CSAN graph, or, finally, ${\lambda_v(q,X)}$ depends both on $q$ and on $X$ and in this case node $v$ depends on itself and its neighbors in the CSAN graph.
On the other hand, a polynomial time algorithm to compute the interaction graph from a circuit representation would give a polynomial algorithm solving Unambiguous-SAT (which is very unlikely following Valiant-Vazirani theorem \cite{Valiant_1986}). Indeed, any ``dirac'' map ${\delta : \{0,1\}^n\to\{0,1\}}$ with ${\delta(x)= 1}$ if and only if ${x_1\cdots x_n = b_1\cdots b_n}$ can be seen as the local map of a CSAN network because it can be written as ${\lambda(\{\rho_i(x_i) : 1\leq i\leq n\})}$ where $\rho_i(x_i)=x_i$ if $b_i=1$ and $\neg x_i$ else, and $\lambda$ is the map 
\[S\subseteq Q\mapsto
  \begin{cases}
    1 &\text{ if }S=\{1\}\\
    0 &\text{ else.}
  \end{cases}
\]
A constant map can also be seen as the local map of some CSAN network.
Therefore, given a Boolean formula $\phi$ with the promise that is at has at most one satisfying assignment, one can easily compute the circuit representation of some CSAN network which has some edge in its interaction graph if and only if $\phi$ is satisfiable: indeed, the construction of $F$ above from a SAT formula always produce a CSAN given the promise on $\phi$.

It follows from the discussion above that a polynomial time algorithm to compute a CSAN representation of a CSAN represented by circuit would give a polynomial time algorithm to solve Unambiguous-SAT.

The following table synthesizes the computational hardness of representation conversions. It shall be read as follows: given a family ${(\mathcal{F},\mathcal{F}^*)}$ listed horizontally and a family ${(\mathcal{H},\mathcal{H}^*)}$ listed vertically, the corresponding entry in the table indicates the complexity of the problem of transforming ${w\in L_F}$ with the promise that ${F_w\in\mathcal{F}\cap\mathcal{H}}$ into $w'\in L_H$ such that ${F_w = H_{w'}}$.

\begin{center}
\begin{tabular}{c|c|c|c|c}
  \diagbox{output}{input}&circuit&CSAN&$\Delta$-bounded degree&matrix\\
  \hline
  circuit&trivial &\DLOG{} &\DLOG{}&\DLOG{}\\
  CSAN&USAT-hard&trivial&\DLOG{}&\DLOG{}\\
  $\Delta$-bounded degree&PTIME&PTIME&trivial&\DLOG{}\\
\end{tabular}
\end{center}

USAT-hard means that any PTIME algorithm would imply a PTIME algorithm for Unambiguous-SAT.

\subsection{Simulation between automata network families}

From now on, a family of automata networks will be given as a pair ${(\mathcal{F},\mathcal{F}^*)}$ where $\mathcal{F}$ is the set of abstract automata networks and $\mathcal{F}^*$ a standard representation.
We can now present our notion of simulation between families: a family $\mathcal{A}$ can simulate another family $\mathcal{B}$ if we are able to \textit{effectively} construct for any $B\in\mathcal{B}$ some automata network ${A\in\mathcal{A}}$ that is able to simulate $B$ in the sense of Definition~\ref{def:bloc-simu}.
More precisely, we ask on one hand that the automata network which performs the simulation do this task in reasonable time and reasonable space in the \emph{size of the simulated automata network}, and, on the other hand, that the construction of the simulator is efficient in the \emph{size of the representation of the simulated one}.

\begin{definition}\label{def:simu-family}
  Let $(\mathcal{F},\mathcal{F}^*)$ and $(\mathcal{H},\mathcal{H}^*)$ be two families with standard representations on alphabets $Q_F$ and $Q_H$ respectively. Let $T,S: \N \to \N$ be two functions. We say that $\mathcal{F}^*$ simulates $\mathcal{H}^*$ in time $T$ and space $S$  if there exists a \DLOG{} Turing machine $M$ such that for each $w \in L_{\mathcal{H}}$ representing some automata network ${H_w \in \mathcal{H} : Q_H^n\to Q_H^n}$, the machine produces a pair $M(w)$ which consists in:
  \begin{itemize}
  \item $w'\in L_{\mathcal{F}}$ with $F_{w'}:Q_F^{n_F}\to Q_F^{n_F}$,
  \item ${T(n)}$ and a representation of a block embedding $\phi:Q^{n_F}\to Q^{n}$, 
  \end{itemize}
  such that ${n_F=S(n)}$ and $F_{w'}$ simulates $H_w$ in time $T= T(n)$ under block embedding $\phi$. 
\end{definition}
From now on, whenever $\mathcal{F}^*$ simulates $\mathcal{H}^*$ in time $T$ and space $S$ we write $\mathcal{H}^* \preccurlyeq^T_S \mathcal{F}^*.$ 

\begin{remark}
  Note that both $T$ and $S$ maps must be \DLOG{} computable from this definition. Moreover, the simulation relation between families is transitive because the class \DLOG{} is closed under composition and simulation between individual automata networks is also transitive. When composing simulations time and space maps $S$ and $T$ get multiplied. 
\end{remark}

\subsection{Decision problems and automata network dynamics}
Studying the complexity of decision problems related to the dynamics of some discrete dynamical system is a very well known and interesting approach for measuring the complexity of the dynamics. In this section we introduce three variants of a classical decision problem that is closely related to the dynamical behavior of automata networks: the prediction problem.
This problem consists in predicting the state of one node of the network at a given time.
We study short term and long term versions of the problem depending on the way the time step is given in input.
In addition, we explore a variant in which we ask if some node has eventually changed without specifying any time step, but only a constant observation time rate $\tau$.
In other words, we check the system for any changement on the state of a particular node every multiple of $\tau$ time steps.
The main point of this subsection is to show that these problems are coherent with our simulation definition in the sense that if some family of automata networks $(\mathcal{F}_2,\mathcal{F}^*_2)$ simulates $(\mathcal{F}_1,\mathcal{F}^*_1)$ then, if some of the latter problem is hard for  $\mathcal{F}_1$ it will also be hard for $\mathcal{F'}_2$.
We will precise this result in the following lines.

Let $(\mathcal{F},\mathcal{F}^*)$ an automata network family and let $L \in \{0,1\}^* \times \{0,1\}^*$ be a parametrized language. We say that $L$ is parametrized by $\mathcal{F}$ if $L$  has $\mathcal{F}^*$ encoded as parameter. We note $L_{\mathcal{F}} \in \{0,1\}^*  $ as the language resulting on fixing $\mathcal{F}^*$ as a constant.

In particular, we are interested in studying prediction problems. We start by defining two variants of this well-known decision problem:
\begin{problem}[Unary Prediction (\PREDU)]\ 
	\label{prob:upred}
	\begin{description}
		\item[Parameters: ] alphabet $Q$, a standard representation $\mathcal{F}^*$ of an automata network family $\mathcal{F}$
		\item[Input: ]\ 
		\begin{enumerate}
			\item a word $w_F \in \mathcal{F}^*$ representing an automata network $F:Q^n \to Q^n$ on alphabet $Q$, with $F \in \mathcal{F}$.
			\item a node $v \in V(F) = [n]$
			\item an initial condition $x \in Q^V$.
			\item a natural number $t$ represented in unary $t \in \textbf{1}.$
		\end{enumerate}
		\item[Question: ] $F^t(x)_v \not = x_v$?
	\end{description}
\end{problem}
\begin{problem}[Binary Prediction (\PREDB)]\ 
	\label{prob:bpred}
	\begin{description}
		\item[Parameters: ] alphabet $Q$, a standard representation $\mathcal{F}^*$ of an automata network family $\mathcal{F}$
		\item[Input: ]\ 
		\begin{enumerate}
			\item a word $w_F \in \mathcal{F}^*$ representing an automata network $F:Q^n \to Q^n$ on alphabet $Q$, with $F \in \mathcal{F}$.
			\item a node $v \in V(F) = [n]$
			\item an initial condition $x \in Q^V$.
			\item a natural number $t$ represented in binary $t \in \{0,1\}^*$.
		\end{enumerate}
		\item[Question: ] $F^t(x)_v \not = x_v$?
	\end{description}
\end{problem}
Note that two problems are essentially the same, the only difference is the representation of time $t$ that we call the \textit{observation time}. We will also call node $v$ the objective node. Roughly,  as it happens with other decision problems, such as integer factorization,  the representation of observation time  will have an impact on the computation complexity of prediction problem. When the context is clear we will refer to both problems simply as $\PRED.$ In order to precise the latter observation we present now some general complexity results concerning $\PRED.$ 
\begin{proposition}
Let $\mathcal{F}$ be a concrete automata network family. The following statements hold:
\begin{enumerate}
	\item $\PREDU_{\mathcal{F}} \in \textbf{P}$
	\item $\PREDB_{\mathcal{F}} \in \textbf{PSPACE}$
\end{enumerate} 
\end{proposition}
Finally, we show that latter problem is coherent with our definition of simulation, in the sense that we can preserve the complexity of $\PRED$. Note that this give us a powerful tool in order to classify concrete automata rules according to the complexity of latter decision problem.
\begin{lemma}
Let $(\mathcal{F},\mathcal{F}^*)$ and $(\mathcal{H}, \mathcal{H}^*)$ be two automata network families. Let $T,S:\N \to \N$ be two polynomial functions such that $\mathcal{H}^* \preccurlyeq^T_S \mathcal{F}^*$ then, $\PRED_{\mathcal{H}^*} \leq^T_{\textbf{L}} \PRED_{\mathcal{F}^*}$\footnote{Here we denote $\leq^T_{\textbf{L}}$ as a $\DLOG$ Turing reduction. The capital letter ``T'' stands for \emph{Turing reduction} and it is not related to the simulation time function which is also denoted by T.} where {\PRED} denotes either \PREDU or \PREDB
\label{lemma:predsim}
\end{lemma}

\begin{proof}
	Let $(w_H,v,x,t)$ be an instance of $\PRED_{\mathcal{H}^*}$. By definition of simulation, there exists a $\DLOG$ algorithm which takes $w_H$ and produces a word $w_F \in L_{\mathcal{F}}$ with $F:Q^{n_F} \to Q^{n_F}$ and a block representation $\phi: Q^{n_F} \to Q^n$ such that $n_F = S(n)$ and $F$ simulates $H$ in time $T(n)$ under block embedding $\phi$. Particularly, there exists a partition of blocks $D_v \subseteq V(F) = [n_F]$ for each $v \in V(H) = [n]$ and a collection of injective patterns, i.e. patterns $p_{i,q} \in Q_F^{D_i}$ such that $p_{i,q} = p_{i,q'} \implies q=q'.$ In addition, we have $\phi \circ H = F^T \circ \phi.$ Let us define the configuration $y \in Q_F^{n_F}$ as $y_{D_i} = p_{i,x_i}$, i.e., $\phi(x)_{D_i} = y_{D_i}$. Note that $y$ is well-defined as the block map is injective. In addition, let us choose an arbitrary vertex $v' \in D_v$ and let us consider now the instance of $\PRED_{\mathcal{F}^*}$ given by $(w_F,v',y,t \times T)$. Note that for each $v' \in D_v$ the transformation $(w_H,v,x,t) \to (w_F,v',y,t \times T)$ can be done in $\DLOG(|w_H|)$ because we can read the representation of $\phi$ for each block $p_{i,x_i}$ and then output the configuration $y$. We claim that  there exists a $\DLOG(|w_H|)$ algorithm that decides if $(w_H,v,x,t) \in L_{\PRED_{\mathcal{H}^*}}$  with oracle calls to $\PRED_{\mathcal{F}^*}$. More precisely, as a consequence of the injectivity of block embedding, we have that $(w_H,v,x,t) \in L_{\PRED_{\mathcal{H}^*}}$ if and only if $(w_F,v',y,t \times T) \in L_{\PRED_{\mathcal{F}^*}}$ for some $v' \in D_v$. In fact, latter algorithm just runs oracle calls of $\PRED_{\mathcal{F}^*}$ for $(w_F,v',y,t \times T)$ for each $v' \in D_v$ and decides if some of these instances is a YES instance and thus, if some node in the simulation block have changed its state. As block-embedding is injective, this necessary means that node $v$ have changed its state in $t$ time steps. Finally, all of this can be done in $\DLOG$ as $n_F = S(n) = n^{\mathcal{O}(1)}$ and $T =  n^{\mathcal{O}(1)}$ and thus, a polynomial amount of calls to each oracle is needed.
\end{proof}

Finally, we would like to study the case in which the observation time is not unique and ask whether the state of some node eventually changes.
However, in order to preserve complexity properties under simulation, we still need to have some sort of restriction on observation times.
This will allow us to avoid giving misleading answers when the simulating network is performing one step of simulation: indeed, it could take several time steps for the simulating network in order to represent one step of the dynamics of the simulated network, so some state change could happen in the intermediate steps while the simulated dynamics involve no state change.
In order to manage this sort of time dilation phenomenon between simulating and simulated systems, we introduce the following decision problem.


\begin{problem}[Prediction change $\PREDC_{\mathcal{F^*}}$]\ 
	\label{prob:cpred}
	\begin{description}
		\item[Parameters: ]  alphabet $Q$, a standard representation $\mathcal{F}^*$ of an automata network family $\mathcal{F}$
		\item[Input: ]\ 
		\begin{enumerate}
			\item a word $w_F \in \mathcal{F}^*$ representing an automata network $F:Q^n \to Q^n$ on alphabet $Q$, with $F \in \mathcal{F}$,
			\item a node $v \in V(G)$,
			\item an initial condition $x \in Q^V$,
			\item a time gap $k \in \mathbb{N}$ in unary.
		\end{enumerate}
		\item[Question: ] $\exists t \in \N: F^{kt}(x)_v \not = x_v$
	\end{description}
\end{problem}

As we did with previous versions of prediction problem, we introduce a general complexity result and then, we show computation complexity is consistent under simulation.

\begin{proposition}
	Let $(\mathcal{F}, \mathcal{F}^*)$ be a automata network family. $\PREDC_{\mathcal{F}} \in \textbf{PSPACE}.$
\end{proposition}

The injectivity of block encodings in our definition of simulation is essential for the following lemma as it guaranties that a state change in the simulating network always represent a state change in the simulated network at the corresponding time steps.

\begin{lemma}
Let $(\mathcal{F},\mathcal{F}^*)$ and $(\mathcal{H}, \mathcal{H}^*)$ be two automata network families and $T,S:\N \to \N$ two polynomial functions such that $\mathcal{H^*} \preccurlyeq^T_{S} \mathcal{F^*}$ then, $\PREDC_{\mathcal{H^*}} \leq^T_{\textbf{L}} \PREDC_{\mathcal{F^*}}$.
\label{lemma:predcsim}
\end{lemma}
\begin{proof}
	Proof is analogous to short term prediction case. Let $(w_H,v,x,k)$ be an instance of $\PREDC_{\mathcal{H}^*}$. Again, by the definition of simulation, there exists a $\DLOG$ algorithm which takes $w_H$ and produces a word $w_F \in L_{\mathcal{F}}$ with $F:Q^{n_F} \to Q^{n_F}$ and a block representation $\phi: Q^{n_F} \to Q^n$ such that $n_F = S(n)$ and $F$ simulates $H$ in time $T(n)$ under block embedding $\phi$. The latter statements means, particularly, that there exists a partition of blocks $D_v \subseteq V(F) = [n_F]$ for each $v \in V(H) = [n]$ and a collection of injective patterns, i.e. patterns $p_{i,q} \in Q_F^{D_i}$ such that $p_{i,q} = p_{i,q'} \implies q=q'$ and also that $\phi \circ H = F^{T(n)} \circ \phi.$ Let us define the configuration $y \in Q_F^{n_F}$ as $y_{D_i} = p_{i,x_i}$, i.e., $\phi(x)_{D_i} = y_{D_i}$. Note, again, that $y$ is well-defined as the block map is injective. Now we proceed in using the same approach than before: for each $v' \in D_v$ we can produce an instance $(w_F,v',y,kT(n))$ of $\PREDC_{\mathcal{F^*}}.$  There exists a $\DLOG$ machine which produces $(w_F,v',y,kT(n))$ for each $v' \in D_v$ and calls for an oracle solving $\PREDC_{\mathcal{F^*}}(w_F,v',y,kT(n))$ and outputs $1$ if there is at least one YES-instance for some $v'$. By definition of simulation and injectivity of block embedding function we have that this algorithm outputs $1$ if and only if $(w_H,v,x,k) \in \PREDC_{\mathcal{H}^*}.$
\end{proof}

To end this subsection, let us show that problems $\PREDC$ and $\PREDB$ are actually orthogonal: depending of the family of automata networks considered, one can be harder than the other and reciprocally.

\begin{theorem}\label{theo:orthogonalpredictions}
  The exists a family with circuit representation ${(\mathcal{F},\mathcal{F}^*)}$ such that ${\PREDB_{\mathcal{F}}}$ is solvable in polynomial time while ${\PREDC_{\mathcal{F}}}$ is NP-hard.
  The exists a family with circuit representation ${(\mathcal{G},\mathcal{G}^*)}$ such that ${\PREDB_{\mathcal{G}}}$ is PSPACE-complete while ${\PREDC_{\mathcal{G}}}$ is solvable in polynomial time.
\end{theorem}
\begin{proof}
  Given a SAT formula $\phi$ with $n$ variables, let us define the automata network $F_\phi$ on ${\{0,1\}^{n+1}}$ which interprets any configuration as a pair ${(b,v)\in \{0,1\}\times\{0,1\}^n}$ where $b$ is the state of node $1$ and $v$ is both a number represented in base 2 and a valuation for $\phi$ and does the following:
  \[F(b,i) =
    \begin{cases}
      (1, i+1\bmod 2^n) &\text{ if $\phi$ is true on valuation $v$,}\\
      (0, i+1\bmod 2^n) &\text{ else.}
    \end{cases}
  \]
  A circuit representation of size polynomial in $n$ can be computed in $\DLOG$ from $\phi$ and we define ${(\mathcal{F},\mathcal{F}^*)}$ as the family obtained by considering all $F_\phi$ for all SAT formulas $\phi$.
  First, ${\PREDB_{\mathcal{F}}}$ can be solved in polynomial time: given $F_\phi$, an initial configuration ${(b,v)}$ and a time $t$, it is sufficient to compute ${v'=v+t-1\bmod 2^n}$ and verify the truth $b'$ of $\phi$ on valuation $v'$ and we have ${(b',v'+1\bmod 2^n)=F^t(b,v)}$.
  To see that ${\PREDC_{\mathcal{F}}}$ is NP-hard, it suffices to note that, on input ${(0,0\cdots0)}$, $F_\phi$ will test successively each possible valuation for $\phi$ and the state of node $1$ will change to $1$ at some time if and only if formula $\phi$ is satisfiable.

  For the second part of the proposition, the key is the construction for any $n$ of an automata network $H_n$ on $Q^n$ that completely trivializes problem ${\PREDC}$ in the following sense: for any configuration ${c\in Q^n}$ and any ${k\leq 2^n}$ and any node $v$, there is some $t$ such that ${c_v\neq F^{kt}(c)_v}$.
  Taking any automata network $F$ with $n$ nodes, the product automata network ${F\times H_n}$ (working on the product of alphabets in such a way that each component evolves independently) has the same property, namely that all instances of ${\PREDC}$ with ${k\leq 2^n}$ have a positive answer.
  From this, taking any family with a PSPACE-hard ${\PREDC}$ problem (they are known to exist, see Corollary~\ref{cor:universality} for details), and replacing each automata network $F$ with $n$ nodes by the product ${F\times H_n}$ (the circuit representation of the product is easily deduced from the representations of each component), we get a family ${(\mathcal{G},\mathcal{G}^*)}$ such that ${\PREDB_{\mathcal{G}}}$ is \textbf{PSPACE}-complete while ${\PREDC_{\mathcal{G}}}$ is easy: on one hand, taking products does not simplify ${\PREDB}$ problem; on the other hand, ${\PREDC}$ becomes trivial (always true) on inputs where the observation interval $k$ is less than ${2^n}$, and if ${k\geq 2^n}$ then the size of the whole orbit graph of the input network is polynomial in $k$ (since $k$ is given in unary), so the entire orbit of the input configuration can be computed explicitly in polynomial time and the ${\PREDC}$ can be answered in polynomial time.
  
  Let us complete the proof by giving an explicit construction of the automata networks $H_n$ over $Q^n$ with the desired property.
  ${Q=\{0,1\}\times\{0,1\}\times\{0,1\}}$ and $H_n$ interprets any configuration as a triplet of Boolean configurations ${(c,i,k)}$ with the following meaning: $k$ is a global counter that will take all possible values between $0$ and $2^n-1$ and loop, $i$ is a local counter that will run from $0$ to $2k$ and $c$ is the component where state changes will be realized at precise time steps to ensure the desired property of $H_n$. The goal is to produce in any orbit and for any $k$ and at any node the sequence of states ${O^k1^k}$ on the $c$-component: such a behavior is sufficient to ensure the desired property on $H_n$. This is obtained by defining ${H_n(c,i,k)=(c',i',k')}$ as follows:
  \begin{itemize}
  \item for any node $v$, ${c'_v=0}$ if $i<k$ and ${c'_v=1}$ else,
  \item $i'=0$ if ${i\geq 2k}$ and $i'+1$ else,
  \item ${k'=k+1\bmod 2^n}$ if ${i\geq 2k}$ and ${k'=k}$ else.
  \end{itemize}
  It is clear that such an $H_n$ admits a polynomial circuit representation \DLOG{} computable from $n$.
\end{proof}

\subsection{Universal automata network families}

Building upon our definition of simulation, we can now define a precise notion of universality. In simple words, an universal family is one that is able to simulate every other automata network under any circuit encoding. Our definition of simulation ensures that the amount of resources needed in order to simulate is controlled so that we can deduce precise complexity results.

Consider some alphabet $Q$ and some polynomial map $P:\N\to\N$.
We denote by $\mathcal{U}_{Q,P}$ the class of all possible functions $F:Q^n \to Q^n$ for any $n\in \N$ that admits a circuit representation of size at most ${P(n)}$. We also denote $\mathcal{U}_{Q,P}^*$ the language of all possible circuit representations of size bounded by $P$ of all functions from $\mathcal{U}_{Q,P}$. Finally for any ${\Delta\geq 1}$, denote by ${\mathcal{B}_{Q,\Delta}}$ the set of automata networks on alphabet $Q$ with a communication graph of degree bounded by $\Delta$ and by ${\mathcal{B}_{Q,\Delta}^*}$ their associated bounded degree representations made of a pair (graph, local maps) as discussed above.

\begin{definition}
  \label{def:universal}
  A family of automata networks $(\mathcal{F},\mathcal{F}^*)$ is :
  \begin{itemize}
  \item \emph{universal} if for any alphabet $Q$ and any polynomial
    map $P$ it can simulate $(\mathcal{U}_{Q,P},\mathcal{U}_{Q,P}^*)$
    in time $T$ and space $S$ where $T$ and $S$ are polynomial
    functions;
  \item \emph{strongly universal} if for any alphabet $Q$ and any degree ${\Delta\geq 1}$ it can simulate ${(\mathcal{B}_{Q,\Delta},\mathcal{B}_{Q,\Delta}^*)}$ in time $T$ ans space $S$ where $T$ and $S$ are linear maps.
\end{itemize}
\end{definition}

\begin{remark}
  The link between the size of automata networks and the size of their representation is the key in the above definitions: a universal family must simulate any individual automata network $F$ (just take $P$ large enough so that ${F\in\mathcal{U}_{Q,P}}$), however it is not required to simulate in polynomial space and time the family of all possible networks without restriction. Actually no family admitting polynomial circuit representation could simulate the family of all networks in polynomial time and space by the Shannon effect (most $n$-ary Boolean function have super-polynomial circuit complexity). In particular the family  ${\mathcal{B}_{Q,\Delta}}$ can't.

  At this point it is clear, by transitivity of simulations, that if some ${\mathcal{B}_{Q,\Delta}}$ happens to be universal then, any strongly universal family is also universal. It turns out that ${\mathcal{B}_{\{0,1\},3}}$ is universal. We will however delay the proof until section~\ref{sec:gmonuniv} below where we prove a more precise result which happens to be very useful to get universality result in concrete families. 
  
  Finally, observe that the fact that $S$ and $T$ are polynomial or linear maps implies that they are computable in $\DLOG$ which is coherent with the reductions presented in Lemma \ref{lemma:predsim} and Lemma \ref{lemma:predcsim}.
\end{remark}


Now, we introduce an important corollary of universality regarding complexity. Roughly speaking, a universal family exhibits all the complexity in terms of dynamical behaviour and computational complexity of prediction problems. Concerning computational complexity, let us introduce the following definition.

\begin{definition}
  We say that an automata network family $\mathcal{F}$ is computationally complex if the following conditions hold:
  \begin{enumerate}
  \item $\PREDU_{\mathcal{F}}$ is $\textbf{P}$-hard.
  \item $\PREDB_{\mathcal{F}}$ is $\textbf{PSPACE}$-hard.
  \item $\PREDC_{\mathcal{F}}$ is $\textbf{PSPACE}$-hard.
  \end{enumerate}
\end{definition}

As a direct consequence of universality, we have the following result where there is no difference between strong or standard universality.

\begin{corollary}\label{cor:universality}
Let $\mathcal{F}$ a (strongly) universal automata network family then $\mathcal{F}$ computationally complex.
\end{corollary}

\begin{proof}
  We observe that ${\mathcal{B}_{Q,\Delta}}$ is computationally complex for large enough $Q$ and $\Delta$, which concludes the proof by definition of (strong) universality and Lemmas~\ref{lemma:predsim} and~\ref{lemma:predcsim}. First, any Turing machine working in bounded space can be directly embedded into a cellular automaton on a periodic configuration which is a particular case of automata network on a bounded degree communication graph (for the {\PREDC} variant we can always add a witness node that changes only when the Turing machine accepts for instance). This direct embedding is such that one step of the automata network correspond to one step of the Turing machine and one node of the network corresponds to one cell of the Turing tape. However, the alphabet of the automata network depends on the tape alphabet and the state set of the Turing machine. To obtain the desired result we need to fix the target alphabet, while allowing more time and/or more space. Such simulations of any Turing machine by fixed alphabet cellular automata with linear space/time distortion are known since a long time \cite{lindgren90}, but a modern formulation would be as follows: if there exists an intrinsically universal cellular automaton \cite{bulk2} with states set $Q$ and neighborhood size $\Delta$ (whatever the dimension), then ${\mathcal{B}_{Q,\Delta}}$ is computationally complex. The 2D cellular automaton of Banks \cite{banks} is intrinsically universal \cite{surveyOllinger} and has two states and $5$ neighbors, which shows that $\mathcal{B}_{Q,\Delta}$ is computationally complex when $\Delta\geq 5$ and $Q$ is not a singleton. The 1D instrinsically universal cellular automaton of Ollinger-Richard \cite{OllingerRichard4states} has $4$ states and $3$ neighbors so $\mathcal{B}_{Q,\Delta}$ is computationally complex when $\Delta\geq 3$ and $|Q|\geq 4$.
\end{proof}

In addition, reader interested in simulation results which does not involve cellular automata but automata networks having a more general graph structure, we provide following references in which authors have presented alternative simulation schemes by using very well-known boolean network families as threshold networks:
\begin{enumerate}
	\item The family $\Theta$ of threshold networks over the binary alphabet $Q=\{0,1\}$ satisfies that $\PREDU_{\Theta}$ is $\textbf{P}$-hard \cite{goles2014computational}.	
	\item The family $\Theta$ of threshold networks over an alphabet $Q=\{0,1\}$ (equipped with block sequential update scheme) satisfies that $\PREDC_{\Theta}$ is $\textbf{PSPACE}$-hard \cite{goles2016pspace}
\end{enumerate}

We now turn to the dynamical consequences of universality. By definition simulations are particular embeddings of orbit graphs into larger ones, but the parameters of the simulation can involve some distortion.

\begin{definition}
  Fix a map ${\rho:\N\to\N}$, we say that the orbit graph $G_F$ of $F$ with $n$ nodes is $\rho$-succinct if $F$ can be represented by circuits of size at most $\rho(n)$. We say that the orbit graph $G_H$ of $H$ with $m$ nodes embeds $G_F$ with distortion $\delta:\N\to\N$ if ${m\leq\delta(n)}$ and there is $T\leq\delta(n)$ such that $G_F$ is a subgraph of $G_{H^T}$.
\end{definition}

\begin{remark}
  The embedding of orbit graphs with distortion obviously modify the relation between the number of nodes of the automata netwroks and the length of paths or cycles in the orbit graph. In particular, with polynomial distortion $\delta$, if $F$ has $n$ nodes and a cyclic orbit of length ${2^n}$ (hence exponential in the number of nodes) then in $H$ it gives a cyclic orbit of size ${O(\delta(n)2^n)}$ for up to $\delta(n)$ nodes, which does not guarantee an exponential length in the number of nodes in general, but just a super-polynomial one (${n\mapsto 2^{n^\alpha}}$ for some $0<\alpha< 1$).
\end{remark}

To fix ideas, we give examples of orbit graphs of bounded degree automata networks with large components corresponding to periodic orbits or transient.

\begin{proposition}\label{prop:bounded-degree-dynamics}
  There is an alphabet $Q$ such that for any ${n\geq 1}$ there is an automata network ${F_n\in\mathcal{B}_{Q,2}}$ whose orbit graph $G_{F_n}$ has the following properties:
  \begin{itemize}
  \item it contains a cycle $C$ of length at least ${2^n}$;
  \item there is a complete binary tree $T$ with $2^n$ leaves connected to some $v_1\in C$, \textit{i.e.} for all ${v\in T}$ there is a path from $v$ to $v_1$;
  \item there is a node ${v_2\in C}$ with a directed path of length $2^{n}$ pointing towards $v_2$;
  \item it possesses at least ${2^n}$ fixed points.
  \end{itemize}
\end{proposition}

\begin{proof}
  First on a component of states ${\{0,1,2\}\subseteq Q}$ the large cycle $C$ is obtained by the following 'odometer' behavior of $F_n$: if ${x_n\in\{0,1,2\}}$ then ${F_n(x)_n = x_n+1 \bmod 3}$, and if both ${x_i,x_{i+1}\in\{0,1,2\}}$ for ${1\leq i<n}$ then 
  \[F_n(x)_i =
    \begin{cases}
      0 &\text{ if }x_i=2\\
      x_i+1 &\text{ else if }x_{i+1}=2,\\
      x_i &\text{ else.}
    \end{cases}
  \]
  $C$ is realized on ${\{0,1,2\}^n}$ as follows. For $x\in\{0,1,2\}^n$ denote by $S_i$ the sequence ${(F^t(x)_i)_{t\geq 0}}$ for any ${1\leq i\leq n}$. Clearly $S_n$ is periodic of period $012$. $S_{n-1}$ is ultimately periodic of period ${200111}$ (of length $6$) and by a straighforward induction we get that $S_1$ is ultimately periodic of period ${20^{3\cdot 2^{n-2}-1}1^{3\cdot 2^{n-2}}}$ which is of length ${3\cdot 2^{n-1}}$.


  For the tree $T$, just add states ${\{a,b\}\subseteq Q}$ with the following behavior: if $x_1\in\{a,b\}$ then ${F_n(x)_1=0}$ and if ${x_{i}\in\{a,b\}}$ and ${x_{i-1}=0}$ then ${F_n(x)_i=0}$ for ${1<i\leq n}$. In any other case, we set ${F(x)_i=x_i}$ for ${x\in\{0,1,2,a,b\}^n}$ and ${1\leq i\leq n}$.

  Using similar mechanisms as above on additional states ${c_0,c_1,c_2\in Q}$, $F_n$ runs another odometer whose behavior is isomorphic to the behavior of $F_n$ on ${\{0,1,2\}^n}$ through ${i\mapsto c_i}$, but with the following exception: when ${x_1=c_2}$ we set ${F_n(x)_1=0}$ and then state $0$ propagates from node $1$ to node $n$ as in the construction of tree $T$. We thus get a transient behavior of length more than ${3\cdot 2^{n-1}}$ which yields to configuration ${0^n}$, which itself (belongs or) yields to cycle $C$.

  Finally, the fixed points are obtained by adding two more states to the alphabet on which the automata network just acts like the identity map.
\end{proof}

We can now state that any universal family must be dynamically rich in a precise sense.

\begin{theorem}\label{them:univ-rich-dynamics}
  Let $\mathcal{F}$ be an automata networks family.
  \begin{itemize}
  \item If $\mathcal{F}$ is universal then, for any polynomial map
    $\rho$, there is a polynomial distortion $\delta$ such that, any
    $\rho$-succinct orbit graph can be embedded into some
    ${F\in\mathcal{F}}$ with distortion $\delta$. In particular
    $\mathcal{F}$ contains networks with super-polynomial periods and
    transients, and a super-polynomial number of disjoint periodic orbits of period at most polynomial.
  \item If $\mathcal{F}$ is strongly universal then it embeds the orbit graph of any bounded-degree automata network with linear distortion. In particular it contains networks with exponential periods and transients, and an exponential number of disjoint periodic orbits of period at most linear.
\end{itemize}
\end{theorem}
\begin{proof}
  This is a direct consequence of Lemma~\ref{lem:simu-dynamic}, Definition~\ref{def:simu-family} and Proposition~\ref{prop:bounded-degree-dynamics} above.
\end{proof}



Of course, we do not claim that computational complexity and dynamical richness as stated above are the only meaningful consequences of universality.
To conclude this subsection about universality, let us show that it allows to prove finer results linking the global dynamics with the interaction graph.

In a directed graph, we say a node $v$ belongs to a strongly connected component if there is a directed path from $v$ to $v$.

\begin{corollary}
  Any universal family $\mathcal{F}$ satisfies the following: there is a constant $\alpha$ with ${0<\alpha\leq 1}$ such that for any $m>0$ there is a network $F\in\mathcal{F}$ with $n\geq m$ nodes such that some node $v$ belonging to a strongly connected component of the interaction graph of $F$ and a periodic configuration $x$ such that the trace at $v$ of the orbit of $x$ is of period at least ${2^{n^\alpha}}$.
  \label{coro:trickyuniversal}
\end{corollary}
\begin{proof}
  Consider a Boolean network $F$ with nodes ${V=\{1,\ldots,m\}}$ that do the following on configuration $x\in \{0,1\}^V$: it interprets ${x_1,\ldots,x_m}$ as an number $k$ written in base $2$ where $x_1$ is the most significant bit and produces ${F(x)}$ which represents number ${k+1 \bmod 2^{m}}$. 

  $F$ is such that node $1$ has a trace of exponential period and belongs to a strongly connected component of the interaction graph of $F$ (because it depends on itself).
  Note that $F$ has a circuit representation which is polynomial in $m$, and take $F'\in\mathcal{F}$ of size polynomial in $m$ that simulates $F$ in polynomial time (by universality of family $\mathcal{F}$).
  Taking the notations of Definition~\ref{def:bloc-simu}, we have that each node $v\in D_i$ for each block $D_i$ is such that the map ${q\in\{0,1\}\mapsto p_{i,q}(v)}$ is either constant or bijective (because $F$ has a Boolean alphabet, see Remark~\ref{rem:bloc-simul}).
  In the last case, the value of the node $v\in D_i$ completely codes the value of the corresponding node $i$ in $F$.
  Take any $v\in B_1$ that has this coding property.
  Since node $1$ depends on itself in $F$, there must be a path from $v$ to some node $v'\in D_1$ that is also coding in the interaction graph of $F'$. 
  Then we can also find a path from $v'$ to some coding node in $D_1$.
  Iterating this reasoning we must find a cycle, and in particular we have a coding node in $D_1$ which belongs to some strongly connected component of the interaction graph of $F'$.
  Since this node is coding the values taken by node $1$ of $F$ and since the simulation is in polynomial time and space, we deduce the super-polynomial lower bound on the period of its trace for a well-chosen periodic configuration.   
\end{proof}


\section{Gadgets and glueing}
\label{sec:gadgetsandglueing}

In the same way as Boolean circuits are defined from Boolean gates, many automata network families can be defined by fixing a finite set of local maps $\G$ that we can freely connect together to form a global network, called a $\G$-network.

Such families can be strongly universal as we will see, even for very simple choices of $\G$, which is an obvious motivation to consider them.
In this section, we introduce a general framework to prove simulation results of a $\G$-network family by some arbitrary family that amounts to a finite set of conditions to check.
From this we will derive a framework to certify strong universality of an arbitrary family just by exhibiting a finite set of networks from the family that verify a finite set of conditions.
As already said above, our goal is to analyze automata networks with symmetric communication graph (CSAN families). Our framework is targeted towards such families.

The idea behind is that of building large automata networks from small automata networks in order to mimic the way a $\G$-network is built from local maps in $\G$.
The difficulty, and the main contribution of this section, is to formalize how small building blocks are glued together and what conditions on them guaranty that the large network correctly simulates the corresponding $\G$-networks.

We will now introduce all the concepts used in this framework progressively.

\subsection{$\G$-networks}

Let $Q$ be a fixed alphabet and $\G$ be any set of maps of type ${g:Q^{i(g)}\rightarrow Q^{o(g)}}$ for some ${i(g),o(g)\in\N}$. We say $g$ is \emph{reducible} if it can be written as a disjoint union of two gates, and \emph{irreducible} otherwise. Said differently, if $G$ is the (bipartite) dependency graph of $g$ describing on which inputs effectively depends each output, then $g$ is irreducible if $G$ is weakly connected.

From $\G$ we can define a natural family of networks:  a $\G$-network is an automata network obtained by wiring outputs to inputs of a number of gates from $\G$. To simplify some later results, we add the technical condition that no output of a gate can be wired to one of its inputs (no self-loop condition).

\begin{definition}\label{def:g-network}
  A $\G$-network is an automata network ${F:Q^V\rightarrow Q^V}$ with set of nodes $V$ associated to a collection of gates ${g_1,\ldots, g_n\in\G}$ with the following properties. Let 
  \begin{align*}
    I &=\{(j,k) : 1\leq j\leq n\text{ and }1\leq k\leq i(g_j)\} \text{ and}\\
    O &=\{(j,k) : 1\leq j\leq n\text{ and }1\leq k\leq o(g_j)\}
    \end{align*}
    be respectively the sets of inputs and outputs of the collection of gates ${(g_j)_{1\leq j\leq n}}$.
    We require ${|V|=|I|=|O|}$ and the existence of two bijective maps ${\alpha:I\rightarrow V}$ and ${\beta: V\rightarrow O}$ with the condition that there is no ${(j,k)\in I}$ such that $\beta(\alpha(j,k))=(j,k')$ for some $k'$ (no self-loop condition). For ${v\in V}$ with $\beta(v) = (j,k)$, let ${I_v = \{\alpha(j,1),\ldots,\alpha(j,i(g_j))\}}$ and denote by $g_v$ the map: ${x\in Q^{I_v}\mapsto g_j(\tilde x)_k}$ where ${\tilde x\in Q^{i(g_j)}}$ is defined by ${\tilde{x}_k = x_{\alpha(j,k)}}$.
    Then $F$ is defined as follows: 
  \[F(x)_v = g_v(x_{I_v}).\]
\end{definition}

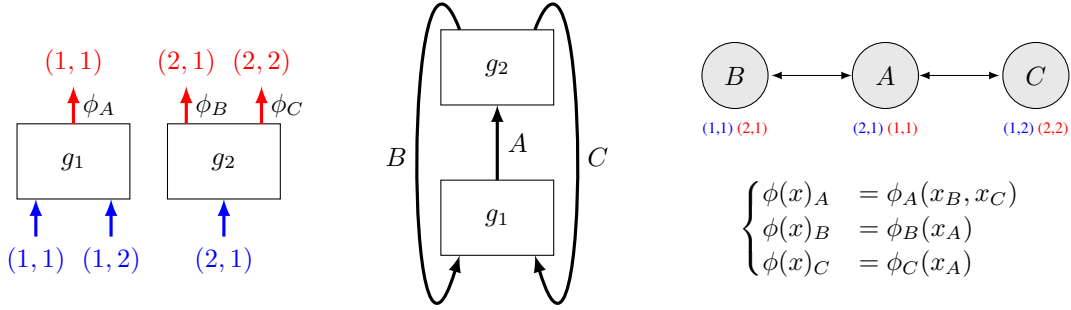
\begin{figure}
  \centering
  \begin{minipage}{.3\linewidth}
    \begin{tikzpicture}[scale=.5]
      \draw (-1.5,1)--(1.5,1)--(1.5,-1)--(-1.5,-1)--cycle;
      \draw (0,0) node {$g_1$};
      \draw[->,red,very thick] (0,1)-- (0,2) node[above] {$(1,1)$};
      \draw[->,blue,very thick] (-1,-2) node[below] {$(1,1)$} -- (-1,-1);
      \draw[->,blue,very thick] (1,-2) node[below] {$(1,2)$} -- (1,-1);
      \draw (0,1.5) node[right] {$\phi_A$};
      \begin{scope}[shift={(4,0)}]
        \draw (-1.5,1)--(1.5,1)--(1.5,-1)--(-1.5,-1)--cycle;
        \draw (0,0) node {$g_2$};
        \draw[->,blue,very thick] (0,-2)  node[below] {$(2,1)$} -- (0,-1);
        \draw[->,red,very thick] (-1,1)  -- (-1,2) node[above] {$(2,1)$};
        \draw[->,red,very thick] (1,1)  -- (1,2) node[above] {$(2,2)$};
        \draw (-1,1.5) node[right] {$\phi_B$};
        \draw (1,1.5) node[right] {$\phi_C$};
      \end{scope}
    \end{tikzpicture}
  \end{minipage}
  \begin{minipage}{.25\linewidth}
    \begin{tikzpicture}[scale=.5]
      \draw (-1.5,1)--(1.5,1)--(1.5,-1)--(-1.5,-1)--cycle;
      \draw (0,0) node {$g_1$};
      \begin{scope}[shift={(0,4)}]
        \draw (-1.5,1)--(1.5,1)--(1.5,-1)--(-1.5,-1)--cycle;
        \draw (0,0) node {$g_2$};
      \end{scope}
      \draw[->,very thick] (0,1)-- node[midway,right] {$A$} (0,3);
      \draw[->,very thick] (-1,5) .. controls (-2.5,9) and (-2.5,-6) .. node[midway,left] {$B$} (-1,-1);
      \draw[->,very thick] (1,5) .. controls (2.5,9) and (2.5,-6) .. node[midway,right] {$C$} (1,-1);
    \end{tikzpicture}
  \end{minipage}
  \begin{minipage}{.3\linewidth}
    \begin{tikzpicture}[shorten >=1pt,node distance=2cm,on grid,auto]
      \tikzstyle{every state}=[fill={rgb:black,1;white,10}]
      
      \node[state]   (q_1)                    {$A$};
      \node[state] (q_2)  [left of=q_1]    {$B$};
      \node[state]           (q_3)  [right of=q_1]    {$C$};
      
      \path[<->]
      (q_1) edge (q_2);
      \path[<->]
      (q_1) edge (q_3);
      \draw (q_1)+(0,-.5) node[below] {\tiny\color{blue}(2,1) \color{red}(1,1)};
      \draw (q_2)+(0,-.5) node[below] {\tiny\color{blue}(1,1) \color{red}(2,1)};
      \draw (q_3)+(0,-.5) node[below] {\tiny\color{blue}(1,2) \color{red}(2,2)};
    \end{tikzpicture}\\
    \[
      \begin{cases}
        \phi(x)_A &= \phi_A(x_B,x_C)\\
        \phi(x)_B &= \phi_B(x_A)\\
        \phi(x)_C &= \phi_C(x_A)
      \end{cases}
    \]
  \end{minipage}
  \caption{On the left a set of maps $\G$ over alphabet $Q$, in the middle an intuitive representation of input/output connections to make a $\G$-network, on the right the corresponding formal $\G$-network ${\phi : Q^3\rightarrow Q^3}$ together with the global map associated to it. The bijections $\alpha$ and $\beta$ from Definition~\ref{def:g-network} are represented in blue and red (respectively).}
  \label{fig:gnetwork}
\end{figure}

\begin{remark}\label{rem:represent-g-networks}
  Once $\G$ is fixed, there is a bound on the degree of dependency graphs of all $\G$-networks. Thus, it is convenient to represent $\G$-networks by the standard representation of bounded degree automata networks (as a pair of a graph and a list of local update maps). Another representation choice following strictly Definition~\ref{def:g-network} consists in giving a list of gates ${g_1,\ldots, g_k\in\G}$, fixing ${V=\{1,\ldots,n\}}$ and give the two bijective maps ${\alpha:I\to V}$ and ${\beta: V\to O}$ describing the connections between gates (maps are given as a simple list of pairs source/image). One can check that these two representations are \DLOG{} equivalent when the gates of $\G$ are irreducible: we can construct the interaction graph and the local maps from the list of gates and maps $\alpha$ and $\beta$ in \DLOG{} (the incoming neighborhood of a node $v$, $I_v$, and its local map $g_v$ are easy to compute as detailed in Definition~\ref{def:g-network}); reciprocally, given the interaction graph $G$ and the list of local maps $(g_v)$, one can recover in \DLOG{} the list of gates and their connections as follows:
  \begin{itemize}
  \item for $v$ from $1$ to $n$ do:
    \begin{itemize}
    \item gather the (finite) incoming neighborhood $N^-(v)$ of $v$ then the (finite) outgoing neighborhood $N^+(N^-(v))$ and iterate this process until it converges (in finite time) to a set $I_v$ of inputs and $O_v$ of outputs with ${v\in O_v}$;
    \item check that all ${v'\in I_v\cup O_v}$ are such that ${v'\geq v}$ otherwise jump to next $v$ in the loop (this guaranties that each gate is generated only once);
    \item since the considered gates are irreducible, $I_v$ and $O_v$ actually correspond to input and output sets of a gate $g\in\G$ that we can recover by finite checks from the local maps of nodes in $O_v$;
    \item output gate $g$ and the pairs source/image to describe $\alpha$ and $\beta$ for nodes in $I_v$ and $O_v$ respectively.
    \end{itemize}
  \end{itemize}
\end{remark}

In the sequel we denote $\Gamma(\mathcal{G})$ the family of all posible $\mathcal{G}$-networks associated to their bounded degree representation.

\subsection{Glueing of automata networks}
In this section we define an operation that allows us to 'glue' two different abstract automata networks on a common part in order to create another one which, roughly, preserve some dynamical properties in the sense that it allows to glue pseudo-orbits of each network to obtain a pseudo-orbit of the glued network.
One might find useful to think about the common part of the two networks as a dowel attaching two pieces of wood: each individual network is a piece of wood with the dowel inserted in it, and the result of the glueing is the attachment of the two pieces with a single dowel (see Figure~\ref{fig:glueingscheme}).

\begin{definition}\label{def:glueing}
  Consider ${F_1:Q^{V_1}\rightarrow Q^{V_1}}$ and ${F_2:Q^{V_2}\rightarrow Q^{V_2}}$  two automata networks with $V_1$ disjoint from $V_2$, $C$ a set disjoint from ${V_1\cup V_2}$, $\varphi_1: C \to V_1$ and $\varphi_2: C \to V_2$ two injective maps with ${\varphi_1(C)\cap\varphi_2(C)=\emptyset}$ and ${C_1,C_2}$ a partition of $C$ in two sets.
  We define \[V'= C\cup (V_1\setminus \varphi_1(C)) \cup (V_2\setminus \varphi_2(C))\] and the map $\alpha : V'\rightarrow V_1\cup V_2$ by 
  \[\alpha(v) =
    \begin{cases}
      v & \text{ if } v\not\in C\\
      \varphi_i(v) & \text{ if } v\in C_i, \text{ for }i=1,2.\\
    \end{cases}
  \]
  We then define the glueing of $F_1$ and $F_2$ over $C$ as the automata network $F' : Q^{V'}\rightarrow Q^{V'}$ where
  \[F'_v = \begin{cases}
      (F_1)_{\alpha(v)}\circ\rho_1 &\text{ if } \alpha(v) \in  V_1, \\
      (F_2)_{\alpha(v)}\circ\rho_2 &\text{ if } \alpha(v) \in  V_2,
    \end{cases}\]
  where ${\rho_i : Q^{V'}\rightarrow Q^{V_i}}$ is defined by 
  \[\rho_i(x)_v =
    \begin{cases}
      x_{\varphi_i^{-1}(v)} &\text{ if } v\in \varphi_i(C),\\
      x_v &\text{ else.}
    \end{cases}
  \]
\end{definition}

\begin{figure}\label{fig:glueingscheme}
	\centering
\begin{tikzpicture}[x=0.75pt,y=0.75pt,yscale=-0.75,xscale=0.75]

\draw   (244,53) .. controls (244,41.95) and (278.14,33) .. (320.25,33) .. controls (362.36,33) and (396.5,41.95) .. (396.5,53) .. controls (396.5,64.05) and (362.36,73) .. (320.25,73) .. controls (278.14,73) and (244,64.05) .. (244,53) -- cycle ;
\draw   (115.5,98) .. controls (120.5,57) and (226.5,100) .. (209.5,124) .. controls (192.5,148) and (441.5,161) .. (235,201) .. controls (28.5,241) and (143,209) .. (79.5,180) .. controls (16,151) and (110.5,139) .. (115.5,98) -- cycle ;
\draw  [dash pattern={on 4.5pt off 4.5pt}] (158,171) .. controls (158,159.95) and (192.14,151) .. (234.25,151) .. controls (276.36,151) and (310.5,159.95) .. (310.5,171) .. controls (310.5,182.05) and (276.36,191) .. (234.25,191) .. controls (192.14,191) and (158,182.05) .. (158,171) -- cycle ;
\draw   (527.49,242.78) .. controls (522.49,283.78) and (416.49,240.78) .. (433.49,216.78) .. controls (450.49,192.78) and (201.49,179.78) .. (407.99,139.78) .. controls (614.49,99.78) and (499.99,131.78) .. (563.49,160.78) .. controls (626.99,189.78) and (532.49,201.78) .. (527.49,242.78) -- cycle ;
\draw  [dash pattern={on 4.5pt off 4.5pt}] (484.99,169.78) .. controls (484.99,180.82) and (450.85,189.78) .. (408.74,189.78) .. controls (366.63,189.78) and (332.49,180.82) .. (332.49,169.78) .. controls (332.49,158.73) and (366.63,149.78) .. (408.74,149.78) .. controls (450.85,149.78) and (484.99,158.73) .. (484.99,169.78) -- cycle ;

\draw   (193.5,285) .. controls (198.5,244) and (304.5,287) .. (287.5,311) .. controls (270.5,335) and (519.5,348) .. (313,388) .. controls (106.5,428) and (221,396) .. (157.5,367) .. controls (94,338) and (188.5,326) .. (193.5,285) -- cycle ;
\draw   (432.49,430.78) .. controls (427.49,471.78) and (321.49,428.78) .. (338.49,404.78) .. controls (355.49,380.78) and (106.49,367.78) .. (312.99,327.78) .. controls (519.49,287.78) and (404.99,319.78) .. (468.49,348.78) .. controls (531.99,377.78) and (437.49,389.78) .. (432.49,430.78) -- cycle ;
\draw   (237,358) .. controls (237,346.95) and (271.14,338) .. (313.25,338) .. controls (355.36,338) and (389.5,346.95) .. (389.5,358) .. controls (389.5,369.05) and (355.36,378) .. (313.25,378) .. controls (271.14,378) and (237,369.05) .. (237,358) -- cycle ;

\draw  [dash pattern={on 4.5pt off 4.5pt}]  (305.5,78) -- (271.01,137.41) ;
\draw [shift={(269.5,140)}, rotate = 300.14] [fill={rgb, 255:red, 0; green, 0; blue, 0 }  ][line width=0.08]  [draw opacity=0] (8.93,-4.29) -- (0,0) -- (8.93,4.29) -- cycle    ;
\draw  [dash pattern={on 4.5pt off 4.5pt}]  (334.5,79) -- (367.99,136.41) ;
\draw [shift={(369.5,139)}, rotate = 239.74] [fill={rgb, 255:red, 0; green, 0; blue, 0 }  ][line width=0.08]  [draw opacity=0] (8.93,-4.29) -- (0,0) -- (8.93,4.29) -- cycle    ;
\draw    (317.5,187) -- (317.5,254) ;
\draw [shift={(317.5,257)}, rotate = 270] [fill={rgb, 255:red, 0; green, 0; blue, 0 }  ][line width=0.08]  [draw opacity=0] (8.93,-4.29) -- (0,0) -- (8.93,4.29) -- cycle    ;

\draw (312,42.4) node [anchor=north west][inner sep=0.75pt]    {$C$};
\draw (228,160.4) node [anchor=north west][inner sep=0.75pt]    {$\varphi _{1}( C)$};
\draw (383,159.4) node [anchor=north west][inner sep=0.75pt]    {$\varphi _{2}( C)$};
\draw (114,132.4) node [anchor=north west][inner sep=0.75pt]    {$V_{1}$};
\draw (518,169.4) node [anchor=north west][inner sep=0.75pt]    {$V_{2}$};
\draw (244,87.4) node [anchor=north west][inner sep=0.75pt]    {$\varphi _{1}$};
\draw (373,88.4) node [anchor=north west][inner sep=0.75pt]    {$\varphi _{2}$};
\draw (305,347.4) node [anchor=north west][inner sep=0.75pt]    {$C$};
\draw (311,284.4) node [anchor=north west][inner sep=0.75pt]    {$V'$};

\end{tikzpicture}
\caption{Scheme of a glueing}
\end{figure}
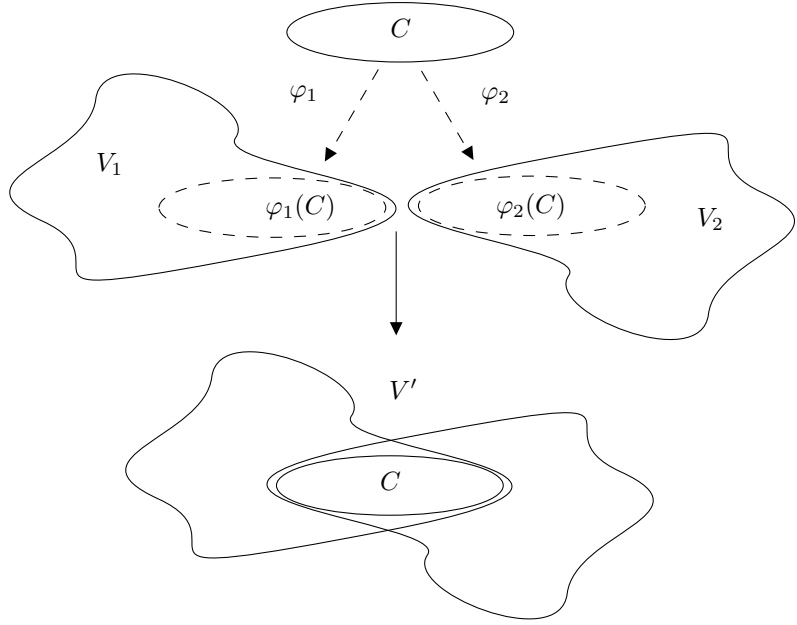

\newcommand\glueing[6]{#1\,{}_{#3}^{#5}\!\!\oplus_{#4}^{#6}\,#2}
When necessary, we will use the notation ${F' = \glueing{F_1}{F_2}{C_1}{C_2}{\phi_1}{\phi_2}}$ to underline the dependency of the glueing operation on its parameters.

Given an automata network ${F:Q^V\rightarrow Q^V}$ and a set ${X\subseteq V}$, we say that a sequence ${(x^t)_{0\leq t\leq T}}$ of configurations from $Q^V$ is a $X$-pseudo-orbit if it respects $F$ as in a normal orbit, except on $X$ where it can be arbitrary, formally: ${x^{t+1}_v = F(x^t)_v}$ for all ${v\in V\setminus X}$ and all ${0\leq t<T}$. The motivation for Definition~\ref{def:glueing} comes from the following lemma.

\begin{lemma}[Pseudo-orbits glueing]\label{lem:pseudo-orbit-glueing}
  Taking the notations of Definition~\ref{def:glueing}, let ${X\subseteq V_1\setminus\varphi_1(C)}$ and ${Y\subseteq V_2\setminus\varphi_2(C)}$ be two (possibly empty) sets. If ${(x^t)_{0\leq t\leq T}}$ is a ${X\cup\varphi_1(C_2)}$-pseudo-orbit for $F_1$ and if ${(y^t)_{0\leq t\leq T}}$ is a ${Y\cup\varphi_2(C_1)}$-pseudo-orbit for $F_2$ and if they verify for all ${0\leq t\leq T}$
  \begin{equation}
    \forall v\in C, x^t_{\varphi_1(v)}=y^t_{\varphi_2(v)},
    \label{eq:sametrace}
  \end{equation} 
  then the sequence ${(z^t)_{0\leq t\leq T}}$ of configurations of $Q^{V'}$ is a ${X\cup Y}$-pseudo-orbit of $F'$, where 
  \[z^t_v =
    \begin{cases}
      x^t_{\alpha(v)} &\text{ if } \alpha(v)\in V_1, \\
      y^t_{\alpha(v)} &\text{ if } \alpha(v)\in V_2.
    \end{cases}
  \]
\end{lemma}

\begin{proof}
  Take any ${v\in V'\setminus(X\cup Y)}$. Suppose first that ${\alpha(v)\in V_1}$. By definition of $F'$, we have ${F'(z^{t})_v = (F_1)_{\alpha(v)}\circ\rho_1 (z^t)}$ but ${\rho_1(z^t) = x^t}$ (using the Equation~\ref{eq:sametrace} in the hypothesis) so ${F'(z^t)_v=F_1(x^t)_{\alpha(v)}}$. Since ${(x^t)}$ is a ${X\cup\varphi_1(C_2)}$-pseudo-orbit and since ${\alpha(v)\not\in X\cup\varphi_1(C_2)}$, we have \[F_1(x^t)_{\alpha(v)}=x^{t+1}_{\alpha(v)}=z^{t+1}_v.\] We conclude that ${z^{t+1}_v=F'(z^{t})_v}$. 
  By a similar reasoning, we obtain the same conclusion if ${\alpha(v)\in V_2}$.
  We deduce that ${(z^t)}$ is a ${X\cup Y}$-pseudo-orbit of $F'$. 
\end{proof}

In the case of a CSAN family where the transition rules are determined by a labeled non-directed graph, the result of a glueing operation has no reason to belong to the family because the symmetry of the interaction graph might be broken (see Figure~\ref{fig:nonsymglueing}).
\begin{figure}
  \centering
  \begin{tikzpicture}
    \draw (2,0) node {$+$};
    \draw (5,0) node {$=$};
    \tikzstyle{every node}=[fill,shape=circle,inner sep=1pt]
    \coordinate[label = above:$\phi_1(C_1)$] (A) at (1,1);
    \coordinate[label = below:$\phi_1(C_2)$] (B) at (1,-1);
    \draw (0,0) node {} -- (A) node {} -- (B) node {} -- cycle;
    \coordinate[label = above:$\phi_2(C_1)$] (C) at (3,1);
    \coordinate[label = below:$\phi_2(C_2)$] (D) at (3,-1);
    \draw (4,0) node {} -- (C) node {} -- (D) node {} -- cycle;
    \draw (6,0) node {} -- (7,1) node {} -- (7,-1) node {} -- (8,0) node {};
    \draw[blue,-triangle 45] (7,1) -- (8,0);
    \draw[blue,-triangle 45] (7,-1) -- (6,0);
  \end{tikzpicture}
  \caption{Symmetry breaking in interaction graph after a glueing operation. Arrows indicate influence of a node (source) on another (target), edges without arrow indicates bi-directional influence. Here $C$ consists in two nodes only.}
  \label{fig:nonsymglueing}
\end{figure}
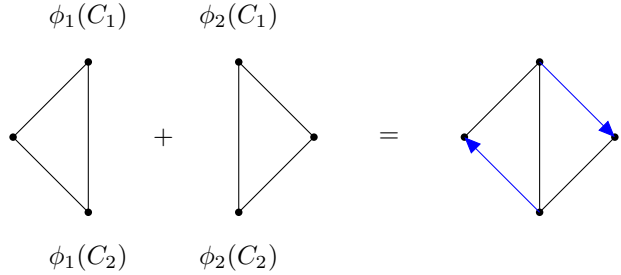
The following lemma gives a sufficient condition in graph theoretical terms for glueing within a concrete family of automata networks. Intuitively, it consists in asking that, in each graph, all the connections of one half of the dowel to the rest of the graph goes through the other half of the dowel. Here the wooden dowel metaphor is particularly relevant: when considering a single piece of wood with the dowel inserted inside, one half of the dowel is 'inside' (touches the piece of wood), the other half is 'outside' (not touching the piece of wood); then, when the two pieces are attached, each position in the wood assembly is locally either like in one piece of wood with the dowel inserted or like in the other one with the dowel inserted.

\begin{lemma}[Glueing for CSAN]
  \label{lem:concreteglueing}
  Let $(G_1,\lambda_1,\rho_1)$ and $(G_2,\lambda_2,\rho_2)$ be two CSAN from the same CSAN family $\mathcal{F}$ where $G_1$ and $G_2$ are disjoint and $F_1$ and $F_2$ are the associated global maps. Taking again the notations of Definition~\ref{def:glueing}, if the following conditions hold
  \begin{itemize}
  \item the labeled graphs induced by $\varphi_1(C)$ and $\varphi_2(C)$ in $G_1$ and $G_2$ are the same (using the identification ${\varphi_1(v)=\varphi_2(v)}$)
  \item ${N_{G_1}(\varphi_1(C_2))\subseteq \varphi_1(C)}$ 
  \item ${N_{G_2}(\varphi_2(C_1))\subseteq \varphi_2(C)}$
  \end{itemize}
  then the glueing ${F' = \glueing{F_1}{F_2}{C_1}{C_2}{\phi_1}{\phi_2}}$ can be defined as the CSAN on graph $G'=(V',E')$ where $V'$ is as in Definition~\ref{def:glueing}
  and each node ${v\in V'}$ has the same label and same labeled neighborhood as either a node of ${(G_1,\lambda_1,\rho_1)}$ or a node of ${(G_2,\lambda_2,\rho_2)}$. In particular $F'$ belongs to ${\mathcal{F}}$.
\end{lemma} 
\begin{proof} 
  Let us define ${\beta_i : V_i\to V'}$ by
  \[\beta_i(v) =
    \begin{cases}
      \phi_i^{-1}(v) &\text{ if }v\in\phi_i(C),\\
      v&\text{ else.}
    \end{cases}
  \]
  Fix ${i=1}$ or $2$. According to Definition~\ref{def:glueing}, if ${v\in V'}$ is such that ${\alpha(v) \in  V_i}$ then ${F'_v=(F_i)_{\alpha(v)}\circ\rho_i}$. By definition of CSAN, this means that for any ${x\in Q^{V'}}$ we have ${F'_v(x) = \psi_{i,\alpha(v)}(x_{|\beta(N_{G_i}(\alpha(v)))})}$ where $\psi_{i,\alpha(v)}$ is a map depending only on the labeled neighborhood of $\alpha(v)$ in $G_i$ as in Definition~\ref{def:csan}. So the dependencies of $v$ in $F'$ are in one-to-one correspondence through $\beta$ with the neighborhood of $\alpha(v)$ in $G_i$. They key observation is that the symmetry of dependencies is preserved, formally for any ${v'\in \beta(N_{G_i}(\alpha(v)))}$:
  \begin{itemize}
  \item either $\alpha(v')\in V_i$ in which case the dependency of $v'$ on $v$ (in map $\psi_{i,\alpha(v')}$) is the same as the dependency of $v$ on $v'$ (in map $\psi_{i,\alpha(v)}$), and both are determined by the undirected labeled edge ${\{\alpha(v),\alpha(v')\}}$ of $G_i$;
  \item or $\alpha(v')\not\in V_{i}$ and in this case necessarily ${v\in C_i}$ and ${v'\in C_{3-i}}$ (because ${N_{G_i}(V_i\setminus\varphi_i(C))\cap\varphi_i(C)\subseteq \varphi_i(C_i)}$ from the hypothesis), so the dependency of $v'$ on $v$ is the same as the dependency of $v$ on $v'$ because the labeled graphs induced by $\phi_1(C)$ and $\phi_2(C)$ in $G_1$ and $G_2$ are the same.
  \end{itemize}

  Concretely, $F'$ is a CSAN that can be defined on graph ${G'=(V',E_1'\cup E_2'\cup E(C))}$ with
  \[E_i' = E(V_i \setminus \varphi_i(C)) \cup \{(u,v_i): u \in V(C_i), v_i \in (V_i \setminus \varphi_i(C)),  (\varphi_i(u),v_i) \in  E_i\},\]
  and labels as follows:
  \begin{itemize}
  \item on $E(C)$ as in both $G_1$ and $G_2$ (which agree through maps $\phi_1$ and $\phi_2$ on $C$), 
  \item on ${E(V_i \setminus \varphi_i(C))}$ as in $G_i$,
  \item for each $u \in V(C_i), v_i \in (V_i \setminus \varphi_i(C))$ such that  $(\varphi_i(u),v_i) \in  E_i$, edge $(u,v_i)$ has same label as $(\varphi_i(u),v_i)$.
  \end{itemize}
  Since any CSAN families (Definition~\ref{def:csan}) is entirely based on local constraints on labels (vertex label plus set of labels of the incident edges), we deduce that $F'$ is in $\mathcal{F}$. 
\end{proof}

\begin{figure}\label{fig:concreteglueing}
\centering
\begin{tikzpicture}[x=0.75pt,y=0.75pt,yscale=-1,xscale=1]

\draw    (110,239) -- (156.5,239) ;
\draw    (300.75,95.25) -- (263.75,123) ;
\draw    (263.75,123) -- (291.5,160) ;
\draw    (209,123) -- (255.5,123) ;
\draw  [fill={rgb, 255:red, 255; green, 255; blue, 255 }  ,fill opacity=1 ] (192.5,123) .. controls (192.5,118.44) and (196.19,114.75) .. (200.75,114.75) .. controls (205.31,114.75) and (209,118.44) .. (209,123) .. controls (209,127.56) and (205.31,131.25) .. (200.75,131.25) .. controls (196.19,131.25) and (192.5,127.56) .. (192.5,123) -- cycle ;
\draw  [fill={rgb, 255:red, 255; green, 255; blue, 255 }  ,fill opacity=1 ] (255.5,123) .. controls (255.5,118.44) and (259.19,114.75) .. (263.75,114.75) .. controls (268.31,114.75) and (272,118.44) .. (272,123) .. controls (272,127.56) and (268.31,131.25) .. (263.75,131.25) .. controls (259.19,131.25) and (255.5,127.56) .. (255.5,123) -- cycle ;
\draw  [fill={rgb, 255:red, 255; green, 255; blue, 255 }  ,fill opacity=1 ] (283.25,160) .. controls (283.25,155.44) and (286.94,151.75) .. (291.5,151.75) .. controls (296.06,151.75) and (299.75,155.44) .. (299.75,160) .. controls (299.75,164.56) and (296.06,168.25) .. (291.5,168.25) .. controls (286.94,168.25) and (283.25,164.56) .. (283.25,160) -- cycle ;
\draw  [fill={rgb, 255:red, 255; green, 255; blue, 255 }  ,fill opacity=1 ] (292.5,95.25) .. controls (292.5,90.69) and (296.19,87) .. (300.75,87) .. controls (305.31,87) and (309,90.69) .. (309,95.25) .. controls (309,99.81) and (305.31,103.5) .. (300.75,103.5) .. controls (296.19,103.5) and (292.5,99.81) .. (292.5,95.25) -- cycle ;

\draw    (64,123) -- (110.5,123) ;
\draw  [fill={rgb, 255:red, 255; green, 255; blue, 255 }  ,fill opacity=1 ] (55.75,123) .. controls (55.75,118.44) and (59.44,114.75) .. (64,114.75) .. controls (68.56,114.75) and (72.25,118.44) .. (72.25,123) .. controls (72.25,127.56) and (68.56,131.25) .. (64,131.25) .. controls (59.44,131.25) and (55.75,127.56) .. (55.75,123) -- cycle ;
\draw    (110.5,123) -- (157,123) ;
\draw  [fill={rgb, 255:red, 255; green, 255; blue, 255 }  ,fill opacity=1 ] (102.25,123) .. controls (102.25,118.44) and (105.94,114.75) .. (110.5,114.75) .. controls (115.06,114.75) and (118.75,118.44) .. (118.75,123) .. controls (118.75,127.56) and (115.06,131.25) .. (110.5,131.25) .. controls (105.94,131.25) and (102.25,127.56) .. (102.25,123) -- cycle ;
\draw  [fill={rgb, 255:red, 255; green, 255; blue, 255 }  ,fill opacity=1 ] (148.75,123) .. controls (148.75,118.44) and (152.44,114.75) .. (157,114.75) .. controls (161.56,114.75) and (165.25,118.44) .. (165.25,123) .. controls (165.25,127.56) and (161.56,131.25) .. (157,131.25) .. controls (152.44,131.25) and (148.75,127.56) .. (148.75,123) -- cycle ;

\draw    (132.95,41.71) -- (209.72,41.71) ;
\draw  [fill={rgb, 255:red, 255; green, 255; blue, 255 }  ,fill opacity=1 ] (119.33,41.71) .. controls (119.33,33.82) and (125.43,27.42) .. (132.95,27.42) .. controls (140.48,27.42) and (146.57,33.82) .. (146.57,41.71) .. controls (146.57,49.6) and (140.48,56) .. (132.95,56) .. controls (125.43,56) and (119.33,49.6) .. (119.33,41.71) -- cycle ;
\draw  [fill={rgb, 255:red, 255; green, 255; blue, 255 }  ,fill opacity=1 ] (196.1,41.71) .. controls (196.1,33.82) and (202.2,27.42) .. (209.72,27.42) .. controls (217.24,27.42) and (223.34,33.82) .. (223.34,41.71) .. controls (223.34,49.6) and (217.24,56) .. (209.72,56) .. controls (202.2,56) and (196.1,49.6) .. (196.1,41.71) -- cycle ;
\draw    (257.75,211.25) -- (220.75,239) ;
\draw    (220.75,239) -- (248.5,276) ;
\draw    (166,239) -- (212.5,239) ;
\draw  [fill={rgb, 255:red, 255; green, 255; blue, 255 }  ,fill opacity=1 ] (149.5,239) .. controls (149.5,234.44) and (153.19,230.75) .. (157.75,230.75) .. controls (162.31,230.75) and (166,234.44) .. (166,239) .. controls (166,243.56) and (162.31,247.25) .. (157.75,247.25) .. controls (153.19,247.25) and (149.5,243.56) .. (149.5,239) -- cycle ;
\draw  [fill={rgb, 255:red, 255; green, 255; blue, 255 }  ,fill opacity=1 ] (212.5,239) .. controls (212.5,234.44) and (216.19,230.75) .. (220.75,230.75) .. controls (225.31,230.75) and (229,234.44) .. (229,239) .. controls (229,243.56) and (225.31,247.25) .. (220.75,247.25) .. controls (216.19,247.25) and (212.5,243.56) .. (212.5,239) -- cycle ;
\draw  [fill={rgb, 255:red, 255; green, 255; blue, 255 }  ,fill opacity=1 ] (240.25,276) .. controls (240.25,271.44) and (243.94,267.75) .. (248.5,267.75) .. controls (253.06,267.75) and (256.75,271.44) .. (256.75,276) .. controls (256.75,280.56) and (253.06,284.25) .. (248.5,284.25) .. controls (243.94,284.25) and (240.25,280.56) .. (240.25,276) -- cycle ;
\draw  [fill={rgb, 255:red, 255; green, 255; blue, 255 }  ,fill opacity=1 ] (249.5,211.25) .. controls (249.5,206.69) and (253.19,203) .. (257.75,203) .. controls (262.31,203) and (266,206.69) .. (266,211.25) .. controls (266,215.81) and (262.31,219.5) .. (257.75,219.5) .. controls (253.19,219.5) and (249.5,215.81) .. (249.5,211.25) -- cycle ;

\draw  [fill={rgb, 255:red, 255; green, 255; blue, 255 }  ,fill opacity=1 ] (101.75,239) .. controls (101.75,234.44) and (105.44,230.75) .. (110,230.75) .. controls (114.56,230.75) and (118.25,234.44) .. (118.25,239) .. controls (118.25,243.56) and (114.56,247.25) .. (110,247.25) .. controls (105.44,247.25) and (101.75,243.56) .. (101.75,239) -- cycle ;
\draw  [dash pattern={on 4.5pt off 4.5pt}] (98.75,111.13) -- (168.75,111.13) -- (168.75,134.88) -- (98.75,134.88) -- cycle ;
\draw  [dash pattern={on 4.5pt off 4.5pt}] (192.5,110.75) -- (271.5,110.75) -- (271.5,134.5) -- (192.5,134.5) -- cycle ;
\draw  [dash pattern={on 4.5pt off 4.5pt}]  (163.85,52) -- (138.96,94.41) ;
\draw [shift={(137.44,97)}, rotate = 300.40999999999997] [fill={rgb, 255:red, 0; green, 0; blue, 0 }  ][line width=0.08]  [draw opacity=0] (8.93,-4.29) -- (0,0) -- (8.93,4.29) -- cycle    ;
\draw  [dash pattern={on 4.5pt off 4.5pt}]  (185.12,52.73) -- (209.27,93.69) ;
\draw [shift={(210.79,96.27)}, rotate = 239.48] [fill={rgb, 255:red, 0; green, 0; blue, 0 }  ][line width=0.08]  [draw opacity=0] (8.93,-4.29) -- (0,0) -- (8.93,4.29) -- cycle    ;
\draw    (178.5,139) -- (178.5,171) ;
\draw [shift={(178.5,174)}, rotate = 270] [fill={rgb, 255:red, 0; green, 0; blue, 0 }  ][line width=0.08]  [draw opacity=0] (8.93,-4.29) -- (0,0) -- (8.93,4.29) -- cycle    ;

\draw (59,112.4+8) node [anchor=north west][inner sep=0.75pt]    {$a$};
\draw (105,113.4+5) node [anchor=north west][inner sep=0.75pt]    {$b$};
\draw (152,113.4+8) node [anchor=north west][inner sep=0.75pt]    {$c$};
\draw (196,113.4+5) node [anchor=north west][inner sep=0.75pt]    {$b$};
\draw (259,113.4+8) node [anchor=north west][inner sep=0.75pt]    {$c$};
\draw (295,85.4+5) node [anchor=north west][inner sep=0.75pt]    {$d$};
\draw (287,150.4+8) node [anchor=north west][inner sep=0.75pt]    {$e$};
\draw (244,266.4+8) node [anchor=north west][inner sep=0.75pt]    {$e$};
\draw (252,201.4+5) node [anchor=north west][inner sep=0.75pt]    {$d$};
\draw (216,229.4+8) node [anchor=north west][inner sep=0.75pt]    {$c$};
\draw (153,229.4+5) node [anchor=north west][inner sep=0.75pt]    {$b$};
\draw (124.01,28.42+8) node [anchor=north west][inner sep=0.75pt]    {$C_{1}$};
\draw (201.61,29.65+8) node [anchor=north west][inner sep=0.75pt]    {$C_{2}$};
\draw (105,228.4+8) node [anchor=north west][inner sep=0.75pt]    {$a$};
\draw (41,77.4) node [anchor=north west][inner sep=0.75pt]    {$G_{1}$};
\draw (256,71.4) node [anchor=north west][inner sep=0.75pt]    {$G_{2}$};
\draw (167,196.4) node [anchor=north west][inner sep=0.75pt]    {$G'$};
\draw (168,10.4) node [anchor=north west][inner sep=0.75pt]    {$C$};
\draw (115.94,61.46) node [anchor=north west][inner sep=0.75pt]    {$\varphi _{1}$};
\draw (211.56,62.19) node [anchor=north west][inner sep=0.75pt]    {$\varphi _{2}$};
\end{tikzpicture}

      
        

\caption{Example of a glueing of two compatibles CSAN. The labeling in nodes of $G_1$, $G_2$ and $G'$ shows equalities between local $\lambda$ maps of these three CSAN.}
\end{figure}
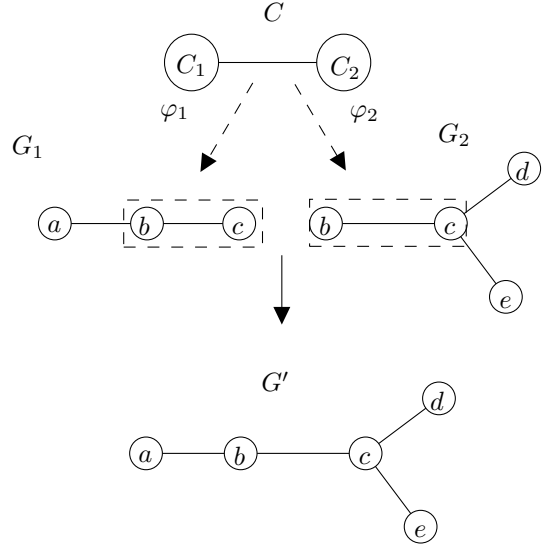

\subsection{$\G$-gadgets, gadget glueing and simulation of $\G$-networks}

We now give a precise meaning to the intuitively simple fact that, if a family of automata networks can coherently simulate a set of small building blocks (gates from $\G$), it should be able to simulate any automata network that can be built out of them ($\G$-networks).

The key idea here is that gates from $\G$ will be represented by networks of the family called $\G$-gadgets, and the wiring between gates to obtain a 
$\G$-network will translate into glueing between $\G$-gadgets. Following this idea there are two main conditions for the family to simulate any $\G$-network:
\begin{itemize}
\item the glueing of gadgets should be freely composable inside the family to allow the building of any $\G$-network;
\item the gadgets corresponding to gates from $\G$ should correctly and coherently simulate the functional relation between inputs and outputs given by their corresponding gate.
\end{itemize}

For clarity, we separate these conditions in two definitions.


We start by developing a definition for gadget glueing. Recall first that Definition~\ref{def:glueing} relies on the identification of a common dowel in the two networks to be glued. Here, as we want to mimic the wiring of gates which connects inputs to outputs, several copies of a fixed network called \textit{glueing interface} will be identified in each gadget, some of them corresponding to input, and the other ones to outputs. In this context,  the only glueing operations we will use are those where some output copies of the interface in a gadget $A$ are glued on input copies of the interface in a gadget $B$ and some input copies of the interface in $A$ are glued on output copies of the interface in $B$. Then, the global dowel used to formally apply Definition~\ref{def:glueing} is a disjoint union of the selected input/output copies of the interface. Figure~\ref{fig:gadgetglue} illustrates with the notations of the following Definition.

\newcommand\QF{Q_{\mathcal F}}

\begin{definition}[Glueing interface and gadgets]\label{def:gadgets-glueing}
  Let ${C=C_i\cup C_o}$ be a fixed set partitioned into two sets.
  A \emph{gadget} with \textit{glueing interface} ${C=C_i\cup C_o}$ is an automata network ${F:Q^{V_F}\rightarrow Q^{V_F}}$ together with two collections of injective maps ${\phi_{F,k}^i:C\rightarrow V_F}$ for ${k\in I(F)}$ and ${\phi_{F,k}^o:C\rightarrow V_F}$ for ${k\in O(F)}$ whose images in $V_F$ are pairwise disjoint and where $I(F)$ and $O(F)$ are disjoint sets which might be empty.

  Given two disjoint gadgets ${(F,(\phi_{F,k}^i),(\phi_{F,k}^o))}$ and ${(G,(\phi_{G,k}^i),(\phi_{G,k}^o))}$ with same alphabet and interface ${C=C_i\cup C_o}$, a \emph{gadget glueing} is a glueing of the form ${H = \glueing{F}{G}{C_F}{C_G}{\phi_F}{\phi_G}}$ defined as follows: 
  \begin{itemize}
  \item a choice of a set $A$ of inputs from $F$ and outputs from $G$ given by injective maps ${\sigma_F : A \to I(F)}$ and ${\sigma_G: A\to O(G)}$,
  \item a choice of a set $B$ of outputs from $F$ and inputs from $G$ given by injective maps ${\tau_F : B \to O(F)}$ and ${\tau_G: B\to I(G)}$ (the set $B$ is disjoint from $A$),
  \item ${C_F}$ is a disjoint union of $|A|$ copies of $C_i$, and $|B|$ copies of $C_o$: ${C_F = A\times C_i\cup B\times C_o}$,
  \item ${C_G}$ is a disjoint union of $|A|$ copies of $C_o$, and $|B|$ copies of $C_i$: ${C_G = A\times C_o\cup B\times C_i}$,
  \item ${\phi_F:C_F\cup C_G\to V_F}$ is such that ${\phi_F(a,c) = \phi_{F,\sigma_F(a)}^i(c)}$ for ${a\in A}$ and ${c\in C}$, and ${\phi_F(b,c) = \phi_{F,\tau_F(b)}^o(c)}$ for ${b\in B}$ and ${c\in C}$,
  \item ${\phi_G:C_F\cup C_G\to V_G}$ is such that ${\phi_G(a,c) = \phi_{G,\sigma_G(a)}^o(c)}$ for ${a\in A}$ and ${c\in C}$, and ${\phi_G(b,c) = \phi_{G,\tau_G(b)}^i(c)}$ for ${b\in B}$ and ${c\in C}$. 
  \end{itemize}
  The resulting network $H$ is a gadget with same alphabet and same interface with ${I(H) = I(F)\setminus\sigma_F(A) \cup I(G)\setminus\tau_G(B)}$ and ${O(H) = O(F)\setminus\tau_F(B)\cup O(F)\setminus\sigma_G(A)}$ and ${\phi_{H,k}^i}$ is $\phi_{F,k}^i$ when $k\in I(F)$ and ${\phi_{G,k}^i}$ when ${k\in I(G)}$, and ${\phi_{H,k}^o}$ is ${\phi_{F,k}^o}$ when ${k\in O(F)}$ and ${\phi_{G,k}^o}$ when ${k\in O(G)}$.

  Given a set of gadgets $X$ with same alphabet and interface, its closure by gadget glueing is the closure of $X$ by the following operations:
  \begin{itemize}
  \item add a disjoint copy of some gadget from the current set,
  \item add the disjoint union of two gadgets from the current set,
  \item add a gadget glueing of two gadgets from the current set.
  \end{itemize}
\end{definition}

\begin{remark}\label{rem:gadgets-glueing}
  The representation of the result of a gadget glueing can be easily computed from the two gadgets $F$ and $G$ and the choices of inputs/outputs given by maps $\sigma_F$, $\sigma_G$, $\tau_F$ and $\tau_G$: precisely, by definition of glueing (Definition~\ref{def:glueing}) the local map of each node of the result automata network is either a local map of $F$ (when in $V_F\setminus \phi_F(C_G)$ or in $C_F$) or a local map of $G$ (when in $V_G\setminus \phi_G(C_F)$ or in $C_G$). Note also that the closure by gadget glueing of a finite set of gadgets $X$ is always a set of automata networks of bounded degree.
\end{remark}

\begin{figure}\label{fig:gadgetglue}
  \begin{center}
    \begin{tikzpicture}[
      mnode/.style={circle,fill=white,draw=black,minimum size=.7cm},
      minode/.style={circle,fill=red!20!white,draw=black,minimum size=.7cm},
      mvnode/.style={circle,fill=green!20!white,draw=black,minimum size=.7cm},
      monode/.style={circle,fill=blue!20!white,draw=black,minimum size=.7cm}
      ]
      \draw[dashed] (1.5,2.5)--(1.5,-5);
      \draw (.5,2) node {inputs};
      \draw (2.5,2) node {outputs};
      \draw (1,0) -- (2,-1) -- (2,1);
      \draw (-1,0) node {$F$};
      \draw (.5,.7) node {$\sigma_F(A)$};
      \draw (2.5,-1.7) node {$\tau_F(B)$};
      \draw (0,0) node[minode] {$\alpha$} --  (1,0) node[monode] {$\beta$} -- (2,1) node[minode] {$\gamma$} -- (3,1) node[monode] {$\delta$};
      \draw (2,-1) node[minode] {$\epsilon$} -- (3,-1) node[monode] {$\zeta$};
      \begin{scope}[shift={(0,-2)}]
        \draw (1,-2) -- (2,-1);
      \draw (.5,-1.3) node {$\sigma_G(B)$};
      \draw (2.5,-1.7) node {$\tau_G(A)$};
        \draw (-1,-1) node {$G$};
        \draw (0,0) node[minode] {$a$} -- (1,0) node[monode] {$b$} -- (2,-1) node[minode] {$c$} -- (3,-1) node[monode] {$d$};
        \draw (0,-2) node[minode] {$e$} -- (1,-2) node[monode] {$f$};
      \end{scope}
      \begin{scope}[shift={(7,.5)}]
        \draw (-1.5,1) node {G};
        \draw (0,0) -- (0,-1);
        \draw (-1,0) -- (0,0);
        \draw[dotted] (0,0) -- (.5,1);
        \draw[dotted] (1,0) -- (1.5,1);
        \draw[dotted] (0,-1) -- (.5,-2);
        \draw[dotted] (1,-1) -- (1.5,-2);
        \draw (0,0) node[mvnode] {$c$} -- (1,0) node[mnode] {$d$};
        \draw (0,-1) node[mvnode] {$f$} -- (1,-1) node[mnode] {$e$};
        \draw (-2,0) node[mnode] {$a$} -- (-1,0) node[mnode] {$b$};
      \end{scope}
      \begin{scope}[shift={(7.5,.5)}]
        \draw (2.5,1) node {F};
        \draw (1,1) -- (1,-2);
        \draw (1,-2) -- (2,0);
        \draw (0,1) node[mnode] {$\alpha$} -- (1,1) node[mvnode] {$\beta$} -- (2,0) node[mnode] {$\gamma$} -- (3,0) node[mnode] {$\delta$};
        \draw (0,-2) node[mnode] {$\zeta$} -- (1,-2) node[mvnode] {$\epsilon$};
      \end{scope}
      \draw[very thick] (4,-2.5) -- (11,-2.5);
      \draw[very thick] (4,2.5) -- (4,-5.5);
      \begin{scope}[shift={(7,-4)}]
        \draw (-1.5,1) node {input};
        \draw (2.5,1) node {output};
        \draw (0,0) -- (0,-1);
        \draw (1,0) -- (1,-1);
        \draw (1,-1) -- (2,0);
        \draw (-1,0) -- (0,0);
        \draw (0,0) node[mnode] {$c$} -- (1,0) node[mnode] {$\beta$} -- (2,0) node[minode] {$\gamma$} -- (3,0) node[monode] {$\delta$};
        \draw (0,-1) node[mnode] {$f$} -- (1,-1) node[mnode] {$\epsilon$};
        \draw (-2,0) node[minode] {$a$} -- (-1,0) node[monode] {$b$};
      \end{scope}
    \end{tikzpicture}
  \end{center}
  \caption{Gadget glueing as in Definition~\ref{def:gadgets-glueing}. On the left, two gadgets with interface ${C=C_i\cup C_o}$ where $C_i$ part in each copy of the interface dowel is in red and $C_o$ part in blue. The gadget glueing is done with input ${\sigma_F(A)}$ on output ${\sigma_G(A)}$ (here $A$ is a singleton) and output ${\tau_F(B)}$ on input ${\tau_G(B)}$ ($B$ is also a singleton). On the upper right, a representation of the global glueing process where nodes in green are those in the copy of $C_F$ in $F$ or in the copy $C_G$ in $G$; dotted links show the bijection between the embeddings of ${C=C_F\cup C_G}$ into $V_F$ and $V_G$ via maps $\phi_F$ and $\phi_G$. On the lower right the resulting gadget with the same interface ${C=C_i\cup C_o}$ as the two initial gadgets.}
\end{figure}
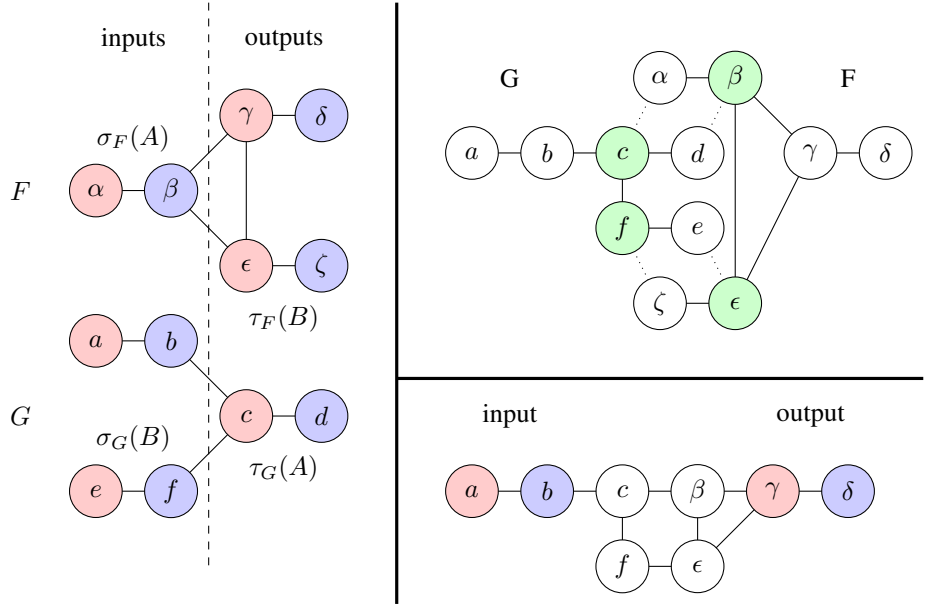

Lemma~\ref{lem:concreteglueing} gives sufficient conditions on a set of gadgets to have its closure by gadget glueing contained in a CSAN family.

\begin{lemma}\label{lem:csan-gadgets-glueing-closure}
  Fix some alphabet $Q$ and some glueing interface ${C=C_i\cup C_o}$ and some CSAN family $\mathcal{F}$.
  Let ${(G_n,\lambda_n,\rho_n)}$ for ${n\in S}$ be a set of CSAN belonging to $\mathcal{F}$ with associated global maps $F_n$. Let ${\phi_{F_n,k}^i}$ for ${k\in I(F_n)}$ and ${\phi_{F_n,k}^o}$ for ${k\in O(F_n)}$ be maps as in Definition~\ref{def:gadgets-glueing} so that ${(F_n,(\phi_{F_n,k}^i),(\phi_{F_n,k}^o))}$ is a gadget with interface ${C=C_i\cup C_o}$. Denote by $X$ the set of such gadgets. If the following conditions hold:
  \begin{itemize}
  \item the labeled graphs induced by ${\phi_{F_n,k}^i(C)}$ and by ${\phi_{F_n,k}^o(C)}$ in $G_n$ are all the same for all $n$ and $k$ with the identification of vertices given by the $\phi_{\ast,\ast}^\ast$ maps,
  \item ${N_{G_n}(\phi_{F_n,k}^i(C_o))\subseteq \phi_{F_n,k}^i(C)}$ for all ${n\in S}$ and all ${k\in I(F_n)}$,
  \item ${N_{G_n}(\phi_{F_n,k}^o(C_i))\subseteq \phi_{F_n,k}^o(C)}$ for all ${n\in S}$ and all ${k\in O(F_n)}$,
  \end{itemize}
  then the closure by gadget glueing of $X$ is included in $\mathcal{F}$.
\end{lemma}
\begin{proof}
  Consider first the gadget glueing $H$ of two gadgets $F_{n}$ and $F_{n'}$ from $X$. Following Definition~\ref{def:gadgets-glueing}, the global dowel ${C_{F_n}\cup C_{F_{n'}}}$ used in such a glueing is a disjoint union of copies of $C$, and its embedding $\phi_{F_n}$ in $G_n$ (resp. $\phi_{F_{n'}}$ in $G_{n'}$) is a disjoint union of maps $\phi_{F_n,\ast}^\ast$ (resp. ${\phi_{F_{n'},\ast}^\ast}$). Therefore the three conditions of Lemma~\ref{lem:concreteglueing} follow from the three conditions of the hypothesis on gadgets from $X$ and we deduce that $H$ belongs to family $\mathcal{F}$. Moreover, it is clear that gadget $H$ then also verifies the three conditions from the hypothesis, and adding a copy of any gadget to the set also verifies the conditions. We deduce that the closure by gadget glueing of $X$ is included in ${\mathcal{F}}$.
   
\end{proof}

The second key aspect to have a coherent set $X$ of $\G$-gadgets is of dynamical nature: there must exists a collection of pseudo-orbits on each gadget satisfying suitable conditions to permit application of Lemma~\ref{lem:pseudo-orbit-glueing} for any gadget glueing in the closure of $X$; moreover, these pseudo-orbits must simulate via an appropriate coding the input/output relations of each gate ${g\in\G}$ in the corresponding gadget. To obtain this, we rely on a standard set of traces on the glueing interface that must be respected on any copy of it in any gadget.

\begin{definition}[Coherent $\G$-gadgets]\label{def:coherent-gadgets}
  Let $\G$ be any set of finite maps over alphabet $Q$ and let ${\mathcal F}$ be any set of abstract automata networks over alphabet $\QF$.
  We say  ${\mathcal F}$ has \emph{coherent $\G$-gadgets} if there exists:
    \begin{itemize}
    \item a unique glueing interface ${C=C_i\cup C_o}$, 
    \item a set $X$ of gadgets ${(F_g,(\phi_{g,k}^i)_{1\leq k\leq i(g)},(\phi_{g,k}^o)_{1\leq k\leq o(g)})}$ for each $g\in\G$ where ${F_g:\QF^{V_g}\rightarrow\QF^{V_g}\in\mathcal{F}}$ and sets $V_g$ and $C$ are pairwise disjoint, and the closure of $X$ by gadget glueing is contained in $\mathcal F$,
    \item a \emph{state configuration} ${s_q\in \QF^{C}}$ for each ${q\in Q}$ such that ${q\mapsto s_q}$ is an injective map,
    \item a \emph{context configuration} ${c_g\in \QF^{{\hat V}_g}}$ for each ${g\in\G}$ where ${{\hat V}_g=V_g\setminus\bigl(\cup_k\phi_{g,k}^i(C)\cup_k \phi_{g,k}^o(C)\bigr)}$,
    \item a time constant $T$,
    \item a \emph{standard trace} $\tau_{q,q'}\in (\QF^C)^{\{0,\ldots,T\}}$ for each pair ${q,q'\in Q}$ such that ${\tau_{q,q'}(0)=s_{q}}$ and ${\tau_{q,q'}(T)=s_{q'}}$,
    \item for each ${g\in\G}$ and for any uples of states ${q_{i,1},\ldots,q_{i,i(g)}\in Q}$ and ${q_{o,1},\ldots,q_{o,o(g)}\in Q}$ and ${q_{i,1}',\ldots,q_{i,i(g)}'\in Q}$ and ${q_{o,1}',\ldots,q_{o,o(g)}'\in Q}$ such that ${g(q_{i,1},\ldots,q_{i,i(g)}) = (q'_{o,1},\ldots,q'_{o,o(g)})}$, a $P_g$-pseudo-orbit ${(x^t)_{0\leq t\leq T}}$ of $F_g$ with ${P_g = \bigcup_{1\leq k\leq i(g)}\phi_{g,k}^i(C_o)\cup \bigcup_{1\leq k\leq o(g)}\phi_{g,k}^o(C_i)}$ and with
      \begin{itemize}
      \item for each ${1\leq k\leq i(g)}$, the trace ${t\mapsto x^t_{\phi_{g,k}^i(C)}}$ is exactly $\tau_{q_{i,k},q_{i,k}'}$,
      \item for each ${1\leq k\leq o(g)}$, the trace ${t\mapsto x^t_{\phi_{g,k}^o(C)}}$ is exactly $\tau_{q_{o,k},q_{o,k}'}$,
      \item ${x^0_{{\hat V}_g} = x^T_{{\hat V}_g} = c_g}$.
      \end{itemize}
    \end{itemize}
\end{definition}

We can now state the key lemma of our framework: having coherent $\G$-gadgets is sufficient to simulate the whole family of $\G$-networks.

\begin{lemma}\label{lem:from-gadgets-to-networks}
  Let $\G$ be a set of irreducible gates.
  If an abstract automata network family $\mathcal{F}$ has coherent $\G$-gadgets then it contains a subfamily of bounded degree networks with the canonical bounded degree representation ${(\mathcal{F}_0,\mathcal{F}_0^*)}$ that simulates $\Gamma(\mathcal{G})$ in time $T$ and space $S$ where $T$ is a constant map and $S$ is bounded by a linear map.
\end{lemma}
\begin{proof}
  We take the notations of Definition~\ref{def:coherent-gadgets}.
  To any $\G$-network $F$ with set of nodes $V$ given as in Definition~\ref{def:g-network} by a list of gates ${g_1,\ldots,g_k\in\G}$ and maps $\alpha$ and $\beta$ (see Remark~\ref{rem:represent-g-networks}) we associate an automata network from $\mathcal{F}$ as follows. First, let $(F_{g_i})_{1\leq i\leq k}$ be the gadgets corresponding to gates $g_i$ and suppose they are all disjoint (by taking disjoint copies when necessary). Then, start from the gadget $F_1=F_{g_1}$ and for any ${1\leq i< k}$ we define ${F_{i+1}}$ as the gadget glueing of $F_i$ and $F_{g_{i+1}}$ on the input/outputs as prescribed by maps $\alpha$ and $\beta$. More precisely, the gadget glueing select the set of inputs ${(j,k)}$ with ${1\leq j\leq i}$ and ${1\leq k\leq i(g_j)}$ such that ${\beta(\alpha(j,k))=(i+1,k')}$ for some ${1\leq k'\leq o(g_{i+1})}$ and glue them on their corresponding output ${(i+1,k')}$ of $g_{i+1}$ (precisely, through maps $\sigma_{F_i}$ and $\sigma_{F_{g_{i+1}}}$ of domain $A_{i+1}$ playing the role of maps $\sigma_F$ and $\sigma_G$ of Definition~\ref{def:gadgets-glueing}), and, symmetrically, selects the inputs ${(i+1,k)}$ with ${1\leq k\leq i(g_{i+1})}$ such that ${\beta(\alpha(i+1,k))=(j,k')}$ for some ${1\leq j\leq i}$ and ${1\leq k'\leq o(g_{j})}$ and glue their corresponding output ${(j,k')}$ (precisely, through maps $\tau_{F_{g_{i+1}}}$ and $\tau_{F_i}$ of domain $B_{i+1}$ playing the role of maps $\tau_G$ and $\tau_F$ from Definition~\ref{def:gadgets-glueing}). If both of these sets of inputs/outputs are empty, the gadget glueing is replaced by a simple disjoint union.

  The final gadget $F_k$ has no input and no output, and a representation of it as a pair graph and local maps can be constructed in \DLOG{}, because the local map of each of its nodes is independent of the glueing sequence above and completely determined by the gadget $F_{g_j}$ it belongs to and whether the node is inside some input or some output dowel or not (see Reamrk~\ref{rem:gadgets-glueing}).

  It now remains to show that the automata network $F_k$ simulates $F$. To fix notations, let $V_k$ be the set of nodes of $F_k$. For each ${v\in V}$, define ${D_v\subseteq V_k}$ as the copy of the dowel that correspond to node $v$ of $F$, \textit{i.e.} that was produced in the gadget glueing of ${F_i}$ with $F_{g_{i+1}}$ for $i$ such that ${\beta(v)=(i+1,k')}$ for some ${1\leq k'\leq o(g_{i+1})}$ (or symmetrically $\alpha(i+1,k)=v$ for some ${1\leq k\leq i(g_j)}$). More precisely, if $a\in A_{i+1}$ is such that ${\sigma_{F_{g_{i+1}}}(a)=(i+1,k')}$ then $D_v = \{a\}\times C$ (symmetrically if $b\in B_{i+1}$ is such that ${\tau_{F_{g_{i+1}}}(b)=(i+1,k)}$ then $D_v = \{b\}\times C$). Also denote by $\rho_v : D_v\to C$ the map such that ${\rho_v(a,c)=c}$ for all ${c\in C}$ (symmetrically, ${\rho_v(b,c)=c}$). With these notations, we have 
  \[V_k = \bigcup_{v\in V} D_v\cup\bigcup_{1\leq i\leq k}\hat{V}_{g_i}\]
  Let us define the block embedding $\phi : Q^V\to Q_{\mathcal{F}}^{V_k}$ as follows 
  \[\phi(x)(v') =
    \begin{cases}
      s_{x_v}(\rho_v(v'))&\text{ if }v'\in D_v,\\
      c_{g_i}(v')&\text{ if }v'\in\hat{V}_{g_i}.\\
    \end{cases}
  \]
  for any ${x\in Q^V}$ and any ${v'\in V_k}$, where $s_q$ for ${q\in Q}$ are the state configurations and $c_g$ for ${g\in\G}$ are the context configurations granted by Definition~\ref{def:coherent-gadgets}. Note that $\phi$ is injective because the map ${q\mapsto s_q}$ is injective. By inductive applications of Lemma~\ref{lem:pseudo-orbit-glueing}, the $P_{g_i}$-pseudo-orbits of each ${F_{g_i}}$ from Definition~\ref{def:coherent-gadgets} can be glued together to form valid orbits of $F_k$ that start from any configuration ${\phi(x)}$ with ${x\in Q^V}$ and ends after $T$ steps in a configuration ${\phi(y)}$ for some ${y\in Q^V}$ which verifies ${y=F(x)}$. Said differently, we have the following equality on $Q^V$: 
  \[\phi\circ F = F_k^T\circ\phi.\]
  Note that $T$ is a constant and that the size of $V_k$ is at most linear in the size of $V$. The lemma follows. 
\end{proof}

\begin{remark}
  Note that in Lemma~\ref{lem:from-gadgets-to-networks} above, the block embedding that is constructed can be viewed as a collection of blocs of bounded size that encode all the information plus a context (see Remark~\ref{rem:bloc-simul}).
\end{remark}

In the case of CSAN families and using Remark~\ref{rem:csan-bounded-degree} we have a simpler formulation of the Lemma.

\begin{corollary}\label{cor:gadgets-mon-csan}
  If $\G$ in a set of irreducible gates and $\mathcal{F}$ a CSAN family which has coherent $\G$-gadgets then $\mathcal{F}$ simulates $\Gamma(\mathcal{G})$ in time $T$ and space $S$ where $T$ is a constant map and $S$ is bounded by a linear map.  
\end{corollary}

\newcommand\Gmon{\G_{m}}
\newcommand\Gmontwo{\G_{m,2}}
\subsection{$\Gmon$-networks and $\Gmontwo$-networks as standard universal families}
\label{sec:gmonuniv}
\newcommand\gateOR{\mathrm{OR}}
\newcommand\gateAND{\mathrm{AND}}
\newcommand\gateCOPY{\mathrm{COPY}}
Let $i,o \in \{1,2\}$ be two numbers. We define the functions $\gateOR, \gateAND: \{0,1\}^i \to \{0,1\}^o$ where $\gateOR(x) = \max(x)$ and $\gateAND(x) = \min(x)$. Note that in the case in which $i = o = 1$ we have $\gateAND(x) = \gateOR(x) = \text{Id}(x) = x$ and also in the case $i = 1$ and $o=2$ we have that $\gateAND(x)=\gateOR(x) = (x, x).$ We define the set $\Gmon = \{\gateAND, \gateOR\}.$ Observe that in this case $o$ and $i$ may take different values. In addition, we define the set $\Gmontwo$ in which we fix $i=o=2.$

It is folklore knowledge that monotone Boolean networks (with AND/OR local maps) can simulate any other network.
Here we make this statement precise within our formalism: $\Gmon$-networks are strongly universal.
Note that there is more work than the classical circuit transformations involving monotone gates because we need to obtain a simulation of any automata network via block embedding. In particular we need to build monotone circuitry that is synchronized and reusable (\textit{i.e.} that can be reinitialized to a standard state before starting a computation on a new input). Moreover, our definitions requires a production of $\Gmon$-networks in \DLOG{}.
The main ingredient for establishing universality of $\Gmon$-networks is an efficient circuit transformation due to Greenlaw, Hoover and Ruzzo in \cite[Theorems 6.2.3 to 6.2.5]{Greenlaw_1995}.
Let us start by proving that this family is strongly universal, which is slightly simpler to prove.

\begin{theorem}\label{theo:gmon-univ}
  The family $\Gamma(\mathcal{G}_m)$ of all $\Gmon$-networks is strongly universal.
\end{theorem}
\begin{proof}
Let $Q$ an arbitrary alphabet and $F:Q^{n} \to Q^{n}$ an arbitrary automata network on alphabet $Q$ such that the communication graph of $F$ has maximum degree $\Delta.$  Let $C:\{0,1\}^{n} \to \{0,1\}^{n}$ be a constant depth circuit representing $F$. Let us assume that $C$  has only  OR, AND and NOT gates. We can also assume that $C$ is synchronous because, as its depth does not depend on the size of the circuit, one can always add fanin one and fanout one OR gates in order to modify layer structure. We are going to use a very similar transformation to the one proposed in \cite[Theorem 6.2.3]{Greenlaw_1995} in order to efficiently construct an automata network in $\Gamma(\mathcal{G}_m)$. In fact, we are going to duplicate the original circuit by considering the coding $x \in \{0,1\} \to (x, 1-x) \in \{0,1\}^{2}.$ Roughly, each gate will have a positive part (which is essentially a copy) and a negative part which is produces the negation of the original output by using De Morgan's laws.  More precisely, we are going to replace each gate in the network by the gadgets shown in Figure \ref{fig:ANDORgatesgad}.  The main idea is that one can represent the function $x \wedge y$ by the coding: $(x \wedge y, \overline{x} \vee \overline{y})$ and $x \vee y$ by the coding: $(x \vee y, \overline{x} \wedge \overline{y}).$ In addition, each time there is a NOT gate, we replace it by a fan in $1$ fan out $1$ OR gadget and we connect positive outputs to negative inputs in the next layer and negative outputs to positive inputs as it is shown in Figure \ref{fig:NOTgatesgad}.  
We are going to call $C^*$ to the circuit constructed by latter transformations. Observe that $C^*$ is such that it holds on $\{0,1\}^i$: 
  \[\phi\circ C = C^*\circ\phi\]
  where $\phi : \{0,1\}^n\to\{0,1\}^{2n}$ is defined for any $n$ by ${\phi(x)_{2j}=x_j}$ and ${\phi(x)_{2j+1}=\neg x_j}$.

  Now consider the coding map ${m_Q:Q\to\{0,1\}^k}$ and let ${n=k|V|}$. Build from $C^*$ the $\Gmontwo$-network $F^*:\{0,1\}^{V^{+}} \to \{0,1\}^{V^{+}}$ that correspond to it (gate by gate) and where the output $j$ is wired to input $j$ for all ${1\leq j\leq 2n}$. Define a block embedding of $Q^V$ into ${\{0,1\}^{V^+}}$ as follows (see Remark~\ref{rem:bloc-simul}):
  \begin{itemize}
  \item for each $v\in V$ let $D_v$ be the set of input nodes in $F^*$ that code $v$ (via $m_Q$ and then double railed logic),
  \item let $C=V^+\setminus \bigcup_v D_v$ be the remaining context block,
  \item let ${p_{v,q}\in \{0,1\}^{D_v}}$ be the pattern coding node $v$ in state $q$,
  \item let $p_C = 0^C$ be the context pattern,
  \item let $\phi: Q^V\to \{0,1\}^{V^+}$ be the associated block embedding map.
  \end{itemize}
  We claim that $F^*$ simulates $F$ via block embedding $\phi$ with time constant equal to the depth of $C^*$ plus $1$. Indeed, $F^* $ can be seen as a directed cycle of $N$ layers where layer $L_{i+1\bmod N}$ only depends on layer $i$. The block embedding is such that for any configuration ${x\in Q^V}$, ${\phi(x)}$ is $0$ on each layer except the layer containing the inputs. On configurations where a single layer $L_i$ is non-zero, $F^*$ will produce a configuration where the only non-zero layer is $L_{i+1\bmod N}$. From there, it follows by construction of $F^*$ that ${\phi\circ F (x) = (F^*)^N\circ\phi(x)}$ for all ${x\in Q^V}$.

The fact that construction is obtainable in $\DLOG$ follows from the same reasoning used to show in \cite[Theorem 6.2.3]{Greenlaw_1995}. In fact, authors show that reduction is actually better as they show it is $\NC^{1}.$  
\end{proof}
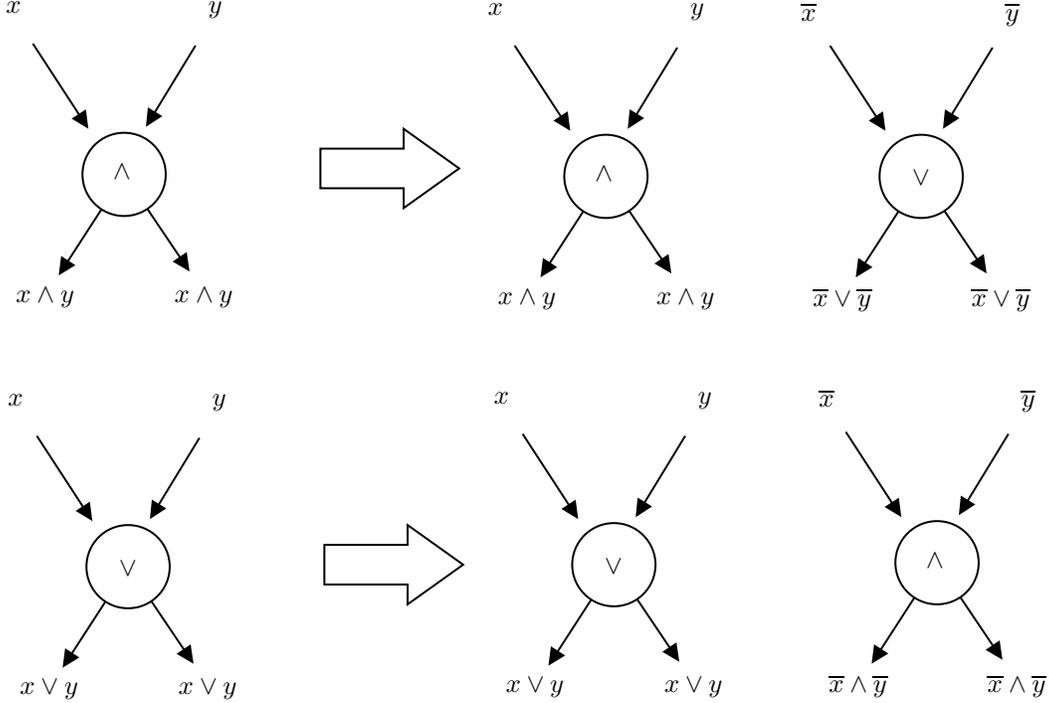
\begin{figure}

\centering

\tikzset{every picture/.style={line width=0.75pt}} 

\begin{tikzpicture}[x=0.75pt,y=0.75pt,yscale=-1,xscale=1]


\draw    (38,34) -- (64.36,74.49) ;
\draw [shift={(66,77)}, rotate = 236.93] [fill={rgb, 255:red, 0; green, 0; blue, 0 }  ][line width=0.08]  [draw opacity=0] (8.93,-4.29) -- (0,0) -- (8.93,4.29) -- cycle    ;
\draw    (120,35) -- (96.56,73.44) ;
\draw [shift={(95,76)}, rotate = 301.37] [fill={rgb, 255:red, 0; green, 0; blue, 0 }  ][line width=0.08]  [draw opacity=0] (8.93,-4.29) -- (0,0) -- (8.93,4.29) -- cycle    ;
\draw    (84,100) -- (115.32,146.51) ;
\draw [shift={(117,149)}, rotate = 236.04] [fill={rgb, 255:red, 0; green, 0; blue, 0 }  ][line width=0.08]  [draw opacity=0] (8.93,-4.29) -- (0,0) -- (8.93,4.29) -- cycle    ;
\draw    (84,100) -- (52.63,148.48) ;
\draw [shift={(51,151)}, rotate = 302.90999999999997] [fill={rgb, 255:red, 0; green, 0; blue, 0 }  ][line width=0.08]  [draw opacity=0] (8.93,-4.29) -- (0,0) -- (8.93,4.29) -- cycle    ;
\draw  [fill={rgb, 255:red, 255; green, 255; blue, 255 }  ,fill opacity=1 ] (63,100) .. controls (63,88.4) and (72.4,79) .. (84,79) .. controls (95.6,79) and (105,88.4) .. (105,100) .. controls (105,111.6) and (95.6,121) .. (84,121) .. controls (72.4,121) and (63,111.6) .. (63,100) -- cycle ;

\draw    (40,232) -- (66.36,272.49) ;
\draw [shift={(68,275)}, rotate = 236.93] [fill={rgb, 255:red, 0; green, 0; blue, 0 }  ][line width=0.08]  [draw opacity=0] (8.93,-4.29) -- (0,0) -- (8.93,4.29) -- cycle    ;
\draw    (122,233) -- (98.56,271.44) ;
\draw [shift={(97,274)}, rotate = 301.37] [fill={rgb, 255:red, 0; green, 0; blue, 0 }  ][line width=0.08]  [draw opacity=0] (8.93,-4.29) -- (0,0) -- (8.93,4.29) -- cycle    ;
\draw    (86,298) -- (117.32,344.51) ;
\draw [shift={(119,347)}, rotate = 236.04] [fill={rgb, 255:red, 0; green, 0; blue, 0 }  ][line width=0.08]  [draw opacity=0] (8.93,-4.29) -- (0,0) -- (8.93,4.29) -- cycle    ;
\draw    (86,298) -- (54.63,346.48) ;
\draw [shift={(53,349)}, rotate = 302.90999999999997] [fill={rgb, 255:red, 0; green, 0; blue, 0 }  ][line width=0.08]  [draw opacity=0] (8.93,-4.29) -- (0,0) -- (8.93,4.29) -- cycle    ;
\draw  [fill={rgb, 255:red, 255; green, 255; blue, 255 }  ,fill opacity=1 ] (65,298) .. controls (65,286.4) and (74.4,277) .. (86,277) .. controls (97.6,277) and (107,286.4) .. (107,298) .. controls (107,309.6) and (97.6,319) .. (86,319) .. controls (74.4,319) and (65,309.6) .. (65,298) -- cycle ;
\draw   (183,87) -- (225,87) -- (225,77) -- (253,97) -- (225,117) -- (225,107) -- (183,107) -- cycle ;
\draw    (281,35) -- (307.36,75.49) ;
\draw [shift={(309,78)}, rotate = 236.93] [fill={rgb, 255:red, 0; green, 0; blue, 0 }  ][line width=0.08]  [draw opacity=0] (8.93,-4.29) -- (0,0) -- (8.93,4.29) -- cycle    ;
\draw    (363,36) -- (339.56,74.44) ;
\draw [shift={(338,77)}, rotate = 301.37] [fill={rgb, 255:red, 0; green, 0; blue, 0 }  ][line width=0.08]  [draw opacity=0] (8.93,-4.29) -- (0,0) -- (8.93,4.29) -- cycle    ;
\draw    (327,101) -- (358.32,147.51) ;
\draw [shift={(360,150)}, rotate = 236.04] [fill={rgb, 255:red, 0; green, 0; blue, 0 }  ][line width=0.08]  [draw opacity=0] (8.93,-4.29) -- (0,0) -- (8.93,4.29) -- cycle    ;
\draw    (327,101) -- (295.63,149.48) ;
\draw [shift={(294,152)}, rotate = 302.90999999999997] [fill={rgb, 255:red, 0; green, 0; blue, 0 }  ][line width=0.08]  [draw opacity=0] (8.93,-4.29) -- (0,0) -- (8.93,4.29) -- cycle    ;
\draw  [fill={rgb, 255:red, 255; green, 255; blue, 255 }  ,fill opacity=1 ] (306,101) .. controls (306,89.4) and (315.4,80) .. (327,80) .. controls (338.6,80) and (348,89.4) .. (348,101) .. controls (348,112.6) and (338.6,122) .. (327,122) .. controls (315.4,122) and (306,112.6) .. (306,101) -- cycle ;

\draw    (440,35) -- (466.36,75.49) ;
\draw [shift={(468,78)}, rotate = 236.93] [fill={rgb, 255:red, 0; green, 0; blue, 0 }  ][line width=0.08]  [draw opacity=0] (8.93,-4.29) -- (0,0) -- (8.93,4.29) -- cycle    ;
\draw    (522,36) -- (498.56,74.44) ;
\draw [shift={(497,77)}, rotate = 301.37] [fill={rgb, 255:red, 0; green, 0; blue, 0 }  ][line width=0.08]  [draw opacity=0] (8.93,-4.29) -- (0,0) -- (8.93,4.29) -- cycle    ;
\draw    (486,101) -- (517.32,147.51) ;
\draw [shift={(519,150)}, rotate = 236.04] [fill={rgb, 255:red, 0; green, 0; blue, 0 }  ][line width=0.08]  [draw opacity=0] (8.93,-4.29) -- (0,0) -- (8.93,4.29) -- cycle    ;
\draw    (486,101) -- (454.63,149.48) ;
\draw [shift={(453,152)}, rotate = 302.90999999999997] [fill={rgb, 255:red, 0; green, 0; blue, 0 }  ][line width=0.08]  [draw opacity=0] (8.93,-4.29) -- (0,0) -- (8.93,4.29) -- cycle    ;
\draw  [fill={rgb, 255:red, 255; green, 255; blue, 255 }  ,fill opacity=1 ] (465,101) .. controls (465,89.4) and (474.4,80) .. (486,80) .. controls (497.6,80) and (507,89.4) .. (507,101) .. controls (507,112.6) and (497.6,122) .. (486,122) .. controls (474.4,122) and (465,112.6) .. (465,101) -- cycle ;

\draw    (285,231) -- (311.36,271.49) ;
\draw [shift={(313,274)}, rotate = 236.93] [fill={rgb, 255:red, 0; green, 0; blue, 0 }  ][line width=0.08]  [draw opacity=0] (8.93,-4.29) -- (0,0) -- (8.93,4.29) -- cycle    ;
\draw    (367,232) -- (343.56,270.44) ;
\draw [shift={(342,273)}, rotate = 301.37] [fill={rgb, 255:red, 0; green, 0; blue, 0 }  ][line width=0.08]  [draw opacity=0] (8.93,-4.29) -- (0,0) -- (8.93,4.29) -- cycle    ;
\draw    (331,297) -- (362.32,343.51) ;
\draw [shift={(364,346)}, rotate = 236.04] [fill={rgb, 255:red, 0; green, 0; blue, 0 }  ][line width=0.08]  [draw opacity=0] (8.93,-4.29) -- (0,0) -- (8.93,4.29) -- cycle    ;
\draw    (331,297) -- (299.63,345.48) ;
\draw [shift={(298,348)}, rotate = 302.90999999999997] [fill={rgb, 255:red, 0; green, 0; blue, 0 }  ][line width=0.08]  [draw opacity=0] (8.93,-4.29) -- (0,0) -- (8.93,4.29) -- cycle    ;
\draw  [fill={rgb, 255:red, 255; green, 255; blue, 255 }  ,fill opacity=1 ] (310,297) .. controls (310,285.4) and (319.4,276) .. (331,276) .. controls (342.6,276) and (352,285.4) .. (352,297) .. controls (352,308.6) and (342.6,318) .. (331,318) .. controls (319.4,318) and (310,308.6) .. (310,297) -- cycle ;

\draw    (448,230) -- (474.36,270.49) ;
\draw [shift={(476,273)}, rotate = 236.93] [fill={rgb, 255:red, 0; green, 0; blue, 0 }  ][line width=0.08]  [draw opacity=0] (8.93,-4.29) -- (0,0) -- (8.93,4.29) -- cycle    ;
\draw    (530,231) -- (506.56,269.44) ;
\draw [shift={(505,272)}, rotate = 301.37] [fill={rgb, 255:red, 0; green, 0; blue, 0 }  ][line width=0.08]  [draw opacity=0] (8.93,-4.29) -- (0,0) -- (8.93,4.29) -- cycle    ;
\draw    (494,296) -- (525.32,342.51) ;
\draw [shift={(527,345)}, rotate = 236.04] [fill={rgb, 255:red, 0; green, 0; blue, 0 }  ][line width=0.08]  [draw opacity=0] (8.93,-4.29) -- (0,0) -- (8.93,4.29) -- cycle    ;
\draw    (494,296) -- (462.63,344.48) ;
\draw [shift={(461,347)}, rotate = 302.90999999999997] [fill={rgb, 255:red, 0; green, 0; blue, 0 }  ][line width=0.08]  [draw opacity=0] (8.93,-4.29) -- (0,0) -- (8.93,4.29) -- cycle    ;
\draw  [fill={rgb, 255:red, 255; green, 255; blue, 255 }  ,fill opacity=1 ] (473,296) .. controls (473,284.4) and (482.4,275) .. (494,275) .. controls (505.6,275) and (515,284.4) .. (515,296) .. controls (515,307.6) and (505.6,317) .. (494,317) .. controls (482.4,317) and (473,307.6) .. (473,296) -- cycle ;

\draw   (185,287) -- (227,287) -- (227,277) -- (255,297) -- (227,317) -- (227,307) -- (185,307) -- cycle ;

\draw (72+5,88.4+5) node [anchor=north west][inner sep=0.75pt]    {$\land $};
\draw (23,11.4) node [anchor=north west][inner sep=0.75pt]    {$x$};
\draw (125,11.4) node [anchor=north west][inner sep=0.75pt]    {$y$};
\draw (28,155.4) node [anchor=north west][inner sep=0.75pt]    {$x\land y$};
\draw (108,155.4) node [anchor=north west][inner sep=0.75pt]    {$x\land y$};
\draw (75+5,288.4+5) node [anchor=north west][inner sep=0.75pt]    {$\lor $};
\draw (24,209.4) node [anchor=north west][inner sep=0.75pt]    {$x$};
\draw (127,209.4) node [anchor=north west][inner sep=0.75pt]    {$y$};
\draw (30,353.4) node [anchor=north west][inner sep=0.75pt]    {$x\lor y$};
\draw (110,353.4) node [anchor=north west][inner sep=0.75pt]    {$x\lor y$};
\draw (351,156.4) node [anchor=north west][inner sep=0.75pt]    {$x\land y$};
\draw (271,156.4) node [anchor=north west][inner sep=0.75pt]    {$x\land y$};
\draw (368,12.4) node [anchor=north west][inner sep=0.75pt]    {$y$};
\draw (266,12.4) node [anchor=north west][inner sep=0.75pt]    {$x$};
\draw (315+5,89.4+5) node [anchor=north west][inner sep=0.75pt]    {$\land $};
\draw (510,156.4) node [anchor=north west][inner sep=0.75pt]    {$\overline{x} \lor \overline{y}$};
\draw (430,156.4) node [anchor=north west][inner sep=0.75pt]    {$\overline{x} \lor \overline{y}$};
\draw (527,12.4) node [anchor=north west][inner sep=0.75pt]    {$\overline{y}$};
\draw (424,12.4) node [anchor=north west][inner sep=0.75pt]    {$\overline{x}$};
\draw (475+5,91.4+5) node [anchor=north west][inner sep=0.75pt]    {$\lor $};
\draw (355,352.4) node [anchor=north west][inner sep=0.75pt]    {$x\lor y$};
\draw (275,352.4) node [anchor=north west][inner sep=0.75pt]    {$x\lor y$};
\draw (372,208.4) node [anchor=north west][inner sep=0.75pt]    {$y$};
\draw (269,208.4) node [anchor=north west][inner sep=0.75pt]    {$x$};
\draw (320+5,287.4+5) node [anchor=north west][inner sep=0.75pt]    {$\lor $};
\draw (518,351.4) node [anchor=north west][inner sep=0.75pt]    {$\overline{x} \land \overline{y}$};
\draw (438,351.4) node [anchor=north west][inner sep=0.75pt]    {$\overline{x} \land \overline{y}$};
\draw (535,207.4) node [anchor=north west][inner sep=0.75pt]    {$\overline{y}$};
\draw (433,207.4) node [anchor=north west][inner sep=0.75pt]    {$\overline{x}$};
\draw (482+5,284.4+5) node [anchor=north west][inner sep=0.75pt]    {$\land $};

\end{tikzpicture}
\caption{AND and OR gadgets for simulating AND/OR gates with fanin and fanout 2. For other values of fanin and fanout gadgets are the same but considering different number of inputs/outputs}
\label{fig:ANDORgatesgad}
\end{figure}
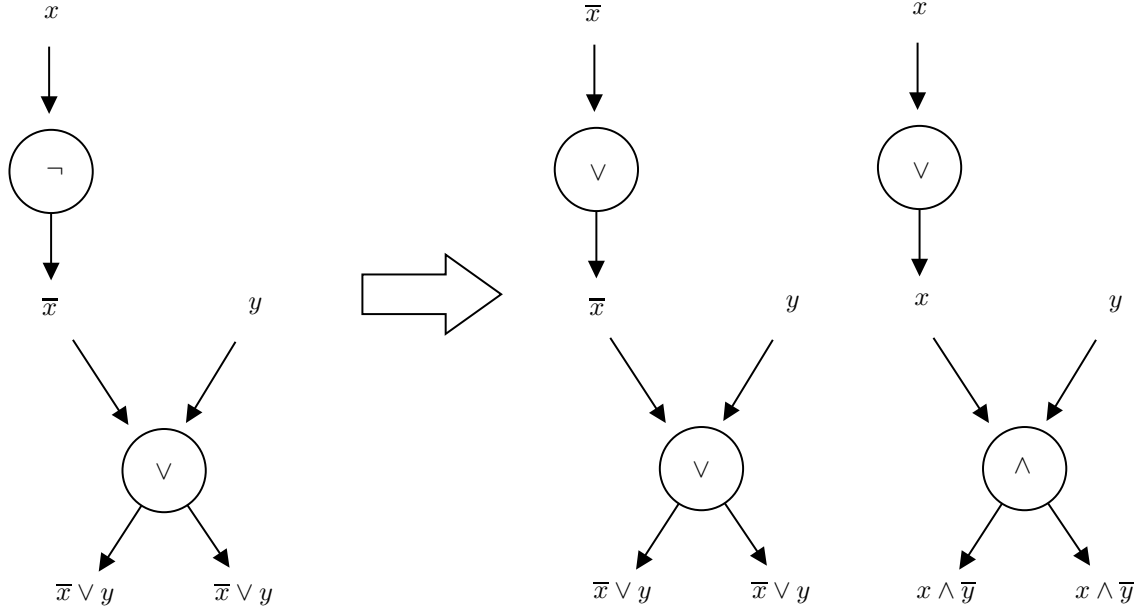
\begin{figure}
\centering

\tikzset{every picture/.style={line width=0.75pt}} 

\begin{tikzpicture}[x=0.75pt,y=0.75pt,yscale=-1,xscale=1]

\draw   (202,181) -- (244,181) -- (244,171) -- (272,191) -- (244,211) -- (244,201) -- (202,201) -- cycle ;
\draw    (44,66) -- (44,97) ;
\draw [shift={(44,100)}, rotate = 270] [fill={rgb, 255:red, 0; green, 0; blue, 0 }  ][line width=0.08]  [draw opacity=0] (8.93,-4.29) -- (0,0) -- (8.93,4.29) -- cycle    ;
\draw    (45,130) -- (45,181) ;
\draw [shift={(45,184)}, rotate = 270] [fill={rgb, 255:red, 0; green, 0; blue, 0 }  ][line width=0.08]  [draw opacity=0] (8.93,-4.29) -- (0,0) -- (8.93,4.29) -- cycle    ;
\draw  [fill={rgb, 255:red, 255; green, 255; blue, 255 }  ,fill opacity=1 ] (24,129) .. controls (24,117.4) and (33.4,108) .. (45,108) .. controls (56.6,108) and (66,117.4) .. (66,129) .. controls (66,140.6) and (56.6,150) .. (45,150) .. controls (33.4,150) and (24,140.6) .. (24,129) -- cycle ;

\draw    (482,64) -- (482,95) ;
\draw [shift={(482,98)}, rotate = 270] [fill={rgb, 255:red, 0; green, 0; blue, 0 }  ][line width=0.08]  [draw opacity=0] (8.93,-4.29) -- (0,0) -- (8.93,4.29) -- cycle    ;
\draw    (483,128) -- (483,179) ;
\draw [shift={(483,182)}, rotate = 270] [fill={rgb, 255:red, 0; green, 0; blue, 0 }  ][line width=0.08]  [draw opacity=0] (8.93,-4.29) -- (0,0) -- (8.93,4.29) -- cycle    ;
\draw  [fill={rgb, 255:red, 255; green, 255; blue, 255 }  ,fill opacity=1 ] (462,127) .. controls (462,115.4) and (471.4,106) .. (483,106) .. controls (494.6,106) and (504,115.4) .. (504,127) .. controls (504,138.6) and (494.6,148) .. (483,148) .. controls (471.4,148) and (462,138.6) .. (462,127) -- cycle ;
\draw    (319,65) -- (319,96) ;
\draw [shift={(319,99)}, rotate = 270] [fill={rgb, 255:red, 0; green, 0; blue, 0 }  ][line width=0.08]  [draw opacity=0] (8.93,-4.29) -- (0,0) -- (8.93,4.29) -- cycle    ;
\draw    (320,129) -- (320,180) ;
\draw [shift={(320,183)}, rotate = 270] [fill={rgb, 255:red, 0; green, 0; blue, 0 }  ][line width=0.08]  [draw opacity=0] (8.93,-4.29) -- (0,0) -- (8.93,4.29) -- cycle    ;
\draw  [fill={rgb, 255:red, 255; green, 255; blue, 255 }  ,fill opacity=1 ] (299,128) .. controls (299,116.4) and (308.4,107) .. (320,107) .. controls (331.6,107) and (341,116.4) .. (341,128) .. controls (341,139.6) and (331.6,149) .. (320,149) .. controls (308.4,149) and (299,139.6) .. (299,128) -- cycle ;
\draw    (327,213) -- (353.36,253.49) ;
\draw [shift={(355,256)}, rotate = 236.93] [fill={rgb, 255:red, 0; green, 0; blue, 0 }  ][line width=0.08]  [draw opacity=0] (8.93,-4.29) -- (0,0) -- (8.93,4.29) -- cycle    ;
\draw    (409,214) -- (385.56,252.44) ;
\draw [shift={(384,255)}, rotate = 301.37] [fill={rgb, 255:red, 0; green, 0; blue, 0 }  ][line width=0.08]  [draw opacity=0] (8.93,-4.29) -- (0,0) -- (8.93,4.29) -- cycle    ;
\draw    (373,279) -- (404.32,325.51) ;
\draw [shift={(406,328)}, rotate = 236.04] [fill={rgb, 255:red, 0; green, 0; blue, 0 }  ][line width=0.08]  [draw opacity=0] (8.93,-4.29) -- (0,0) -- (8.93,4.29) -- cycle    ;
\draw    (373,279) -- (341.63,327.48) ;
\draw [shift={(340,330)}, rotate = 302.90999999999997] [fill={rgb, 255:red, 0; green, 0; blue, 0 }  ][line width=0.08]  [draw opacity=0] (8.93,-4.29) -- (0,0) -- (8.93,4.29) -- cycle    ;
\draw  [fill={rgb, 255:red, 255; green, 255; blue, 255 }  ,fill opacity=1 ] (352,279) .. controls (352,267.4) and (361.4,258) .. (373,258) .. controls (384.6,258) and (394,267.4) .. (394,279) .. controls (394,290.6) and (384.6,300) .. (373,300) .. controls (361.4,300) and (352,290.6) .. (352,279) -- cycle ;
\draw    (490,213) -- (516.36,253.49) ;
\draw [shift={(518,256)}, rotate = 236.93] [fill={rgb, 255:red, 0; green, 0; blue, 0 }  ][line width=0.08]  [draw opacity=0] (8.93,-4.29) -- (0,0) -- (8.93,4.29) -- cycle    ;
\draw    (572,214) -- (548.56,252.44) ;
\draw [shift={(547,255)}, rotate = 301.37] [fill={rgb, 255:red, 0; green, 0; blue, 0 }  ][line width=0.08]  [draw opacity=0] (8.93,-4.29) -- (0,0) -- (8.93,4.29) -- cycle    ;
\draw    (536,279) -- (567.32,325.51) ;
\draw [shift={(569,328)}, rotate = 236.04] [fill={rgb, 255:red, 0; green, 0; blue, 0 }  ][line width=0.08]  [draw opacity=0] (8.93,-4.29) -- (0,0) -- (8.93,4.29) -- cycle    ;
\draw    (536,279) -- (504.63,327.48) ;
\draw [shift={(503,330)}, rotate = 302.90999999999997] [fill={rgb, 255:red, 0; green, 0; blue, 0 }  ][line width=0.08]  [draw opacity=0] (8.93,-4.29) -- (0,0) -- (8.93,4.29) -- cycle    ;
\draw  [fill={rgb, 255:red, 255; green, 255; blue, 255 }  ,fill opacity=1 ] (515,279) .. controls (515,267.4) and (524.4,258) .. (536,258) .. controls (547.6,258) and (557,267.4) .. (557,279) .. controls (557,290.6) and (547.6,300) .. (536,300) .. controls (524.4,300) and (515,290.6) .. (515,279) -- cycle ;

\draw    (56,214) -- (82.36,254.49) ;
\draw [shift={(84,257)}, rotate = 236.93] [fill={rgb, 255:red, 0; green, 0; blue, 0 }  ][line width=0.08]  [draw opacity=0] (8.93,-4.29) -- (0,0) -- (8.93,4.29) -- cycle    ;
\draw    (138,215) -- (114.56,253.44) ;
\draw [shift={(113,256)}, rotate = 301.37] [fill={rgb, 255:red, 0; green, 0; blue, 0 }  ][line width=0.08]  [draw opacity=0] (8.93,-4.29) -- (0,0) -- (8.93,4.29) -- cycle    ;
\draw    (102,280) -- (133.32,326.51) ;
\draw [shift={(135,329)}, rotate = 236.04] [fill={rgb, 255:red, 0; green, 0; blue, 0 }  ][line width=0.08]  [draw opacity=0] (8.93,-4.29) -- (0,0) -- (8.93,4.29) -- cycle    ;
\draw    (102,280) -- (70.63,328.48) ;
\draw [shift={(69,331)}, rotate = 302.90999999999997] [fill={rgb, 255:red, 0; green, 0; blue, 0 }  ][line width=0.08]  [draw opacity=0] (8.93,-4.29) -- (0,0) -- (8.93,4.29) -- cycle    ;
\draw  [fill={rgb, 255:red, 255; green, 255; blue, 255 }  ,fill opacity=1 ] (81,280) .. controls (81,268.4) and (90.4,259) .. (102,259) .. controls (113.6,259) and (123,268.4) .. (123,280) .. controls (123,291.6) and (113.6,301) .. (102,301) .. controls (90.4,301) and (81,291.6) .. (81,280) -- cycle ;

\draw (91+5,270.4+5) node [anchor=north west][inner sep=0.75pt]    {$\lor $};
\draw (143,191.4) node [anchor=north west][inner sep=0.75pt]    {$y$};
\draw (46,335.4) node [anchor=north west][inner sep=0.75pt]    {$\overline{x} \lor y$};
\draw (126,335.4) node [anchor=north west][inner sep=0.75pt]    {$\overline{x} \lor y$};
\draw (313,43.4) node [anchor=north west][inner sep=0.75pt]    {$\overline{x}$};
\draw (362+5,269.4+5) node [anchor=north west][inner sep=0.75pt]    {$\lor $};
\draw (414,190.4) node [anchor=north west][inner sep=0.75pt]    {$y$};
\draw (317,334.4) node [anchor=north west][inner sep=0.75pt]    {$\overline{x} \lor y$};
\draw (397,334.4) node [anchor=north west][inner sep=0.75pt]    {$\overline{x} \lor y$};
\draw (524+5,267.4+5) node [anchor=north west][inner sep=0.75pt]    {$\land $};
\draw (577,190.4) node [anchor=north west][inner sep=0.75pt]    {$y$};
\draw (480,334.4) node [anchor=north west][inner sep=0.75pt]    {$x\land \overline{y}$};
\draw (560,334.4) node [anchor=north west][inner sep=0.75pt]    {$x\land \overline{y}$};
\draw (310+5,119.4+5) node [anchor=north west][inner sep=0.75pt]    {$\lor $};
\draw (315,191.4) node [anchor=north west][inner sep=0.75pt]    {$\overline{x}$};
\draw (478,42.4) node [anchor=north west][inner sep=0.75pt]    {$x$};
\draw (473+5,118.4+5) node [anchor=north west][inner sep=0.75pt]    {$\lor $};
\draw (479,189.4) node [anchor=north west][inner sep=0.75pt]    {$x$};
\draw (38+3,119.4+5) node [anchor=north west][inner sep=0.75pt]    {$\neg $};
\draw (40,44.4) node [anchor=north west][inner sep=0.75pt]    {$x$};
\draw (39,191.4) node [anchor=north west][inner sep=0.75pt]    {$\overline{x}$};

\end{tikzpicture}

\caption{NOT gadget wiring for circuit simulation using gates from $\Gmon.$ In this case a NOT gate is connected to an OR gate in the original circuit. Copies of the NOT gate in the circuit performing simulation are connected to the copies of the OR gate switched: positive part is connected to negative part of the OR gate and viceversa.}
\label{fig:NOTgatesgad}

\end{figure}

\begin{theorem}
The family $\Gamma(\Gmontwo)$ simulates in constant time and linear space the family $\Gamma(\Gmon)$, i.e. there exists a constant function $T: \N \to \N$ and a linear function $S: \N \to \N$ such that $\Gamma(\Gmon) \preccurlyeq^{T}_{S}\Gamma(\Gmontwo)$
\label{theo:gmontwosimgmon}
\end{theorem}

\begin{proof}
Let $F:Q^{V} \to Q^{V}$ be an arbitrary $\Gmon$-network coded by its standard representation defined by a list of gates $g_{1}, \hdots, g_{n}$ and two functions $\alpha$ an $\beta$ mapping inputs to nodes in $F$ and nodes in $F$ to outputs respectively. We are going to construct in $\DLOG$ a $\Gmontwo$-network $G$ that simulates $F$ in time $T = \mathcal{O}(1)$ and space $S = \mathcal{O}(|V|)$ where $H$ is the communication graph of $F$. In order to do that, we are going to replace each gate $g_{k}$ by a small gadget. More precisely, we are going to introduce the following coding function: $x \in \{0,1\} \to (x,x,0) \in \{0,1\}^{3}$. We are going to define gadgets for each gate. Let us take $k \in \{1,\hdots,n\}$ and call $g^{*}_{k}$ the corresponding gadget associated to $g_{k}$. Let us that suppose $g_{k}$ is an OR gate and that it has fanin $2$ and fanout $1$ then, we define $g^{*}: \{0,1\}^{6} \to \{0,1\}^{6}$ as a function that for each input of the form $(x,x,y,y,0,0)$ produces the output $g^{*}((x,x,y,y,0,0)) = (x \vee y, x \vee y, 0,0,0,0)$. The case fanin $1$ and fanout $1$ is given by $g^{*}((x,x,0,0,0,0)) = (x, x, 0,0,0,0)$,  the case fanin 2 and fanout 2 is given by $g^{*}((x,x,y,y,0,0)) = (x \vee y, x \vee y, x \vee y, x \vee y,0,0)$ and finally case fanin $1$ and fanout $2$ is given by the same latter function but on input (x,x,0,0,0,0). The AND case is completely analogous. We are going to implement the previous functions as small (constant depth) synchronized circuits that we call block gadgets. More precisely, we are going to identify functions $g^{*}$ with its correspondent block gadget. The detail on the construction of these circuits that define latter functions are provided in Figures \ref{fig:simgmongmon21}, \ref{fig:simgmongmon22} and \ref{fig:simgmongmon23}.

Once we have defined the structure of block gadgets, we have to manage connections between them and  also manage the fixed $0$ inputs that we have added in addition  to the zeros that are produced by the blocks as outputs. In order to do that, let us assume that gates $g_{i} $ and $g_{j}$ are connected. Note from the discussion on coding above that AND/OR gadgets have between $2$ and $4$ inputs and outputs fixed to $0$. In particular, as it is shown in Figures \ref{fig:simgmongmon21}, \ref{fig:simgmongmon22} and \ref{fig:simgmongmon23}, all the block gadgets have the same amount of zeros in the input and in the output with the exception of the fanin $1$ fanout $2$ gates and the fanin $2$ fanout $1$ gates. However, as $\G$-networks are closed systems (the amount of inputs must be the same that the amount of outputs) we have that, for each  fanin $1$ fanout $2$  gate,  it must be a fanin $2$ fanout $1$ gate and vice versa (otherwise there would be more input than outputs or more outputs than inputs). In other words, there is a bijection between the set of fanin $1$ fanout $2$  gates and the set of fanin $2$ fanout $1$. Observe that fanin $2$ fanout $1$ gates consume $2$ zeros in input but produce $4$ zeros in output while fanin $1$ fanout $2$ gates consume $4$ in input and produce $2$ zeros in output (see Figures \ref{fig:simgmongmon22} and \ref{fig:simgmongmon23}). So, between $g^{*}_{i}$ and $g^{*}_{j}$ we have to distinct two cases: a) if both gates have the same number of inputs and outputs, connections are managed in the obvious way i.e., outputs corresponding to the computation performed by original gate are assigned between $g^{*}_{i}$ and $g^{*}_{j}$  and each gate uses the same zeros they produce to feed its inputs. b) if $g^{*}_{i}$ or $g^{*}_{j}$ have more inputs than outputs or vice versa, we have to manage the extra zeros (needed or produced). Without lost of generality, we assume that $g^{*}_{i}$ is fanin $2$ fanout $1$. Then, by latter observation it must exists another gate $g_{k}$ and thus, a gadget block $g^{*}_{k}$ with fanin $1$ and fanout $2$. We simply connect extra zeros produced by  $g^{*}_{i}$ to block  $g^{*}_{k}$ and we do the same we did in previous case in order to manage connections.

Note that $F^{*}$  is constructible in $\DLOG$ as it suffices to read the standard representation of $F$ and produce the associated block gadgets which have constant size. In addition we have that previous encoding $g \to g^*$  induce a block map $\phi: \{0,1\}^{V} \to \{0,1\}^{V^{+}}$  where $|V^{+}| =  \mathcal{O} ( |V| )$ and that $\phi \circ F = F^{*T} \circ \phi$ where $T=6$ is the size of each gadget block in $F^{*}$. We conclude that $F^{*} \in \Gamma(\Gmontwo)$ simulates  $F$ in space $|V^{+}| = \mathcal{O} ( |V| )$ and time $T =6$ and thus,  $ \Gamma(\Gmon)  \preccurlyeq^T_{S} \Gamma(\Gmontwo)$ where $T$ is constant and $S$ is a linear function.  .

 \end{proof}
\begin{figure}
\centering

\tikzset{every picture/.style={line width=0.75pt}} 



\caption{Block gadgets for simulating Fanin $2$ Fanout $1$ AND/OR gates using only gates in $\Gmontwo$. Squared zeros represent the amount of zeros that can be used as inputs for the same block. Circled zeros correspond to extra zeros that need to be assigned to a Fanin $1$ Fanout $2$ gate.}
\label{fig:simgmongmon21}
\end{figure}

\begin{figure}
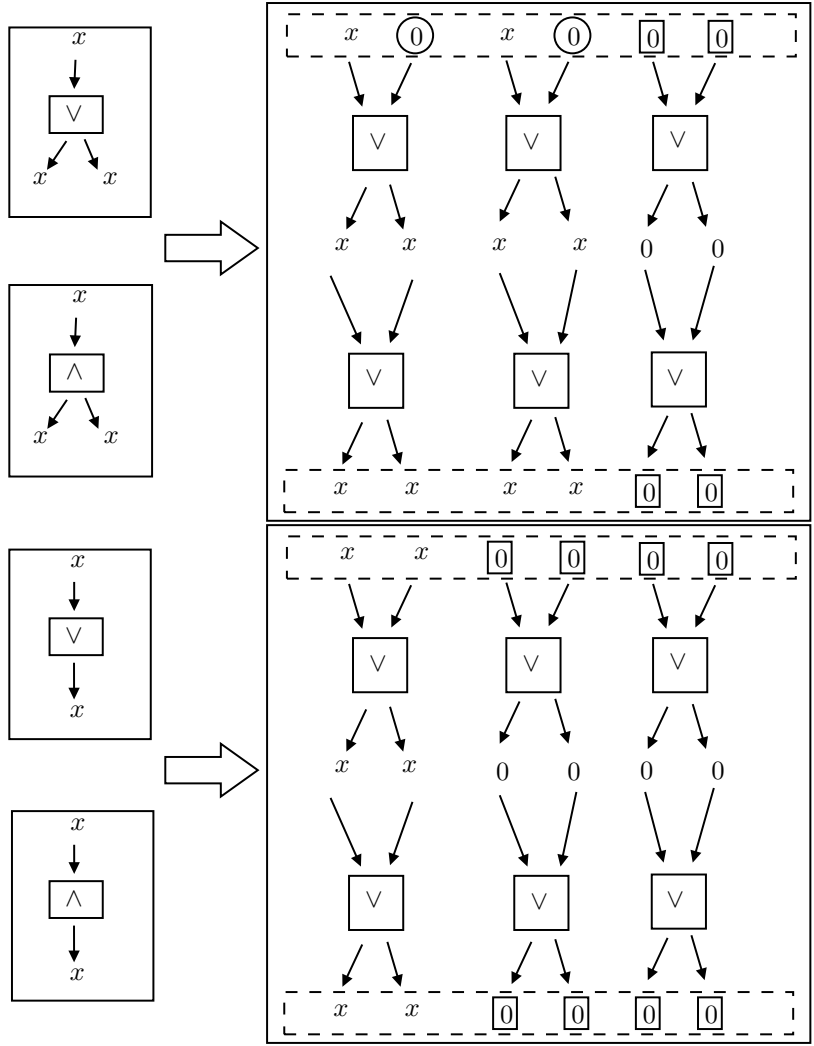


\centering
\tikzset{every picture/.style={line width=0.75pt}} 

%

\caption{Block gadgets for simulating  Fanin $1$ Fanout $2$ and Fanin $1$ Fanout $1$ AND/OR gates using only gates in $\Gmontwo$. Squared zeros represent the amount of zeros that can be used as inputs for the same block. Circled zeros correspond to extra zeros that need to be received from a Fanin $2$ Fanout $1$ gate.}
\label{fig:simgmongmon22}
\end{figure}

\begin{figure}
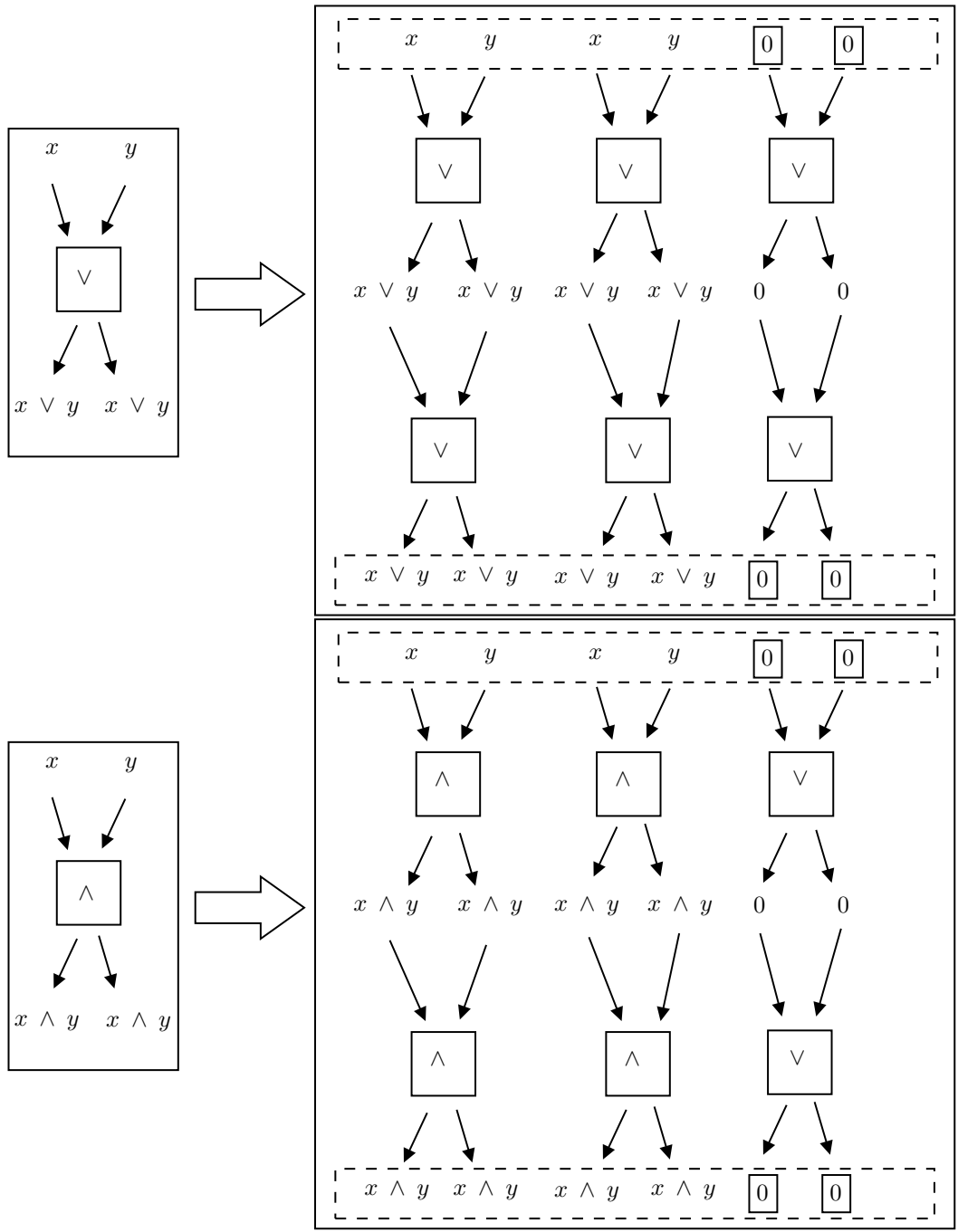

\centering

\tikzset{every picture/.style={line width=0.75pt}} 

%

\caption{Block gadgets for simulating  fanin $2$ fanout $2$ AND/OR gates using only gates in $\Gmontwo$. Squared zeros represent the amount of zeros that can be used as inputs for the same block. }
\label{fig:simgmongmon23}
\end{figure}
\begin{corollary}
The family $\Gamma(\Gmontwo)$ is strongly universal.
\end{corollary}

\begin{proof}
Result is direct from Theorem \ref{theo:gmon-univ} ($\Gamma(\Gmon)$ is strongly universal)  and Theorem \ref{theo:gmontwosimgmon} ($ \Gamma(\Gmon)  \preccurlyeq^T_{S} \Gamma(\Gmontwo)$ where $T$ is constant and $S$ is a linear function).
\end{proof}


Now we show the universality of $\Gamma(\Gmon)$. The proof is essentially a consequence of \cite[Theorem 6.2.5]{Greenlaw_1995}.
Roughly, latter result starts with alternated monotone circuit which has only fanin $2$ and fanout $2$ gates (previous results in the same reference show that one can always reduce to this case starting from an arbitrary circuit) and  gives an $\textbf{NC}^{1}$ construction of a synchronous circuit preserving latter properties.
We need additional care here because we want a reusable circuit whose output is fed back to its input.
Note also that the construction uses quadratic space in the number of gates of the circuit given in input, so we cannot show strong universality this way but only universality.
\begin{theorem}\label{teo:gmonuniv}
The family $\Gamma(\Gmon)$ of all $\Gmon$-networks is universal
\end{theorem}
\begin{proof}
Let $F:Q^{k} \to Q^{k}$ be some arbitrary network with a circuit representation $C:\{0,1\}^{n} \to \{0,1\}^{n}$ such that $n = k^{\mathcal{O}(1)}.$ By \cite[Theorem 6.2.5]{Greenlaw_1995} we can assume that there exists a circuit  $C':\{0,1\}^{n'} \to \{0,1\}^{n'}$ where $n' = \mathcal{O}(n^{2})$ such that $C'$ is synchronous alternated and monotone. In addition, every gate in $C'$ has fanin and fanout $2$. We remark that latter reference do not only provides the standard encoding of $C'$ but also give us a $\DLOG$ algorithm (it is actually $\textbf{NC}^{1}$) which takes the standard representation of $C:\{0,1\}^{n} \to \{0,1\}^{n}$ and produces $C'$. We are going to slightly modify latter algorithm in order to construct not only a circuit but a $\Gmon$-network. In fact, the only critical point is to manage the identification between outputs and inputs. This is not direct from the result by Ruzzo et al. as their algorithm involves duplication of inputs and also adding constant inputs. In order to manage this, it suffices to simply modify their construction in order to mark original, copies and constant inputs. Then, as $\Gmon$ includes COPY gates and also AND/OR gates with fanout $1$, one can always produce copies of certain input if we need more, or erase extra copies by adding and small tree of $\mathcal{O}(\log(n))$ depth. Same goes for constant inputs. Formally, at the end of the algorithm, the $\DLOG$ can read extra information regarding copies and constant inputs, and then can construct $\mathcal{O}(\log(n))$ depth circuit that produces a coherent encoding for inputs and outputs. This latter construction defines a $\Gmon$-network $G:\{0,1\}^{n''} \to \{0,1\}^{n''}$ and an encoding $\phi:Q^{n} \to \{0,1\}^{n''}$ where $n'' = \mathcal{O}(n^{2})$  such that $\phi \circ F = G^{T} \circ \phi$ where $T = \mathcal{O}(\text{depth(C')} + \log(n))$. Thus, $\Gmon$ is universal.
\end{proof}

We can now state the following direct corollary.
\begin{corollary}
Let $\mathcal{F}$ be a strongly universal automata network family. Then, $\mathcal{F}$ is universal.
\end{corollary}
\begin{proof}
In order to show the result, it suffices to exhibit a $\G$-network family ($\G$-networks are bounded degree networks) which is strongly universal and universal at the same time. By Theorem \ref{teo:gmonuniv} we take $\G = \Gmon$ and thus, corollary holds.
\end{proof}

\begin{corollary}\label{cor:univfrommon}
  Let $\G$ be either $\Gmon$ or $\Gmontwo$.
  Any family $\mathcal{F}$ that has coherent $\G$-gadgets contains a subfamily of bound degree networks with bounded degree representation which is  (strongly) universal. Any CSAN family with coherent $\G$-gadgets is  (strongly) universal.
\end{corollary}

\subsection{Closure and synchronous closure}

Although monotone gates are sometimes easier to realize in concrete dynamical system which make the above results useful, there is nothing special about them to achieve universality: any set of gates that are expressive enough for Boolean functions yields the same universality result. Given a set of maps $\G$ over alphabet $Q$, we define its \emph{closure} ${\overline{\G}}$ as the set of maps that are computed by circuits that can be built using only gates from $\G$. More precisely, $\overline{\G}$ is the closure of $\G$ by composition, \textit{i.e.} forming from maps ${g_1:Q^{I_1}\to Q^{O_1}}$ and ${g_2:Q^{I_2}\to Q^{O_2}}$ (with ${I_1, I_2, O_1, O_2}$ disjoint) a composition $g$ by plugging a subset of outputs $O\subseteq O_2$ of $g_1$ into a subset of inputs $I\subseteq I_2$ of $g_2$, thus obtaining ${g: Q^{I_1\cup I_2\setminus I}\to Q^{O_1\setminus O\cup O_2}}$ with 
\[g(x)_o =
  \begin{cases}
    g_1(x_{I_1})_o &\text{ if }o\in O_1\setminus O\\
    g_2(y)_o &\text{ if }o\in O_2
  \end{cases}
\]
where ${y_j = x_j}$ for ${j\in I_2\setminus I}$ and ${y_j= g_1(x_{I_1})_{\pi(j)}}$ where ${\pi:I\to O}$ is the chosen bijection between $I$ and $O$ (the wiring of outputs of $g_1$ to inputs of $g_2$).
A composition is \emph{synchronous} if either ${I=\emptyset}$ or ${I=I_2}$. We then define the synchronous closure ${\overline{\G}^2}$ as the closure by synchronous composition. The synchronous composition correspond to synchronous circuits with gates in $\G$. A $\G$-circuit is a sequence of compositions starting from elements of $\G$. It is synchronous if the compositions are synchronous. The depth of a $\G$-circuit is the maximal length of a path from an input to an output. In the case of a synchronous circuits, all such path are of equal length.

\begin{figure}
  \centering
  \begin{minipage}{.2\linewidth}
    \begin{tikzpicture}[scale=.5]
      \draw (-1.5,1)--(2.5,1)--(2.5,-1)--(-1.5,-1)--cycle;
      \draw[->,very thick] (-1,-2)  -- (-1,-1);
      \draw[->,very thick] (0,-2) -- (0,-1);
      \draw[->,very thick] (1,-2) -- (1,-1);
      \draw[->,very thick] (2,-2) -- (2,-1);
      \draw[->,very thick] (3,-2) -- (3,2);
      \draw[->,very thick] (-1,1)  -- (-1,2);
      \draw[->,very thick] (0,1) -- (0,2);
      \draw[->,very thick] (1,1) -- (1,5);
      \draw[->,very thick] (2,1) -- (2,2);
      \begin{scope}[shift={(0,3)}]
        \draw (-1.5,1)--(.5,1)--(.5,-1)--(-1.5,-1)--cycle;
        \draw[->,very thick] (-1,1)  -- (-1,2);
        \draw[->,very thick] (0,1) -- (0,2);
      \end{scope}
      \begin{scope}[shift={(3,3)}]
        \draw (-1.5,1)--(.5,1)--(.5,-1)--(-1.5,-1)--cycle;
        \draw[->,very thick] (-1,1)  -- (-1,2);
        \draw[->,very thick] (0,1) -- (0,2);
      \end{scope}
    \end{tikzpicture}
  \end{minipage}
  \begin{minipage}{.2\linewidth}
    \begin{tikzpicture}[scale=.5]
      \draw (-1.5,1)--(2.5,1)--(2.5,-1)--(-1.5,-1)--cycle;
      \draw[->,very thick] (-1,-2)  -- (-1,-1);
      \draw[->,very thick] (0,-2) -- (0,-1);
      \draw[->,very thick] (1,-2) -- (1,-1);
      \draw[->,very thick] (2,-2) -- (2,-1);
      \draw[->,very thick] (-1,1)  -- (-1,2);
      \draw[->,very thick] (0,1) -- (0,2);
      \draw[->,very thick] (1,1) -- (1,2);
      \draw[->,very thick] (2,1) -- (2,2);
      \begin{scope}[shift={(0,3)}]
        \draw (-1.5,1)--(.5,1)--(.5,-1)--(-1.5,-1)--cycle;
        \draw[->,very thick] (-1,1)  -- (-1,2);
        \draw[->,very thick] (0,1) -- (0,2);
      \end{scope}
      \begin{scope}[shift={(2,3)}]
        \draw (-1.5,1)--(.5,1)--(.5,-1)--(-1.5,-1)--cycle;
        \draw[->,very thick] (-1,1)  -- (-1,2);
        \draw[->,very thick] (0,1) -- (0,2);
      \end{scope}
    \end{tikzpicture}
  \end{minipage}
  \caption{Non-synchronous composition (on the left) and synchronous composition (on the right).}
  \label{fig:sync-nonsync-composition}
\end{figure}
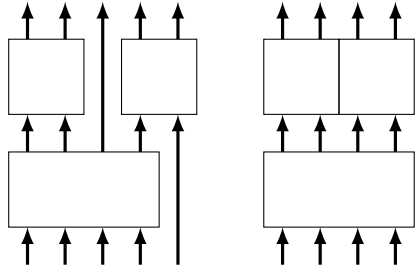

\begin{remark}
  The above definitions are very close to the classical notion of clones \cite{post_lattice}. However, we stress that, in our case, projections maps ${Q^k\to Q}$ are generally not available. This is important because in a given dynamical systems, erasing information might be impossible (think about reversible systems) and hiding it into some non-coding part might be complicated.
\end{remark}

\begin{proposition}\label{prop:closuregadgets}
  Fix some alphabet $Q$ and consider two finite sets of maps $\G$ and $\G'$ over alphabet $Q$ such that:
  \begin{itemize}
  \item either contains the identity map $Q\to Q$ and is such that $\overline{\G}$ contains $\G'$,
  \item or there is an integer $k$ such that ${\overline{\G}^2_k}$, the set of elements of $\overline{\G}^2$ that can be realized by a circuit of depth $k$, contains $\G'$.
  \end{itemize}
  Then, any family $\mathcal{F}$ that has coherent $\G$-gadgets has coherent $\G'$-gadgets.
\end{proposition}

\begin{proof}
  Suppose first that the first item holds.  Since ${\overline{\G}}$ contains $\G'$ there must exist a circuit made of gates from $\G$ that produces any given element $g\in\G'$. One then wants to apply gadget glueing on gadgets from $\G$ to mimic the composition and thus obtain a gadget corresponding to $g$. However this doesn't work as simply because propagation delay is a priori not respected at each gate in the circuit composition yielding $g$ and there is a risk that information arrives distinct delays at different outputs. However, since $\G$ contains the identity map, there is a corresponding gadget in the family that actually implements a delay line. This additionnal gadget solves the problem: it is straightforward to transform by padding with identity gates all circuit with gates in $\G$ into synchronous ones. Moreover, by padding again, we can assume that the finite set of such circuits computing elements of $\Gmon$ are all of same depth. It is then straightforward to translate this set of circuits into coherent $\Gmon$-gadgets by iterating gadget glueing and using Lemma~\ref{lem:pseudo-orbit-glueing}.

  If the second item holds the situation is actually simpler because the synchronous closure contains only synchronous circuits of gates from $\Gmontwo$ so we can directly translate the circuits producing the maps of $\Gmontwo$ into gadgets via gadget glueing by Lemma~\ref{lem:pseudo-orbit-glueing} as in the previous case. Moreover, the hypothesis is that all elements of $\G'$ are realized by circuit of same depth so we get gadgets that share the same time constant.
   
\end{proof}

\subsection{Super-polynomial periods without universality}

A universal family must exhibit super-polynomial periods, however universality is far from necessary to have this dynamical feature. In this subsection we define the family of wire networks to illustrate this.

 In order to do that, we need the following classical result about the growth of Chebyshev function and prime number theorem. 

\begin{lemma} \cite{hardy2008introduction}
  Let $m \geq 2$ and $\mathcal{P}(m) = \{p \leq m \text{ } | \text{ }  p \text{ prime}\}$. If we define $\pi(m) = |\mathcal{P}(m)|$ and  $\theta(m)  = \sum \limits_{p \in \mathcal{P}(m)}\log (p)$  then we have ${\pi (m) \sim \frac{m}{\log(m)}}$ and ${\theta(m) \sim m}$.
  \label{lemma:primes}
\end{lemma}
\newcommand\Gwire{\G_{w}}

By using the Lemma \ref{lemma:primes} we can construct automata networks with non-polynomial cycles simply by making disjoint union of rotations (\textit{i.e.} network whose interaction graph is a cycle that just rotate the configuration at each step). Indeed, it is sufficient to consider rotations on cycle whose length are successive prime numbers. It turns out that these automata networks are exactly $\Gwire$-networks where $\Gwire$ is a single 'wire gate': ${\Gwire =\{id_B\}}$ where ${id_B}$ is the identity map over $\{0,1\}$.

Formally, according to Definition \ref{def:g-network}, for any $\Gwire$-network $F:Q^V \to Q^V$ there exist a partition $V=C_1\cup C_2 \hdots \cup C_k$ where $C_i = \{u^i_1\hdots,u^i_{l_i}\}$ with ${l_i\geq 2}$ for each $i = 1,\hdots,k$ and $F(x)_{u^i_{s+1\bmod l_i}}= x_{u^i_s}$ for any $x \in Q^V$ and $0\leq s \leq l_i$.


\begin{theorem}\label{theo:non-polyn-cycl}
  Any family $\mathcal{F}$ that has coherent $\Gwire$-gadgets has superpolynomial cycles, more precisely: there is some ${\alpha>0}$ such that for infinitely many ${n\in\N}$, there exists a network $F_n\in\mathcal{F}$ with ${O(n)}$ nodes and a periodic orbit of size ${\Omega(\exp (n^\alpha))}$. 
\end{theorem}
\begin{proof}
  Taking the notations of Lemma~\ref{lemma:primes}, define for any $n$ the $\Gwire$-network $G_n$ made of disjoint union of circuits of each prime length less than $n$. $G_n$ has size at most ${n\pi(n)}$ and if we consider a configuration $x$ which is in state $1$ at exactly one node in each of the $\pi(n)$ disjoint circuit, it is clear that the orbit of $x$ is periodic of period $\exp{\theta(n)}$. Therefore, from Lemma~\ref{lemma:primes}, for any $n$, $G_n$ is a circuit of size ${m\leq n\pi(n)}$ with a periodic orbit of size $\theta(n)\in\Omega(\exp(\sqrt{m\log m}))$.
 By hypothesis there are linear maps $T$ and $S$ such that for any $n$, there is $F_n$ that simulates $G_n$ (by Lemma~\ref{lem:from-gadgets-to-networks}), therefore $F_n$ also has a super-polynomial cycle by Lemma~\ref{lem:simu-dynamic}. 
\end{proof}

\newcommand\Gconj{\G_{conj}}
\newcommand\Fconj{\mathcal{F}_{conj}}
\subsection{Conjunctive networks and $\Gconj$-networks}

Let $G=(V,E)$ be any directed graph. The conjunctive network associated to $G$ is the automata network ${F_G:\{0,1\}^V\to\{0,1\}^V}$ given by ${F(x)_i = \wedge_{j\in N^-(i)}x_j}$ where $N^-(i)$ denotes the incoming neighborhood of $i$. Conjunctive networks are thus completely determined by the interaction graph and a circuit representation can be deduced from this graph in \DLOG{}. We define the family $\Fconj$ as the set of conjunctive networks together with the standard representation $\Fconj^*$ which are just directed graphs encoded as finite words in a canonical way.

\begin{remark}
  We can of course do the same with disjunctive networks. Any conjunctive network $F_G$ on graph $G$ is conjugated to the disjunctive network $F'_G$ on the same graph by the negation map ${\rho : \{0,1\}^V\to \{0,1\}^V}$ defined by ${\rho(x)_i = 1-x_i}$, formally ${\rho\circ F_G = F'_G\circ \rho}$. In particular, this means that the families of conjunctive and disjunctive networks simulate each other. In the sequel we will only state results for conjunctive networks while they hold for disjunctive networks as well.
\end{remark}

Let us now consider the set $\Gconj = \{\gateAND,\gateCOPY\}$.
$\Gconj$-networks are nothing else than conjunctive networks with the following degree constraints: each node has either in-degree $1$ and out-degree $2$, or in-degree $2$ and out-degree $1$. The following theorem shows that, up to simulation, these constraints are harmless.

\begin{theorem}\label{theo:Gconj-networks}
  The family of $\Gconj$-networks simulates the family ${(\Fconj,\Fconj^*)}$ of conjunctive networks in linear time and polynomial space.
\end{theorem}
\begin{proof}
  Let $F$ be an arbitrary conjunctive network on graph ${G=(V,E)}$ with $n$ nodes. Its maximal in/out degree is at most $n$. For each node of indegree ${i\leq n}$ we can make a tree-like $\Gconj$-gadget with $i$ inputs and $1$ output that computes the conjunction of its $i$ inputs in exactly ${n}$ steps: more precisely, we can build a sub-network of size ${O(n)}$ with $i$ identified 'input' nodes of fanin $1$ and one identified output node of fanout $1$ such that for any ${t\in\N}$ the state of the output node at time $t+n$ is the conjunction of the states of the input nodes at time $t$ (the only sensible aspect is to maintain synchronization in the gadget, see Figure~\ref{fig:faningadget}).
  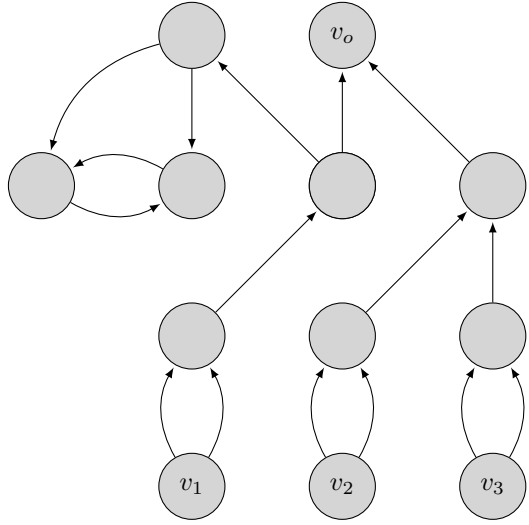
\begin{figure}
    \centering
    \begin{tikzpicture}[shorten >=1pt,node distance=2cm,on grid,auto]
      \tikzstyle{every state}=[fill={rgb:black,2;white,10}]
      \node[state] (q_1)                    {};
      \node[state] (q_2)  [right of=q_1]    {};
      \node[state] (q_3)  [right of=q_2]    {};
      \node[state] (i_1)  [below of=q_1]    {$v_1$};
      \node[state] (i_2)  [below of=q_2]    {$v_2$};
      \node[state] (i_3)  [below of=q_3]    {$v_3$};
      \node[state] (q_4)  [above of=q_3]    {};
      \node[state] (q_5)  [above of=q_2]    {};
      \node[state] (q_6)  [above of=q_5]    {$v_o$};
      \node[state] (q_7)  [above of=q_1]    {};
      \node[state] (q_8)  [above of=q_7]    {};
      \node[state] (q_9)  [right of=q_7]    {};
      \node[state] (q'_6)  [left of=q_7]    {};
      \path[->]
      (q_2) edge (q_4);
      \path[->]
      (q_1) edge (q_5);
      \path[->]
      (q_5) edge (q_8);
      \path[->]
      (q_3) edge (q_4);
      \path[->]
      (q_5) edge (q_6);
      \path[->]
      (q_4) edge (q_6);
      \path[->]
      (q_8) edge (q_7);
      \path[->]
      (q_8) edge[bend right] (q'_6);
      \path[->]
      (q'_6) edge[bend right] (q_7);
      \path[->]
       (q_7) edge[bend right] (q'_6);
      \path[->]
      (i_1) edge[bend left] (q_1);
      \path[->]
      (i_1) edge[bend right] (q_1);
      \path[->]
      (i_2) edge[bend left] (q_2);
      \path[->]
      (i_2) edge[bend right] (q_2);
      \path[->]
      (i_3) edge[bend left] (q_3);
      \path[->]
      (i_3) edge[bend right] (q_3);
    \end{tikzpicture}
    \caption{Fanin gadget of degree $3$. For any configuration $x$, ${F^3(x)_{v_o}=x_{v_1}\wedge x_{v_2}\wedge x_{v_3}}$.}
    \label{fig:faningadget}
  \end{figure}

  We do the same for copying the output of a gate $i$ times and dealing with arbitrary fanout. Then we replace each node of $F$ by a meta node made of the two gadgets to deal with fanin/fanout and connect everything together according to graph $G$ (note that fanin/fanout is granted to be $1$ in the gadgets so connections respect the degree constraints). We obtain in \DLOG{} a $\Gconj$-network of size polynomial in $n$ that simulates $F$ in linear time. 
\end{proof}

\begin{remark}
  The family of conjunctive networks can produce super-polynomial periods but is not universal. There are several ways to show this. It is for instance impossible to produce super-polynomial transients within the family \cite[Theorem 3.20]{De_Schutter_1999} so Corollary~\ref{cor:universality} conclude. One could also use Corollary~\ref{coro:trickyuniversal} since a node in a strongly connected component of a conjunctive network must have a trace period of at most the size of the component (actually much more in known about periods in conjunctive networks through the concept of loop number or cyclicity, see \cite{De_Schutter_1999}).
\end{remark}

\newcommand\Gwfa{\G_{t}}
\newcommand\stdAND{\mathrm{AND}_{\{0,1\}}}
\newcommand\specAND{\mathrm{AND}_{2}}
\newcommand\gateLoop{\mathrm{\Lambda}}
\newcommand\gateId{\mathrm{Id}}
\newcommand\gateCpy{\mathrm{\Upsilon}}

\subsection{Super-polynomial transients and periods without universality}

Let us consider in this section alphabet ${Q=\{0,1,2\}}$. We are going to define a set $\Gwfa$ such that $\Gwfa$-networks exhibit super-polynomial transients but are not universal. To help intuition, $\Gwfa$-networks can be though as standard conjunctive networks on ${\{0,1\}}$ that can in some circumstances produce state $2$ which is a spreading state (a node switches to state $2$ if one of its incoming neighbors is in state $2$). The extra state $2$ will serve to mark super-polynomial transients, but it cannot escape a strongly connected component once it appears in and, as we will see, $\Gwfa$-networks are therefore too limited in their ability to produce large periodic behavior inside strongly connected components.

$\Gwfa$ is made of the following maps:
\begin{align*}
  \stdAND &: (x,y) \mapsto
                    \begin{cases}
                      2&\text{ if $2\in\{x,y\}$}\\
                      x\wedge y&\text{ else.}
                    \end{cases}\\
  \specAND &: (x,y) \mapsto
                    \begin{cases}
                      2&\text{ if $2\in\{x,y\}$ or $x=y=1$}\\
                      0&\text{ else.}
                    \end{cases}\\
  \gateLoop &: (x,y) \mapsto
                    \begin{cases}
                      2&\text{ if $2\in\{x,y\}$}\\
                      x&\text{ else.}
                    \end{cases}\\
  \gateId &: x\mapsto x.\\
  \gateCpy &: x\mapsto (x,x).
\end{align*}

 $\Gwfa$-networks can produce non-polynomial periods by disjoint union of rotations of prime lengths as in Theorem~\ref{theo:non-polyn-cycl}, but they can also wait for a global synchronization of all rotations and freeze the result of the test for this synchronization condition inside a small feedback loop attached to a ``controlled AND map''.

More precisely, as shown in Figure~\ref{fig:freezeand} we can use in the context of any $\Gwfa$-network a small module ${T(x)}$ of made of five nodes with the following property: if the $\gateLoop$ node of the module is in state $0$ in some initial configuration, then it stays in state $0$ as long as nodes $x$ is not in state $1$, and when $x=1$ at some time step $t$ then from step $t+2$ on the $\gateLoop$ node is in state $2$ at least one step every two steps. This module is the key to control transient behavior.

\begin{figure}
  \centering
  \begin{tikzpicture}[shorten >=1pt,node distance=2cm,on grid,auto]
    \tikzstyle{every state}=[fill={rgb:black,2;white,10}]
    
    \fill[fill=white!90!gray] (-4,0)--(-2,-2)--(-4,-2)--cycle;

    \node[state] (q_1)                    {$\specAND$};
    \node[state] (q_2)  [left of=q_1]    {$\gateCpy_1$};
    \node[state] (q_3)  [below of=q_1]    {$\gateCpy_2$};
    \node[state] (q_4)  [right of=q_1]    {$\gateLoop$};
    \node[state] (q_5)  [below of=q_2]    {$x$};
    \node[state] (q_6)  [below of=q_4]    {$\gateId$};

    \path[<->]
    (q_6) edge (q_4);
    \path[->]
    (q_2) edge (q_1);
    \path[->]
    (q_1) edge (q_4);
    \path[->]
    (q_3) edge (q_1);
    \path[->]
    (q_5) edge (q_2);
    \path[->]
    (q_5) edge (q_3);
    \path[dotted,very thick] (q_5)+(-2,2) edge (q_5);
    \path[dotted,very thick] (q_5)+(-2,0) edge (q_5);
  \end{tikzpicture}
  
  \caption{Freezing the result of a test in a $\Gwfa$-network. The module $T(x)$ is made of the nodes marked $\gateCpy$, $\specAND$, $\gateLoop$ and $\gateId.$  Observe that each node represents some output of its corresponding label (for more details on $\G$-networks see Definition \ref{def:g-network}). Each gate has one output with the exception of the gate $\Upsilon$ which is represented by two nodes. The module $T(x)$ reads the value of node $x$ belonging to an arbitrary $\Gwfa$-network (represented in light gray inside dotted lines).  The output $\Lambda$ is fed back to its control input via the $\gateId$ node (self-loops are forbidden in $\Gwfa$-networks). Note that $x$  as well as the rest of the network is not influenced by the behavior of the gates of the module $T(x)$.}
  \label{fig:freezeand}
\end{figure}

Besides, the map $\stdAND$ behaves like standard Boolean AND map when its inputs are in ${\{0,1\}}$. More generally, by combining such maps in a tree-like fashion, one can build modules ${A(x_1,\ldots,x_k)}$ for any number $k$ of inputs with a special output node which has the following property for some time delay $\Delta\in O(\log(k))$: the output node at time ${t+\Delta}$ is in state $1$ if and only if all nodes $x_i$ (with ${1\leq i\leq k}$) are in state $1$ at time $t$.

Combining these two ingredients, we can build upon the construction of Theorem~\ref{theo:non-polyn-cycl} to obtain non-polynomial transients in any family having coherent $\Gwfa$-gadgets.

\begin{theorem}
  \label{theo:non-poly-transient}
  Any family $\mathcal{F}$ that has coherent $\Gwfa$-gadgets has superpolynomial transients, more precisely: there is some ${\alpha>0}$ such that for any ${n\in\N}$, there exists a network $F_n\in\mathcal{F}$ with ${O(n)}$ nodes and a configuration $x$ such that ${F_n^t(x)}$ is not in an attractor of $F_n$ with ${t\in\Omega(\exp ( n^\alpha))}$.
\end{theorem}
\begin{proof}
  Like in Theorem~\ref{theo:non-polyn-cycl}, the key of the proof is to show that there is a $\Gwfa$-network with transient length as in the theorem statement, then the property immediately holds for networks of the family $\mathcal{F}$ by Lemma~\ref{lem:from-gadgets-to-networks} and Lemma~\ref{lem:simu-dynamic}.

  For any ${n>0}$ we construct a $\Gwfa$-network $G_n$ made of two parts:
  \begin{itemize}
  \item the 'bottom' part of $G_n$ uses a polynomial set of nodes $B_n$ and consists in a disjoint union of circuits for each prime length less than $n$ as in Theorem~\ref{theo:non-polyn-cycl}, but where for each prime $p$, the circuit of length $p$ has a node $v_p$ which implements a copy gate $\gateCOPY$, thus not only sending its value to the next node in the circuit, but also outputting it to the second part of $G_n$;
  \item the 'top' part of $G_n$ is made of a module ${A(x_1,\ldots,x_k)}$ connected to all nodes $v_p$ as inputs and whose output is connected to a test module $T(x)$ as in Figure~\ref{fig:freezeand}.
  \end{itemize}
  Note that the size of $G_n$ is polynomial in $n$.  With this construction we have the following property as soon as the modules ${A(x_1,\ldots,x_k)}$ and $T(x)$ are initialized to state $0$ everywhere: as long as nodes $v_p$ are not simultaneously in state $1$ then the output of the test module $T(x)$ stays in state $0$; moreover, if at some time $t$ nodes $v_p$ are simultaneously in state $t$, then after time $t+O(\log(t))$ the output node of module $T(x)$ is in state $1$ one step every two steps. This means that $t+O(\log(t))$ is a lower bound on the transient of the considered orbit. To conclude the theorem it is sufficient to consider the initial configuration where all nodes are in state $0$ except the successor of node $v_p$ in each circuit of prime length $p$, which are in state $1$. In this case it is clear that the first time $t$ at which all nodes $v_p$ are in state $1$ is the product of prime numbers less than $n$. As in theorem~\ref{theo:non-polyn-cycl}, we conclude thanks to Lemma~\ref{lemma:primes}. 
\end{proof}

As said above, $\Gwfa$-networks are limited in their ability to produce large periods. More precisely, as shown by the following lemma, their behavior is close enough to conjunctive networks so that it can be analyzed as the superposition of the propagation/creation of state $2$ above the behavior of a classical Boolean conjunctive network. To any $\Gwfa$-network $F$ we associate the Boolean conjunctive network $F^*$ with alphabet $\{0,1\}$ as follows: nodes with local map $\stdAND$ or $\specAND$ are simply transformed into nodes with Boolean conjunctive local maps on the same neighbors, nodes with local maps $\gateCpy$ or $\gateId$ are left unchanged (only their alphabet changes), and nodes with map $\gateLoop(x,y)$ are transformed into a node with only $x$ as incoming neighborhood. 

\begin{lemma}\label{lem:superconjunctive}
  Let $F$ be a $\Gwfa$-network with node set $V$ and $F^*$ its associated Boolean conjunctive network. Consider any ${x\in\{0,1,2\}^V}$ and any ${x^*\in\{0,1\}^V}$ such that the following holds: 
  \[\forall v\in V: x_v\in\{0,1\}\Rightarrow x^*_v=x_v,\]
  then the same holds after one step of each network: 
  \[\forall v\in V: F(x)_v\in\{0,1\}\Rightarrow F^*(x^*)_v=F(x)_v.\]
\end{lemma}
\begin{proof}
  It is sufficient to check that if $F(x)_v\neq 2$, it means that all its incoming neighbors are in ${\{0,1\}}$ so $x$ and $x^*$ are equal on these incoming neighbors, and that it only depend on neighbor $a$ in the case of a local map ${\gateLoop(a,b)}$. In any case, we deduce ${F^*(x^*)_v=F(x)_v}$ by definition of $F^\ast$. 
\end{proof}

$\Gwfa$-networks are close to Boolean conjunctive networks as shown by the previous lemma. The following result shows that this translates into strong limitations in their ability to produce large periods and prevents them to be universal.

\begin{theorem}
  \label{theo:transient-nonuniversal}
  The family of $\Gwfa$-networks is not universal.
\end{theorem}
\begin{proof}
  Consider a Boolean conjunctive automata network $F$, a configuration $x$ with periodic orbit under $F$ and some node $v$ such that there is a walk of length $L$ from $v$ to $v$. We claim that $x_v=F^L(x)_v$ so the trace at node $v$ in $x$ is periodic of period less than $L$. Indeed, in a conjunctive network state $0$ is spreading so clearly if $x_v=0$ then $F^L(x)_v=0$ and, more generally, ${F^{kL}(x)_v=0}$ for any ${k\geq 1}$. On the contrary, if ${x_v=1}$ then we can't have ${F^L(x)_v=0}$ because then ${F^{Pk}(x)_v=0}$ with $P$ the period of $x$ which would imply $x_v=0$.

  With the same reasoning, if we consider any $\Gwfa$-network $F$, any configuration $x$ with periodic orbit and some node $v$ such that there is a walk of length $L$ from $v$ to $v$, then it holds: 
  \[x_v=2 \Leftrightarrow F^L(x)_v=2.\]
  We deduce thanks to Lemma~\ref{lem:superconjunctive} that for any configuration $x$ with periodic orbit of some $\Gwfa$-network $F$ with $n$ nodes, and for any node $v$ belonging to some strongly connected component, the period of the trace at $v$ starting from $x$ is less than ${n^2}$: it is a periodic pattern of presence of state $2$ of length less than $n$ superposed on a periodic trace on $\{0,1\}$ of length less than $n$. We conclude that the family of $\Gwfa$-networks cannot be universal thanks to Theorem~\ref{coro:trickyuniversal}. 
\end{proof}

\subsection{Simulations and update mode extensions}

In the the next sections, we establish various strong universality results concerning families which are block sequential, local clocks or periodic extensions, \textit{i.e.} of the form ${\blocfamily{\mathcal{F}}{b}}$ or ${\clockfamily{\mathcal{F}}{c}}$ or ${\periodfamily{\mathcal{F}}{p}}$ for some CSAN family $\mathcal{F}$.
Automata networks from such families have a composite alphabet ${Q\times Q'}$ where the extension component $Q'$ is used to code the update mode.
If such an automata network $F$ simulates some automata network $G$, it could potentially use distinct values on both components of states in the block embedding in order to represent different initial configurations for $G$.
Since the $Q'$ component determines the update mode, this means that $F$ could potentially ``be in distinct update modes'' in the same simulation of $G$, depending on the initial configuration of $G$ to simulate.
It would then be illegitimate to say that some automata network under some update mode is simulating $G$.
On the contrary, if the block embedding is such that the $Q'$ component is constant whatever the configuration of $G$ to be represented, then all the orbits of $F$ used to simulate $G$ correspond to the exact same update mode.
In this case, we say that the simulation is \emph{update mode coherent}, and it is legitimate to speak about an automata network under a specific update mode being simulating another one.

First, we would like to make it clear that all results of Section~\ref{sec:concretesimulationresults} are obtained by update mode coherent simulations:
it can be checked that all the pseudo orbits of all the gadgets used are indeed update mode coherent (the extension component of state is always the same, whatever the configuration being encoded).

Second, we would like point out that in fact if some family ${\blocfamily{\mathcal{F}}{b}}$ or ${\clockfamily{\mathcal{F}}{c}}$ or ${\periodfamily{\mathcal{F}}{p}}$ is (strongly) universal then it is also (strongly) universal using only update mode coherent simulations.
The reason for this is as follows: inside the extension component $Q'$ coding the update mode at some node $v$, there are (at most) two components, one being constant (for local clocks or periodic extensions) and one having a deterministic cyclic behavior of the type ${i\mapsto i+1 \bmod p}$ where $p$ is the period of the update mode at the considered node $v$. Since the block embedding of any simulation must be injective, if node $v$ is part of some block to encode states $q$ and $q'$ of some simulated automata network $G$, and if $G$ allows both transitions from $q$ to $q'$ and from $q$ to $q$, then:
\begin{itemize}
\item the constant part of the extension component $Q'$ must be the same in the coding of $q$ and $q'$ (to make possible the transition from $q$ to $q'$);
\item the time constant of the simulation must be a multiple of the period of the update mode at $v$ (to allow transition from $q$ to $q$), and therefore the cyclic part of the extension component must hold the same value in the coding of $q$ and $q'$ (to allow transition from $q$ to $q'$).
\end{itemize}
We deduce that if $G$ can transition from any state $q$ to any state $q'$ at each node, then simulation of $G$ is forced to be update mode coherent.
Finally, observe that any automata network in the family of $\Gmon$-networks has the property of allowing a transition from any $q$ to any $q'$ at any node, therefore a (strongly) universal family ${\blocfamily{\mathcal{F}}{b}}$ or ${\clockfamily{\mathcal{F}}{c}}$ or ${\periodfamily{\mathcal{F}}{p}}$ must simulate the family of $\Gmon$-networks using update mode coherent simulations only. By using the fact that $\Gmon$-networks are strongly universal, we deduce that the considered family can actually simulate any network using update mode coherent simulations.

\section{Effect of asynchronism: a case study of symmetric networks}
\label{sec:concretesimulationresults}


In this section, we focus on studying concrete symmetric automata network (CSAN) families. We use previous theoretical framework on complexity of automata networks families in order to classify different CSAN families according to their dynamical behavior under different update schemes. More precisely, we focus on two different families of CSAN: signed conjunctive networks and min-max networks. We distinguish three main subfamilies inside signed conjunctive networks (for detailed definitions see Definition \ref{def:conjfamily} and Definition \ref{def:minmaxfamily} ) :
\begin{itemize}
 \item symmetric conjunctive networks, which are regular conjunctive networks (in which all edge labels are the identity function); 
 \item locally positive signed conjunctive networks, in which we allow negative edges (labeled by the switch function: $\text{Switch}(x) = 1-x$) but with a local constraint forcing the existence of at least one positive edge in each neighborhood (one edge labeled by the identity function); and 
 \item general signed conjunctive networks, in which there could be negative edges without any constraint (possibly all edges can be negative). 
 \end{itemize}
 In addition, we consider the latter presented three update schemes: block sequential, local clocks and general periodic update scheme. We classify previous families according to their dynamical behavior and simulation capabilities by using the  framework presented in previous sections
Before we enter into the detail, we present in Table \ref{tab:resultsummary} a summary of the main results obtained. 
\begin{table}[]
\centering
\resizebox{0.75\textwidth}{!}{%
\begin{tabular}{|c|c|c|c|c|}
\hline
\textit{Family/Update scheme} & Block sequential & Local clocks & General periodic  \\ \hline
Conjunctive networks          & BPA              & BPA          & \textbf{SPA}     \\ \hline
Locally positive              & BPA              & \textbf{SU}  & SU                    \\ \hline
Signed conjunctive networks   & \textbf{SU}      & SU           & SU              \\ \hline
Min-max networks              & \textbf{SU}      & SU           & SU                   \\ \hline
\end{tabular}%
}
\caption{Summary of the main results on complexity of the dynamics of the network families studied in the current chapter, depending on different update schemes. BPA = Bounded period attractors. SPA = Superpolynomial attractors. SU = Strong universality. Black fonts indicate the emergence of complex behavior such as long period attractors or universality. }
\label{tab:resultsummary}
\end{table}
Observe that in each row of Table \ref{tab:resultsummary}, we show how the dynamical behavior of some CSAN families changes as we change the update scheme. In particular, the most simple ones, such as conjunctive and locally positive networks exhibit a relatively simple dynamical behavior (they have bounded period attractors). Contrarily, the last two families have strong universality even for block sequential update schemes. In addition, we would like to remark that there is not only a hierarchy for update schemes (block sequential update schemes are a particular case of local clocks and both are a particular cases of periodic update schemes) but that network  families are also somehow related as conjunctive networks are a sub-family of all the other ones. 
Additionally, one can also observe that there is some sort of ''diagonal emergence'' of strong universality in Table \ref{tab:resultsummary} consisting in the fact that it seems to exists a trade-off between the complexity in the definition of some network families and the complexity of the corresponding update schemes. In other words, simple families seem to need more complex update schemes in order to be universal and as we pass to more complex rules one can observe this property for simpler update schemes. 

\subsection{Symmetric Conjunctive networks}

As said previously, symmetric conjunctive network form a particular CSAN family on alphabet ${\{0,1\}}$ where all edges are labeled by the identity map and all nodes have the same 'conjunctive' local map 
\[\lambda(q,X) =
  \begin{cases}
    0 &\text{ if }0\in X,\\
    1 &\text{ else.}
  \end{cases}
\]

First, symmetric conjunctive networks being a particular case of symmetric threshold networks it follows from the classical results of \cite{PaperGoles,GolesO80}, that they always converge to some fixed point or cycle of length two. Therefore they are not dynamically complex with parallel update schedule. It turns out that with a periodic update schedule of period 3 they can break this limitation and produce super-polynomial cycles.

Particularly, we show that symmetric conjunctive networks with a periodic update schedule of period $3$  can break latter limitation on attractor period and produce superpolynomial cycles.  Observe that all the graphs over we are defining networks in this family (and also in the other concrete examples we will be exploring in next sections) are non-directed (symmetric). If we observe carefully the effect of periodic update schemes,  we note that we are actually changing the interaction graph of the network by considering different orders for updating nodes and thus, breaking the symmetry in the different connections that nodes have in the network. Moreover, we corroborate this remark by showing that actually, we can simulate arbitrary conjunctive networks by using a periodic update scheme. We accomplish this by applying our formalism on simulation and gadget glueing. In fact, we show that the family of symmetric conjunctive networks admits coherent $\G_{\text{conj}}$-gadgets and thus it is capable of simulating the family $(\mathcal{F}_{\text{conj}},\mathcal{F}^{*}_{\text{conj}})$. Particularly, this implies that this family admits attractors of superpolynomial period.


\begin{proposition}
  Let $p\geq 1$ and denote by ${\periodfamily{\mathcal{F}}{p}_{\text{sym-conj}}}$ be the family of symmetric conjunctive networks under periodic update schedules of period $p$. Then, ${\periodfamily{\mathcal{F}}{p}_{\text{sym-conj}}}$ is not universal and therefore it does not admit coherent $\Gmon$-gadgets.
\end{proposition}
\begin{proof}
  We actually show that the longest transient of any ${F\in\mathcal{F}}$ with $n$ nodes is ${O(n^2)}$, then the conclusions follow by Corollary \ref{coro:trickyuniversal} and Theorem \ref{teo:gmonuniv}. Recall that the alphabet is ${\{0,1\}\times\{0,\cdots,p\}\times 2^{0,\cdots,p}}$ and, by definition, on any given configuration $x$ the component ${\{0,\cdots,p\}\times 2^{0,\cdots,p}}$ is periodic of period $p$ independently of the behavior on the first ${\{0,1\}}$ component. Moreover, the behavior on the ${\{0,1\}}$ component is that of a fixed (non-symmetric) conjunctive network $F'$ in the following sense: 
  \[F^{t+p}(x) = F'(F^t(x)), \forall t\geq 0.\]
  By \cite[Theorem 3.20]{De_Schutter_1999}, the transient of any orbit of $F'$ is ${O(n^2)}$. We deduce that the transient of the orbit of $x$ under $F$ is also ${O(n^2)}$. 
\end{proof}

\begin{theorem}
  Let  ${\periodfamily{\mathcal{F}}{3}_{\text{sym-conj}}}$ be the family of symmetric conjunctive networks under periodic update schedule of period $3$. ${\periodfamily{\mathcal{F}}{3}_{\text{sym-conj}}}$ has coherent $\Gconj$-gadgets and therefore simulates the family of conjunctive networks ${(\Fconj,\Fconj^*)}$ and for instance can produce superpolynomial cycles.
\end{theorem}
\begin{proof}
  The conclusion about simulations of conjunctive networks and superpolynomial cycles follows from previous results as soon as we prove that the family has coherent $\Gconj$-gadgets: Theorem~\ref{theo:Gconj-networks} for simulation of conjunctive networks and Proposition~\ref{prop:closuregadgets} to show that the family has coherent $\Gwire$-gadgets (because the identity map is obtained by composition of $\gateCOPY$ and $\gateAND$) and then Theorem~\ref{theo:non-polyn-cycl} for super-polynomial cycles.

  We now describe the coherent $\Gconj$-gadgets within family  ${\periodfamily{\mathcal{F}}{3}_{\text{sym-conj}}}$ using notations from Definition~\ref{def:coherent-gadgets}. Let $F_{\gateCOPY} \in {\periodfamily{\mathcal{F}}{3}_{\text{sym-conj}}}$ be defined on the following graph with nodes ${V_{\gateCOPY}=\{v_1,v_2,v_3,v_4,v_5,v_4',v_5'\}}$:
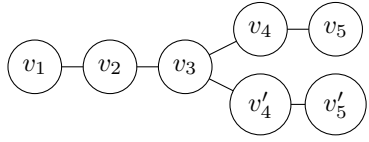
  \begin{center}
    \begin{tikzpicture}
      \draw (2,0) -- ++(1,-.5) node[fill=white,circle,draw=black] {$v_4'$} -- ++(1,0) node[fill=white,circle,draw=black] {$v_5'$};
      \draw (0,0) node[fill=white,circle,draw=black] {$v_1$} -- ++(1,0) node[fill=white,circle,draw=black] {$v_2$} -- ++(1,0) node[fill=white,circle,draw=black] {$v_3$} -- ++(1,.5) node[fill=white,circle,draw=black] {$v_4$} -- ++(1,0) node[fill=white,circle,draw=black] {$v_5$};
    \end{tikzpicture}
  \end{center}

  Let also $F_{\gateAND}\in\mathcal{F}$ be defined on the following graph with nodes ${V_{\gateAND}=\{v_1,v_2,v_1',v_2',v_3,v_4,v_5\}}$:
  \begin{center}
    \begin{tikzpicture}
      \draw (0,-.5) node[fill=white,circle,draw=black] {$v_1'$} -- ++(1,0) node[fill=white,circle,draw=black] {$v_2'$} -- ++(1,.5);
      \draw (0,.5) node[fill=white,circle,draw=black] {$v_1$} -- ++(1,0) node[fill=white,circle,draw=black] {$v_2$} -- ++(1,-.5) node[fill=white,circle,draw=black] {$v_3$} -- ++(1,0) node[fill=white,circle,draw=black] {$v_4$} -- ++(1,0) node[fill=white,circle,draw=black] {$v_5$};
    \end{tikzpicture}
  \end{center}
  
Recall that both $F_{\gateCOPY}$ and $F_{\gateAND}$ have alphabet ${Q=\{0,1\}\times\{0,1,2\}\times 2^{\{0,1,2\}}}$. 
  Now let ${C=C_i\cup C_o}$ be the glueing interface with ${C_i=\{i\}}$ and ${C_o=\{o\}}$. $F_{\gateCOPY}$ is seen as a gadget with one input and two outputs for the gate $\gateCOPY\in\Gconj$ as follows:
  \begin{itemize}
  \item $\phi_{\gateCOPY,1}^i(i)=v_2$ and ${\phi_{\gateCOPY,1}^i(o)=v_1}$;
  \item $\phi_{\gateCOPY,1}^o(i)=v_5$ and ${\phi_{\gateCOPY,1}^o(o)=v_4}$.
  \item $\phi_{\gateCOPY,2}^o(i)=v_5'$ and ${\phi_{\gateCOPY,2}^o(o)=v_4'}$.
  \end{itemize}
  $F_{\gateAND}$ is seen as a gadget with two inputs and one output for the gate $\gateAND\in\Gconj$ as follows:
  \begin{itemize}
  \item $\phi_{\gateAND,1}^i(i)=v_2$ and ${\phi_{\gateAND,1}^i(o)=v_1}$;
  \item $\phi_{\gateAND,2}^i(i)=v_2'$ and ${\phi_{\gateAND,2}^i(o)=v_1'}$.
  \item $\phi_{\gateAND,1}^o(i)=v_5$ and ${\phi_{\gateAND,1}^o(o)=v_4}$.
  \end{itemize}
  Note in particular that the conditions of Lemma~\ref{lem:csan-gadgets-glueing-closure} are satisfied so the closure by gadget glueing of this gadget stays in our family $\mathcal{F}$.
  Indeed $F_{\gateCOPY}$ has the same induced label graphs on all images of $C$ under $\phi$ maps, and the neighborhood of ${\phi_{\gateCOPY,1}^i(C_o)=v_1}$ is $v_2$ which belongs to ${\phi_{\gateCOPY,1}^i(C)}$, and similarly for other $\phi$ maps. Corresponding properties hold also for $F_{\gateAND}$.
  We now define the following elements required by Definition~\ref{def:coherent-gadgets}:
  \begin{itemize}
  \item the two state configurations $s_q$ for ${q\in\{0,1\}}$ are defined by ${s_q(i)=(q,0,\{0,2\})}$ and ${s_q(o)=(1,0,\{0,1\})}$;
  \item the context configuration is defined by ${c(v_3)=(1,0,\{1,2\})}$;
  \item the time constant is ${T=3}$
  \item the standard trace $\tau_{q,q'}$ over the glueing interface from ${q\in\{0,1\}}$ to ${q'\in\{0,1\}}$ is given by:
  \begin{footnotesize}
    \[ \begin{array}{c|c|c}
        \mathrm{time} & i & o\\
        \hline
        0 & (q,0,\{0,2\}) & (1,0,\{0,1\})\\
        1 & (1,1,\{0,2\}) & (q,1,\{0,1\})\\
        2 & (1,2,\{0,2\}) & (1,2,\{0,1\})\\
        3 & (q',0,\{0,2\}) & (1,0,\{0,1\})\\
       \end{array}
     \]
     \end{footnotesize}
   \item for any ${q_i,q_i^T,q_o,q_{o'}\in\{0,1\}}$ we have the following ${\{v_1,v_5,v_5'\}}$-pseudo orbit for $F_{\gateCOPY}$:
   \begin{tiny}
    \[ \begin{array}{c|c|c|c|c|c|c|c}
        \mathrm{time} & v_1 & v_2 & v_3 & v_4 & v_5 & v_4' & v_5' \\
        \hline
        0 & (q_i,0,\{0,2\}) & (1,0,\{0,1\}) & (1,0,\{1,2\}) & (q_o,0,\{0,2\}) & (1,0,\{0,1\}) & (q_o',0,\{0,2\}) & (1,0,\{0,1\}) \\
        1 & (1,1,\{0,2\}) & (q_i,1,\{0,1\}) & (1,1,\{1,2\}) & (1,1,\{0,2\}) & (q_o,1,\{0,1\}) & (1,1,\{0,2\}) & (q_o',1,\{0,1\})\\
        2 & (1,2,\{0,2\}) & (1,2,\{0,1\}) & (q_i,2,\{1,2\}) & (1,2,\{0,2\}) & (1,2,\{0,1\}) & (1,2,\{0,2\}) & (1,2,\{0,1\})\\
        3 & (q_i',0,\{0,2\}) & (1,0,\{0,1\}) & (1,0,\{1,2\}) & (q_i,0,\{0,2\}) & (1,0,\{0,1\}) & (q_i,0,\{0,2\}) & (1,0,\{0,1\}) \\
       \end{array}
     \]
     \end{tiny}
   \item for any ${q_i,q_i^T,q_{i'},q_{i'}^T,q_o\in\{0,1\}}$ we have the following ${\{v_1,v_1',v_5\}}$-pseudo orbit for $F_{\gateAND}$:
   \begin{tiny}
    \[ \begin{array}{c|c|c|c|c|c|c|c}
        \mathrm{time} & v_1 & v_2 & v_1' & v_2' & v_3 & v_4 & v_5  \\
        \hline
        0 & (q_i,0,\{0,2\}) & (1,0,\{0,1\}) & (q_{i'},0,\{0,2\}) & (1,0,\{0,1\}) & (1,0,\{1,2\}) & (q_o,0,\{0,2\}) & (1,0,\{0,1\}) \\
        1 & (1,1,\{0,2\}) & (q_i,1,\{0,1\}) & (1,1,\{0,2\}) & (q_{i'},1,\{0,1\}) & (1,1,\{1,2\}) & (1,1,\{0,2\}) & (q_o,1,\{0,1\})\\
        2 & (1,2,\{0,2\}) & (1,2,\{0,1\}) & (1,2,\{0,2\}) & (1,2,\{0,1\}) & (q_i\wedge q_{i'},2,\{1,2\}) & (1,2,\{0,2\}) & (1,2,\{0,1\}) \\
        3 & (q_i^T,0,\{0,2\}) & (1,0,\{0,1\}) & (q_{i'}^T,0,\{0,2\}) & (1,0,\{0,1\}) & (1,0,\{1,2\}) & (q_i\wedge q_{i'},0,\{0,2\}) & (1,0,\{0,1\}) \\
       \end{array}
     \]
         \end{tiny}

   \end{itemize}
\end{proof}

Interestingly, local clock update schedules on conjunctive networks are not able to produce superpolynomial cycles. To avoid to much conflicts in indices in notations, we denote by ${x^{Q}}$, $x^c$ and $x^m$ the three components of a configuration $x$ in some local clocks extension network.

\begin{lemma} \label{lem:locconj}

  Let us fix any $c>0$ and consider the family  ${\clockfamily{\mathcal{F}}{c}_{sym-conj}}$  of all symmetric conjunctive networks under local clocks update scheme with clock period $c$. Fix $n>0$ and let $F:  Q_{c}^{n} \to  Q_{c}^{n} \in {\clockfamily{\mathcal{F}}{c}_{sym-conj}}.$ For any configuration $(x^{Q},x^{c},x^{m}) \in Q_c^{n}$ we have that the period of the attractor reached from $(x^{Q},x^{c},x^{m})$ is at most $2 \lcm\{x^m_{v}: v \in \{1,\hdots,n\}\}.$ Moreover, for each attractor  $\overline{x}  \in \text{Att}(F)$, the set of nodes whose $Q$ component is not constant in $\overline{x}$ induces a bipartite subgraph.

\end{lemma}
\begin{proof}
	Let $(x^{Q},x^{c},x^{m}) \in Q_{c}^{n}$ be a configuration and  $(\overline{x}^{Q}, \overline{x}^{c},\overline{x}^{m})  \in \text{Att}(F)$ be an attractor that is reachable from $(x^{Q},x^{c},x^{m})$ and that is not a fixed point, i.e. $p(\overline{x}) \geq 2$. In order to simplify the notation, we are going to denote $\overline{x}^{Q}$ by $\overline{x}$. Let $i \in \{0,\hdots, n-1\}$ be a coordinate such that $\overline{x}_i$ changes its state, i.e there exists some $t \in \N$ such that $\overline{x}(0)_i \not = \overline{x}(t)_i$. Without loss of generality we assume that $\overline{x}(0)_i = 1$ and $\overline{x}^c(0)_i = 0$ and $\overline{x}(1)_i = 0$. Consider $t_1$ as the first time step such that the coordinate $i$ changes its state from $0$ to $1$, which means that, $t_1$ is the first time step such that $\overline{x}(t_1) = 0$ and $\overline{x}(t_1+1) = 1$. Observe that $t_{1} = sx^m_{i}$ for some $s \geq 1$. Note that, for all $j \in N(i)$, $\overline{x}(t_1)_j = 1$. Moreover, we have that for all $j \in N(i)$:  and that $\overline{x}(s)_j = 1$ for all $s \in [1, t_{1}]$ (see Figure \ref{fig:schemeloc}). This is because, since the interaction graph is symmetric, $j$ cannot be in state $0$ during interval $[1,t_1]$ otherwise both $i$ and $j$ would stay in state $0$ forever, thus contradicting the hypothesis on $i$. We deduce that $x^m_{j} \geq t_{1} = sx^m_{i}$ (otherwise $j$ would update its state and become $0$ on interval $[1,t_1]$). Now consider $t_0$ to be the first time step in which the node $i$ changes its state from $1$ to $0$, i.e. $\overline{x}(t_0) = 1$ and $\overline{x}(t_0+1) = 0$. Observe that, $t_{1} < t_{0} = t_{1} + s^{*}x^m_{i}$ for some $s^{*} \geq 1$. In addition, there must exist some neighbor $k \in N(i)$ satisfying that $\overline{x}(t_0)_{k} = 0$, otherwise $i$ cannot change to $0$ (it requires at least one neighbor in $0$ in order to change its state from $1$ to $0$). Observe that node $k$ satisfies $x^m_{k} \leq x^m_{i}$ because it needs to update to $0$ before node $i$. More precisely, by the definition of $t_{0}$ we have that $i$ is fixed in state between $t_{1}$ and $t_{0}$ (see Figure \ref{fig:schemeloc}). Additionally, we have that $\overline{x}_{k}(t_{1}) = 1,$ $\overline{x}_{k}(t_{0}) = 0$ and also we have that $\overline{x}(t_{0}-x^m_{i})_{k} = 1$ (otherwise it contradicts the minimality of $t_{0}$).  Finally, since $i$ is remains in state $0$ on the interval $[t_{0}, t_{0} + x^m_{i}]$ then,  $k$ cannot be updated in the same interval. Thus, $x^m_{k} \leq x^m_{i}.$  Moreover, the latter observations imply that $i$ and $k$ are synchronized, i.e. $(\overline{x}^{c}(0))_{k} = (\overline{x}^{c}(0))_{i},$ $x^m_{k} = x^m_{i},$ $t_{1} = x^m_{i}$ and $t_{0} = t_{1} + x^m_{i}.$   Note that also, we have that for all $t$, $\overline{x}(t)_i = 1- \overline{x}(t)_k.$
We have shown that the period of any node $v$ is at most ${2x^m_v}$, so we deduce $p(\overline{x}) \leq 2 \lcm\{x^m_{v}: v \in \{1,\hdots,n\}\}$.

	At this point, we know that $i$ must have at least one neighbor that is not constant in $\overline{x}$ and that it is synchronized. Let us assume that there is a non constant neighbor $\ell$ of $i$ that satisfies $x^m_{\ell} > x^m_{i}.$ On the other hand, we have that $\ell$ is in state $1$ on the interval $[0,t_{1}]$ (see Figure \ref{fig:schemeloc}) because otherwise $i$ cannot switch to state $1$ at the time step $t_{1}$. Observe that, by hypothesis, $\ell$ cannot change its state on intervals of the form $[rx^m_{i}, (r+1)x^m_{i})$ for $r \in \N$ even since $i$ is in state $0$ on those intervals (otherwise $i$ cannot switch back to $1$ because it would have a neighbor in $0$).  However, for $r$ even, $i$ is in state $1$ on intervals of the form $[rx^m_{i}, (r+1)x^m_{i}).$ Suppose that $\ell$ changes its value for the first time on an interval of the form $[r^{*}x^m_{i}, (r^{*}+1)x^m_{i})$ for some $r^{*} \in \N$ odd, i.e. $\overline{x}(s)_{\ell} = 0$ for some $s \in [r^{*}x^m_{i}, (r^{*}+1)x^m_{i}).$ Observe now that  $\overline{x}((r^{*}+2)x^m_{i}) = 1$ since $i$ must return to state $1$ but  $\ell$ cannot change its state in $[(r^{*}+1)x^m_{i}, (r^{*}+2)x^m_{i})$ because $i$ is in state $0$. Then, we must have  $x^m_{\ell}  \leq x^m_{i},$ which contradicts the hypothesis. We conclude that every non constant neighbor of $i$ is synchronized. Repeating the same argument now for any non constant neighbor of $i$ we have that all the nodes in the connected component containing $i$ have local delay $x^m_{i}$. Iterating this same technique now for each $i$ in the network, we deduce that $\overline{x}$ is such that $p(\overline{x}) \leq 2 \lcm\{x^m_{v}: v \in \{1,\hdots,n\}\}$ since locally, each connected component containing some node $i$ is synchronized and thus, each non-constant node is switching its state every $2 x^m_{i}$ time steps. In addition, for each node $i$ every non-constant neighbor is in the state $0$ whenever $i$ is in the state $1$. Thus, the set of nodes which are not constant for  $(\overline{x}^{Q}, \overline{x}^{c},\overline{x}^{m})$ i.e. $S(\overline{x}^{Q}, \overline{x}^{c},\overline{x}^{m}),$ induces a two colorable subgraph. The result holds. 
\end{proof}
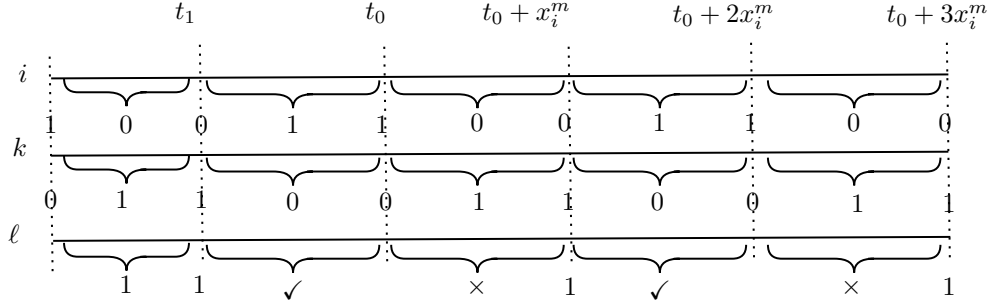
\begin{figure}
\tikzset{every picture/.style={line width=0.75pt}} 
\centering
\begin{tikzpicture}[x=0.75pt,y=0.75pt,yscale=-1,xscale=1]

\draw    (72,114) -- (524.5,112) ;
\draw    (72,153) -- (524.5,151) ;
\draw    (73,196) -- (525.5,194) ;
\draw  [dash pattern={on 0.84pt off 2.51pt}]  (240,98) -- (241.5,208) ;
\draw  [dash pattern={on 0.84pt off 2.51pt}]  (333,98) -- (334.5,208) ;
\draw  [dash pattern={on 0.84pt off 2.51pt}]  (425,96) -- (426.5,206) ;
\draw  [dash pattern={on 0.84pt off 2.51pt}]  (71.25,103) -- (72.75,213) ;
\draw   (150.5,115) .. controls (150.5,119.67) and (152.83,122) .. (157.5,122) -- (184,122) .. controls (190.67,122) and (194,124.33) .. (194,129) .. controls (194,124.33) and (197.33,122) .. (204,122)(201,122) -- (230.5,122) .. controls (235.17,122) and (237.5,119.67) .. (237.5,115) ;
\draw   (150.5,154) .. controls (150.5,158.67) and (152.83,161) .. (157.5,161) -- (184,161) .. controls (190.67,161) and (194,163.33) .. (194,168) .. controls (194,163.33) and (197.33,161) .. (204,161)(201,161) -- (230.5,161) .. controls (235.17,161) and (237.5,158.67) .. (237.5,154) ;
\draw   (150.5,198) .. controls (150.5,202.67) and (152.83,205) .. (157.5,205) -- (184,205) .. controls (190.67,205) and (194,207.33) .. (194,212) .. controls (194,207.33) and (197.33,205) .. (204,205)(201,205) -- (230.5,205) .. controls (235.17,205) and (237.5,202.67) .. (237.5,198) ;
\draw   (243.5,115) .. controls (243.5,119.67) and (245.83,122) .. (250.5,122) -- (277,122) .. controls (283.67,122) and (287,124.33) .. (287,129) .. controls (287,124.33) and (290.33,122) .. (297,122)(294,122) -- (323.5,122) .. controls (328.17,122) and (330.5,119.67) .. (330.5,115) ;
\draw   (243.5,154) .. controls (243.5,158.67) and (245.83,161) .. (250.5,161) -- (277,161) .. controls (283.67,161) and (287,163.33) .. (287,168) .. controls (287,163.33) and (290.33,161) .. (297,161)(294,161) -- (323.5,161) .. controls (328.17,161) and (330.5,158.67) .. (330.5,154) ;
\draw   (243.5,198) .. controls (243.5,202.67) and (245.83,205) .. (250.5,205) -- (277,205) .. controls (283.67,205) and (287,207.33) .. (287,212) .. controls (287,207.33) and (290.33,205) .. (297,205)(294,205) -- (323.5,205) .. controls (328.17,205) and (330.5,202.67) .. (330.5,198) ;

\draw   (335.5,115) .. controls (335.5,119.67) and (337.83,122) .. (342.5,122) -- (369,122) .. controls (375.67,122) and (379,124.33) .. (379,129) .. controls (379,124.33) and (382.33,122) .. (389,122)(386,122) -- (415.5,122) .. controls (420.17,122) and (422.5,119.67) .. (422.5,115) ;
\draw   (335.5,154) .. controls (335.5,158.67) and (337.83,161) .. (342.5,161) -- (369,161) .. controls (375.67,161) and (379,163.33) .. (379,168) .. controls (379,163.33) and (382.33,161) .. (389,161)(386,161) -- (415.5,161) .. controls (420.17,161) and (422.5,158.67) .. (422.5,154) ;
\draw   (335.5,198) .. controls (335.5,202.67) and (337.83,205) .. (342.5,205) -- (369,205) .. controls (375.67,205) and (379,207.33) .. (379,212) .. controls (379,207.33) and (382.33,205) .. (389,205)(386,205) -- (415.5,205) .. controls (420.17,205) and (422.5,202.67) .. (422.5,198) ;

\draw   (433.5,115) .. controls (433.5,119.67) and (435.83,122) .. (440.5,122) -- (467,122) .. controls (473.67,122) and (477,124.33) .. (477,129) .. controls (477,124.33) and (480.33,122) .. (487,122)(484,122) -- (513.5,122) .. controls (518.17,122) and (520.5,119.67) .. (520.5,115) ;
\draw   (433.5,154) .. controls (433.5,158.67) and (435.83,161) .. (440.5,161) -- (467,161) .. controls (473.67,161) and (477,163.33) .. (477,168) .. controls (477,163.33) and (480.33,161) .. (487,161)(484,161) -- (513.5,161) .. controls (518.17,161) and (520.5,158.67) .. (520.5,154) ;
\draw   (433.5,198) .. controls (433.5,202.67) and (435.83,205) .. (440.5,205) -- (467,205) .. controls (473.67,205) and (477,207.33) .. (477,212) .. controls (477,207.33) and (480.33,205) .. (487,205)(484,205) -- (513.5,205) .. controls (518.17,205) and (520.5,202.67) .. (520.5,198) ;

\draw   (78.5,114) .. controls (78.5,118.67) and (80.83,121) .. (85.5,121) -- (100,121) .. controls (106.67,121) and (110,123.33) .. (110,128) .. controls (110,123.33) and (113.33,121) .. (120,121)(117,121) -- (134.5,121) .. controls (139.17,121) and (141.5,118.67) .. (141.5,114) ;
\draw   (78.5,153) .. controls (78.5,157.67) and (80.83,160) .. (85.5,160) -- (100,160) .. controls (106.67,160) and (110,162.33) .. (110,167) .. controls (110,162.33) and (113.33,160) .. (120,160)(117,160) -- (134.5,160) .. controls (139.17,160) and (141.5,157.67) .. (141.5,153) ;
\draw   (78.5,197) .. controls (78.5,201.67) and (80.83,204) .. (85.5,204) -- (100,204) .. controls (106.67,204) and (110,206.33) .. (110,211) .. controls (110,206.33) and (113.33,204) .. (120,204)(117,204) -- (134.5,204) .. controls (139.17,204) and (141.5,201.67) .. (141.5,197) ;

\draw  [dash pattern={on 0.84pt off 2.51pt}]  (523.75,96) -- (525.25,206) ;
\draw  [dash pattern={on 0.84pt off 2.51pt}]  (147,97) -- (148.5,207) ;

\draw (463+6,213.4) node [anchor=north west][inner sep=0.75pt]    {$\times $};
\draw (372,214.4) node [anchor=north west][inner sep=0.75pt]    {$\checkmark $};
\draw (187,214.4) node [anchor=north west][inner sep=0.75pt]    {$\checkmark $};
\draw (274+6,213.4) node [anchor=north west][inner sep=0.75pt]    {$\times $};
\draw (67,169.4) node [anchor=north west][inner sep=0.75pt]    {$0$};
\draw (67,131.4) node [anchor=north west][inner sep=0.75pt]    {$1$};
\draw (492,75.4) node [anchor=north west][inner sep=0.75pt]    {$t_{0} +3 x^{m}_{i}$};
\draw (384,75.4) node [anchor=north west][inner sep=0.75pt]    {$t_{0} +2 x^{m}_{i}$};
\draw (288,74.4) node [anchor=north west][inner sep=0.75pt]    {$t_{0} + x^{m}_{i}$};
\draw (520,213.4) node [anchor=north west][inner sep=0.75pt]    {$1$};
\draw (329,213.4) node [anchor=north west][inner sep=0.75pt]    {$1$};
\draw (105,212.4) node [anchor=north west][inner sep=0.75pt]    {$1$};
\draw (142,212.4) node [anchor=north west][inner sep=0.75pt]    {$1$};
\draw (518,130.4) node [anchor=north west][inner sep=0.75pt]    {$0$};
\draw (520,171.4) node [anchor=north west][inner sep=0.75pt]    {$1$};
\draw (421,170.4) node [anchor=north west][inner sep=0.75pt]    {$0$};
\draw (420,130.4) node [anchor=north west][inner sep=0.75pt]    {$1$};
\draw (373,170.4) node [anchor=north west][inner sep=0.75pt]    {$0$};
\draw (328,170.4) node [anchor=north west][inner sep=0.75pt]    {$1$};
\draw (472,130.4) node [anchor=north west][inner sep=0.75pt]    {$0$};
\draw (474,171.4) node [anchor=north west][inner sep=0.75pt]    {$1$};
\draw (236,170.4) node [anchor=north west][inner sep=0.75pt]    {$0$};
\draw (283,170.4) node [anchor=north west][inner sep=0.75pt]    {$1$};
\draw (282,130.4) node [anchor=north west][inner sep=0.75pt]    {$0$};
\draw (234,131.4) node [anchor=north west][inner sep=0.75pt]    {$1$};
\draw (326,130.4) node [anchor=north west][inner sep=0.75pt]    {$0$};
\draw (374,130.4) node [anchor=north west][inner sep=0.75pt]    {$1$};
\draw (189,170.4) node [anchor=north west][inner sep=0.75pt]    {$0$};
\draw (189,131.4) node [anchor=north west][inner sep=0.75pt]    {$1$};
\draw (143,131.4) node [anchor=north west][inner sep=0.75pt]    {$0$};
\draw (143,169.4) node [anchor=north west][inner sep=0.75pt]    {$1$};
\draw (103,169.4) node [anchor=north west][inner sep=0.75pt]    {$1$};
\draw (105,131.4) node [anchor=north west][inner sep=0.75pt]    {$0$};
\draw (49,186.4) node [anchor=north west][inner sep=0.75pt]    {$\ell $};
\draw (51,142.4) node [anchor=north west][inner sep=0.75pt]    {$k$};
\draw (54,105.4) node [anchor=north west][inner sep=0.75pt]    {$i$};
\draw (229,74.4) node [anchor=north west][inner sep=0.75pt]    {$t_{0}$};
\draw (133,74.4) node [anchor=north west][inner sep=0.75pt]    {$t_{1}$};

\end{tikzpicture}
\label{fig:schemeloc}
\caption{Scheme of the dynamics of nodes $i$, $k$ and $\ell$ defined in the proof of Lemma \ref{lem:locconj}. The checkmarks indicate where it is feasible for $\ell$ to be updated and the crosses mark the intervals on which $\ell$ can change its state.}
\end{figure}

As seen above, there is a qualitative jump between local clocks and periodic update schedules for conjunctive networks in the size of dynamical cycles. However, for transient and prediction problems, even general periodic update schedules fail to produce maximal complexity (under standard complexity classes separations assumptions).

\begin{theorem}
  Let $p\in\N$ and consider the family ${\periodfamily{\mathcal{F}}{p}_{\text{sym-conj}}}$ of symmetric conjunctive networks under periodic update schedules of period $p$.  ${\periodfamily{\mathcal{F}}{p}_{\text{sym-conj}}}$ is neither dynamically nor computationally complex: more precisely, the transients of any network in ${\periodfamily{\mathcal{F}}{p}_{\text{sym-conj}}}$  with $n$ nodes are of length at most ${O(n^2)},$ the problem $\textsc{PREDu}_{\periodfamily{\mathcal{F}}{p}_{\text{sym-conj}}}$ can be solved by a $\textbf{NC}^2$ algorithm and $\textsc{PREDb}_{\periodfamily{\mathcal{F}}{p}_{\text{sym-conj}}}$ can be solved in polynomial time.
\end{theorem}
\begin{proof}
  Let $F \in {\periodfamily{\mathcal{F}}{p}_{\text{sym-conj}}}$ with $n$ nodes and consider any initial configuration $x$. By definition, the orbit of $x$ under $F^p$ is constant on the second and third component of sates, and the action of $F^p$ on the first component when starting from $x$ is a particular non-symmetric conjunctive network $F_x$ that can be seen as an arbitrary Boolean matrix $M_x$. First, by~\cite[Theorem 3.20]{De_Schutter_1999}, the transient of the orbit of $x$ under $F_x$ is of length at most ${2n^2-3n+2}$. We deduce that the transient of $x$ under $F$ is in ${O(n^2)}$.

  Second, it is easy to compute $M_x$ from $F$ and $x$ in $\textbf{NC}^1$. Moreover, matrix multiplication can be done in $\textbf{NC}^{1}$ and by fast exponentiation circuits we can compute ${M_x^t}$ with polynomial circuits of depth ${O(\log(t)\log(n))}$. With a constant computational overhead, we can therefore efficiently compute ${F^t(x)_v}$ and the complexity upper bounds on $\textsc{PREDu}_{\periodfamily{\mathcal{F}}{p}_{\text{sym-conj}}}$ and $\textsc{PREDb}_{\periodfamily{\mathcal{F}}{p}_{\text{sym-conj}}}$ follow. 
\end{proof}

\begin{remark}
We remark that even when symmetric conjunctive networks under periodic update schemes are not dynamically nor computationally complex, it has been shown that, with a different type of update schemes, it is possible to construct coherent $\Gmontwo$-gadgets and thus, show strong universality for this latter family. Particularly, we allude to the case of firing memory update schemes, which can also be studied in the context of our asynchronous extensions framework, as same as periodic update schemes. Just as we have mentioned before, this type of update scheme introduces internal clocks in each node, which will actually depend on the dynamics of the network (contrarily to the case of local clocks, where clocks are fixed and independent from the dynamics of the network). This key aspect gives the network strong dynamical and computational capabilities which can be used to simulate monotone boolean networks. Interested reader is referred to in order \cite{goles2020firing} to see details.
\end{remark}
 

\subsection{Locally positive symmetric signed conjunctive networks}

In this section, we study a generalization of conjunctive networks that we call \textit{locally positive} symmetric conjunctive networks.  Using latter notation for CSAN families we denote this family by $\mathcal{F}_{\text{locally-pos}}.$ In this particular case, we allow edges to have negative signs (which will switch the state of the corresponding neighbor) but with a local constraint: no neighborhood in which all the connections (remember that all edges are undirected) are negative is allowed. More precisely, a locally positive symmetric conjunctive network is a CSAN $(G,\lambda,\rho)$ in for any $v \in V(G)$ we have $\lambda_v(q,S) = \bigwedge \limits_{q \in S} q$ and there exists $w \in N(v): \rho(vw) = \text{Id}.$

We will show for this family that the \emph{threshold of universality} when changing update modes is between block sequential update schemes and local clocks update schemes.
More precisely,  we show, on one hand, that the family remains dynamically constrained under block sequential schedule, and, on the other hand, we show that a local clocks version of this  family is strongly universal as a consequence of its capability of simulating coherent $\Gmontwo$-gadgets. A similar construction will be used also for the general case in which we allow any kind of sign label over the edges of the network as we will show in next sections.

\begin{theorem}
  Fix any ${b\geq 1}$ and consider ${\blocfamily{\mathcal{F}}{b}_{\text{locally-pos}}}$ the family of all locally positive symmetric signed conjunctive networks under blocksequential schedule with at most $b$ blocks. Any periodic orbit of any ${F\in {\blocfamily{\mathcal{F}}{b}_{\text{locally-pos}}}}$ has length $1$ or ${2b}$.
\end{theorem}
\begin{proof}
  Take some configuration $x$ in a periodic orbit. If no node changes its state in the orbit then $x$ is actually a fixed-point. Otherwise take some node $i$ that changes its state and consider a maximal time interval ${I=[t_1;t_2]}$ with ${t_1>0}$ during which $i$ is in state $0$: ${\forall t\in I}$, ${F^{t}(x)_i=0}$ but ${F^{t_1-1}(x)_i=F^{t_2+1}(x)_i=1}$. Let ${j}$ be any positive neighborhood (i.e. such that ${\rho(i,j)}$ is the identity). First, we must have ${\forall t\in I}$, ${F^{t}(x)_j=1}$ because supposing ${F^t(x)_i=F^t(x)_j=0}$ implies ${F^{t'}(x)_i=F^{t'}(x)_j=0}$ for all ${t'\geq t}$ which would contradict the hypothesis that $i$ changes its state in the orbit. Thus ${t_2-t_1+1=b}$ because it is by definition a multiple of $b$ and if it were strictly larger than $b$ then node $j$ would be updated in the interval ${[t_1;t_2-1]}$ and therefore would turn into state $0$ in the interval $I$ which is impossible. The same argument actually shows that $i$ and $j$ must be updated synchronously. Therefore it is updated at time $t_2$ and we must have ${F^{t_2+1}(x)_j=0}$. This implies that ${F^{t_2+b+1}(x)_i=0}$ and shows that the maximal time interval starting from $t_2+1$ during which $i$ is in state $1$ is of length exactly $b$. We can then iterate this reasoning starting at time ${t_2+b+1}$ and we deduce that the orbit of $x$ at node $i$ alternates $b$ steps in state $0$ and $b$ steps in state $1$ forever. The same holds for any node that changes it state and finally we have shown that the orbit of $x$ is of period $1$ or $2b$. 
\end{proof}

%

Now, we want to show that coherent AND/OR gadgets can be implemented  in $\clockfamily{\mathcal{F}}{c}{\text{locally-pos}}$, i.e. we want to show that this family has coherent $\Gmontwo$ gadgets where $\Gmontwo = \{\text{AND}_{2},\text{OR}_{2}\}$ where $\text{OR}_{2}: \{0,1\}^{2} \to \{0,1\}^{2}$ is such that $\text{OR}_{2}(x,y) = (x \vee y, x \vee y$) and function $\text{AND}_{2}: \{0,1\}^{2} \to \{0,1\}^{2}$ is such that $\text{AND}_{2}(x,y) = (x \wedge y, x \wedge y)$. However, in order to accomplish this task we need a construction that we will be using for the next subsection.  Particularly, we need to implement the gadgets that are  shown in Figures \ref{fig:ORstruct} and \ref{fig:ANDstruct}.  Then, we will adapt latter gadgets in order to make it work for locally positive symmetric conjunctive networks. Thus, as a consequence, we will have that $\clockfamily{\mathcal{F}}{c}_{\text{locally-pos}}$ is strongly universal.

As  results on general signed symmetric conjunctive networks  are required first in order to show the proof of the main theorem of this section, we will just state the main result and then, we will show the proof of next theorem in the next section, in order to simplify things. By doing this we will respect the order given by hierarchy between different families and preserve coherence of results at the same time.

\begin{theorem} \label{teo:localuniv}
There exist $c>0$ such that the family  $\clockfamily{\mathcal{F}}{c}_{\text{locally-pos}}$  of all locally positive symmetric conjunctive networks under local clocks update scheme with clock parameter $c$ has coherent $\Gmontwo$-gadgets.
\end{theorem}
Finally, as a direct consequence of latter theorem we have the following corollary:
\begin{corollary}
	There exist $c>0$ such that the family $\clockfamily{\mathcal{F}}{c}_{\text{locally-pos}}$  of all locally positive symmetric conjunctive networks under local clocks update scheme with clock parameter $c$ is strongly universal. In particular, $\clockfamily{\mathcal{F}}{c}_{\text{locally-pos}}$  is both dynamically and computationally complex.
\end{corollary}
\begin{proof}
	Proof is a direct consequence of Theorem \ref{them:univ-rich-dynamics}, Corollary \ref{cor:universality} and Corollary \ref{cor:univfrommon}.
\end{proof}

\subsection{Symmetric signed conjunctive networks}
In this section we study conjunctive networks with negative edges without any local constraint in the number of positive edges. Formally the symmetric signed conjunctive networks family is a CSAN family in $\{0,1\}$ in which $\lambda_v: Q \times 2^Q \to Q$ is given by $\lambda(q,S) = 0 \text{ if } 0 \in S$ and $\lambda(q,S) = 1 \text{ if } 0 \not \in S$ and for any $e \in E$ we have $\rho_e \in \{\text{Id},\text{Switch}\}.$ where $\text{Switch}(x) = 1-x$. We denote previous family as $\mathcal{F}_{\text{sign-sym-conj}}$ and for a different update schemes we consider the notation $\blocfamily{\mathcal{F}}{b}_{\text{sign-sym-conj}}, \clockfamily{\mathcal{F}}{c}_{\text{sign-sym-conj}}$ and $\periodfamily{\mathcal{F}}{p}_{\text{sign-sym-conj}}$ for block sequential, local clocks and periodic versions of this family respectively.

We start by remarking that for the parallel update scheme, $\mathcal{F}_{\text{sign-sym-conj}}$ family is not universal as it is a type of threshold family and thus it have bounded period transient and attractors (see \cite{PaperGoles,GolesO80}). Then, a natural question is whether this remains true for other update schemes. In this sense, we are going to show that, when we consider the next update scheme in our hierarchy, the block sequential update scheme then, $\mathcal{F}_{\text{sign-sym-conj}}$ family is strongly universal.
%
 

\subsubsection{Block sequential case}
In this section, we will show that $\blocfamily{\mathcal{F}}{b}_{\text{sign-sym-conj}}$ is strongly universal as a consequence of its capability to implement coherent $\Gmontwo$-gadgets. In addition, we conclude that, as a direct consequence of the latter property, that previous family is both dynamically complex and computationally complex. This means that for this family, complex behavior is exhibited under block sequential update schemes. 
\begin{figure}
\centering

\begin{tikzpicture}[x=0.55pt,y=0.55pt,yscale=-1,xscale=1]

\draw  [line width=1.5]  (121,322) -- (156,357) -- (121,392) -- (86,357) -- cycle ;

\draw  [line width=1.5]  (121,60) -- (156,95) -- (121,130) -- (86,95) -- cycle ;

\draw  [line width=1.5]  (334,320) -- (369,355) -- (334,390) -- (299,355) -- cycle ;

\draw  [line width=1.5]  (334,58) -- (369,93) -- (334,128) -- (299,93) -- cycle ;

\draw  [line width=1.5]  (539,320) -- (574,355) -- (539,390) -- (504,355) -- cycle ;

\draw  [line width=1.5]  (539,58) -- (574,93) -- (539,128) -- (504,93) -- cycle ;

\draw    (121,130) -- (120.5,327) ;
\draw    (334,128) -- (333.5,325) ;
\draw    (539,128) -- (538.5,325) ;
\draw    (343.75,224) -- (523.94,224) ;
\draw    (122.5,225) -- (381.75,224) ;
\draw  [fill={rgb, 255:red, 255; green, 255; blue, 255 }  ,fill opacity=1 ] (284.19,197) -- (389.31,197) -- (389.31,251) -- (284.19,251) -- cycle ;
\draw  [fill={rgb, 255:red, 255; green, 255; blue, 255 }  ,fill opacity=1 ] (489.39,197) -- (594.5,197) -- (594.5,251) -- (489.39,251) -- cycle ;
\draw  [fill={rgb, 255:red, 255; green, 255; blue, 255 }  ,fill opacity=1 ] (75,197) -- (180.11,197) -- (180.11,251) -- (75,251) -- cycle ;

\draw  [dash pattern={on 4.5pt off 4.5pt}] (118.57,36) -- (175.5,36) -- (83.4,117) -- (26.47,117) -- cycle ;
\draw  [dash pattern={on 4.5pt off 4.5pt}] (120.57,298) -- (177.5,298) -- (85.4,379) -- (28.47,379) -- cycle ;
\draw  [dash pattern={on 4.5pt off 4.5pt}] (376.57,69) -- (433.5,69) -- (341.4,150) -- (284.47,150) -- cycle ;
\draw  [dash pattern={on 4.5pt off 4.5pt}] (329.57,294) -- (386.5,294) -- (294.4,375) -- (237.47,375) -- cycle ;
\draw  [dash pattern={on 4.5pt off 4.5pt}] (535.57,294) -- (592.5,294) -- (500.4,375) -- (443.47,375) -- cycle ;
\draw  [dash pattern={on 4.5pt off 4.5pt}] (571.57,69) -- (628.5,69) -- (536.4,150) -- (479.47,150) -- cycle ;
\draw  [dash pattern={on 5.63pt off 4.5pt}][line width=1.5]  (65.5,168) -- (601.5,168) -- (601.5,278) -- (65.5,278) -- cycle ;
\draw  [dash pattern={on 4.5pt off 4.5pt}] (101.56,113.5) .. controls (110.98,92.79) and (136.27,76) .. (158.06,76) .. controls (179.84,76) and (189.86,92.79) .. (180.44,113.5) .. controls (171.02,134.21) and (145.73,151) .. (123.94,151) .. controls (102.16,151) and (92.14,134.21) .. (101.56,113.5) -- cycle ;
\draw  [dash pattern={on 4.5pt off 4.5pt}] (103.56,379.5) .. controls (112.98,358.79) and (138.27,342) .. (160.06,342) .. controls (181.84,342) and (191.86,358.79) .. (182.44,379.5) .. controls (173.02,400.21) and (147.73,417) .. (125.94,417) .. controls (104.16,417) and (94.14,400.21) .. (103.56,379.5) -- cycle ;
\draw  [dash pattern={on 4.5pt off 4.5pt}] (314.56,371.5) .. controls (323.98,350.79) and (349.27,334) .. (371.06,334) .. controls (392.84,334) and (402.86,350.79) .. (393.44,371.5) .. controls (384.02,392.21) and (358.73,409) .. (336.94,409) .. controls (315.16,409) and (305.14,392.21) .. (314.56,371.5) -- cycle ;
\draw  [dash pattern={on 4.5pt off 4.5pt}] (522.56,372.5) .. controls (531.98,351.79) and (557.27,335) .. (579.06,335) .. controls (600.84,335) and (610.86,351.79) .. (601.44,372.5) .. controls (592.02,393.21) and (566.73,410) .. (544.94,410) .. controls (523.16,410) and (513.14,393.21) .. (522.56,372.5) -- cycle ;
\draw  [dash pattern={on 4.5pt off 4.5pt}] (276.56,73.5) .. controls (285.98,52.79) and (311.27,36) .. (333.06,36) .. controls (354.84,36) and (364.86,52.79) .. (355.44,73.5) .. controls (346.02,94.21) and (320.73,111) .. (298.94,111) .. controls (277.16,111) and (267.14,94.21) .. (276.56,73.5) -- cycle ;
\draw  [dash pattern={on 4.5pt off 4.5pt}] (479.56,71.5) .. controls (488.98,50.79) and (514.27,34) .. (536.06,34) .. controls (557.84,34) and (567.86,50.79) .. (558.44,71.5) .. controls (549.02,92.21) and (523.73,109) .. (501.94,109) .. controls (480.16,109) and (470.14,92.21) .. (479.56,71.5) -- cycle ;
\draw    (25.8,225.33) -- (72.5,225) ;
\draw    (594.8,224.33) -- (641.5,224) ;

\draw  [fill={rgb, 255:red, 255; green, 255; blue, 255 }  ,fill opacity=1 ]  (121, 392) circle [x radius= 13.6, y radius= 13.6]   ;
\draw (121,392) node    {$0$};
\draw  [fill={rgb, 255:red, 255; green, 255; blue, 255 }  ,fill opacity=1 ]  (86, 357) circle [x radius= 13.6, y radius= 13.6]   ;
\draw (86,357) node    {$0$};
\draw  [fill={rgb, 255:red, 255; green, 255; blue, 255 }  ,fill opacity=1 ]  (121, 322) circle [x radius= 13.6, y radius= 13.6]   ;
\draw (121,322) node    {$0$};
\draw  [fill={rgb, 255:red, 255; green, 255; blue, 255 }  ,fill opacity=1 ]  (156, 357) circle [x radius= 13.6, y radius= 13.6]   ;
\draw (156,357) node    {$0$};
\draw  [fill={rgb, 255:red, 255; green, 255; blue, 255 }  ,fill opacity=1 ]  (86, 95) circle [x radius= 13.6, y radius= 13.6]   ;
\draw (86,95) node    {$0$};
\draw  [fill={rgb, 255:red, 255; green, 255; blue, 255 }  ,fill opacity=1 ]  (121, 60) circle [x radius= 13.6, y radius= 13.6]   ;
\draw (121,60) node    {$0$};
\draw  [fill={rgb, 255:red, 255; green, 255; blue, 255 }  ,fill opacity=1 ]  (156, 95) circle [x radius= 13.6, y radius= 13.6]   ;
\draw (156,95) node    {$0$};
\draw  [fill={rgb, 255:red, 255; green, 255; blue, 255 }  ,fill opacity=1 ]  (121, 130) circle [x radius= 13.6, y radius= 13.6]   ;
\draw (121,130) node    {$0$};
\draw  [fill={rgb, 255:red, 255; green, 255; blue, 255 }  ,fill opacity=1 ]  (369, 355) circle [x radius= 13.6, y radius= 13.6]   ;
\draw (369,355) node    {$1$};
\draw  [fill={rgb, 255:red, 255; green, 255; blue, 255 }  ,fill opacity=1 ]  (334, 320) circle [x radius= 13.6, y radius= 13.6]   ;
\draw (334,320) node    {$0$};
\draw  [fill={rgb, 255:red, 255; green, 255; blue, 255 }  ,fill opacity=1 ]  (299, 355) circle [x radius= 13.6, y radius= 13.6]   ;
\draw (299,355) node    {$0$};
\draw  [fill={rgb, 255:red, 255; green, 255; blue, 255 }  ,fill opacity=1 ]  (334, 390) circle [x radius= 13.6, y radius= 13.6]   ;
\draw (334,390) node    {$1$};
\draw  [fill={rgb, 255:red, 255; green, 255; blue, 255 }  ,fill opacity=1 ]  (334, 128) circle [x radius= 13.6, y radius= 13.6]   ;
\draw (334,128) node    {$1$};
\draw  [fill={rgb, 255:red, 255; green, 255; blue, 255 }  ,fill opacity=1 ]  (369, 93) circle [x radius= 13.6, y radius= 13.6]   ;
\draw (369,93) node    {$1$};
\draw  [fill={rgb, 255:red, 255; green, 255; blue, 255 }  ,fill opacity=1 ]  (334, 58) circle [x radius= 13.6, y radius= 13.6]   ;
\draw (334,58) node    {$0$};
\draw  [fill={rgb, 255:red, 255; green, 255; blue, 255 }  ,fill opacity=1 ]  (299, 93) circle [x radius= 13.6, y radius= 13.6]   ;
\draw (299,93) node    {$0$};
\draw  [fill={rgb, 255:red, 255; green, 255; blue, 255 }  ,fill opacity=1 ]  (574, 355) circle [x radius= 13.6, y radius= 13.6]   ;
\draw (574,355) node    {$0$};
\draw  [fill={rgb, 255:red, 255; green, 255; blue, 255 }  ,fill opacity=1 ]  (539, 320) circle [x radius= 13.6, y radius= 13.6]   ;
\draw (539,320) node    {$1$};
\draw  [fill={rgb, 255:red, 255; green, 255; blue, 255 }  ,fill opacity=1 ]  (504, 355) circle [x radius= 13.6, y radius= 13.6]   ;
\draw (504,355) node    {$1$};
\draw  [fill={rgb, 255:red, 255; green, 255; blue, 255 }  ,fill opacity=1 ]  (539, 390) circle [x radius= 13.6, y radius= 13.6]   ;
\draw (539,390) node    {$0$};
\draw  [fill={rgb, 255:red, 255; green, 255; blue, 255 }  ,fill opacity=1 ]  (539, 128) circle [x radius= 13.6, y radius= 13.6]   ;
\draw (539,128) node    {$0$};
\draw  [fill={rgb, 255:red, 255; green, 255; blue, 255 }  ,fill opacity=1 ]  (574, 93) circle [x radius= 13.6, y radius= 13.6]   ;
\draw (574,93) node    {$0$};
\draw  [fill={rgb, 255:red, 255; green, 255; blue, 255 }  ,fill opacity=1 ]  (539, 58) circle [x radius= 13.6, y radius= 13.6]   ;
\draw (539,58) node    {$0$};
\draw  [fill={rgb, 255:red, 255; green, 255; blue, 255 }  ,fill opacity=1 ]  (504, 93) circle [x radius= 13.6, y radius= 13.6]   ;
\draw (504,93) node    {$0$};
\draw  [fill={rgb, 255:red, 155; green, 155; blue, 155 }  ,fill opacity=1 ]  (291,1) -- (339,1) -- (339,25) -- (291,25) -- cycle  ;
\draw (315,13) node    {$t\ =0$};
\draw (125.56,224) node    {$0$};
\draw (333.75,223.5) node    {$0$};
\draw (539.94,222) node    {$z$};
\draw  [fill={rgb, 255:red, 155; green, 155; blue, 155 }  ,fill opacity=1 ]  (209, 190) circle [x radius= 13.6, y radius= 13.6]   ;
\draw (209,190) node    {$1$};
\draw  [fill={rgb, 255:red, 155; green, 155; blue, 155 }  ,fill opacity=1 ]  (50, 61) circle [x radius= 13.6, y radius= 13.6]   ;
\draw (50,61) node    {$2$};
\draw  [fill={rgb, 255:red, 155; green, 155; blue, 155 }  ,fill opacity=1 ]  (42, 322) circle [x radius= 13.6, y radius= 13.6]   ;
\draw (42,322) node    {$1$};
\draw  [fill={rgb, 255:red, 155; green, 155; blue, 155 }  ,fill opacity=1 ]  (257, 321) circle [x radius= 13.6, y radius= 13.6]   ;
\draw (257,321) node    {$1$};
\draw  [fill={rgb, 255:red, 155; green, 155; blue, 155 }  ,fill opacity=1 ]  (469, 319) circle [x radius= 13.6, y radius= 13.6]   ;
\draw (469,319) node    {$1$};
\draw  [fill={rgb, 255:red, 155; green, 155; blue, 155 }  ,fill opacity=1 ]  (619, 118) circle [x radius= 13.6, y radius= 13.6]   ;
\draw (619,118) node    {$2$};
\draw  [fill={rgb, 255:red, 155; green, 155; blue, 155 }  ,fill opacity=1 ]  (409, 118) circle [x radius= 13.6, y radius= 13.6]   ;
\draw (409,118) node    {$2$};
\draw  [fill={rgb, 255:red, 155; green, 155; blue, 155 }  ,fill opacity=1 ]  (262, 61) circle [x radius= 13.6, y radius= 13.6]   ;
\draw (262,61) node    {$3$};
\draw  [fill={rgb, 255:red, 155; green, 155; blue, 155 }  ,fill opacity=1 ]  (392, 405) circle [x radius= 13.6, y radius= 13.6]   ;
\draw (392,405) node    {$3$};
\draw  [fill={rgb, 255:red, 155; green, 155; blue, 155 }  ,fill opacity=1 ]  (612, 401) circle [x radius= 13.6, y radius= 13.6]   ;
\draw (612,401) node    {$3$};
\draw  [fill={rgb, 255:red, 155; green, 155; blue, 155 }  ,fill opacity=1 ]  (195, 403) circle [x radius= 13.6, y radius= 13.6]   ;
\draw (195,403) node    {$3$};
\draw (17,223) node    {$x$};
\draw  [fill={rgb, 255:red, 155; green, 155; blue, 155 }  ,fill opacity=1 ]  (203, 118) circle [x radius= 13.6, y radius= 13.6]   ;
\draw (203,118) node    {$3$};
\draw  [fill={rgb, 255:red, 155; green, 155; blue, 155 }  ,fill opacity=1 ]  (466, 61) circle [x radius= 13.6, y radius= 13.6]   ;
\draw (466,61) node    {$3$};


\end{tikzpicture}
\begin{tikzpicture}[x=0.55pt,y=0.55pt,yscale=-1,xscale=1]

\draw  [line width=1.5]  (129,325) -- (164,360) -- (129,395) -- (94,360) -- cycle ;

\draw  [line width=1.5]  (129,63) -- (164,98) -- (129,133) -- (94,98) -- cycle ;

\draw  [line width=1.5]  (342,323) -- (377,358) -- (342,393) -- (307,358) -- cycle ;

\draw  [line width=1.5]  (342,61) -- (377,96) -- (342,131) -- (307,96) -- cycle ;

\draw  [line width=1.5]  (547,323) -- (582,358) -- (547,393) -- (512,358) -- cycle ;

\draw  [line width=1.5]  (547,61) -- (582,96) -- (547,131) -- (512,96) -- cycle ;

\draw    (601.5,226.77) -- (642.8,226.5) ;
\draw    (129,133) -- (128.5,330) ;
\draw    (342,131) -- (341.5,328) ;
\draw    (547,131) -- (546.5,328) ;
\draw    (351.75,227) -- (531.94,227) ;
\draw    (130.5,228) -- (389.75,227) ;
\draw  [fill={rgb, 255:red, 255; green, 255; blue, 255 }  ,fill opacity=1 ] (292.19,200) -- (397.31,200) -- (397.31,254) -- (292.19,254) -- cycle ;
\draw  [fill={rgb, 255:red, 255; green, 255; blue, 255 }  ,fill opacity=1 ] (497.39,200) -- (602.5,200) -- (602.5,254) -- (497.39,254) -- cycle ;
\draw  [fill={rgb, 255:red, 255; green, 255; blue, 255 }  ,fill opacity=1 ] (83,200) -- (188.11,200) -- (188.11,254) -- (83,254) -- cycle ;

\draw  [dash pattern={on 4.5pt off 4.5pt}] (126.57,39) -- (183.5,39) -- (91.4,120) -- (34.47,120) -- cycle ;
\draw  [dash pattern={on 4.5pt off 4.5pt}] (128.57,301) -- (185.5,301) -- (93.4,382) -- (36.47,382) -- cycle ;
\draw  [dash pattern={on 4.5pt off 4.5pt}] (384.57,72) -- (441.5,72) -- (349.4,153) -- (292.47,153) -- cycle ;
\draw  [dash pattern={on 4.5pt off 4.5pt}] (337.57,297) -- (394.5,297) -- (302.4,378) -- (245.47,378) -- cycle ;
\draw  [dash pattern={on 4.5pt off 4.5pt}] (543.57,297) -- (600.5,297) -- (508.4,378) -- (451.47,378) -- cycle ;
\draw  [dash pattern={on 4.5pt off 4.5pt}] (579.57,72) -- (636.5,72) -- (544.4,153) -- (487.47,153) -- cycle ;
\draw  [dash pattern={on 5.63pt off 4.5pt}][line width=1.5]  (73.5,172) -- (609.5,172) -- (609.5,282) -- (73.5,282) -- cycle ;
\draw  [dash pattern={on 4.5pt off 4.5pt}] (109.56,116.5) .. controls (118.98,95.79) and (144.27,79) .. (166.06,79) .. controls (187.84,79) and (197.86,95.79) .. (188.44,116.5) .. controls (179.02,137.21) and (153.73,154) .. (131.94,154) .. controls (110.16,154) and (100.14,137.21) .. (109.56,116.5) -- cycle ;
\draw  [dash pattern={on 4.5pt off 4.5pt}] (111.56,382.5) .. controls (120.98,361.79) and (146.27,345) .. (168.06,345) .. controls (189.84,345) and (199.86,361.79) .. (190.44,382.5) .. controls (181.02,403.21) and (155.73,420) .. (133.94,420) .. controls (112.16,420) and (102.14,403.21) .. (111.56,382.5) -- cycle ;
\draw  [dash pattern={on 4.5pt off 4.5pt}] (322.56,374.5) .. controls (331.98,353.79) and (357.27,337) .. (379.06,337) .. controls (400.84,337) and (410.86,353.79) .. (401.44,374.5) .. controls (392.02,395.21) and (366.73,412) .. (344.94,412) .. controls (323.16,412) and (313.14,395.21) .. (322.56,374.5) -- cycle ;
\draw  [dash pattern={on 4.5pt off 4.5pt}] (530.56,375.5) .. controls (539.98,354.79) and (565.27,338) .. (587.06,338) .. controls (608.84,338) and (618.86,354.79) .. (609.44,375.5) .. controls (600.02,396.21) and (574.73,413) .. (552.94,413) .. controls (531.16,413) and (521.14,396.21) .. (530.56,375.5) -- cycle ;
\draw  [dash pattern={on 4.5pt off 4.5pt}] (284.56,76.5) .. controls (293.98,55.79) and (319.27,39) .. (341.06,39) .. controls (362.84,39) and (372.86,55.79) .. (363.44,76.5) .. controls (354.02,97.21) and (328.73,114) .. (306.94,114) .. controls (285.16,114) and (275.14,97.21) .. (284.56,76.5) -- cycle ;
\draw  [dash pattern={on 4.5pt off 4.5pt}] (487.56,74.5) .. controls (496.98,53.79) and (522.27,37) .. (544.06,37) .. controls (565.84,37) and (575.86,53.79) .. (566.44,74.5) .. controls (557.02,95.21) and (531.73,112) .. (509.94,112) .. controls (488.16,112) and (478.14,95.21) .. (487.56,74.5) -- cycle ;
\draw    (34.8,227.33) -- (81.5,227) ;

\draw  [fill={rgb, 255:red, 255; green, 255; blue, 255 }  ,fill opacity=1 ]  (129, 395) circle [x radius= 13.6, y radius= 13.6]   ;
\draw (129,395) node    {$0$};
\draw  [fill={rgb, 255:red, 255; green, 255; blue, 255 }  ,fill opacity=1 ]  (94, 360) circle [x radius= 13.6, y radius= 13.6]   ;
\draw (94,360) node    {$1$};
\draw  [fill={rgb, 255:red, 255; green, 255; blue, 255 }  ,fill opacity=1 ]  (129, 325) circle [x radius= 13.6, y radius= 13.6]   ;
\draw (129,325) node    {$1$};
\draw  [fill={rgb, 255:red, 255; green, 255; blue, 255 }  ,fill opacity=1 ]  (164, 360) circle [x radius= 13.6, y radius= 13.6]   ;
\draw (164,360) node    {$0$};
\draw  [fill={rgb, 255:red, 255; green, 255; blue, 255 }  ,fill opacity=1 ]  (94, 98) circle [x radius= 13.6, y radius= 13.6]   ;
\draw (94,98) node    {$1$};
\draw  [fill={rgb, 255:red, 255; green, 255; blue, 255 }  ,fill opacity=1 ]  (129, 63) circle [x radius= 13.6, y radius= 13.6]   ;
\draw (129,63) node    {$1$};
\draw  [fill={rgb, 255:red, 255; green, 255; blue, 255 }  ,fill opacity=1 ]  (164, 98) circle [x radius= 13.6, y radius= 13.6]   ;
\draw (164,98) node    {$0$};
\draw  [fill={rgb, 255:red, 255; green, 255; blue, 255 }  ,fill opacity=1 ]  (129, 133) circle [x radius= 13.6, y radius= 13.6]   ;
\draw (129,133) node    {$0$};
\draw  [fill={rgb, 255:red, 255; green, 255; blue, 255 }  ,fill opacity=1 ]  (377, 358) circle [x radius= 13.6, y radius= 13.6]   ;
\draw (377,358) node    {$0$};
\draw  [fill={rgb, 255:red, 255; green, 255; blue, 255 }  ,fill opacity=1 ]  (342, 323) circle [x radius= 13.6, y radius= 13.6]   ;
\draw (342,323) node    {$0$};
\draw  [fill={rgb, 255:red, 255; green, 255; blue, 255 }  ,fill opacity=1 ]  (307, 358) circle [x radius= 13.6, y radius= 13.6]   ;
\draw (307,358) node    {$0$};
\draw  [fill={rgb, 255:red, 255; green, 255; blue, 255 }  ,fill opacity=1 ]  (342, 393) circle [x radius= 13.6, y radius= 13.6]   ;
\draw (342,393) node    {$0$};
\draw  [fill={rgb, 255:red, 255; green, 255; blue, 255 }  ,fill opacity=1 ]  (342, 131) circle [x radius= 13.6, y radius= 13.6]   ;
\draw (342,131) node    {$0$};
\draw  [fill={rgb, 255:red, 255; green, 255; blue, 255 }  ,fill opacity=1 ]  (377, 96) circle [x radius= 13.6, y radius= 13.6]   ;
\draw (377,96) node    {$0$};
\draw  [fill={rgb, 255:red, 255; green, 255; blue, 255 }  ,fill opacity=1 ]  (342, 61) circle [x radius= 13.6, y radius= 13.6]   ;
\draw (342,61) node    {$1$};
\draw  [fill={rgb, 255:red, 255; green, 255; blue, 255 }  ,fill opacity=1 ]  (307, 96) circle [x radius= 13.6, y radius= 13.6]   ;
\draw (307,96) node    {$1$};
\draw  [fill={rgb, 255:red, 255; green, 255; blue, 255 }  ,fill opacity=1 ]  (582, 358) circle [x radius= 13.6, y radius= 13.6]   ;
\draw (582,358) node    {$1$};
\draw  [fill={rgb, 255:red, 255; green, 255; blue, 255 }  ,fill opacity=1 ]  (547, 323) circle [x radius= 13.6, y radius= 13.6]   ;
\draw (547,323) node    {$0$};
\draw  [fill={rgb, 255:red, 255; green, 255; blue, 255 }  ,fill opacity=1 ]  (512, 358) circle [x radius= 13.6, y radius= 13.6]   ;
\draw (512,358) node    {$0$};
\draw  [fill={rgb, 255:red, 255; green, 255; blue, 255 }  ,fill opacity=1 ]  (547, 393) circle [x radius= 13.6, y radius= 13.6]   ;
\draw (547,393) node    {$1$};
\draw  [fill={rgb, 255:red, 255; green, 255; blue, 255 }  ,fill opacity=1 ]  (547, 131) circle [x radius= 13.6, y radius= 13.6]   ;
\draw (547,131) node    {$1$};
\draw  [fill={rgb, 255:red, 255; green, 255; blue, 255 }  ,fill opacity=1 ]  (582, 96) circle [x radius= 13.6, y radius= 13.6]   ;
\draw (582,96) node    {$1$};
\draw  [fill={rgb, 255:red, 255; green, 255; blue, 255 }  ,fill opacity=1 ]  (547, 61) circle [x radius= 13.6, y radius= 13.6]   ;
\draw (547,61) node    {$0$};
\draw  [fill={rgb, 255:red, 255; green, 255; blue, 255 }  ,fill opacity=1 ]  (512, 96) circle [x radius= 13.6, y radius= 13.6]   ;
\draw (512,96) node    {$0$};
\draw (133.56,226) node    {$\overline{x}$};
\draw (341.75,226.5) node    {$0$};
\draw (547.94,225) node    {$0$};
\draw  [fill={rgb, 255:red, 155; green, 155; blue, 155 }  ,fill opacity=1 ]  (217, 193) circle [x radius= 13.6, y radius= 13.6]   ;
\draw (217,193) node    {$1$};
\draw  [fill={rgb, 255:red, 155; green, 155; blue, 155 }  ,fill opacity=1 ]  (58, 64) circle [x radius= 13.6, y radius= 13.6]   ;
\draw (58,64) node    {$2$};
\draw  [fill={rgb, 255:red, 155; green, 155; blue, 155 }  ,fill opacity=1 ]  (50, 325) circle [x radius= 13.6, y radius= 13.6]   ;
\draw (50,325) node    {$1$};
\draw  [fill={rgb, 255:red, 155; green, 155; blue, 155 }  ,fill opacity=1 ]  (265, 324) circle [x radius= 13.6, y radius= 13.6]   ;
\draw (265,324) node    {$1$};
\draw  [fill={rgb, 255:red, 155; green, 155; blue, 155 }  ,fill opacity=1 ]  (477, 322) circle [x radius= 13.6, y radius= 13.6]   ;
\draw (477,322) node    {$1$};
\draw  [fill={rgb, 255:red, 155; green, 155; blue, 155 }  ,fill opacity=1 ]  (627, 121) circle [x radius= 13.6, y radius= 13.6]   ;
\draw (627,121) node    {$2$};
\draw  [fill={rgb, 255:red, 155; green, 155; blue, 155 }  ,fill opacity=1 ]  (417, 121) circle [x radius= 13.6, y radius= 13.6]   ;
\draw (417,121) node    {$2$};
\draw  [fill={rgb, 255:red, 155; green, 155; blue, 155 }  ,fill opacity=1 ]  (270, 64) circle [x radius= 13.6, y radius= 13.6]   ;
\draw (270,64) node    {$3$};
\draw  [fill={rgb, 255:red, 155; green, 155; blue, 155 }  ,fill opacity=1 ]  (400, 408) circle [x radius= 13.6, y radius= 13.6]   ;
\draw (400,408) node    {$3$};
\draw  [fill={rgb, 255:red, 155; green, 155; blue, 155 }  ,fill opacity=1 ]  (620, 404) circle [x radius= 13.6, y radius= 13.6]   ;
\draw (620,404) node    {$3$};
\draw  [fill={rgb, 255:red, 155; green, 155; blue, 155 }  ,fill opacity=1 ]  (203, 406) circle [x radius= 13.6, y radius= 13.6]   ;
\draw (203,406) node    {$3$};
\draw  [fill={rgb, 255:red, 155; green, 155; blue, 155 }  ,fill opacity=1 ]  (211, 121) circle [x radius= 13.6, y radius= 13.6]   ;
\draw (211,121) node    {$3$};
\draw  [fill={rgb, 255:red, 155; green, 155; blue, 155 }  ,fill opacity=1 ]  (474, 64) circle [x radius= 13.6, y radius= 13.6]   ;
\draw (474,64) node    {$3$};
\draw  [fill={rgb, 255:red, 155; green, 155; blue, 155 }  ,fill opacity=1 ]  (306,3) -- (354,3) -- (354,27) -- (306,27) -- cycle  ;
\draw (330,15) node    {$t\ =1$};
\draw (25,225) node    {$0$};

\end{tikzpicture}
\caption{One step of the dynamics of the NOT part of NOR gadget implemented by a signed symmetric conjunctive network. Dotted circles and triangles represent blocks. Numbers in gray represent the updating order of each block. Each time step $t$ is taken after three time steps (one for each block). Total simulation time is $T = 9$.}
\label{fig:NOT1andnot}
\end{figure}
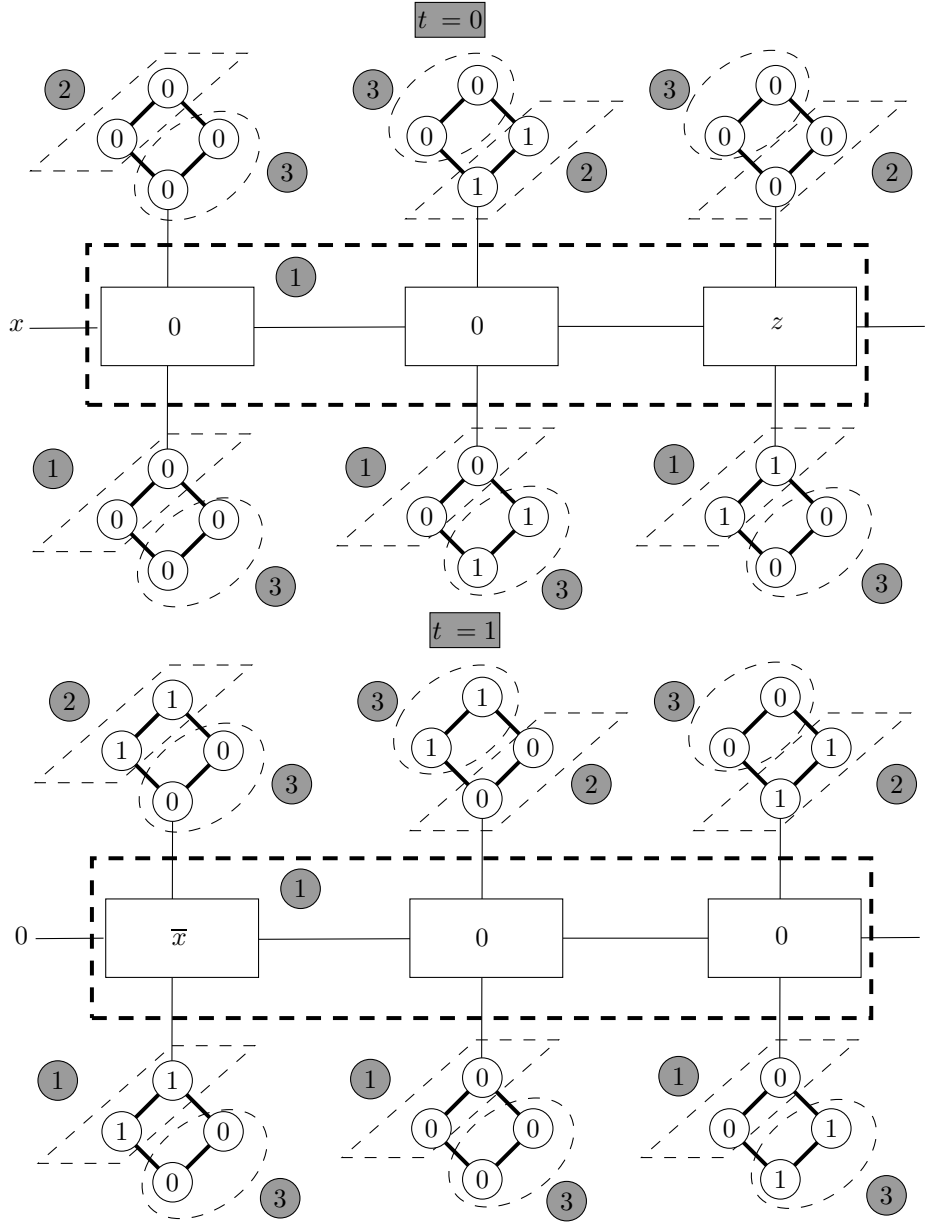

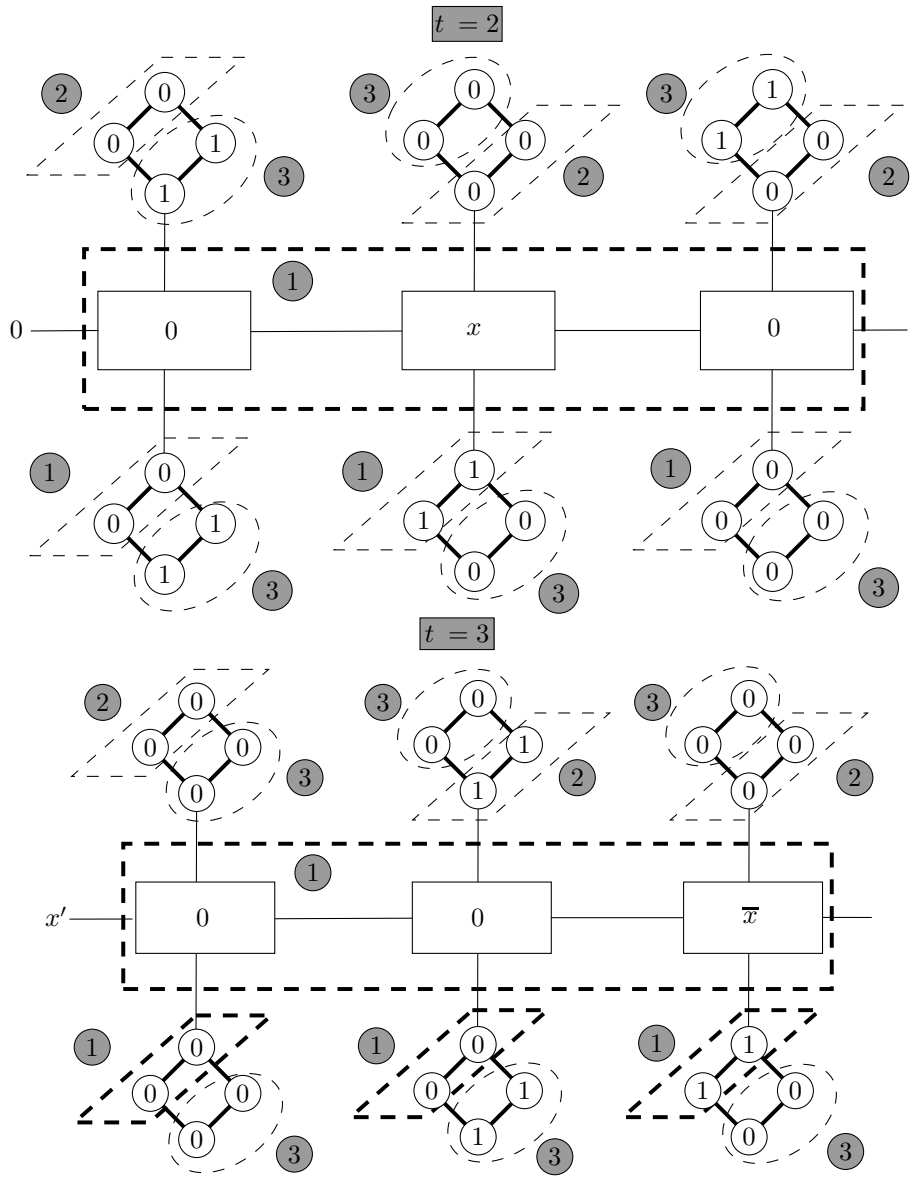
\begin{figure}
\centering
	
	\begin{tikzpicture}[x=0.55pt,y=0.55pt,yscale=-1,xscale=1]
	
	\draw  [line width=1.5]  (123,327) -- (158,362) -- (123,397) -- (88,362) -- cycle ;
	
	\draw  [line width=1.5]  (123,65) -- (158,100) -- (123,135) -- (88,100) -- cycle ;

	\draw  [line width=1.5]  (336,325) -- (371,360) -- (336,395) -- (301,360) -- cycle ;
	
	\draw  [line width=1.5]  (336,63) -- (371,98) -- (336,133) -- (301,98) -- cycle ;

	\draw  [line width=1.5]  (541,325) -- (576,360) -- (541,395) -- (506,360) -- cycle ;
	
	\draw  [line width=1.5]  (541,63) -- (576,98) -- (541,133) -- (506,98) -- cycle ;

	\draw    (595.5,228.77) -- (633.8,228.67) ;
	\draw    (123,135) -- (122.5,332) ;
	\draw    (336,133) -- (335.5,330) ;
	\draw    (541,133) -- (540.5,330) ;
	\draw    (345.75,229) -- (525.94,229) ;
	\draw    (124.5,230) -- (383.75,229) ;
	\draw  [fill={rgb, 255:red, 255; green, 255; blue, 255 }  ,fill opacity=1 ] (286.19,202) -- (391.31,202) -- (391.31,256) -- (286.19,256) -- cycle ;
	\draw  [fill={rgb, 255:red, 255; green, 255; blue, 255 }  ,fill opacity=1 ] (491.39,202) -- (596.5,202) -- (596.5,256) -- (491.39,256) -- cycle ;
	\draw  [fill={rgb, 255:red, 255; green, 255; blue, 255 }  ,fill opacity=1 ] (77,202) -- (182.11,202) -- (182.11,256) -- (77,256) -- cycle ;

	\draw  [dash pattern={on 4.5pt off 4.5pt}] (120.57,41) -- (177.5,41) -- (85.4,122) -- (28.47,122) -- cycle ;
	\draw  [dash pattern={on 4.5pt off 4.5pt}] (122.57,303) -- (179.5,303) -- (87.4,384) -- (30.47,384) -- cycle ;
	\draw  [dash pattern={on 4.5pt off 4.5pt}] (378.57,74) -- (435.5,74) -- (343.4,155) -- (286.47,155) -- cycle ;
	\draw  [dash pattern={on 4.5pt off 4.5pt}] (331.57,299) -- (388.5,299) -- (296.4,380) -- (239.47,380) -- cycle ;
	\draw  [dash pattern={on 4.5pt off 4.5pt}] (537.57,299) -- (594.5,299) -- (502.4,380) -- (445.47,380) -- cycle ;
	\draw  [dash pattern={on 4.5pt off 4.5pt}] (573.57,74) -- (630.5,74) -- (538.4,155) -- (481.47,155) -- cycle ;
	\draw  [dash pattern={on 5.63pt off 4.5pt}][line width=1.5]  (67.5,173) -- (603.5,173) -- (603.5,283) -- (67.5,283) -- cycle ;
	\draw  [dash pattern={on 4.5pt off 4.5pt}] (103.56,118.5) .. controls (112.98,97.79) and (138.27,81) .. (160.06,81) .. controls (181.84,81) and (191.86,97.79) .. (182.44,118.5) .. controls (173.02,139.21) and (147.73,156) .. (125.94,156) .. controls (104.16,156) and (94.14,139.21) .. (103.56,118.5) -- cycle ;
	\draw  [dash pattern={on 4.5pt off 4.5pt}] (105.56,384.5) .. controls (114.98,363.79) and (140.27,347) .. (162.06,347) .. controls (183.84,347) and (193.86,363.79) .. (184.44,384.5) .. controls (175.02,405.21) and (149.73,422) .. (127.94,422) .. controls (106.16,422) and (96.14,405.21) .. (105.56,384.5) -- cycle ;
	\draw  [dash pattern={on 4.5pt off 4.5pt}] (316.56,376.5) .. controls (325.98,355.79) and (351.27,339) .. (373.06,339) .. controls (394.84,339) and (404.86,355.79) .. (395.44,376.5) .. controls (386.02,397.21) and (360.73,414) .. (338.94,414) .. controls (317.16,414) and (307.14,397.21) .. (316.56,376.5) -- cycle ;
	\draw  [dash pattern={on 4.5pt off 4.5pt}] (524.56,377.5) .. controls (533.98,356.79) and (559.27,340) .. (581.06,340) .. controls (602.84,340) and (612.86,356.79) .. (603.44,377.5) .. controls (594.02,398.21) and (568.73,415) .. (546.94,415) .. controls (525.16,415) and (515.14,398.21) .. (524.56,377.5) -- cycle ;
	\draw  [dash pattern={on 4.5pt off 4.5pt}] (278.56,78.5) .. controls (287.98,57.79) and (313.27,41) .. (335.06,41) .. controls (356.84,41) and (366.86,57.79) .. (357.44,78.5) .. controls (348.02,99.21) and (322.73,116) .. (300.94,116) .. controls (279.16,116) and (269.14,99.21) .. (278.56,78.5) -- cycle ;
	\draw  [dash pattern={on 4.5pt off 4.5pt}] (481.56,76.5) .. controls (490.98,55.79) and (516.27,39) .. (538.06,39) .. controls (559.84,39) and (569.86,55.79) .. (560.44,76.5) .. controls (551.02,97.21) and (525.73,114) .. (503.94,114) .. controls (482.16,114) and (472.14,97.21) .. (481.56,76.5) -- cycle ;
	\draw    (30.8,229.33) -- (77.5,229) ;
	
	\draw  [fill={rgb, 255:red, 155; green, 155; blue, 155 }  ,fill opacity=1 ]  (468, 66) circle [x radius= 13.6, y radius= 13.6]   ;
	\draw (468,66) node    {$3$};
	\draw  [fill={rgb, 255:red, 155; green, 155; blue, 155 }  ,fill opacity=1 ]  (205, 123) circle [x radius= 13.6, y radius= 13.6]   ;
	\draw (205,123) node    {$3$};
	\draw  [fill={rgb, 255:red, 155; green, 155; blue, 155 }  ,fill opacity=1 ]  (197, 408) circle [x radius= 13.6, y radius= 13.6]   ;
	\draw (197,408) node    {$3$};
	\draw  [fill={rgb, 255:red, 155; green, 155; blue, 155 }  ,fill opacity=1 ]  (614, 406) circle [x radius= 13.6, y radius= 13.6]   ;
	\draw (614,406) node    {$3$};
	\draw  [fill={rgb, 255:red, 155; green, 155; blue, 155 }  ,fill opacity=1 ]  (394, 410) circle [x radius= 13.6, y radius= 13.6]   ;
	\draw (394,410) node    {$3$};
	\draw  [fill={rgb, 255:red, 155; green, 155; blue, 155 }  ,fill opacity=1 ]  (264, 66) circle [x radius= 13.6, y radius= 13.6]   ;
	\draw (264,66) node    {$3$};
	\draw  [fill={rgb, 255:red, 155; green, 155; blue, 155 }  ,fill opacity=1 ]  (411, 123) circle [x radius= 13.6, y radius= 13.6]   ;
	\draw (411,123) node    {$2$};
	\draw  [fill={rgb, 255:red, 155; green, 155; blue, 155 }  ,fill opacity=1 ]  (621, 123) circle [x radius= 13.6, y radius= 13.6]   ;
	\draw (621,123) node    {$2$};
	\draw  [fill={rgb, 255:red, 155; green, 155; blue, 155 }  ,fill opacity=1 ]  (471, 324) circle [x radius= 13.6, y radius= 13.6]   ;
	\draw (471,324) node    {$1$};
	\draw  [fill={rgb, 255:red, 155; green, 155; blue, 155 }  ,fill opacity=1 ]  (259, 326) circle [x radius= 13.6, y radius= 13.6]   ;
	\draw (259,326) node    {$1$};
	\draw  [fill={rgb, 255:red, 155; green, 155; blue, 155 }  ,fill opacity=1 ]  (44, 327) circle [x radius= 13.6, y radius= 13.6]   ;
	\draw (44,327) node    {$1$};
	\draw  [fill={rgb, 255:red, 155; green, 155; blue, 155 }  ,fill opacity=1 ]  (52, 66) circle [x radius= 13.6, y radius= 13.6]   ;
	\draw (52,66) node    {$2$};
	\draw  [fill={rgb, 255:red, 155; green, 155; blue, 155 }  ,fill opacity=1 ]  (211, 195) circle [x radius= 13.6, y radius= 13.6]   ;
	\draw (211,195) node    {$1$};
	\draw (541.94,227) node    {$0$};
	\draw (335.75,228.5) node    {$x$};
	\draw (127.56,229) node    {$0$};
	\draw  [fill={rgb, 255:red, 255; green, 255; blue, 255 }  ,fill opacity=1 ]  (541, 133) circle [x radius= 13.6, y radius= 13.6]   ;
	\draw (541,133) node    {$0$};
	\draw  [fill={rgb, 255:red, 255; green, 255; blue, 255 }  ,fill opacity=1 ]  (576, 98) circle [x radius= 13.6, y radius= 13.6]   ;
	\draw (576,98) node    {$0$};
	\draw  [fill={rgb, 255:red, 255; green, 255; blue, 255 }  ,fill opacity=1 ]  (541, 63) circle [x radius= 13.6, y radius= 13.6]   ;
	\draw (541,63) node    {$1$};
	\draw  [fill={rgb, 255:red, 255; green, 255; blue, 255 }  ,fill opacity=1 ]  (506, 98) circle [x radius= 13.6, y radius= 13.6]   ;
	\draw (506,98) node    {$1$};
	\draw  [fill={rgb, 255:red, 255; green, 255; blue, 255 }  ,fill opacity=1 ]  (576, 360) circle [x radius= 13.6, y radius= 13.6]   ;
	\draw (576,360) node    {$0$};
	\draw  [fill={rgb, 255:red, 255; green, 255; blue, 255 }  ,fill opacity=1 ]  (541, 325) circle [x radius= 13.6, y radius= 13.6]   ;
	\draw (541,325) node    {$0$};
	\draw  [fill={rgb, 255:red, 255; green, 255; blue, 255 }  ,fill opacity=1 ]  (506, 360) circle [x radius= 13.6, y radius= 13.6]   ;
	\draw (506,360) node    {$0$};
	\draw  [fill={rgb, 255:red, 255; green, 255; blue, 255 }  ,fill opacity=1 ]  (541, 395) circle [x radius= 13.6, y radius= 13.6]   ;
	\draw (541,395) node    {$0$};
	\draw  [fill={rgb, 255:red, 255; green, 255; blue, 255 }  ,fill opacity=1 ]  (336, 133) circle [x radius= 13.6, y radius= 13.6]   ;
	\draw (336,133) node    {$0$};
	\draw  [fill={rgb, 255:red, 255; green, 255; blue, 255 }  ,fill opacity=1 ]  (371, 98) circle [x radius= 13.6, y radius= 13.6]   ;
	\draw (371,98) node    {$0$};
	\draw  [fill={rgb, 255:red, 255; green, 255; blue, 255 }  ,fill opacity=1 ]  (336, 63) circle [x radius= 13.6, y radius= 13.6]   ;
	\draw (336,63) node    {$0$};
	\draw  [fill={rgb, 255:red, 255; green, 255; blue, 255 }  ,fill opacity=1 ]  (301, 98) circle [x radius= 13.6, y radius= 13.6]   ;
	\draw (301,98) node    {$0$};
	\draw  [fill={rgb, 255:red, 255; green, 255; blue, 255 }  ,fill opacity=1 ]  (371, 360) circle [x radius= 13.6, y radius= 13.6]   ;
	\draw (371,360) node    {$0$};
	\draw  [fill={rgb, 255:red, 255; green, 255; blue, 255 }  ,fill opacity=1 ]  (336, 325) circle [x radius= 13.6, y radius= 13.6]   ;
	\draw (336,325) node    {$1$};
	\draw  [fill={rgb, 255:red, 255; green, 255; blue, 255 }  ,fill opacity=1 ]  (301, 360) circle [x radius= 13.6, y radius= 13.6]   ;
	\draw (301,360) node    {$1$};
	\draw  [fill={rgb, 255:red, 255; green, 255; blue, 255 }  ,fill opacity=1 ]  (336, 395) circle [x radius= 13.6, y radius= 13.6]   ;
	\draw (336,395) node    {$0$};
	\draw  [fill={rgb, 255:red, 255; green, 255; blue, 255 }  ,fill opacity=1 ]  (123, 135) circle [x radius= 13.6, y radius= 13.6]   ;
	\draw (123,135) node    {$1$};
	\draw  [fill={rgb, 255:red, 255; green, 255; blue, 255 }  ,fill opacity=1 ]  (158, 100) circle [x radius= 13.6, y radius= 13.6]   ;
	\draw (158,100) node    {$1$};
	\draw  [fill={rgb, 255:red, 255; green, 255; blue, 255 }  ,fill opacity=1 ]  (123, 65) circle [x radius= 13.6, y radius= 13.6]   ;
	\draw (123,65) node    {$0$};
	\draw  [fill={rgb, 255:red, 255; green, 255; blue, 255 }  ,fill opacity=1 ]  (88, 100) circle [x radius= 13.6, y radius= 13.6]   ;
	\draw (88,100) node    {$0$};
	\draw  [fill={rgb, 255:red, 255; green, 255; blue, 255 }  ,fill opacity=1 ]  (158, 362) circle [x radius= 13.6, y radius= 13.6]   ;
	\draw (158,362) node    {$1$};
	\draw  [fill={rgb, 255:red, 255; green, 255; blue, 255 }  ,fill opacity=1 ]  (123, 327) circle [x radius= 13.6, y radius= 13.6]   ;
	\draw (123,327) node    {$0$};
	\draw  [fill={rgb, 255:red, 255; green, 255; blue, 255 }  ,fill opacity=1 ]  (88, 362) circle [x radius= 13.6, y radius= 13.6]   ;
	\draw (88,362) node    {$0$};
	\draw  [fill={rgb, 255:red, 255; green, 255; blue, 255 }  ,fill opacity=1 ]  (123, 397) circle [x radius= 13.6, y radius= 13.6]   ;
	\draw (123,397) node    {$1$};
	\draw  [fill={rgb, 255:red, 155; green, 155; blue, 155 }  ,fill opacity=1 ]  (307,6) -- (355,6) -- (355,30) -- (307,30) -- cycle  ;
	\draw (331,18) node    {$t\ =2$};
	\draw (21,228) node    {$0$};

	\end{tikzpicture}

\begin{tikzpicture}[x=0.5pt,y=0.5pt,yscale=-1,xscale=1]

\draw  [line width=1.5]  (123,331) -- (158,366) -- (123,401) -- (88,366) -- cycle ;

\draw  [line width=1.5]  (123,69) -- (158,104) -- (123,139) -- (88,104) -- cycle ;

\draw  [line width=1.5]  (336,329) -- (371,364) -- (336,399) -- (301,364) -- cycle ;

\draw  [line width=1.5]  (336,67) -- (371,102) -- (336,137) -- (301,102) -- cycle ;

\draw  [line width=1.5]  (541,329) -- (576,364) -- (541,399) -- (506,364) -- cycle ;

\draw  [line width=1.5]  (541,67) -- (576,102) -- (541,137) -- (506,102) -- cycle ;

\draw    (26.8,234.08) -- (74.5,234) ;
\draw    (633.8,232.83) -- (596.5,232.77) ;
\draw    (123,139) -- (122.5,336) ;
\draw    (336,137) -- (335.5,334) ;
\draw    (541,137) -- (540.5,334) ;
\draw    (345.75,233) -- (525.94,233) ;
\draw    (124.5,234) -- (383.75,233) ;
\draw  [fill={rgb, 255:red, 255; green, 255; blue, 255 }  ,fill opacity=1 ] (286.19,206) -- (391.31,206) -- (391.31,260) -- (286.19,260) -- cycle ;
\draw  [fill={rgb, 255:red, 255; green, 255; blue, 255 }  ,fill opacity=1 ] (491.39,206) -- (596.5,206) -- (596.5,260) -- (491.39,260) -- cycle ;
\draw  [fill={rgb, 255:red, 255; green, 255; blue, 255 }  ,fill opacity=1 ] (77,206) -- (182.11,206) -- (182.11,260) -- (77,260) -- cycle ;

\draw  [dash pattern={on 4.5pt off 4.5pt}] (120.57,45) -- (177.5,45) -- (85.4,126) -- (28.47,126) -- cycle ;
\draw  [dash pattern={on 5.63pt off 4.5pt}][line width=1.5]  (122.57,307) -- (179.5,307) -- (87.4,388) -- (30.47,388) -- cycle ;
\draw  [dash pattern={on 4.5pt off 4.5pt}] (378.57,78) -- (435.5,78) -- (343.4,159) -- (286.47,159) -- cycle ;
\draw  [dash pattern={on 5.63pt off 4.5pt}][line width=1.5]  (331.57,303) -- (388.5,303) -- (296.4,384) -- (239.47,384) -- cycle ;
\draw  [dash pattern={on 5.63pt off 4.5pt}][line width=1.5]  (537.57,303) -- (594.5,303) -- (502.4,384) -- (445.47,384) -- cycle ;
\draw  [dash pattern={on 4.5pt off 4.5pt}] (573.57,78) -- (630.5,78) -- (538.4,159) -- (481.47,159) -- cycle ;
\draw  [dash pattern={on 5.63pt off 4.5pt}][line width=1.5]  (67.5,177) -- (603.5,177) -- (603.5,287) -- (67.5,287) -- cycle ;
\draw  [dash pattern={on 4.5pt off 4.5pt}] (103.56,122.5) .. controls (112.98,101.79) and (138.27,85) .. (160.06,85) .. controls (181.84,85) and (191.86,101.79) .. (182.44,122.5) .. controls (173.02,143.21) and (147.73,160) .. (125.94,160) .. controls (104.16,160) and (94.14,143.21) .. (103.56,122.5) -- cycle ;
\draw  [dash pattern={on 4.5pt off 4.5pt}] (105.56,388.5) .. controls (114.98,367.79) and (140.27,351) .. (162.06,351) .. controls (183.84,351) and (193.86,367.79) .. (184.44,388.5) .. controls (175.02,409.21) and (149.73,426) .. (127.94,426) .. controls (106.16,426) and (96.14,409.21) .. (105.56,388.5) -- cycle ;
\draw  [dash pattern={on 4.5pt off 4.5pt}] (316.56,380.5) .. controls (325.98,359.79) and (351.27,343) .. (373.06,343) .. controls (394.84,343) and (404.86,359.79) .. (395.44,380.5) .. controls (386.02,401.21) and (360.73,418) .. (338.94,418) .. controls (317.16,418) and (307.14,401.21) .. (316.56,380.5) -- cycle ;
\draw  [dash pattern={on 4.5pt off 4.5pt}] (524.56,381.5) .. controls (533.98,360.79) and (559.27,344) .. (581.06,344) .. controls (602.84,344) and (612.86,360.79) .. (603.44,381.5) .. controls (594.02,402.21) and (568.73,419) .. (546.94,419) .. controls (525.16,419) and (515.14,402.21) .. (524.56,381.5) -- cycle ;
\draw  [dash pattern={on 4.5pt off 4.5pt}] (278.56,82.5) .. controls (287.98,61.79) and (313.27,45) .. (335.06,45) .. controls (356.84,45) and (366.86,61.79) .. (357.44,82.5) .. controls (348.02,103.21) and (322.73,120) .. (300.94,120) .. controls (279.16,120) and (269.14,103.21) .. (278.56,82.5) -- cycle ;
\draw  [dash pattern={on 4.5pt off 4.5pt}] (481.56,80.5) .. controls (490.98,59.79) and (516.27,43) .. (538.06,43) .. controls (559.84,43) and (569.86,59.79) .. (560.44,80.5) .. controls (551.02,101.21) and (525.73,118) .. (503.94,118) .. controls (482.16,118) and (472.14,101.21) .. (481.56,80.5) -- cycle ;

\draw  [fill={rgb, 255:red, 155; green, 155; blue, 155 }  ,fill opacity=1 ]  (468, 70) circle [x radius= 13.6, y radius= 13.6]   ;
\draw (468,70) node    {$3$};
\draw  [fill={rgb, 255:red, 155; green, 155; blue, 155 }  ,fill opacity=1 ]  (205, 127) circle [x radius= 13.6, y radius= 13.6]   ;
\draw (205,127) node    {$3$};
\draw (17,232) node    {$x'$};
\draw  [fill={rgb, 255:red, 155; green, 155; blue, 155 }  ,fill opacity=1 ]  (197, 412) circle [x radius= 13.6, y radius= 13.6]   ;
\draw (197,412) node    {$3$};
\draw  [fill={rgb, 255:red, 155; green, 155; blue, 155 }  ,fill opacity=1 ]  (614, 410) circle [x radius= 13.6, y radius= 13.6]   ;
\draw (614,410) node    {$3$};
\draw  [fill={rgb, 255:red, 155; green, 155; blue, 155 }  ,fill opacity=1 ]  (394, 414) circle [x radius= 13.6, y radius= 13.6]   ;
\draw (394,414) node    {$3$};
\draw  [fill={rgb, 255:red, 155; green, 155; blue, 155 }  ,fill opacity=1 ]  (264, 70) circle [x radius= 13.6, y radius= 13.6]   ;
\draw (264,70) node    {$3$};
\draw  [fill={rgb, 255:red, 155; green, 155; blue, 155 }  ,fill opacity=1 ]  (411, 127) circle [x radius= 13.6, y radius= 13.6]   ;
\draw (411,127) node    {$2$};
\draw  [fill={rgb, 255:red, 155; green, 155; blue, 155 }  ,fill opacity=1 ]  (621, 127) circle [x radius= 13.6, y radius= 13.6]   ;
\draw (621,127) node    {$2$};
\draw  [fill={rgb, 255:red, 155; green, 155; blue, 155 }  ,fill opacity=1 ]  (471, 328) circle [x radius= 13.6, y radius= 13.6]   ;
\draw (471,328) node    {$1$};
\draw  [fill={rgb, 255:red, 155; green, 155; blue, 155 }  ,fill opacity=1 ]  (259, 330) circle [x radius= 13.6, y radius= 13.6]   ;
\draw (259,330) node    {$1$};
\draw  [fill={rgb, 255:red, 155; green, 155; blue, 155 }  ,fill opacity=1 ]  (44, 331) circle [x radius= 13.6, y radius= 13.6]   ;
\draw (44,331) node    {$1$};
\draw  [fill={rgb, 255:red, 155; green, 155; blue, 155 }  ,fill opacity=1 ]  (52, 70) circle [x radius= 13.6, y radius= 13.6]   ;
\draw (52,70) node    {$2$};
\draw  [fill={rgb, 255:red, 155; green, 155; blue, 155 }  ,fill opacity=1 ]  (211, 199) circle [x radius= 13.6, y radius= 13.6]   ;
\draw (211,199) node    {$1$};
\draw (541.94,231) node    {$\overline{x}$};
\draw (335.75,232.5) node    {$0$};
\draw (127.56,233) node    {$0$};
\draw  [fill={rgb, 255:red, 255; green, 255; blue, 255 }  ,fill opacity=1 ]  (541, 137) circle [x radius= 13.6, y radius= 13.6]   ;
\draw (541,137) node    {$0$};
\draw  [fill={rgb, 255:red, 255; green, 255; blue, 255 }  ,fill opacity=1 ]  (576, 102) circle [x radius= 13.6, y radius= 13.6]   ;
\draw (576,102) node    {$0$};
\draw  [fill={rgb, 255:red, 255; green, 255; blue, 255 }  ,fill opacity=1 ]  (541, 67) circle [x radius= 13.6, y radius= 13.6]   ;
\draw (541,67) node    {$0$};
\draw  [fill={rgb, 255:red, 255; green, 255; blue, 255 }  ,fill opacity=1 ]  (506, 102) circle [x radius= 13.6, y radius= 13.6]   ;
\draw (506,102) node    {$0$};
\draw  [fill={rgb, 255:red, 255; green, 255; blue, 255 }  ,fill opacity=1 ]  (576, 364) circle [x radius= 13.6, y radius= 13.6]   ;
\draw (576,364) node    {$0$};
\draw  [fill={rgb, 255:red, 255; green, 255; blue, 255 }  ,fill opacity=1 ]  (541, 329) circle [x radius= 13.6, y radius= 13.6]   ;
\draw (541,329) node    {$1$};
\draw  [fill={rgb, 255:red, 255; green, 255; blue, 255 }  ,fill opacity=1 ]  (506, 364) circle [x radius= 13.6, y radius= 13.6]   ;
\draw (506,364) node    {$1$};
\draw  [fill={rgb, 255:red, 255; green, 255; blue, 255 }  ,fill opacity=1 ]  (541, 399) circle [x radius= 13.6, y radius= 13.6]   ;
\draw (541,399) node    {$0$};
\draw  [fill={rgb, 255:red, 255; green, 255; blue, 255 }  ,fill opacity=1 ]  (336, 137) circle [x radius= 13.6, y radius= 13.6]   ;
\draw (336,137) node    {$1$};
\draw  [fill={rgb, 255:red, 255; green, 255; blue, 255 }  ,fill opacity=1 ]  (371, 102) circle [x radius= 13.6, y radius= 13.6]   ;
\draw (371,102) node    {$1$};
\draw  [fill={rgb, 255:red, 255; green, 255; blue, 255 }  ,fill opacity=1 ]  (336, 67) circle [x radius= 13.6, y radius= 13.6]   ;
\draw (336,67) node    {$0$};
\draw  [fill={rgb, 255:red, 255; green, 255; blue, 255 }  ,fill opacity=1 ]  (301, 102) circle [x radius= 13.6, y radius= 13.6]   ;
\draw (301,102) node    {$0$};
\draw  [fill={rgb, 255:red, 255; green, 255; blue, 255 }  ,fill opacity=1 ]  (371, 364) circle [x radius= 13.6, y radius= 13.6]   ;
\draw (371,364) node    {$1$};
\draw  [fill={rgb, 255:red, 255; green, 255; blue, 255 }  ,fill opacity=1 ]  (336, 329) circle [x radius= 13.6, y radius= 13.6]   ;
\draw (336,329) node    {$0$};
\draw  [fill={rgb, 255:red, 255; green, 255; blue, 255 }  ,fill opacity=1 ]  (301, 364) circle [x radius= 13.6, y radius= 13.6]   ;
\draw (301,364) node    {$0$};
\draw  [fill={rgb, 255:red, 255; green, 255; blue, 255 }  ,fill opacity=1 ]  (336, 399) circle [x radius= 13.6, y radius= 13.6]   ;
\draw (336,399) node    {$1$};
\draw  [fill={rgb, 255:red, 255; green, 255; blue, 255 }  ,fill opacity=1 ]  (123, 139) circle [x radius= 13.6, y radius= 13.6]   ;
\draw (123,139) node    {$0$};
\draw  [fill={rgb, 255:red, 255; green, 255; blue, 255 }  ,fill opacity=1 ]  (158, 104) circle [x radius= 13.6, y radius= 13.6]   ;
\draw (158,104) node    {$0$};
\draw  [fill={rgb, 255:red, 255; green, 255; blue, 255 }  ,fill opacity=1 ]  (123, 69) circle [x radius= 13.6, y radius= 13.6]   ;
\draw (123,69) node    {$0$};
\draw  [fill={rgb, 255:red, 255; green, 255; blue, 255 }  ,fill opacity=1 ]  (88, 104) circle [x radius= 13.6, y radius= 13.6]   ;
\draw (88,104) node    {$0$};
\draw  [fill={rgb, 255:red, 255; green, 255; blue, 255 }  ,fill opacity=1 ]  (158, 366) circle [x radius= 13.6, y radius= 13.6]   ;
\draw (158,366) node    {$0$};
\draw  [fill={rgb, 255:red, 255; green, 255; blue, 255 }  ,fill opacity=1 ]  (123, 331) circle [x radius= 13.6, y radius= 13.6]   ;
\draw (123,331) node    {$0$};
\draw  [fill={rgb, 255:red, 255; green, 255; blue, 255 }  ,fill opacity=1 ]  (88, 366) circle [x radius= 13.6, y radius= 13.6]   ;
\draw (88,366) node    {$0$};
\draw  [fill={rgb, 255:red, 255; green, 255; blue, 255 }  ,fill opacity=1 ]  (123, 401) circle [x radius= 13.6, y radius= 13.6]   ;
\draw (123,401) node    {$0$};
\draw  [fill={rgb, 255:red, 155; green, 155; blue, 155 }  ,fill opacity=1 ]  (296-4,6) -- (344+4,6) -- (344+4,30) -- (296-4,30) -- cycle  ;
\draw (320,18) node    {$t\ =3$};

\end{tikzpicture}
\caption{Two last steps of the dynamics described by the NOT part of NOR gadget implemented by a symmetric signed conjunctive network. Dotted circles and triangles represent blocks. Numbers in gray represent the updating order of each block. Each time step $t$ is taken after three time steps (one for each block). Total simulation time is $T = 9$.}	
\label{fig:NOT2andnot}
\end{figure}
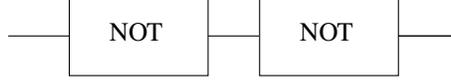
\begin{figure}[t!]
\centering

\begin{tikzpicture}[x=0.75pt,y=0.75pt,yscale=-1,xscale=1]

\draw    (272,48) -- (337.8,48) ;
\draw    (368,48) -- (433.8,48) ;
\draw    (206.2,48) -- (272,48) ;
\draw  [fill={rgb, 255:red, 255; green, 255; blue, 255 }  ,fill opacity=1 ] (237,28) -- (307,28) -- (307,68) -- (237,68) -- cycle ;
\draw    (367,48) -- (432.8,48) ;
\draw  [fill={rgb, 255:red, 255; green, 255; blue, 255 }  ,fill opacity=1 ] (333,28) -- (403,28) -- (403,68) -- (333,68) -- cycle ;

\draw (256,40) node [anchor=north west][inner sep=0.75pt]   [align=left] {NOT};
\draw (352,40) node [anchor=north west][inner sep=0.75pt]   [align=left] {NOT};

\end{tikzpicture}

	\caption{Wire gadget implemented on a signed symmetric conjuntive network. $2$ copies of NOT gadget are combined in order to form a wire. Simulation time is $T=6.$}
	\label{fig:wireandnot}
\end{figure}

As we will be using the same structures to show the main result, we start by showing a less powerful result. We show that $\blocfamily{\mathcal{F}}{b}_{\text{sign-sym-conj}}$ is able to simulate $\G_w$-networks. This is only a way to motivate the main result by showing how key structures work in a particular simple context. As we will see when we present the main result, we do not use exactly this result in the actual proof of main theorem but we apply the same type of ideas as related structures are part of the gadgets we construct. We remark that all involved gadgets have only negative labels, i.e. each edge is labeled by the $\text{Switch}$ function which takes a bit $x$ and produces $1-x.$

\begin{lemma}
	The family $\blocfamily{\mathcal{F}}{3}_{\text{sign-sym-conj}}$ of all signed symmetric conjunctive networks under block sequential update schemes of at most $3$ blocks  has coherent $\G_w$-gadgets.
\end{lemma}
\begin{proof}
	We define a gadget simulating $\text{Id}$ by considering two copies of the NOT gadget presented in Figures \ref{fig:NOT1andnot} and \ref{fig:NOT2andnot}. Observe that this gadget has $3$ central nodes (marked inside a thick dotted rectangle in Figures \ref{fig:NOT1andnot} and \ref{fig:NOT2andnot}) together with $2$ copies of a $4$ nodes cycle graph. Generally speaking, the dynamics on this cycle graphs works as a clock which allows information to flow through the central part in only one direction (from left to right). In addition, they allow the gadget to erase information once it has been transmitted. This latter property allows the gadget to clean itself in order to receive new information.  The wire gadget is composed by two copies of the NOT gadget as is presented in Figure \ref{fig:wireandnot}. Since this gadget is composed by two copies of the NOT gadget, this gadget is defined by a path graph with $6$ central nodes together with $2 \times 2 \times 3=12$ cycle graphs (two copies for each node). We enumerate nodes in the central part from left to write by the following ordering: $\{0,1,2,3,4,5\}$. Additionally, since the functioning of each copy of the NOT gadget is based in an ordered partition of $3$ blocks, we take the union of corresponding blocks in each partition in order to define 3 larger blocks for the wire gadget.  For one of this larger blocks we use the notation $\{0,1,2\}.$ Thus, observe that each wire takes $T=6\times 3 = 18$ time steps in order transport the information. This is because for each round of $3$ time steps in which we update each block, we make the signal pass through exactly one node. We define the glueing interface as $C_i = \{i\}$ and $C_o = \{o\}$. We map input and output in the following way: $\phi^i(i)={1},$ $\phi^i(o)={0}$, $\phi^o(i)= \{5\}$ $\phi^o(o) = \{4\}$. Note that in each case the neighborhood of output and input part is completely contained in $C_i$ and $C_o$ respectively and also the image of $C$ by functions $\phi$ is always the same (two nodes path graph). Thus, the glueing interface satisfies conditions of Lemma \ref{lem:csan-gadgets-glueing-closure}. Now we enumerate the remaining elements required by Definition \ref{def:coherent-gadgets}:
	\begin{itemize}
	\item State configurations $s_q(i) = (q,0)$ and $s_q(o) = (q,1)$ for any $q \in \{0,1\}.$
	\item Context configurations are given in Figure \ref{fig:NOT1andnot} as well as the iteration order which defines the blocks.
	\item Standard trace and pseudo orbit for nodes $\{0,1,4,5\}$ are defined in Figure \ref{fig:NOT1andnot} and Figure \ref{fig:NOT2andnot}.
	\end{itemize}
\end{proof}

We are now in conditions to introduce the main result:
\begin{lemma}
The family $\blocfamily{\mathcal{F}}{3}_{\text{sign-sym-conj}}$ of all signed symmetric conjunctive networks under block sequential update schemes of at most $3$ blocks has coherent $\Gmontwo$-gadgets.
\end{lemma}
\begin{proof}
	We will show that the gadgets in the Figures \ref{fig:ORstruct} and \ref{fig:ANDstruct} are coherent $\Gmontwo$-gadgets. Note that both these gadgets are made of several wires made of NOT gadgets (we call the two in the left hand side of the figure input wires and the other two at the right hand side output wires) and a computation gadget. Observe that, each of these structures (wires and computation gadget) needs exactly $3$ blocks. In addition, note that in Table \ref{tab:Clockdym-signed} the dynamics of the $4$-cycles that are attached to each node is shown. We recall that the function of these clocks is to allow information to flow in one direction only (all the interactions are symmetric so this is not straightforward) and to erase information once it has been copied or processed by the nodes in the gadget. We use the following notation in order to represent nodes in these structures: $c_{s,i,j,p}$ where $s$ is the number of the NOT gadget (a $3$-node-path together with two $4$-cycles for each node, see Figure \ref{fig:clockorder}) to which the cycle is attached, so $s \in \{1,2,3,4,5,6,7,8,9\}$ where the order is taken from left to right (for example in Figure \ref{fig:ORstruct} the copy associated to $(v_{1},v_{2},v_{3})$ comes first then, $(v'_{1},v'_{2},v'_{3})$, then $(v_{4},v_{5},v_{6})$ and so on). Additionally,  $i$ is in the position of the cycle in the NOT block, so $i \in \{1,2,3\}$, $j$ is the position of the node relative to the $4$-cycle graph considered in counter clockwise order (see Figure \ref{fig:clockorder}), so $j \in \{1,2,3,4\}$ and $p$ is the position of the cycle in the central structure of the gadget. Observe that, since there are two copies for each node, we denote them by \textit{upper} and \textit{lower} so $p \in \{u,l\}$. Since input wires and output wires are needed to be updated at the same time and are independent (we have two copies for each one) we combine their blocks in the obvious way (we take the union of pairs of blocks that have the same update order). More precisely, we combine $3\times 9$ blocks of each particular part (there are $8$ copies of NOT and one computation gadget having $3$ blocks each one) in order to define again only $3$ blocks. Precise definition, using notation shown in Figures \ref{fig:ORstruct} and \ref{fig:ANDstruct}, is the following: 
	\begin{itemize}
	 \item $B_{0} = \bigcup \limits_{i=1}^{12} \{v_{i}\} \cup \{v'_{i}\} \cup \bigcup \limits_{s,i}\{c_{s,i,1,l},c_{s,i,2,l}\};$
	 \item $B_{1} = \bigcup \limits_{s}\{c_{s,1,3,u},c_{s,1,4,u}, c_{s,2,1,u},c_{s,2,2,u},c_{s,3,1,u},c_{s,3,2,u}\};$ and 
	 \item $B_{2} =  \bigcup \limits_{s,i} \{c_{s,i,3,l},c_{s,i,4,l}\}.$
	\end{itemize}
	
Now we are going to use information in Tables \ref{tab:ANDdym-signed}, \ref{tab:ORdym-signed} and \ref{tab:Clockdym-signed} in order to show that these gadgets satisfy the conditions of Definition \ref{def:coherent-gadgets}.  Observe that, on the one hand, for the OR gadget (see Figure \ref{fig:ORstruct}),  input wires compute the result in $3 \times 3$ (it needs to carry the signal through the three nodes in the wire and each of this intermediate steps takes three steps, one for each block) time steps and computation gadget takes $3 \times 3$ as well. On the other hand, for the AND gadget (see Figure \ref{fig:ANDstruct}) input wires compute desired in $3 \times 6$ time steps. We define the associated network as $F_{\text{AND}}: (\{0,1\} \times \{0,1,2\})^{15 \times 2 \times 4} \to (\{0,1\} \times \{0,1,2\})^{15 \times 2 \times 4}$ and $F_{\text{OR}}: (\{0,1\} \times \{0,1,2\})^{15 \times 2 \times 4} \to (\{0,1\} \times \{0,1,2\})^{15 \times 2 \times 4}.$ In fact we have that:
\begin{enumerate}
	\item There is a unique glueing interface given by $C = C_{i} \cup C_{o}$ where:
	\begin{itemize}
	\item $C_{i} = \{i\} \cup \{a(i,1),a(i,2),a(i,3),a(i,4)\} \cup  \{a'(i,1),a'(i,2),a'(i,3),a'(i,4)\};$ and
	\item \sloppy $C_{o} = \{o',o\} \cup \{a(o',1),a(o',2),a(o',3),a(o',4)\} \cup \{a(o,1),a(o,2),a(o,3),a(o,4)\} \cup \{a'(o',1),a'(o',2),a'(o',3),a'(o',4)\} \cup \{a'(o,1),a'(o,2),a'(o,3),a'(o,4)\} $
	\end{itemize}	
	We define labelling functions $\phi^{i}_{\text{AND},k}$, $\phi^{o}_{\text{AND},k}$, $\phi^{i}_{\text{OR},k}$, $\phi^{o}_{\text{OR},k}$ (for the sake of simplicity we show the definition for the AND gadget since the one for the OR gadget is completely analogous)  for $k = 1,2$ as:
	\begin{itemize}
	\item  $\phi^{i}_{\text{AND},1}(i) = v_{1}$, $\phi^{i}_{\text{AND},1}(o') = v_{2}$ and $\phi^{i}_{\text{AND},1}(o) = v_{3};$
	\item  \sloppy $\phi^{i}_{\text{AND},1}(a(i,r)) = c_{1,1,r,u}$, $\phi^{i}_{\text{AND},1}(a'(i,r)) = c_{1,1,r,l}$ for $r=1,2,3,4,$ where $1$ corresponds to the NOT gadget which starts with $v_{1}$ in Figure \ref{fig:ANDstruct}.
	\item  \sloppy $\phi^{i}_{\text{AND},1}(a(o',r)) = c_{1,2,r,u}$, $\phi^{i}_{\text{AND},1}(a'(o',r)) = c_{1,2,r,l}$ for $r=1,2,3,4,$ where $1$ corresponds to the NOT gadget which starts with $v_{1}$ in Figure \ref{fig:ANDstruct};
	\item  \sloppy $\phi^{i}_{\text{AND},1}(a(o,r)) = c_{1,3,r,u}$, $\phi^{i}_{\text{AND},1}(a'(o,r)) = c_{1,3,r,l}$ for $r=1,2,3,4,$ where $1$ corresponds to the NOT gadget which starts with $v'_{1}$ in Figure \ref{fig:ANDstruct};
	\item $\phi^{i}_{\text{AND},2}(i) = v'_{1}$, $\phi^{i}_{\text{AND},2}(o') = v'_{2}$ and $\phi^{i}_{\text{AND},2}(o) = v'_{3};$ 
	\item  \sloppy $\phi^{i}_{\text{AND},2}(a(i,r)) = c_{1',1,r,u}$, $\phi^{i}_{\text{AND},1}(a'(i,r)) = c_{1',1,r,l}$ for $r=1,2,3,4,$ where $2$ correspond to the NOT gadget which starts with $v'_{1}$ in Figure \ref{fig:ANDstruct};
	\item  \sloppy $\phi^{i}_{\text{AND},2}(a(o',r)) = c_{1',2,r,u}$, $\phi^{i}_{\text{AND},1}(a'(o',r)) = c_{1',2,r,l}$ for $r=1,2,3,4,$ where $2$ corresponds to the NOT gadget which starts with $v'_{1}$ in Figure \ref{fig:ANDstruct};
	\item  \sloppy $\phi^{i}_{\text{AND},2}(a(o,r)) = c_{1',3,r,u}$, $\phi^{i}_{\text{AND},1}(a'(o,r)) = c_{1',3,r,l}$ for $r=1,2,3,4,$ where $2$ corresponds to the NOT gadget which starts with $v'_{1}$ in Figure \ref{fig:ANDstruct}.
	\item $\phi^{o}_{\text{AND},1}(i) = v_{10}$, $\phi^{o}_{\text{AND},1}(o') = v_{11}$ and $\phi^{o}_{\text{AND},1}(o) = v_{12};$
\item  \sloppy $\phi^{o}_{\text{AND},1}(a(i,r)) = c_{5,1,r,u}$, $\phi^{o}_{\text{AND},1}(a'(i,r)) = c_{5,1,r,l}$ for $r=1,2,3,4,$ where $8$ corresponds to the NOT gadget which starts with $v_{10}$ in Figure \ref{fig:ANDstruct}.
	\item  \sloppy $\phi^{o}_{\text{AND},1}(a(o',r)) = c_{5,2,r,u}$, $\phi^{o}_{\text{AND},1}(a'(o',r)) = c_{5,2,r,l}$ for $r=1,2,3,4,$ where $8$ corresponds to the NOT gadget which starts with $v_{10}$ in Figure \ref{fig:ANDstruct};
	\item  \sloppy $\phi^{o}_{\text{AND},1}(a(o,r)) = c_{5,3,r,u}$, $\phi^{o}_{\text{AND},1}(a'(o,r)) = c_{5,3,r,l}$ for $r=1,2,3,4,$ where $8$ corresponds to the NOT gadget which starts with $v_{10}$ in Figure \ref{fig:ANDstruct};	
	\item $\phi^{o}_{\text{AND},2}(i) = v'_{10}$, $\phi^{o}_{\text{AND},2}(o') = v'_{11}$ and $\phi^{o}_{\text{AND},2}(o) = v'_{12};$
	\item  \sloppy $\phi^{i}_{\text{AND},2}(a(i,r)) = c_{5',1,r,u}$, $\phi^{i}_{\text{AND},2}(a'(i,r)) = c_{5',1,r,l}$ for $r=1,2,3,4,$ where $9$ correspond to the NOT gadget which starts with $v'_{10}$ in Figure \ref{fig:ANDstruct};
		\item  \sloppy $\phi^{o}_{\text{AND},2}(a(o',r)) = c_{5',2,r,u}$, $\phi^{o}_{\text{AND},2}(a'(o',r)) = c_{5',2,r,l}$ for $r=1,2,3,4,$ where $9$ corresponds to the NOT gadget which starts with $v_{10}$ in Figure \ref{fig:ANDstruct};
	\item  \sloppy $\phi^{o}_{\text{AND},2}(a(o,r)) = c_{5',3,r,u}$, $\phi^{o}_{\text{AND},2}(a'(o,r)) = c_{5',3,r,l}$ for $r=1,2,3,4,$ where $9$ corresponds to the NOT gadget which starts with $v'_{10}$ in Figure \ref{fig:ANDstruct}.		
	\end{itemize}
	\item State configurations are defined for each $q \in \{0,1\}$ as $s_{q}(i) = (q,1)$ and $s_{q}(o') = s_{q}(o) = (0,1)$ and the state configuration of the nodes in the clocks i.e. the ones labeled by $a$ are constant and shown in Table \ref{tab:Clockdym-signed}. The block number for each of these nodes can is the same as the original NOT gadget \ref{fig:NOT1andnot}.
	\item Context configurations are described in Tables \ref{tab:ANDdym-signed}, \ref{tab:ORdym-signed}  and \ref{tab:Clockdym-signed} as the ones related to the cycles of length $4$ connected to central path of the gadgets and nodes in the path which are not part of the glueing interface.
	\item Standard trace is defined in Tables \ref{tab:ANDdym-signed} and \ref{tab:ORdym-signed}  (which contain the information related to the dynamics of nodes $v_{1},v'_{1},v_{2},v'_{2},v_{3},v'_{3},v_{10},v'_{10},v_{11},v'_{11},v_{12},v'_{12}$) and in Table \ref{tab:Clockdym-signed} (which contains the dynamics of the nodes in the $4$-cycles).
	\item Simulation constant is $T= 3 \times 12$ as it is shown in Tables \ref{tab:ANDdym-signed},\ref{tab:ORdym-signed} and 
	\item Pseudo-orbit is given by the dynamics shown in in Tables \ref{tab:ANDdym-signed}, \ref{tab:ORdym-signed}, and \ref{tab:Clockdym-signed} where $x,y,x',y',z$ are variables.
\end{enumerate}

\end{proof}

\begin{figure}
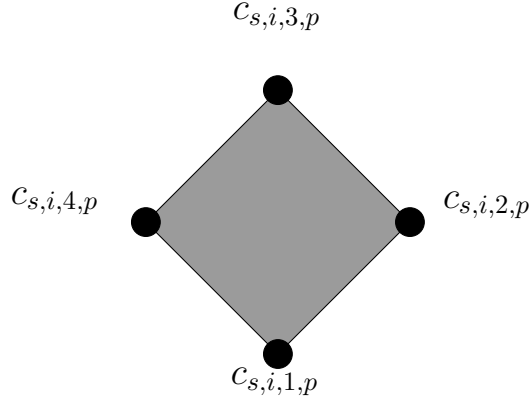

\centering

\caption{Scheme of labelling for $4$-cycles in AND/OR gadgets. Notation is given by the following guidelines: $s$ represent the associated group of three nodes, second two coordinates indicate its position relative to the original gadget (there are $3$ clocks) and its position in the $4$-cycle graph (considering counter clock-wise order), and $u,l$ stands for upper or lower according to its position in the gadget.}
\label{fig:clockorder}
\end{figure}

\begin{figure}
	\centering
	%
	\caption{One step of the dynamics of the computation gadget inside NOR gadget implemented by a signed symmetric conjunctive network. Dotted circles and triangles represent blocks. Numbers in gray represent the updating order of each block. Each time step $t$ is taken after three time steps (one for each block). Total simulation time is $T = 9$.}
	\label{fig:andnotnand1}
\end{figure}

\begin{figure}
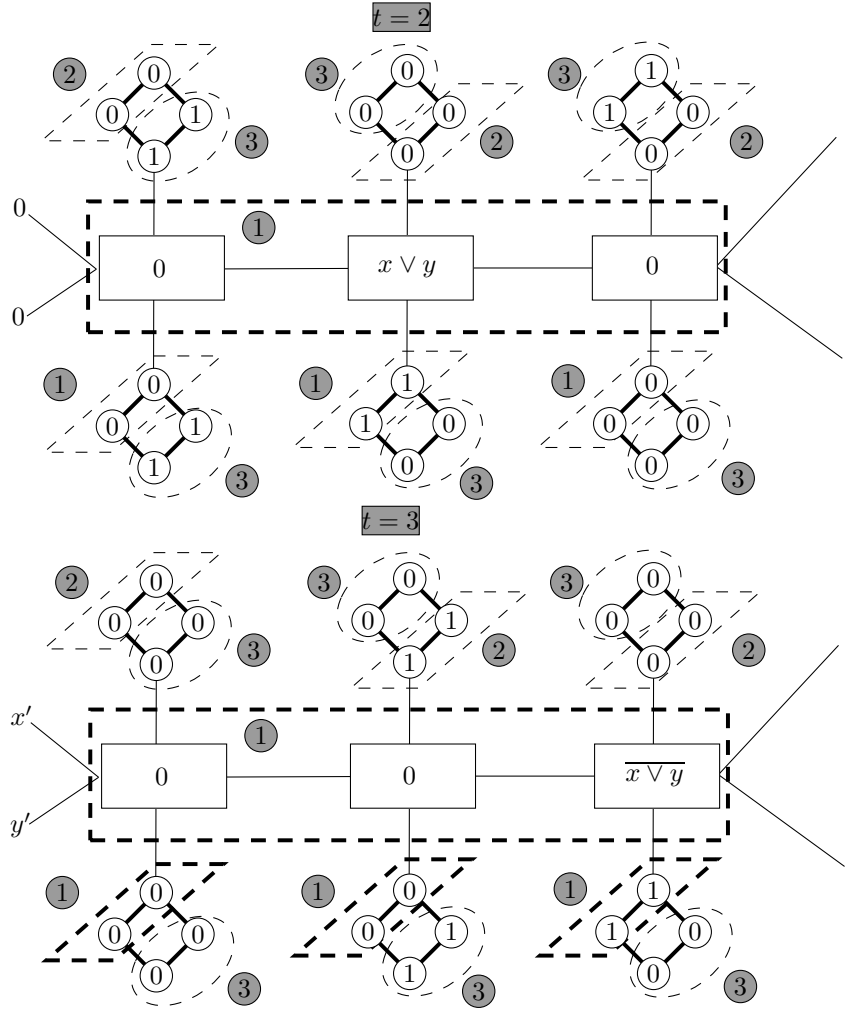

	\centering
%
\caption{Last two steps of the dynamics of the computation gadget inside NOR gadget implemented by an AND-not network. Dotted circles and triangles represent blocks. Numbers in gray represent the updating order of each block. Each time step $t$ is taken after three time steps (one for each block). Total simulation time is $T = 12$.}
\label{fig:andnotnand2}			
\end{figure}
\begin{figure}
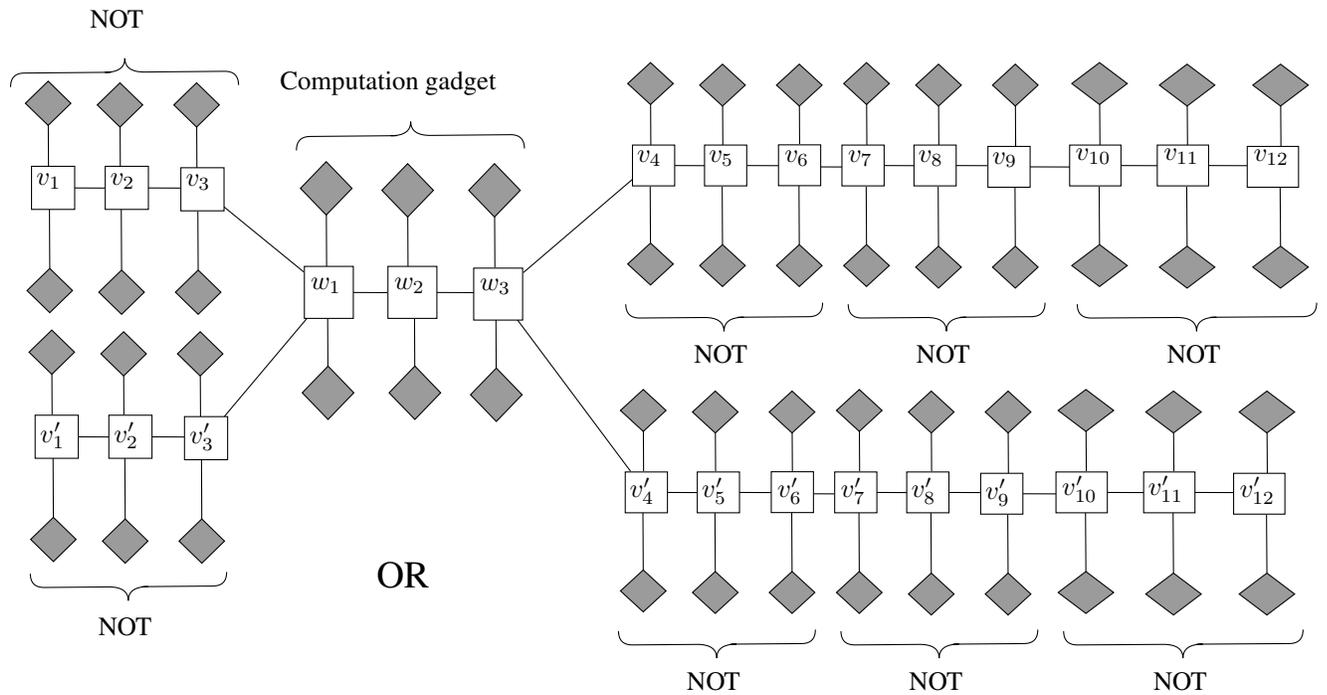

\centering
%

	\caption{OR gadget structure. In order to produce a OR gadget, wire gadget and NOT gadget are combined with computation part showed in Figures \ref{fig:andnotnand1} and \ref{fig:andnotnand2}.}
	\label{fig:ORstruct}
\end{figure}

\begin{figure}
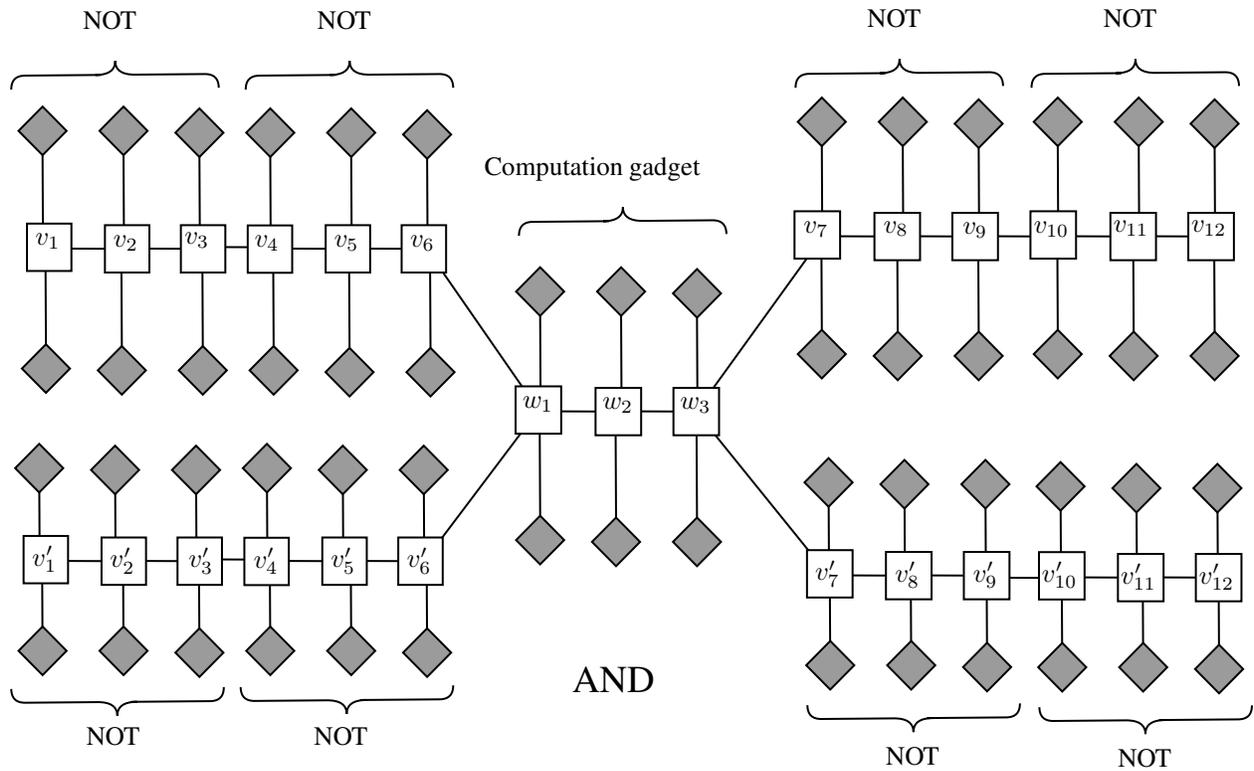

\centering

\tikzset{every picture/.style={line width=0.75pt}} 

%
	\caption{AND gadget structure. In order to implement an AND gadget, wire gadget and NOT gadget are combined with computation part showed in Figures \ref{fig:andnotnand1} and \ref{fig:andnotnand2}.}
	\label{fig:ANDstruct}
\end{figure}

\begin{table}[h]
\centering
\resizebox{\textwidth}{!}{%
\begin{tabular}{|l|l|l|l|l|l|l|l|l|l|l|l|l|l|l|l|l|l|l|l|l|l|l|l|l|l|l|l|}
\hline
Node/Time & $v_1$ & $v_2$          & $v_3$ & $v_4$          & $v_5$ & $v_6$          & $v'_1$ & $v'_2$         & $v'_3$ & $v'_4$         & $v'_5$ & $v'_6$         & $w_1$        & $w_2$                   & $w_3$        & $v_7$                   & $v_8$        & $v_9$                   & $v_{10}$     & $v_{11}$       & $v_{12}$ & $v'_7$                  & $v'_8$       & $v'_9$                  & $v'_{10}$    & $v'_{11}$      & $v'_{12}$ \\ \hline
0         & x     & 0              & 0     & 0              & 0     & 0              & y      & 0              & 0      & 0              & 0      & 0              & 0            & 0                       & 0            & 0                       & 0            & 0                       & z            & 0              & 0        & 0                       & 0            & 0                       & z            & 0              & 0         \\ \hline
3         & 0     & $\overline{x}$ & 0     & 0              & 1     & 0              & 0      & $\overline{y}$ & 0      & 0              & 1      & 0              & 0            & 1                       & 0            & 0                       & 1            & 0                       & 0            & $\overline{z}$ & 0        & 0                       & 1            & 0                       & 0            & $\overline{z}$ & 0         \\ \hline
6         & 0     & 0              & x     & 0              & 0     & 0              & 0      & 0              & y      & 0              & 0      & 0              & 0            & 0                       & 0            & 0                       & 0            & 0                       & 0            & 0              & z        & 0                       & 0            & 0                       & 0            & 0              & z         \\
\hline
\end{tabular}%
}
\caption{Dynamics for central gadgets in AND gadget implemented over a symmetric signed conjunctive network. Notation is the same of the one shown in Figure \ref{fig:ANDstruct}}
\label{tab:ANDdym-signed}
\end{table}

\begin{table}[h]
\centering
\resizebox{\textwidth}{!}{%
\begin{tabular}{|l|l|l|l|l|l|l|l|l|l|l|l|l|l|l|l|l|l|l|l|l|l|l|l|l|l|l|l|}
\hline
Node/Time & $v_1$ & $v_2$          & $v_3$ & $v'_1$ & $v'_2$         & $v'_3$ & $w_1$                 & $w_2$      & $w_3$                 & $v_4$      & $v_5$                 & $v_6$     & $v_7$                 & $v_8$      & $v_9$                 & $v_{10}$   & $v_{11}$       & $v_{12}$ & $v'_4$     & $v'_5$                & $v'_6$     & $v'_7$                & $v'_8$     & $v'_9$                & $v'_{10}$  & $v'_{11}$      & $v'_{12}$ \\ \hline
0         & x     & 0              & 0     & y      & 0              & 0      & 0                     & 0          & 0                     & 0          & 0                     & 0         & 0                     & 0          & 0                     & z          & 0              & 0        & 0          & 0                     & 0          & 0                     & 0          & 0                     & z          & 0              & 0         \\ \hline
3         & 0     & $\overline{x}$ & 0     & 0      & $\overline{y}$ & 0      & 0                     & 1          & 0                     & 0          & 1                     & 0         & 0                     & 1          & 0                     & 0          & $\overline{z}$ & 0        & 0          & 1                     & 0          & 0                     & 1          & 0                     & 0          & $\overline{z}$ & 0         \\ \hline
6         & 0     & 0              & x     & 0      & 0              & y      & 0                     & 0          & 0                     & 0          & 0                     & 0         & 0                     & 0          & 0                     & 0          & 0              & z        & 0          & 0                     & 0          & 0                     & 0          & 0                     & 0          & 0              & z         \\
\hline
\end{tabular}%
}
\caption{Dynamics for central gadgets in OR gadget implemented over a symmetric signed conjunctive network. Notation is the same of the one shown in Figure \ref{fig:ORstruct}}
\label{tab:ORdym-signed}
\end{table}

\begin{table}[h]
\centering
\resizebox{\textwidth}{!}{%
\begin{tabular}{|l|l|l|l|l|l|l|l|l|l|l|l|l|}
\hline
Node/Time & $c_{s,1,1,u}$ & $c_{s,1,2,u}$ & $c_{s,1,3,u}$ & $c_{s,1,4,u}$ & $c_{s,2,1,u}$ & $c_{s,2,2,u}$ & $c_{s,2,3,u}$ & $c_{s,2,4,u}$ & $c_{s,3,1,u}$ & $c_{s,3,2,u}$ & $c_{s,3,3,u}$ & $c_{s,3,4,u}$ \\ \hline
0         & 0             & 0             & 1             & 1             & 0             & 0             & 1             & 1             & 1             & 1             & 0             & 0             \\ \hline
3         & 1             & 1             & 0             & 0             & 0             & 0             & 0             & 0             & 0             & 0             & 1             & 1             \\ \hline
6         & 0             & 0             & 0             & 0             & 1             & 1             & 0             & 0             & 0             & 0             & 0             & 0             \\ \hline
9         & 0             & 0             & 1             & 1             & 0             & 0             & 1             & 1             & 1             & 1             & 0             & 0             \\ \hline
Node/Time & $c_{s,1,1,l}$ & $c_{s,1,2,l}$ & $c_{s,1,3,l}$ & $c_{s,1,4,l}$ & $c_{s,2,1,l}$ & $c_{s,2,2,l}$ & $c_{s,2,3,l}$ & $c_{s,2,4,l}$ & $c_{s,3,1,l}$ & $c_{s,3,2,l}$ & $c_{s,3,3,b}$ & $c_{s,3,4,b}$ \\ \hline
0         & 1             & 1             & 0             & 0             & 0             & 0             & 0             & 0             & 0             & 0             & 1             & 1             \\ \hline
3         & 0             & 0             & 1             & 1             & 1             & 1             & 0             & 0             & 0             & 0             & 0             & 0             \\
\hline
\end{tabular}
}
\caption{Dynamics for context in AND/OR gadgets implemented on symmetric signed conjunctive networks. Notation is given by the following guidelines: $s$ represent the associated group of three nodes, second two coordinates indicate its position relative to the original gadget (there are $3$ clocks) and its position in the $4$-cycle graph (considering counter clock-wise order), and $u,l$ stands for upper or lower according to its position in the gadget.}
\label{tab:Clockdym-signed}
\end{table}

\begin{corollary}
	The family $\blocfamily{\mathcal{F}}{2}_{\text{sign-sym-conj}}$ of all signed symmetric conjunctive networks under block sequential update schemes with at most $3$ blocks is strongly universal. In particular, it is complex both dynamically and computationally complex.
	\label{cor:ANDNOTuniv}
\end{corollary}

\begin{proof}
	Strong universality holds from the fact that family is capable of simulating $\Gmontwo$ in linear space and constant time (See Corollary \ref{cor:univfrommon}). Family is dynamically and computationally complex as a direct consequence of strong universality (see Theorem \ref{them:univ-rich-dynamics} and Corollary \ref{cor:universality}).
\end{proof}

We are now in conditions to resume the proof of Theorem \ref{teo:localuniv}, we give again the statement of the theorem and we provide the complete proof.
\begin{center}
\begin{adjustbox}{minipage=0.85\textwidth,precode=\dbox}
There exist $c>0$ such that the family  $\clockfamily{\mathcal{F}}{c}{\text{locally-pos}}$  of all locally positive symmetric conjunctive networks under local clocks update scheme with clock parameter $c$ has coherent $\Gmontwo$-gadgets.
\end{adjustbox}
\end{center}
\begin{proof}
	We start by observing that the gadgets in Figures \ref{fig:ORstruct} and \ref{fig:ANDstruct} can be implemented in $\clockfamily{\mathcal{F}}{c}_{\text{locally-pos}}$   for some $c$. This can be easily done by adding a positive node to each node in the gadget. The main idea here is that each of these artificial positive nodes will play no role in calculations and will stay in state $1$ most of the time. In fact, it suffices that these positive neighbors reach state $1$ before critical steps of computation are performed inside the gadget.
	
	In order to illustrate this idea, let us consider two different cases and analyze why computation gadget still works in this case:
	
	\begin{enumerate}
		\item  \textbf{Nodes in $4$-cycles}: observe that these nodes have a fixed trajectory that is independent on the input that computation part is handling. Thus, it suffices to note that each node in the context effectively changes its state (they are in an attractor of period $3 \times 3$ as it is shown in Figure \ref{tab:Clockdym-signed}). As a consequence of this latter observation, we can set the local period of each positive neighbor so it is updated when its neighbor in the clock is in state $1$. More precisely, we fix the corresponding local period value to $9$ and correctly initialize them so each positive neighbor is updated exactly when their correspondent node is in state $1$.
		
		\item \textbf{Central nodes}: Observe that, in this case, we have that in the pseudo-orbit given in Table \ref{tab:ANDdym-signed} and Table \ref{tab:ORdym-signed} each node eventually reaches the state $1$ independently from the value of $x,y,z,x'$ and $y'$.  Thus, as same as the nodes that are in the $4$-cycles, it suffices to set up the local clock of each positive neighbor in order to be updated while its neighbor in the gadget is in state $1$. More precisely, it suffices to set up clocks following values in Table \ref{tab:ANDdym-signed} and Table \ref{tab:ORdym-signed} and set clocks to be updated every $18$ time-steps. Note that this work since nodes in the central part are in the first block so positive neighbors are updated at the same time as its neighbors but only when nodes in the gadget are in state $1$. 		\end{enumerate}
		Thus, gadgets  in Figures \ref{fig:ORstruct} and \ref{fig:ANDstruct} can be implemented as same as we did for general symmetric signed conjunctive networks and desired result holds.	

\end{proof}

\subsection{Symmetric min-max networks}

In this section, we study min-max networks. This is also a particular CSAN family in which local functions take the maximum or minimum value of some set of states. More precisely, a min-max network with ordered alphabet ${Q}$ is a CSAN $(G, \lambda,\rho)$ characterized by local functions $\lambda_v(q,S) = \min S$ or $\lambda_v(S) = \max S$ for every $v \in V(G)$ and trivial edge labels $\rho_e \in \{\text{Id}\}$ for each $ e \in E(G)$. Note that in the particular case in which $Q = \{0,1\}$ we have $\lambda_v(S) = \bigvee \limits_{x \in S } x$ or $\lambda_v(q,S) = \bigwedge \limits_{x \in S }s$ for every $v \in V(G).$ We call these particular CSAN family given by alphabet $Q = \{0,1\}$ the family of AND-OR networks and we write $\mathcal{F}_{\text{AND-OR}}$, $\blocfamily{\mathcal{F}}{b}_{\text{AND-OR}}$,  $\clockfamily{\mathcal{F}}{c}_{\text{AND-OR}}$ and $\periodfamily{\mathcal{F}}{p}_{\text{AND-OR}}$ to denote different update schemes as we did before. Analogously, we use the notation $\mathcal{F}_{\text{MIN-MAX}}$, $\blocfamily{\mathcal{F}}{b}_{\text{MIN-MAX}}$,  $\clockfamily{\mathcal{F}}{c}_{\text{MIN-MAX}}$ and $\periodfamily{\mathcal{F}}{p}_{\text{MIN-MAX}}$.

In the Boolean case, $\max$ and $\min$ functions are threshold functions because ${\min(x_1,\ldots,x_k)=1 \Leftrightarrow \sum x_i= k}$ and ${\max(x_1,\ldots,x_k)=0 \Leftrightarrow \sum x_i=0}$. Therefore their periodic orbits are of length at most $2$ by the results in \cite{PaperGoles}. In the non-Boolean case, they are not threshold functions. However, as shown by the following lemma the general alphabet case can be understood through multiple factorings onto the Boolean case.   
\begin{lemma}\label{lem:minmaxfactor}
	Let $n \geq 2$ and let $\mathcal{A} = (G,\lambda,\rho)$ be a min-max automata network with alphabet $Q$ such that $|V(G)| = n$. Let $F$ be a global rule for $\mathcal{A}$. There exists an AND-OR automata network $\mathcal{A}^* =   (G,\lambda,\rho)$ and a global rule $F^*: \{0,1\}^n \to \{0,1\}^n $  such that for every $\alpha \in Q$ the function $\pi^{\alpha}: Q^n \to \{0,1\}^n$ is such that $\pi^{\alpha} \circ F = F^{*} \circ \pi^{\alpha}$ where $\pi^{\alpha}(x)_i = \begin{cases}
	1 & \text{ if  } x_i \geq \alpha \\
	0 & \text{ otherwise.} 
	\end{cases}$
	%
\end{lemma}
\begin{proof}
  Let us take $F^*$ as global function given by the same min/max labels than $F$ but considering the fact that in the alphabet $\{0,1\}$ we have $\min(x,y) = x\wedge y$ and $\max (x,y) = x \vee y$. Fix $\alpha \in Q$ and let $x \in Q^n$. Additionally, let us fix $i \in V$. Suppose that $F(x)_i = \max(x_{i_1}, \hdots, x_{i_k})$ then $F^*(\pi^{\alpha}(x))_i = \bigvee\limits_{s = 1}^{k} \pi^{\alpha}(x)_{i_s}$ where $N(i) = \{i_1,\hdots i_k\}$. \sloppy Note that $$\pi^{\alpha}(F(x))_i = \begin{cases}
    1 & \textit{if } \max(x_{i_1}, \hdots, x_{i_k}) \geq \alpha, \\
    0 & \text{ otherwise}.
  \end{cases}$$ Then, it is clear that we have  $\pi^{\alpha}(F(x))_i = 1$ if and only if $\pi^{\alpha}(x)_{i_s} = 1$ for some $s \in \{1, \hdots, k\}$ and thus if and only if $\bigvee \limits_{s = 1}^{k}  \pi^{\alpha}(x)_{i_s} = 1$ which is equivalent to $F^{*} \circ \pi (x)_{i}=1$. The case in which $F(x)_i = \min(x_{i_1}, \hdots, x_{i_k})$ is analogous since we have that $F^*(\pi^{\alpha}(x))_i = \bigwedge \limits_{s = 1}^{k}\pi^{\alpha}(x)_{i_s}.$  \end{proof}

We deduce that periodic orbits in MIN-MAX networks are of length at most $2$ and therefore the family cannot be universal.

\begin{corollary}
  Let $F$ be any MIN-MAX network over alphabet $Q$ and $x$ any limit configuration (\textit{i.e.} such that ${F^t(x)=x}$ for some $t$). Then, ${F^2(x)=x}$. Therefore the family of MIN-MAX networks cannot be universal (considered under the parallel update scheme).
\end{corollary}
\begin{proof}
  We show that ${F^2(x)_i = x_i}$ for any node $i$. Let $q$ be the maximum state appearing in the sequence ${(F^t(x)_i)_{t\in\N}}$. Using Lemma~\ref{lem:minmaxfactor} with projection $\pi^q$ and the fact that periodic orbits of AND-OR networks have period at most $2$ we deduce that if ${F^t(x)_i=q}$ then ${F^{t+2}(x)_i=q}$ for some $t$. Suppose without loss of generality that $x_i=q$. If ${F(x)_i=q}$ we are done because then ${(F^t(x)_i)_{t\in\N}}$ is constant equal to $q$. Otherwise let $q'$ be the minimum state appearing in the sequence ${(F^t(x)_i)_{t\in\N}}$. Necessarily ${q'<q}$, and using again Lemma~\ref{lem:minmaxfactor} with projection $\pi^{q'}$ we deduce that if ${F^t(x)_i=q'}$ then ${F^{t+2}(x)_i=q'}$. With our assumption that ${x_i=q}$ it must be the case that ${F^{2k+1}(x)_i=q'}$ for some ${k\geq 0}$. From the previous facts and periodicity of the orbit of $x$ we deduce that ${F^t(x)_i}$ is $q$ when $t$ is even and $q'$ when $t$ is odd.    
  The claim about non-universality follows from Theorem~\ref{them:univ-rich-dynamics}.   
\end{proof}
\subsection{Block sequential update schemes}
We have shown that MIN-MAX networks are very limited under the parallel update mode. We now consider them under block sequential update schedules. We actually show that AND-OR networks under such update modes can simulate AND-NOT networks and therefore inherit the universality property. Since MIN-MAX networks on any alphabet $Q$ with $|Q|>1$  simulates Boolean AND-OR networks (by just restricting their alphabet to size $2$), we only focus on AND-OR networks.

\subsubsection{Strong universality}
As  we did for $\blocfamily{\mathcal{F}}{b}_{\text{sign-sym-conj}}$ we will show that $\blocfamily{\mathcal{F}}{b}_{\text{AND-OR}}$ is also strongly universal as a direct consequence of the fact that we can simulate   $\blocfamily{\mathcal{F}}{b}_{\text{sign-sym-conj}}$ in linear space and constant time. We accomplish this by using again the coding trick of ``double railed logic". In fact, given the interaction graph of a signed symmetric conjunctive network having $n$ nodes, we simply double each node  and thus, our simulator has $2n$ nodes. Simulation is made in real time so $T = 1$. We precise these ideas in the following lemma:

\begin{lemma}
	Let  $\periodfamily{\mathcal{F}}{p}_{\text{AND-OR}}$ be the family of AND-OR networks updated according to some arbitrary periodic update scheme of period $p$. Let $T$ be the constant function equal to $1$ and $S: \N \to \N$ be defined as $S(n) = 2n.$ Then, we have that $\periodfamily{\mathcal{F}}{p}_{\text{sign-sym-conj}} \preccurlyeq^S_T \periodfamily{\mathcal{F}}{p}_{\text{AND-OR}}.$ 
	\label{lemma:ANDORsim}
\end{lemma}
\begin{proof}
	Let us a fix a periodic update scheme of period $p \in \N.$ Let us take the graph representation of some arbitrary network  $(G, \lambda, \rho)$ in $\periodfamily{\mathcal{F}}{p}_{\text{sign-sym-conj}}.$ We construct a network in $\periodfamily{\mathcal{F}}{p}_{\text{AND-OR}}$ capable to simulate the latter network in the following way:  first, we represent each state $q \in \{0,1\}$ by $(q, \overline{q})$ where $\overline{q} = 1-q;$ then, for each $v \in V(G)$ we consider two nodes $v', \overline{v}.$ By doing this, we  store the original state of $v$ in $v'$ and $\overline{v}$  store its complement. More precisely, we replace each node in the network by the gadget in Figure \ref{fig:ANDORsim}. As it is shown in the same figure, we define $\lambda'(v') \equiv \text{AND}$ and $\lambda' (\overline{v}) \equiv \text{OR}.$ Let us define the set $V' = \{(v',\overline{v}):v\in V\}.$ We define a set of edges in $V'$ denoted $E'$ as follows: for each edge $(u,v)$ in $E(G)$ we add the following edges depending on $\rho((u,v))$:
	\begin{itemize}
		\item if $\rho((u,v)) = \text{Id},$ we add the edges $(u',v')$ and $(\overline{u},\overline{v}).$
		\item if $\rho((u,v)) = \text{Switch},$ we add the edges $(u',\overline{v})$ and $(\overline{u},v').$
	\end{itemize}
Note that this immediately defines a local rule for the simulator network that we call $f'_{w}$ for each $w \in V'$. We show that this local rule effectively simulates  $(G, \lambda, \rho).$ More precisely, let us call $f_u$  the local rule of $(G, \lambda, \rho)$ for each $u \in V(G)$. In addition, let us call $N(u)_{+}$ to the neighborhood of $u$ in $G$ such that any $v \in  N(u)_{+}$ satisfies $\rho((u,v)) = \text{Id}$ and also let us define  $N(u)_{-}$  such that $v \in N(u)_{-}$ if and only if $\rho((u,v)) = \text{Switch}.$ Finally, for each $x \in \{0,1\}^n$ let us call $x' \in  \{0,1\}^{2n} $  the configuration defined by $x'_{u'} =  x_u$ and $x'_{\overline{u}} = \overline{x_u},$ for each $u \in V(G)$. Fix $u \in V,$ we have for $(u',\overline{u})$ the following result:

\begin{itemize}
	\item \sloppy $f'_{u'}(x'|_{N(u')}) = (\bigwedge \limits_{v' \in N_{G'}(u'): v \in N(u)_{+}} x'_{v'}) \wedge (\bigwedge \limits_{\overline{v} \in N_{G'}(u'): v \in N(u)_{-}} x'_{\overline{v}}) = (\bigwedge \limits_{v \in N(u)_{+}} x_{v}) \wedge (\bigwedge \limits_{v \in N(u)_{-}} \overline{x}_v) = f_u(x) $
	\item \sloppy $f'_{\overline{u}}(x'|_{N(\overline{u})}) = (\bigvee \limits_{\overline{v} \in N_{G'}(\overline{u}): v \in N(u)_{+}} x'_{\overline{v}}) \vee (\bigvee \limits_{v' \in N_{G'}(\overline{u}): v \in N(u)_{-}} x'_{v}) = (\bigvee \limits_{v \in N(u)_{+}} \overline{x}_v) \vee (\bigvee \limits_{v \in N(u)_{-}} x_v) = \overline{(\bigwedge \limits_{v \in N(u)_{+}} x_{v}) \wedge (\bigwedge \limits_{v \in N(u)_{-}} \overline{x}_v)} = \overline{f_u(x)} $
\end{itemize}

And thus, we have that $x'_{u'} \to f_u(x)$ and $x'_{\overline{u}} \to \overline{f_u(x)}.$ 

Finally, in order to be coherent with the given periodic update scheme, it suffices to define an update scheme with period $p$ such that the nodes $(u', \overline{u})$ are updated at each time step at which $u \in V(G)$ is.  We conclude that the network  $(G', \lambda', \rho') \in \periodfamily{\mathcal{F}}{p}_{\text{AND-OR}}$ (where $\rho(u,v) = \text{Id}$ for each $(u,v) \in E$), simulates $(G, \lambda, \rho)$ in constant time $T=1$ and linear space $S(n) = 2n$.
\end{proof}
\begin{theorem}
	There exists some $b>0$ such that the family $\blocfamily{\mathcal{F}}{b}_{\text{AND-OR}}$ of all AND-OR networks under block sequential update scheme of block parameter $b$ is strongly universal. In particular, $\blocfamily{\mathcal{F}}{b}_{\text{AND-OR}}$ is dynamically and computationally complex.
\end{theorem}
\begin{proof}
	The result is a direct consequence of Lemma \ref{lemma:ANDORsim} and strong universality of $\blocfamily{\mathcal{F}}{b}_{\text{sign-sym-conj}}$ given by Corollary \ref{cor:ANDNOTuniv}.
\end{proof}

\begin{figure}
\centering

\tikzset{every picture/.style={line width=0.75pt}} 

\begin{tikzpicture}[x=0.75pt,y=0.75pt,yscale=-1,xscale=1]

\draw    (124.4,86.4) -- (234.8,86.33) ;
\draw  [fill={rgb, 255:red, 255; green, 255; blue, 255 }  ,fill opacity=1 ] (219.4,87.33) .. controls (219.4,78.83) and (226.29,71.93) .. (234.8,71.93) .. controls (243.31,71.93) and (250.2,78.83) .. (250.2,87.33) .. controls (250.2,95.84) and (243.31,102.73) .. (234.8,102.73) .. controls (226.29,102.73) and (219.4,95.84) .. (219.4,87.33) -- cycle ;
\draw  [fill={rgb, 255:red, 255; green, 255; blue, 255 }  ,fill opacity=1 ] (109,86.4) .. controls (109,77.89) and (115.89,71) .. (124.4,71) .. controls (132.91,71) and (139.8,77.89) .. (139.8,86.4) .. controls (139.8,94.91) and (132.91,101.8) .. (124.4,101.8) .. controls (115.89,101.8) and (109,94.91) .. (109,86.4) -- cycle ;

\draw    (409,48.47) -- (519.4,48.4) ;
\draw    (409,113.47) -- (519.4,113.4) ;
\draw  [fill={rgb, 255:red, 255; green, 255; blue, 255 }  ,fill opacity=1 ] (389,113.4) .. controls (389,104.89) and (395.89,98) .. (404.4,98) .. controls (412.91,98) and (419.8,104.89) .. (419.8,113.4) .. controls (419.8,121.91) and (412.91,128.8) .. (404.4,128.8) .. controls (395.89,128.8) and (389,121.91) .. (389,113.4) -- cycle ;

\draw  [fill={rgb, 255:red, 255; green, 255; blue, 255 }  ,fill opacity=1 ] (504,113.4) .. controls (504,104.89) and (510.89,98) .. (519.4,98) .. controls (527.91,98) and (534.8,104.89) .. (534.8,113.4) .. controls (534.8,121.91) and (527.91,128.8) .. (519.4,128.8) .. controls (510.89,128.8) and (504,121.91) .. (504,113.4) -- cycle ;

\draw  [fill={rgb, 255:red, 255; green, 255; blue, 255 }  ,fill opacity=1 ] (389,47.4) .. controls (389,38.89) and (395.89,32) .. (404.4,32) .. controls (412.91,32) and (419.8,38.89) .. (419.8,47.4) .. controls (419.8,55.91) and (412.91,62.8) .. (404.4,62.8) .. controls (395.89,62.8) and (389,55.91) .. (389,47.4) -- cycle ;

\draw  [fill={rgb, 255:red, 255; green, 255; blue, 255 }  ,fill opacity=1 ] (504,48.4) .. controls (504,39.89) and (510.89,33) .. (519.4,33) .. controls (527.91,33) and (534.8,39.89) .. (534.8,48.4) .. controls (534.8,56.91) and (527.91,63.8) .. (519.4,63.8) .. controls (510.89,63.8) and (504,56.91) .. (504,48.4) -- cycle ;

\draw   (290,69) -- (332,69) -- (332,59) -- (360,79) -- (332,99) -- (332,89) -- (290,89) -- cycle ;

\draw    (128.4,229.4) -- (238.8,229.33) ;
\draw  [fill={rgb, 255:red, 255; green, 255; blue, 255 }  ,fill opacity=1 ] (223.4,230.33) .. controls (223.4,221.83) and (230.29,214.93) .. (238.8,214.93) .. controls (247.31,214.93) and (254.2,221.83) .. (254.2,230.33) .. controls (254.2,238.84) and (247.31,245.73) .. (238.8,245.73) .. controls (230.29,245.73) and (223.4,238.84) .. (223.4,230.33) -- cycle ;
\draw  [fill={rgb, 255:red, 255; green, 255; blue, 255 }  ,fill opacity=1 ] (113,229.4) .. controls (113,220.89) and (119.89,214) .. (128.4,214) .. controls (136.91,214) and (143.8,220.89) .. (143.8,229.4) .. controls (143.8,237.91) and (136.91,244.8) .. (128.4,244.8) .. controls (119.89,244.8) and (113,237.91) .. (113,229.4) -- cycle ;
\draw    (413,191.47) -- (523.4,256.4) ;
\draw    (413,256.47) -- (523.4,191.4) ;
\draw  [fill={rgb, 255:red, 255; green, 255; blue, 255 }  ,fill opacity=1 ] (393,256.4) .. controls (393,247.89) and (399.89,241) .. (408.4,241) .. controls (416.91,241) and (423.8,247.89) .. (423.8,256.4) .. controls (423.8,264.91) and (416.91,271.8) .. (408.4,271.8) .. controls (399.89,271.8) and (393,264.91) .. (393,256.4) -- cycle ;

\draw  [fill={rgb, 255:red, 255; green, 255; blue, 255 }  ,fill opacity=1 ] (508,256.4) .. controls (508,247.89) and (514.89,241) .. (523.4,241) .. controls (531.91,241) and (538.8,247.89) .. (538.8,256.4) .. controls (538.8,264.91) and (531.91,271.8) .. (523.4,271.8) .. controls (514.89,271.8) and (508,264.91) .. (508,256.4) -- cycle ;

\draw  [fill={rgb, 255:red, 255; green, 255; blue, 255 }  ,fill opacity=1 ] (393,190.4) .. controls (393,181.89) and (399.89,175) .. (408.4,175) .. controls (416.91,175) and (423.8,181.89) .. (423.8,190.4) .. controls (423.8,198.91) and (416.91,205.8) .. (408.4,205.8) .. controls (399.89,205.8) and (393,198.91) .. (393,190.4) -- cycle ;

\draw  [fill={rgb, 255:red, 255; green, 255; blue, 255 }  ,fill opacity=1 ] (508,191.4) .. controls (508,182.89) and (514.89,176) .. (523.4,176) .. controls (531.91,176) and (538.8,182.89) .. (538.8,191.4) .. controls (538.8,199.91) and (531.91,206.8) .. (523.4,206.8) .. controls (514.89,206.8) and (508,199.91) .. (508,191.4) -- cycle ;

\draw   (294,212) -- (336,212) -- (336,202) -- (364,222) -- (336,242) -- (336,232) -- (294,232) -- cycle ;

\draw (394+3,104.4+6) node [anchor=north west][inner sep=0.75pt]    {$\lor $};
\draw (394+3,38.4+6) node [anchor=north west][inner sep=0.75pt]    {$\land $};
\draw (509+3,39.4+6) node [anchor=north west][inner sep=0.75pt]    {$\land $};
\draw (509+3,104.4+6) node [anchor=north west][inner sep=0.75pt]    {$\lor $};
\draw (115,45.4+10) node [anchor=north west][inner sep=0.75pt]    {$u$};
\draw (229,45.4+10) node [anchor=north west][inner sep=0.75pt]    {$v$};
\draw (173,56.4+10) node [anchor=north west][inner sep=0.75pt]    {$\text{Id}$};
\draw (395,3.4+10) node [anchor=north west][inner sep=0.75pt]    {$u'$};
\draw (397,72.4+10) node [anchor=north west][inner sep=0.75pt]    {$\overline{u}$};
\draw (515,4.4+10) node [anchor=north west][inner sep=0.75pt]    {$v'$};
\draw (512,72.4+10) node [anchor=north west][inner sep=0.75pt]    {$\overline{v}$};
\draw (513+3,182.4+6) node [anchor=north west][inner sep=0.75pt]    {$\land $};
\draw (398+3,181.4+6) node [anchor=north west][inner sep=0.75pt]    {$\land $};
\draw (513+3,247.4+6) node [anchor=north west][inner sep=0.75pt]    {$\lor $};
\draw (398+3,247.4+6) node [anchor=north west][inner sep=0.75pt]    {$\lor $};
\draw (516,215.4+10) node [anchor=north west][inner sep=0.75pt]    {$\overline{v}$};
\draw (519,147.4+10) node [anchor=north west][inner sep=0.75pt]    {$v'$};
\draw (401,215.4+10) node [anchor=north west][inner sep=0.75pt]    {$\overline{u}$};
\draw (399,146.4+10) node [anchor=north west][inner sep=0.75pt]    {$u'$};
\draw (161,199.4+10) node [anchor=north west][inner sep=0.75pt]    {$\text{Switch}$};
\draw (233,188.4+10) node [anchor=north west][inner sep=0.75pt]    {$v$};
\draw (119,188.4+10) node [anchor=north west][inner sep=0.75pt]    {$u$};

\end{tikzpicture}
\caption{Gadget used for simulation of an AND-NOT network with arbitrary periodic update scheme implemented over an AND-OR network with periodic update scheme.}
\label{fig:ANDORsim}
\end{figure}
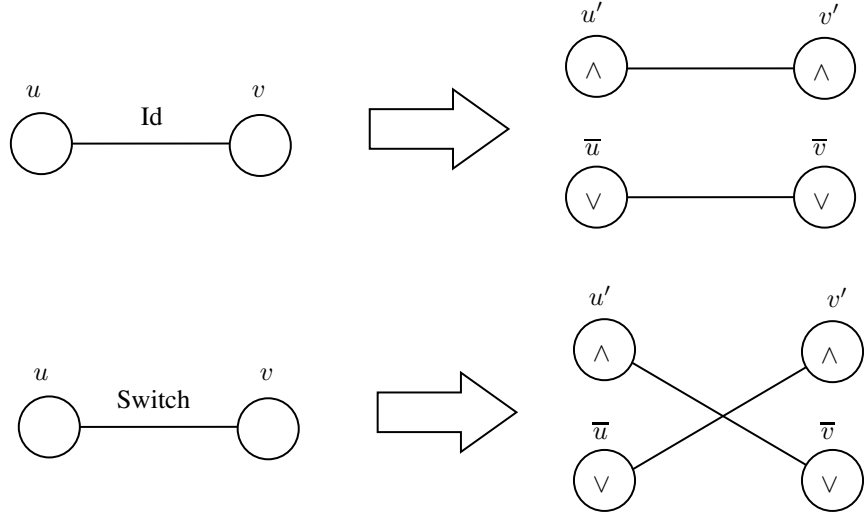

\section{Perspectives}
\label{sec:perspectives}

The main results of this paper extend previous results \cite{goles2020firing,goles2016pspace,goles2014computational} on the strong interplay between local interactions and update modes, but are still limited to finite hierarchies of local interactions and update modes. A natural research direction is to try to further extend this analysis to other families of networks and update modes. For instance, we have no example of a CSAN family which is (strongly) universal under periodic update schemes, but not universal under local clocks update schemes.

Besides, our results were obtained as an application of various concepts and tools, and we believe that several research directions around them are worth being considered. We detail some of them below.

\paragraph{Glueing} We think it would be interesting to understand the properties of the glueing process itself and see what information on the result of the glueing process can be deduced from the knowledge of each network to be glued. We are particularly interested in dynamical properties. In addition, it would be very interesting to explore if latter process can be seen in the opposite way, i.e., given an automata network, determine if it is possible to decompose the network into glued blocks satisfying some particular properties as gadgets do.

\paragraph{Simulations and universality}
Two notions of universality are introduced in this paper. We see how to build families which are universal but not strongly universal by adding a somewhat artificial mechanism that slows down polynomially any useful computation made by networks in the family. However, we don't have any natural example so far. In the same spirit, we can ask how a strongly universal family can fail to have coherent $\Gmon$-gadgets (recall that Corollary~\ref{cor:univfrommon} only gives a sufficient condition to be strongly universal). We don't think that strongly universality implies coherent $\Gmon$-gadgets in all generality, but it might be true under some additional hypothesis, and possibly in natural families like $\G$-networks.

\paragraph{Families, $\G$-networks and gadgets}
Proposition \ref{prop:closuregadgets} together with theorems~\ref{theo:non-polyn-cycl} and \ref{theo:transient-nonuniversal} provide an interesting starting point to explore the link between different gate sets and the richness of their synchronous closure and the associated family of $\G$-networks. It is natural to further study the hierarchy between sets of gates and we believe that a promising direction would be to study reversible gate sets such as Toffoli or Fredkin gates.
 
\paragraph{Update schedules} As we have pointed out in the update schemes section (see Remark \ref{rem:blockpar}), it would be interesting to study other types of update schemes which are not captured by general periodic update schemes such as firing memory schemes \cite{goles2020firing,glrs20}, or other modes where the decision to apply an update at a node depends on its state (for instance we can interpret reaction-difussion systems \cite{goles1997reaction} as threshold networks with an update mode with memory). We can consider even more general ones inspired by the most permissive semantics studied in \cite{chatain2018most}. In some cases, Definition~\ref{def:asyncextension} of asynchronous extension should be adapted in order to capture latter update schemes.


%
\bibliographystyle{plain}

\bibliography{paper.bib}

\end{document}